\documentclass[12pt,openright,twoside]{report}
\usepackage{amsmath,amsfonts,amssymb}
\usepackage{mathrsfs}
\usepackage{mathtools}
\usepackage{caption}
\usepackage[shortlabels,inline]{enumitem}

\usepackage{tikz}
\usepackage{tikz-cd}
\usepackage{float}
\usepackage{graphicx}
\graphicspath{ {figures/} }

\usepackage[a4paper,inner=4cm,outer=2.5cm]{geometry}

\usepackage[numbers,sort]{natbib} 	

\usepackage{appendix}

\usepackage{aliascnt} 
\usepackage[pdfusetitle,
bookmarks=true,
pagebackref,
colorlinks,
linkcolor=blue,
citecolor=blue,
urlcolor=blue]{hyperref}

\usepackage{amsthm}

\theoremstyle{plain}
\newaliascnt{theorem}{dummy}
\newtheorem{theorem}[theorem]{Theorem}
\aliascntresetthe{theorem}

\newaliascnt{proposition}{dummy}
\newtheorem{proposition}[proposition]{Proposition}
\aliascntresetthe{proposition}

\newaliascnt{corollary}{dummy}
\newtheorem{corollary}[corollary]{Corollary}
\aliascntresetthe{corollary}

\newaliascnt{lemma}{dummy}
\newtheorem{lemma}[lemma]{Lemma}
\aliascntresetthe{lemma}

\newaliascnt{conjecture}{dummy}

\aliascntresetthe{conjecture}

\theoremstyle{definition}
\newaliascnt{definition}{dummy}
\newtheorem{definition}[definition]{Definition}
\aliascntresetthe{definition}

\newaliascnt{example}{dummy}

\aliascntresetthe{example}

\theoremstyle{remark}
\newaliascnt{remark}{dummy}
\newtheorem{remark}[remark]{Remark}
\aliascntresetthe{remark}

\allowdisplaybreaks

\makeatletter
\newcommand{\titlezh}[1]{\gdef\@titlezh{#1}}%
\newcommand{\@titlezh}{\@latex@warning@no@line{No \noexpand\titlezh given}}
\newcommand{\degree}[1]{\gdef\@degree{#1}}%
\newcommand{\@degree}{\@latex@warning@no@line{No \noexpand\degree given}}
\newcommand{\programme}[1]{\gdef\@programme{#1}}%
\newcommand{\@programme}{\@latex@warning@no@line{No \noexpand\programme given}}
\newcommand{\institution}[1]{\gdef\@institution{#1}}%
\newcommand{\@institution}{\@latex@warning@no@line{No \noexpand\institution given}}
\newcommand{\committee}[1]{\gdef\@committee{#1}}%
\newcommand{\@committee}{\@latex@warning@no@line{No \noexpand\committee given}}

\renewcommand{\maketitle}{%
	\thispagestyle{empty}
	\pdfbookmark{Titlepage}{titlepage}
	\vspace*{5mm}
	\begin{center}
	\Large{ \bf \@title}
	\end{center}
	\vspace*{20mm}
	
	\begin{center}
	\@author
	\end{center}
	\vspace*{20mm}

	\begin{center}
	A Thesis Submitted in Partial Fulfilment \\
	of the Requirements for the Degree of \\
	\@degree \\
	in \\
	\@programme
	\end{center}
	\vspace*{20mm}

	\begin{center}
	\@institution \\
	\@date
	\end{center}

}

\newcommand{\addloftotoc}{%
	\cleardoublepage
	\phantomsection
	\addcontentsline{toc}{chapter}{\listfigurename}
}

\newcommand{\addbibtotoc}{%
	\phantomsection
	\addcontentsline{toc}{chapter}{\bibname}
}

\newcommand{\addtocbm}{%
	\cleardoublepage
	\pdfbookmark{\contentsname}{toc}
}

\makeatother

\usepackage{amsmath}
\usepackage{amssymb}
\usepackage{mathtools}
\usepackage{pgfplots}
\usepackage{pgfplotstable}	
\usepackage{algorithm,algpseudocode}
\usepackage{float}

\title{One-Shot Coding and Applications}
\author{LIU, Yanxiao}
\degree{Doctor of Philosophy}
\programme{Information Engineering}
\institution{The Chinese University of Hong Kong}
\date{July 2025}

\begin{document}
\maketitle 

\pagenumbering{roman}
\setcounter{page}{0}

	\chapter*{Abstract}
	\addcontentsline{toc}{chapter}{Abstract}

\emph{One-shot information theory} addresses scenarios in source coding and
channel coding where the signal blocklength is assumed to be $1$. 
In this case, each source and channel can be used only once, and the sources and channels are arbitrary and not required to be memoryless or ergodic. 
We study the \emph{achievability} part of one-shot information theory, i.e., we consider explicit coding schemes in the one-shot scenario.
The objective is to derive one-shot achievability results that can imply existing (first-order and second-order) asymptotic results when applied to memoryless sources and channels, or applied to systems with memory that behave ergodically.

Poisson functional representation was first proposed as a one-shot channel simulation technique by Li and El Gamal~\cite{li2018strong} for proving a strong functional representation lemma. It was later extended to the Poisson matching lemma by Li and Anantharam~\cite{li2021unified}, which provided a unified one-shot coding scheme for a broad class of information-theoretic problems. The main contribution of this thesis is to extend the applicability of Poisson functional representation to various more complicated scenarios, where the original version cannot be applied directly and further extensions must be developed. Below, we highlight some of the key contributions.

\begin{enumerate}
    \item In Chapter~\ref{chp:nnc}, we design a unified one-shot coding framework for the communication and compression of messages among multiple nodes across a \emph{general} acyclic noisy network. This framework can be viewed as a one-shot counterpart to the unified random coding bound studied by Lee and Chung~\cite{lee2018unified}, as well as the noisy network coding developed by Lim \emph{et al.}~\cite{lim2011noisy}. Our general framework not only recovers a wide range of existing one-shot and asymptotic results but also provides novel one-shot achievability results for various network information theory problems.

    \item In Chapter~\ref{chp:hiding}, we examine two classes of secrecy problems where the channel conditions are \emph{unknown} to the encoder and the decoder, based on the Poisson matching lemma and a covering argument. 
    We provide one-shot achievability results for a generalized information hiding setting~\cite{moulin2003information} and the compound wiretap channel~\cite{liang2009compound}, each of which recovers many existing problems as special cases.

    \item In Chapter~\ref{chp:ppr}, leveraging the Poisson functional representation, we design a novel construction called \emph{Poisson private representation} that can compress arbitrary differential privacy mechanisms. 
    It is the first scheme that achieves a close-to-optimal compression size (within a logarithmic gap), exactness of the output distribution (thus preserving all the desirable statistical properties of the original privacy mechanism, such as unbiasedness and Gaussianity), while ensuring local differential privacy. 
New trade-offs among communication, accuracy, and central and local differential privacy are established, and experimental advantages are demonstrated across different applications.
\end{enumerate}

	\clearpage

%
	\chapter*{Acknowledgement}
	\addcontentsline{toc}{chapter}{Acknowledgement}

This dissertation would not have been possible without the support of my advisors, mentors, collaborators, colleagues, friends, family, and many other individuals. 
I am deeply indebted to all of them for their kind assistance throughout this pleasant journey.

First and foremost, I would like to express my deepest gratitude to my advisors, Prof.\ Cheuk Ting Li and Prof.\ Raymond W.\ Yeung, for their support throughout my time at CUHK. 
I feel honored to hold the distinguished title of being one of the first two PhD students of Prof.\ Li, which afforded me privileges that very few PhD students enjoy: whenever I came up with vague ideas or felt confused, Prof.\ Li always made time to meet with me with great patience. 
Most of the time, these ideas would later prove to be useless, especially at the beginning of my PhD journey, but Prof.\ Li was always patient enough to listen, take them seriously, and think through them with me. 
Such support is crucial for a young PhD student. 
Along with his guidance on proving theorems, writing academic papers, preparing presentations, and allowing me to explore freely in various areas of information theory, I wish to express my deepest gratitude to him. 
If I have the opportunity to advise students in the future, I will remember the patience and generosity I received from Prof.\ Li, and I hope I can inspire my students as he inspired me. 
I also feel privileged to be co-advised by Prof.\ Yeung, one of the giants in information theory. 
As I gained more experience, I learned more and more from Prof.\ Yeung about how to embrace unexpected results in research, discover the extraordinary in the ordinary, and distill complicated findings into clear insights. 
His two courses laid the foundation of my knowledge in information theory. 
I sincerely thank both of my advisors for their outstanding guidance and support.

Throughout my PhD journey, I had the privilege of visiting Stanford University for half a year. 
I am grateful to Prof.\ Li for supporting my visit and to Prof.\ Ayfer \"{O}zg\"{u}r for hosting me. 
This experience was crucial for my academic development, as I attended several workshops in California, finished an interesting project with Prof.\ Li and Prof. \"{O}zg\"{u}r, and made many friends in the Bay Area. 
I am also thankful to Prof.\ \"{O}zg\"{u}r for her guidance not only on the project we worked on but also on research methodology in general. 
She was so kind and supportive that I felt re-energized after every meeting with her.

I would like to extend my heartfelt thanks to Prof.\ Chandra Nair. Though not officially my advisor, he was like one to me. 
From him, I learned not only research skills but also the attitude a great scholar should possess. 
The first time I felt capable of working something out independently was during a project in his course. 
No matter how naïve my questions were, he was always willing to answer them and support my ideas. 
I have learned a lot from his persistence, creativity, discipline, and life philosophy.

I would also like to thank Prof.\ Amin Gohari for his guidance on how to explain sophisticated ideas clearly and how to ask good questions to elaborate on concepts. 
I am grateful to Prof.\ Pascal O.\ Vontobel for his excellent course on coding theory and his inspiring papers, from which I learned how to explain complex ideas with detailed and illustrative examples. 
It has been a privilege to study information theory in a department with so many experts in the field, a rarity anywhere in the world. 

Moreover, I would also like to thank Prof.\ Vincent Y. F. Tan and Prof.\ Angela Yingjun Zhang for serving on my thesis committee. Their comments and suggestions are valuable to this thesis.

My PhD journey would not have been possible without the help of Prof.\ Shenghao Yang, who guided me for more than three years when I was an undergraduate and introduced me to the field of information theory. 
He taught me how to conduct good research, solve difficult problems step by step, and write first-class academic papers, even before I began my PhD studies, and he continued to support me throughout the journey. 
This journey would also not have been possible without Prof.\ Ximing Fu, who collaborated with me on a project closely and taught me many things.

I would also like to extend my thanks to professors who offered me valuable advice when I was an undergraduate deciding to pursue a PhD at CUHK, including Prof.\ Jianwei Huang, who provided extremely helpful advice on how to plan my PhD studies wisely; Prof.\ Gongqiu Zhang, who hosted my final-year project and advised me on PhD opportunities; Prof.\ Qin Wang, who guided me through multiple courses and provided me invaluable advices; Prof.\ Kenneth Shum, who encouraged me to pursue information theory despite its challenges; and Prof.\ Christopher Kluz and Dr.\ Lucas Scripter.

I would like to thank my other coauthors, each of whom taught me a lot during our collaborations, including Prof.\ Yi Chen, Yijun Fan, Chih Wei Ling, Wei-Ning Chen, and Sepehr Heidari Advary. 
Special thanks to Chih Wei and Yijun, who supported me many times when I was frustrated.

Except my coauthors, my life at CUHK would have been much more tedious without my other friends from the department: Jianguo Zhao, Xiang Li, Chin Wa (Ken) Lau, Jinpei Zhao, Zijie Chen, Zhaobang Zhu, Yi Liu, Jiaxin Qing, Xiaohong Cai, Yulin Chen, Yuwen Huang, Junda Zhou, Yicheng Cui, Binghong Wu, Zhenduo Wen, Chenyu Wang, and Chun Hei (Michael) Shiu. 
Many thanks also to the friends I met at Stanford: Dan Song and Andy Dong.

I am deeply appreciative of Prof.\ Chak Wong, who was the most important mentor to me outside of research. 
He taught me how to read, how to think, how to live a happy life, and how to overcome difficulties. 
Whenever I feel lost, I think about what Chak shared with me and the books we read together, and I regain my passion and confidence. 
I cannot be more grateful.

I have been fortunate to have great friends who took care of me when I was unwell: Qingyan Chen, Xiaoyu Yue, Taolin Liu, and Yue Qin. 
I would also like to thank my great friends with whom I can discuss not only research but also personal life: Chengchang Liu, Guodong Li, Licheng Mao, Jie Wang, and Yanyan Dong. 
My deepest thanks also go to Xiao Tan for her encouragement.

Last but not least, I would like to thank Dr.\ Raymond Tung for his treatment, along with the other doctors, nurses, and helpers from Prince of Wales Hospital whose names I do not know. 
Without their careful treatment and support, I would not have been able to write this thesis. 
There are no words that can fully express my gratitude to them.

Lastly, I dedicate this thesis to my parents, Jiayong Liu and Baorong Ma, whose unconditional love and support have been my constant motivation. 
I am forever indebted to them.

	\clearpage

%
	\chapter*{List of Publications}
	\addcontentsline{toc}{chapter}{List of Publications}
	During the PhD studies, Yanxiao Liu has published the following works: 

\begin{enumerate}
    \item Liu, Yanxiao, Sepehr Heidari Advary, and Cheuk Ting Li. ``Nonasymptotic Oblivious Relaying and Variable-Length Noisy Lossy Source Coding.'' 2025 IEEE International Symposium on Information Theory (ISIT). 
    
    \item Liu, Yanxiao, Wei-Ning Chen, Ayfer \"{O}zg\"{u}r, and Cheuk Ting Li. ``Universal exact compression of differentially private mechanisms.'' Advances in Neural Information Processing Systems 37 (2024): 91492-91531.
    
    \item Liu, Yanxiao, and Cheuk Ting Li. ``One-Shot Information Hiding.'' 2024 IEEE Information Theory Workshop (ITW), pp. 169-174. IEEE, 2024. (c) 2024 IEEE. Reprinted, with permission. 
    
    \item Liu, Yanxiao, and Cheuk Ting Li. ``One-shot coding over general noisy networks.'' 2024 IEEE International Symposium on Information Theory (ISIT), pp. 3124-3129. IEEE, 2024. 
    Full version was submitted to IEEE Transactions on Information Theory. (c) 2024 IEEE. Reprinted, with permission. 
    
    \item Ling, Chih Wei, Yanxiao Liu, and Cheuk Ting Li. ``Weighted parity-check codes for channels with state and asymmetric channels.'' IEEE Transactions on Information Theory (2024). Short version was presented at 2023 IEEE International Symposium on Information Theory (ISIT), pp. 3103-3108. IEEE, 2022.
    
    \item Yang, Shenghao, Jun Ma, and Yanxiao Liu. ``Wireless network scheduling with discrete propagation delays: Theorems and algorithms.'' IEEE Transactions on Information Theory 70, no. 3 (2023): 1852-1875. 
    Short version was presented at IEEE INFOCOM 2021-IEEE Conference on Computer Communications, pp. 1-10.
    
    \item Fan, Yijun, Yanxiao Liu, Yi Chen, Shenghao Yang, and Raymond W. Yeung. ``Reliable throughput of generalized collision channel without synchronization.'' 2023 IEEE International Symposium on Information Theory (ISIT), pp. 2697-2702. IEEE, 2023.
    
    \item Fan, Yijun, Yanxiao Liu, and Shenghao Yang. ``Continuity of link scheduling rate region for wireless networks with propagation delays.'' 2022 IEEE International Symposium on Information Theory (ISIT), pp. 730-735. IEEE, 2022.
\end{enumerate}

This thesis covers the second, third and fourth publications. 
	\clearpage

\hypersetup{linkcolor=black}
\addtocbm
\tableofcontents 

\addloftotoc
\listoffigures 

\hypersetup{linkcolor=blue}
\clearpage
\pagenumbering{arabic}
\setcounter{page}{1}


\chapter{Introduction} \label{chp:intro}

\section{Background of One-Shot Information Theory}
\label{sec::bkgrd_1shot}

In information theory, which originated from Shannon~\cite{shannon1948mathematical}, the goal is to determine optimal and reliable transmission rates over channels, or optimal compression rates for sources. 
Conventional information-theoretic analyses often rely on the asymptotic equipartition property, typicality-based proofs, and the law of large numbers to characterize the behavior of channels and sources in the asymptotic regime~\cite{el2011network}.

\begin{figure}[htpb]
	\centering
	\includegraphics[scale=0.37]{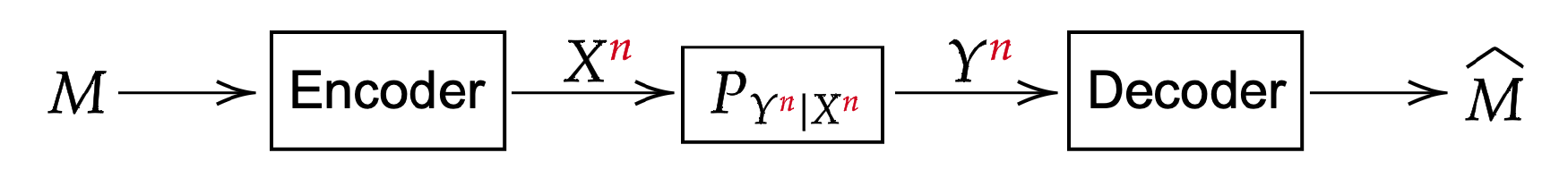}
	\caption{Channel coding setting in the large blocklength limit.} 
	\label{fig::channel_asymp}
\end{figure}

Take channel coding as an example. Figure~\ref{fig::channel_asymp} illustrates the conventional channel coding setting: a message $M$ of length $k$ is encoded to an input sequence $X^n = (X_1,\ldots, X_n)$; a decoder observes $Y^n = (Y_1,\ldots, Y_n)$ through a discrete memoryless channel, and outputs $\hat{M}$. 
Shannon's channel coding theorem~\cite{shannon1948mathematical} states that when the blocklength $n$ is large, i.e.,  $n\rightarrow \infty$, the channel capacity, which is defined as the maximum communication rate $k/n$ in bits per channel transmissions such that $\mathbf{P}(M \neq \hat{M})$ can be made
arbitrarily small~\cite{shannon1948mathematical}, is given by 
\begin{equation}
    C = \max_{P_X}\, I(X; Y),
    \label{eq::channel_capacity}
\end{equation}
where $I(X;Y)$ is the mutual information between $X$ and $Y$.

However, a critical practical issue is that packet lengths are never infinite and can, in fact, be very short in real-world applications—for example, in ultra-reliable low-latency communications~\cite{durisi2016toward}. 
Motivated by this, \emph{finite blocklength information theory} has been extensively studied over the past decade. 
The goal is to provide nonasymptotic guarantees in scenarios where the number of channel uses is limited~\cite{kostina2013lossy, wang2011dispersion, tan2013dispersions, polyanskiy2010channel}. That is, in Figure~\ref{fig::channel_asymp}, when $n$ is finite, what is the guarantee on the error probability $\mathbf{P}(M \neq \hat{M})$?

An even more general scenario is the \emph{one-shot} setting~\cite{feinstein1954new, shannon1957certain, hayashi2009information, polyanskiy2010channel, verdu2012non, yassaee2013technique, liu2015one, song2016likelihood, watanabe2015nonasymptotic, li2021unified}, where the blocklength is assumed to be $1$. 
That is, each source and channel can be used only \emph{once}.
Note that ``one-shot'' does not mean transmitting only $1$ bit. Instead, it represents the most general case, where the sources and channels can be arbitrary.
This line of research is primarily motivated by the generality of the setting: no assumptions are imposed on the sources or channels (e.g., memorylessness, ergodicity, etc.). 
The difficulty is that well-known techniques such as joint typicality and time sharing are not applicable. This setting is more general than some finite-blocklength cases; for instance, the finite-blocklength bounds in~\cite{yassaee2013non} do not seem to yield one-shot results due to their use of the method of types.

\begin{figure}[htpb]
	\centering
	\includegraphics[scale=0.37]{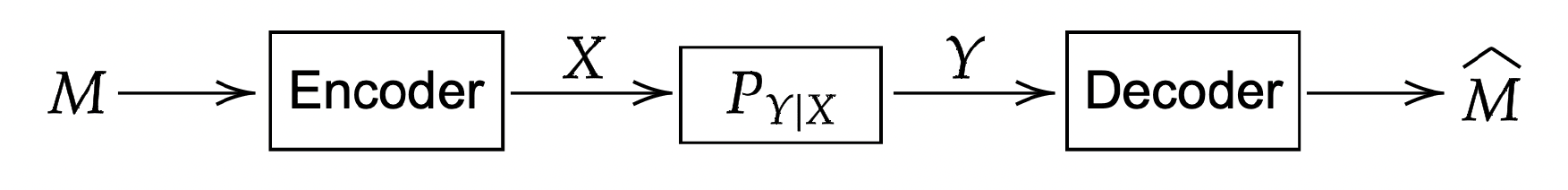}
	\caption{Channel coding setting in the one-shot regime.} 
	\label{fig::channel_oneshot}
\end{figure}

We use the one-shot channel coding setting as an example in Figure~\ref{fig::channel_oneshot}. 
Upon observing a message $M\sim \mathrm{Unif}[1:\mathsf{L}]$, the encoder produces $X$ that is sent through the channel $P_{Y|X}$. The decoder observes $Y$ and recovers $\hat{M}$ with error probability $P_e := \mathbf{P}\{M\neq\hat{M} \}$. 
Take the dependence testing bound by Polyanskiy, Poor and Verdú~\cite{polyanskiy2010channel} as an example, we have: 
\begin{equation}
    P_e\leq \mathbf{E}\left[\min\left\{
    \frac{\mathsf{L}-1}{2}\cdot 2^{-\iota_{X;Y}(X;Y)}
    , 1\right\}\right],\label{eq::1shot_eg_channel1}
\end{equation}
where $\iota_{X;Y}(X;Y) := \log \left(\frac{\mathrm{d} P_{X|Y}(x |y)}{\mathrm{d} P_X(x)}\right)$ is the information density and $\frac{\mathrm{d} P_{X|Y}(x |y)}{\mathrm{d} P_X(x)} =\frac{\mathrm{d} P_{X|Y}(\cdot |y)}{\mathrm{d} P_X}(x)$ denotes the Radon-Nikodym derivative.

We show how the one-shot result~\eqref{eq::1shot_eg_channel1} recovers the asymptotic channel capacity~\eqref{eq::channel_capacity}. 
Consider $\mathsf{L} = 2^{nR}$, due to the channel being discrete memoryless $P_{Y^n|X^n}(y^n|x^n) = \prod^n_{i=1} P_{Y|X}(y_i| x_i)$ in the absence of feedback, applying the one-shot result~\eqref{eq::1shot_eg_channel1} gives 
\begin{equation}
    P_e\leq \mathbf{E}\left[\min\{(2^{nR} - 1)\cdot 2^{-1-\sum^n_{i=1} \iota_{X;Y}(X_i; Y_i)}, 1\}\right]\label{eq::1shot_eg_channel2}
\end{equation}
where $(X_i; Y_i)\sim P_XP_{Y|X}$ i.i.d. for $i=1,\ldots,n$. 

When $n\rightarrow \infty$, by the law of large numbers we know $\sum^n_{i=1} \iota_{X;Y}(X_i; Y_i)\approx nI(X;Y)$, and therefore by~\eqref{eq::1shot_eg_channel2} we know $P_e\rightarrow 0$ if $R< I(X;Y)$, which recovers the channel capacity $C$ in~\eqref{eq::channel_capacity}.

One-shot settings are general, and we expect good one-shot achievability results can imply existing (first-order and second-order) asymptotic results when applied to memoryless sources and channels as above presents, or applied to systems with memory that behave ergodically~\cite{verdu1994general}.
For point-to-point channel coding, the achievability of the channel capacity is implied by the one-shot bounds by Feinstein~\cite{feinstein1954new} and Shannon~\cite{shannon1957certain}, which are precursors of the dependence testing bound~\cite{polyanskiy2010channel} in~\eqref{eq::1shot_eg_channel1}.

For settings more complex than the point-to-point channel, one-shot coding schemes have also been studied.
We briefly review existing one-shot results for multi-user coding settings, and this part also appeared in~\cite{liu2024one}.
In~\cite{verdu2012non}, 
one-shot versions of the covering and packing lemmas have been proposed and applied to various problems in multiuser information theory, for example, multiple access channels and broadcast channels. 
In~\cite{yassaee2013technique}, a proof technique based on stochastic likelihood encoders and decoders has been used to derive various one-shot achievability results in several multi-user settings, including broadcast channels, multiterminal source coding and multiple description coding. 
A one-shot mutual covering lemma has been proposed in~\cite{liu2015one} for broadcast channels, which recovers Marton's inner bound. 
In~\cite{song2016likelihood}, the multiterminal source coding inner bound has been examined by a likelihood encoder. 
A finite-blocklength version of the random binning technique has been used in~\cite{yassaee2013non} to derive second order regions for broadcast channels. 
Recently, in~\cite{li2021unified}, a technique called the \emph{Poisson matching lemma} has been introduced to prove various one-shot achievability results for a range of coding settings, and the achievable one-shot bounds improve upon the best known one-shot bounds in several settings with shorter proofs. This technique has been applied to unequal message protection~\cite{khisti2024unequal}, hypothesis testing~\cite{guo2024hypothesis} and secret key generation~\cite{hentila2024communication}.
The Poisson matching lemma is based on the \emph{Poisson functional representation}~\cite{li2018strong}, which has been applied to various fields recently, e.g., neural estimation~\cite{lei2022neural}, minimax learning~\cite{li2020minimax} and differential privacy~\cite{liu2024universal}, together with other related techniques. 
We will provide more details on the Poisson functional representation~\cite{li2018strong} in Chapter~\ref{chp:existing_techniques}.

\section{Background of Differential Privacy}

In this section we review the background of differential privacy, part of which also appeared in~\cite{liu2024universal}.

In modern data science and wireless communications, there is a growing dependence on large amounts of high-quality data, often generated by edge devices (e.g., photos and videos captured by smartphones, or messages hosted by social networks). However, this data inherently contains personal information, making it susceptible to privacy breaches during acquisition, collection, or utilization. For instance, despite the significant recent advancement in foundational models \citep{bommasani2021opportunities}, studies have shown that these models can accidentally memorize their training data. This poses a risk where malicious users, even with just API access, can extract substantial portions of sensitive information \citep{carlini2021extracting, carlini2023extracting}.

In recent years, differential privacy (DP) \citep{dwork2006calibrating} has emerged as a powerful framework for safeguarding users' privacy by ensuring that local data is properly randomized before leaving users' devices. 
With the local data $X$, a DP mechanism $\mathcal{A}$ satisfies $(\epsilon, \delta)$-DP maps $X$ to the output $Z = \mathcal{A}(X)\in\mathcal{Z}$, where $\mathcal{A}(\cdot)$ is a randomized function, such that for any neighboring $(x,x')$ and $\mathcal{S}\subseteq \mathcal{Z}$, \begin{equation}
    \Pr ( Z \in \mathcal{S}\, |\, X=x ) \leq  e^\varepsilon \Pr ( Z \in \mathcal{S}\, |\, X=x' ) + \delta,
    \label{eq::DP}
\end{equation}
where neighboring $(x,x')$ are neighboring if they differ in a single data point.
This definition can be understood as follows:  A differentially private mechanism (satisfying~\eqref{eq::DP}) guarantees that small changes in its input lead to only insignificant changes in its output.  If condition~\eqref{eq::DP} is violated, an adversary could infer whether specific data was included in the input.  At a high level, differential privacy prevents attackers from gaining significant knowledge about the input by observing changes in the output~\citep{dwork2006calibrating}.

Apart from privacy concerns, communicating local data from edge devices to the central server often becomes a bottleneck in the system pipeline, especially with high-dimensional data common in many machine learning scenarios. This leads to the following fundamental question: how can we efficiently communicate differentially privatized data? 

Recent works have shown that a wide range of differential privacy mechanisms can be ``simulated'' and ``compressed'' using shared randomness, resulting in a ``compressed mechanism'' which has a smaller communication cost compared to the original mechanism, while retaining the (perhaps slightly weakened) privacy guarantee.
This can be done via rejection sampling \citep{feldman2021lossless}, importance sampling \citep{shah2022optimal, triastcyn2021dp}, 
or dithered quantization \citep{lang2023joint,shahmiri2024communication,hasircioglu2023communication,hegazy2024compression,yan2023layered} with each approach having its own advantages and disadvantages. For example, importance-sampling-based methods \citep{shah2022optimal, triastcyn2021dp} and the rejection-sampling-based method \citep{feldman2021lossless} can simulate a wide range of privacy mechanisms; however, the output distribution of the induced mechanism does
not perfectly match the original mechanism. This is limiting in scenarios where the original mechanism is designed to satisfy some desired statistical properties, e.g. it is often desirable for the local randomizer to be unbiased or to be ``summable'' noise such as Gaussian or other infinitely divisible distributions. Since the induced mechanism is different from the original one, these statistical properties are not preserved. On the other hand, dithered-quantization-based approaches \citep{hegazy2022randomized,lang2023joint,shahmiri2024communication,hasircioglu2023communication,hegazy2024compression,yan2023layered} can ensure a correct simulated distribution, but they can only simulate additive noise mechanisms. More importantly, 
dithered quantization relies on shared randomness between the user and the server, and the server needs to know the dither for decoding. 
This annuls the local privacy guarantee on the user data, unless we are willing to assume a trusted aggregator \citep{hasircioglu2023communication}, use an additional secure aggregation step \citep{hegazy2024compression}, or restrict attention to specific privacy mechanisms (e.g., one-dimensional Laplace \citep{shahmiri2024communication}).

\section{One-Shot Codes Meet Differential Privacy}

In this section, we discuss the interplay between information theory—specifically, one-shot coding schemes—and differential privacy, which motivates the studies presented in this thesis. Our goal is to develop a unified one-shot coding framework that is applicable to both network information theory problems and differential privacy mechanisms. We emphasize that our focus is on information-theoretic coding schemes, rather than information-theoretic measures of privacy. For discussions on the latter, we refer the reader to~\cite{unsal2023information}.

As discussed in Section~\ref{sec::bkgrd_1shot}, one-shot information theory addresses the most fundamental setting, where no assumptions are made about the signal length or the structure of the channel/source. 
Similar to the asymptotic analyses in network information theory~\cite{el2011network}, the goal remains to investigate the fundamental limits of signal transmission over noisy channels or source compression. 
To ensure reliable message reconstruction—and thus high accuracy for downstream tasks—effective one-shot coding schemes must be constructed to withstand noise. 
In this thesis, we present various one-shot codes based on early works on Poisson functional representation~\cite{li2018strong, li2021unified}.

However, when considering privacy, the requirement may initially appear to conflict with the goal of communication: achieving differential privacy typically involves deliberately adding noise, which increases the message entropy and reduces its compressibility. 
This tension is sometimes referred to as the ``communication-privacy-accuracy trilemma''~\cite{chen2022breaking}. 
Nevertheless, it has been shown that through careful encoding and controlled noise injection for privacy, it is possible to simultaneously achieve communication efficiency and privacy, while still maintaining accuracy for various tasks~\cite{feldman2021lossless, bassily2017practical, acharya2019communication, chen2022breaking}. 
In this sense, noise is \emph{utilized} as a tool for both compression and privacy.

One might initially think that controlling noise is difficult. 
However, the idea of using random noise for source coding dates back to~\cite{zamir2002universal, ziv1985universal}, which proposed additive noise as a tool for universally good lossy compression schemes. 
These works represent early-stage studies of \emph{channel simulation}. 
Channel simulation, also known as \emph{channel synthesis} or \emph{reverse channel coding}, aims to \emph{simulate} a noisy channel using as few communication bits as possible~\cite{cuff2013distributed, bennett2002entanglement}. 
This approach has the potential to simultaneously achieve communication efficiency and privacy. 
We will provide a more detailed review of channel simulation in Chapter~\ref{chp:ppr}. 
Various channel simulation schemes—such as dithered quantization~\cite{ziv1985universal}, rejection sampling~\cite{harsha2007communication}, and importance sampling (or more specifically, minimal random coding)~\cite{havasi2019minimal}—have been employed to compress differential privacy mechanisms~\cite{bun2019heavy, feldman2021lossless, triastcyn2021dp, shah2022optimal, hegazy2022randomized, shahmiri2024communication, yan2023layered}.
Readers are referred to~\cite{li2024channel} for a comprehensive review.

Though we will first discuss one-shot codes based on the Poisson functional representation~\cite{li2018strong, li2021unified} for various network information theory problems, it is worth noting that the Poisson functional representation is also a good channel simulation scheme~\cite{li2018strong} (with some improvements in analysis found in~\cite{li2024pointwise, li2021unified, li2025discrete}). 
With unlimited common randomness, it can provide the smallest known bound on the expected length for one-shot channel simulation. 
However, it is a deterministic mapping: the input, together with the common randomness, deterministically determines the output. 
As a result, a small change in the input (in the sense of differential privacy) can lead to a deterministic change in the compressed output, making this method unsuitable for directly ensuring privacy. 
To address this issue, we propose a way to \emph{randomize} the Poisson functional representation. 
This randomized variant will be shown to preserve differential privacy while achieving a compression size close to the optimal.

In summary, though at first glance the goals of privacy protection and efficient communication may appear to be in conflict, recent works have shown that this ``paradox'' can be resolved, and channel simulation emerges as a promising candidate for achieving both. 
The various one-shot codes proposed in this thesis are based on the Poisson functional representation~\cite{li2018strong}, which is also a state-of-the-art channel simulation scheme. 
Although its deterministic nature poses challenges for direct application to differential privacy mechanisms, we show that it can be \emph{randomized} to provide privacy protection. 
From this unified perspective, the Poisson functional representation serves as a bridge between one-shot information theory and differential privacy. 
We note that while importance sampling has also been used in both network information theory~\cite{phan2024importance} and differential privacy~\cite{triastcyn2021dp, shah2022optimal}, it is \emph{not exact}; that is, the output distribution is distorted, resulting in only approximate simulation. 
We will elaborate on this distinction in Chapter~\ref{chp:ppr}. 

\section{Our Contributions}

\subsection{Contributions in Chapter~\ref{chp:nnc}}

In Chapter~\ref{chp:nnc}, we present a unified one-shot coding framework designed for the communication and compression of messages among multiple nodes across a general acyclic noisy network. Our setting can be seen as a one-shot version of the acyclic discrete memoryless network studied by Lee and Chung~\cite{lee2018unified}, and noisy network coding studied by Lim, Kim, El Gamal and Chung~\cite{lim2011noisy}. 
We design a proof technique, called the exponential process refinement lemma, that is rooted in the Poisson matching lemma by Li and Anantharam, and can significantly simplify the analyses of one-shot coding over multi-hop networks. Our one-shot coding theorem not only recovers a wide range of existing asymptotic results, but also yields novel one-shot achievability results in different multi-hop network information theory problems, such as compress-and-forward and partial-decode-and-forward bounds for a one-shot (primitive) relay channel, and a bound for one-shot cascade multiterminal source coding. In a broader context, our framework provides a unified one-shot bound applicable to any combination of source coding, channel coding and coding for computing problems. This chapter is based on~\cite{liu2024one}.

\subsection{Contributions in Chapter~\ref{chp:hiding}}
In Chapter~\ref{chp:hiding}, we present one-shot information-theoretic analyses of two secrecy problems: a generalization of the information hiding problem~\cite{moulin2003information} and the compound wiretap channel~\cite{liang2009compound}.
The former admits a game-theoretic formulation, where one party (the information hider and decoder) seeks to embed secret messages into a host signal for later reconstruction, while the opposing party (an attacker) attempts to remove or degrade the embedded information.
The latter generalizes Wyner's wiretap channel by allowing multiple potential channel states, making it more suitable for the rapidly changing characteristics of modern wireless communications.
Although these two secrecy problems seem unrelated, we study both utilizing a covering argument and similar techniques under a unified framework. 
We derive one-shot achievability results for both problems using techniques based on the Poisson matching lemma, which enables us to handle both discrete and continuous cases.
We show that our one-shot results readily recover existing asymptotic results.
Unlike previous asymptotic results, ours apply to any source distribution and any class of channels, not necessarily memoryless or ergodic. 
This chapter is partially based on~\cite{liu2024hiding}.

\subsection{Contributions in Chapter~\ref{chp:ppr}}

In Chapter~\ref{chp:ppr}, we introduce a novel construction, called Poisson private representation (PPR), designed to compress and simulate any local randomizer while ensuring local differential privacy, hence reduce the communication cost of differential privacy mechanisms, 
Unlike previous simulation-based local differential privacy mechanisms, PPR exactly preserves the joint distribution of the data and the output of the original local randomizer. Hence, the PPR-compressed privacy mechanism retains all desirable statistical properties of the original privacy mechanism such as unbiasedness and Gaussianity. Moreover, PPR achieves a compression size within a logarithmic gap from the theoretical lower bound. Using the PPR, we give a new order-wise trade-off between communication, accuracy, central and local differential privacy for distributed mean estimation. Experiment results on distributed mean estimation show that PPR consistently gives a better trade-off between communication, accuracy and central differential privacy compared to the coordinate subsampled Gaussian mechanism, while also providing local differential privacy.
This chapter is based on~\cite{liu2024universal}.



\chapter{Poisson Functional Representation}
\label{chp:existing_techniques}

In this chapter, we review the Poisson functional representation~\cite{li2018strong} and discuss some technical background on related schemes. 
We begin by introducing our notations.

\section{Notations}

We assume the logarithm and entropy are to the base $2$ unless otherwise stated, and logarithm to the base $e$ is denoted as $\ln(x)$. 
For a statement $S$, we use $\mathbf{1}\{S\}$ to denote its indicator, i.e., $\mathbf{1}\{S\}$ is $1$ if $S$ holds and otherwise $\mathbf{1}\{S\} = 0$.
$\delta_{a}$ denotes the degenerate distribution $\mathbf{P}\{X=a\}=1$.

We use $[i..j]$ to denote $\{i,i+1,\ldots,j \}$ and $[j] := [1..j]$. 
For a set $\mathcal{S} \subseteq [k]$ and random sequence $U_1,\ldots,U_k$, we write $U^k := (U_1,\ldots,U_k)$, $U_\mathcal{S}:= (U_j)_{j\in \mathcal{S}}$.  
For two random variables $X,Y$, the information density is defined as \begin{equation*}
    \iota_{X;Y}(x;y) = \log \left(\frac{\mathrm{d} P_{X|Y}(x |y)}{\mathrm{d} P_X(x)}\right),
\end{equation*} 
where $\frac{\mathrm{d} P_{X|Y}(x |y)}{\mathrm{d} P_X(x)}$ denotes the Radon-Nikodym derivative. 
For two random variables $X,Y$, we sometimes omit the subscript and write $\iota(X;Y)$ instead of $\iota_{X;Y}(X;Y)$ if the random variables are clear from the context. 
In discrete case, the conditional information density is defined to be 
\begin{equation*}
    \iota_{X;Y|Z}(x;y|z):=\log \left(\frac{P_{X,Y|Z}(x,y|z)}{P_{X|Z}(x|z)P_{Y|Z}(y|z)} \right),
\end{equation*}

The total variation (TV) distance between two distributions $P,Q$ over $\mathcal{X}$ is $\Vert P-Q\Vert_{\mathrm{TV}} := \sup_{A \subseteq \mathcal{X} \; \text{measurable}} |P(A)-Q(A)|$.

For two distributions $P$ and $Q$, we write $P \ll Q$ to denote that $P$ is absolutely continuous with respect to $Q$.

\section{Poisson Functional Representation}

In this section, we introduce the Poisson Functional Representation~\cite{li2018strong} and the Poisson Matching Lemma~\cite{li2021unified}, which serve as the building blocks of this thesis.

The Poisson functional representation was introduced in \citep{li2018strong} as a channel simulation scheme, where a strong functional representation lemma is proved. 
Related constructions for Monte Carlo simulations can be found in \citep{maddison2016poisson}. 
Together with other techniques, the Poisson functional representation~\cite{li2018strong} has been applied to various fields, including neural estimation~\cite{lei2022neural} and minimax learning~\cite{li2018minimax}. 
As discussed in \citep{flamich2022fast}, the Poisson functional representation can be viewed as a certain variant of the $\mathrm{A}^*$ sampling \citep{maddison2014sampling, maddison2016poisson}, and hence an optimized version with better runtime for one-dimensional unimodal distribution has been proposed in~\citep{flamich2022fast}. 
A greedy-search version can be found in~\cite{flamich2024greedy}.

Based on the Poisson functional representation, the Poisson Matching Lemma was proposed in~\cite{li2021unified}, and it has been shown to improve upon previously known one-shot bounds in various settings with simpler analyses. 
Recent applications of the Poisson Matching Lemma include unequal message protection~\cite{khisti2024unequal}, hypothesis testing~\cite{guo2024hypothesis}, and secret key generation~\cite{hentila2024communication}.

We start with the discrete case, which we refer to as the exponential functional representation~\cite{li2018strong}.

\begin{definition}[Exponential Functional Representation~\cite{li2018strong}]
    Consider a finite set $\mathcal{U}$. Let $\mathbf{U}:=(Z_{u})_{u\in\mathcal{U}}$ be i.i.d. $\mathrm{Exp}(1)$ random variables.\footnote{$\mathrm{Exp}(1)$ random variables follow an exponential distribution with rate parameter $1$.} Given a distribution $P$ over $\mathcal{U}$, 
        \begin{equation}
        \mathbf{U}_{P} :=\mathrm{argmin}_{u}\frac{Z_{u}}{P(u)} \label{eq:pfr}
        \end{equation}
    is called the exponential functional representation \cite{li2018strong}. 
\end{definition}
By~\cite{li2018strong}, we have $\mathbf{U}_{P} \sim P$.

The exponential functional representation~\cite{li2018strong} is designed for finite alphabets, which is the case in Chapter~\ref{chp:nnc}. 
When the space is continuous, as in Chapter~\ref{chp:hiding} and Chapter~\ref{chp:ppr}, a generalization via Poisson processes is utilized~\cite{li2018strong, li2021unified}. 
Further discussions and detailed derivations of the connection between the two cases can be found in~\cite{li2021unified, li2024channel}; we omit them here. 
We introduce the generalization, called the Poisson functional representation~\cite{li2018strong}, as follows.

\begin{definition}[Poisson Functional Representation~\cite{li2018strong}]
Let $(T_{i})_{i}$ be a Poisson process with rate $1$ (i.e., $T_{1},T_{2}-T_{1},T_{3}-T_{2},\ldots\stackrel{iid}{\sim}\mathrm{Exp}(1)$), independent of $\bar{U}_{i}\stackrel{iid}{\sim} Q$ for $i=1,2,\ldots$, and we denote $\mathbf{U}:= (\bar{U}_i)_i$. 
$(\bar{U}_i, T_{i})_{i}$ is a Poisson process with intensity 
measure $Q\times\lambda_{[0,\infty)}$ \cite{last2017lectures}, where $\lambda_{[0,\infty)}$ is the Lebesgue measure over $[0,\infty)$.
Fix any distribution $P$ over $\mathcal{U}$ that is absolutely continuous with respect to $Q$. Let 
\begin{equation}
    \tilde{T}_{i}:=T_{i} \cdot \Big(\frac{\mathrm{d}P}{\mathrm{d}Q}(\bar{U}_i)\Big)^{-1}, 
    \label{eq:PFR}
\end{equation}
where $\frac{\mathrm{d}P}{\mathrm{d}Q}(\cdot)$ is the Radon-Nikodym derivative. 
By the mapping theorem \citep{last2017lectures}, $(\bar{U}_i,\tilde{T}_{i})$ is a Poisson process with intensity
measure $P\times\lambda_{[0,\infty)}$. 
Then the Poisson functional representation~\cite{li2018strong} selects 
\begin{equation*}
    \mathbf{U}_P := \bar{U}_{K}, 
\end{equation*} where
\begin{equation*}
    K:=\mathrm{argmin}_{i}\tilde{T}_{i}.
\end{equation*}
\end{definition}
Note that since the $T_i$'s are continuous, with probability $1$, there do not exist two equal values among $\tilde{T}_i$'s. 
The Poisson functional representation~\cite{li2018strong} holds
for general $Q$ which may be discrete or continuous.

The Poisson functional representation~\cite{li2018strong} selects a sample following the target distribution $P$ among samples from another distribution $Q$, i.e., $\mathbf{U}_{P} \sim P$. 
It draws a random sequence $(\bar{U}_i)_i$ from $Q$ and a sequence of times $(T_i)_i$ according to a Poisson process. 
If we select the sample $\bar{U}_i$ with the smallest $T_i$, then the selected sample will follow distribution $Q$. 
To obtain a sample from $P$ instead, we multiply the time by the factor $(\frac{\mathrm{d}P}{\mathrm{d}Q}(\bar{U}_i))^{-1}$ in~\eqref{eq:PFR} to give $\tilde{T}_i$, so the $\bar{U}_i$ with the smallest $\tilde{T}_i$ will follow $P$.

The Poisson functional representation~\cite{li2018strong} was originally developed to prove the strong functional representation lemma, and possibly tighter guarantees via different analyses can be found in~\cite{li2024pointwise, li2025discrete}.

The way this Poisson process is used in communication settings (e.g., in~\cite{li2021unified}) is that the encoder would query the process using the prior distribution of the signal to obtain the signal to be sent, and the decoder would query using the posterior distribution of the signal given the noisy observation to obtain the message. 
There is no error in the communication if the two queries return the same sample.
The probability of error can be bounded by the Poisson matching lemma~\cite{li2021unified}, which will be discussed in Section~\ref{sec::PML}.

\section{Poisson Matching Lemma}
\label{sec::PML}

In this section, we introduce a technique that is based on the Poisson Functional Representation, called the Poisson matching lemma~\cite{li2021unified}, which has been shown to be able to provide good one-shot achievability results on a large class of network information theory problems~\cite{li2021unified}. 

The Poisson matching lemma has been shown to be quite useful in proving one-shot achievability results of network information theory~\cite{li2021unified, liu2024one}. 
It is rooted in the Poisson functional representation~\cite{li2018strong} that is reviewed as follows.

\begin{lemma}[Poisson matching lemma~\cite{li2021unified}]
Consider two distributions $P_1, P_2 \ll Q$. Almost surely, we have
\begin{equation*} 
    \mathbf{P}\big(
    \mathbf{U}_{P_2}\neq \mathbf{U}_{P_1} \,\big|\, \mathbf{U}_{P_1} \big) \leq 1 - \Big( 1 + \frac{\mathrm{d}P_1}{\mathrm{d}P_2}(\mathbf{U}_{P_1})
    \Big)^{-1}. 
\end{equation*}
\end{lemma}

The Poisson matching lemma~\cite{li2021unified} provides a bound
on the probability of mismatch between the Poisson functional
representations applied on different distributions. 
Various information theory problems have been studied by using the Poisson matching lemma~\cite{li2021unified}. 
In chapter~\ref{chp:nnc}, we will extend it to a tool that can provide one-shot achievability results over arbitrary acyclic noisy networks, and hence recover many one-shot results in~\cite{li2021unified}.

\section{Discussions on Other Existing Techniques}

Compared to the one-shot coding scheme in~\cite{yassaee2013technique}, the Poisson matching lemma utilizes a Poisson process to create a codebook, instead of the conventional i.i.d. random codebook~\cite{yassaee2013technique}, and each codeword is assigned a bias $T_i$. The scheme is thus a biased maximum likelihood decoder, rather than a stochastic decoder as in~\cite{yassaee2013technique}. The idea of using a biased, or \emph{soft}, coding scheme has been extended to linear codes, known as \emph{weighted parity-check codes}; see~\cite{ling2022weighted, ling2023weighted_tit}.

Except for the one-shot coding scheme based on Poisson functional representation~\cite{li2018strong, li2021unified}, other unified frameworks for one-shot coding include the schemes based on random binning~\cite{yassaee2013technique, yassaee2013non, yassaee2015one}, the likelihood encoder~\cite{song2016likelihood}, and importance sampling~\cite{phan2024importance}. Other one-shot coding schemes~\cite{feinstein1954new, shannon1957certain, hayashi2009information, verdu2012non, watanabe2015nonasymptotic} have been reviewed in detail in Chapter~\ref{chp:intro}; here, we discuss some connections between one-shot codes and channel simulation via the likelihood encoder~\cite{song2016likelihood} and importance sampling~\cite{phan2024importance}, since we will utilize channel simulation to compress differential privacy mechanisms in Chapter~\ref{chp:ppr}.

As mentioned above, the Poisson functional representation can be viewed as a variant of A* sampling~\cite{maddison2014sampling, maddison2016poisson}, and a very useful property is that the output sample follows \emph{exactly} the input distribution. 
This can also be understood as a remote sampling problem; we refer readers to~\cite{li2024channel} for a comprehensive explanation. 
The Poisson functional representation can be used to prove the strong functional representation lemma~\cite{li2018strong}: 
to simulate a channel $P_{Y|X}$, the communication cost (expected number of bits) required is bounded by
\begin{equation*}
    I(X; Y) + \log\big(I(X; Y) + 1\big) + 5,
\end{equation*}
which with some finer analyses can be improved to$I(X; Y) + \log\big(I(X; Y) + 2\big) + 3$~\cite{li2025discrete, li2024pointwise}.

If one considers greedy rejection sampling for channel simulation, a weaker guarantee has been proved by~\cite{harsha2010communication}, and improved by~\cite{braverman2014public}, as follows: 
\begin{equation*}
    I(X; Y) + \log\big(I(X; Y) + 1\big) + c,
\end{equation*}
where $c$ is an unspecified constant.
The guarantee by using greedy rejection sampling was later improved by~\cite{flamich2023adaptive} to
\begin{equation*}
    I(X; Y) + \log\big(I(X; Y) +\log(4e)\big)  +\log(4e) + 1. 
\end{equation*}

Another popular sampling scheme is importance sampling, which was considered by~\cite{cuff2009communication} for asymptotic channel simulation, and later dubbed the \emph{likelihood encoder}~\cite{song2016likelihood} for one-shot coding. 
In machine learning, a similar scheme named \emph{minimal random coding} was also studied by~\cite{havasi2019minimal}, where it was applied to model compression and later to lossy image compression~\cite{flamich2020compressing}. 
Besides the likelihood encoder, a recent work~\cite{phan2024importance} used importance sampling (an \emph{importance matching lemma} that shares some similarity with the Poisson matching lemma~\cite{li2021unified}) for coding in information theory, and it has the potential to be extended to other information theory problems. 
We would like to emphasize here that a major difference compared to the Poisson functional representation is that the output sample does not exactly follow the input distribution.

This difference also appears in the study of compressing differential privacy mechanisms. 
Compression of differential privacy mechanisms can be viewed as a channel simulation problem, where the channel is subject to an additional privacy constraint. 
Importance sampling (or more specifically, minimal random coding~\cite{havasi2019minimal}) has been used for compressing differential privacy mechanisms~\cite{triastcyn2021dp, shah2022optimal}. 
However, as mentioned in the previous paragraph, since minimal random coding is not exact, the output distribution is only approximate. 
On the other hand, rejection sampling can also be utilized~\cite{feldman2021lossless, bun2019heavy}, although the communication cost and privacy guarantees were not close to optimal. 
In Chapter~\ref{chp:ppr}, we will utilize a variant of the Poisson functional representation~\cite{li2018strong} to compress differential privacy mechanisms.

In summary, importance sampling has been applied in both one-shot coding and the compression of differential privacy mechanisms. 
In contrast, the Poisson functional representation~\cite{li2018strong} offers an exact simulation framework with close-to-optimal communication cost guarantees. 
In the following chapter, we will demonstrate how to extend the construction of the Poisson functional representation to various problems, ensuring good performance.


\chapter{One-Shot Coding over General Noisy Networks} \label{chp:nnc}

\section{Overview}

In this chapter, we study a general class of networks, which we call \emph{acyclic discrete networks}, where there are $N$ nodes connected by noisy channels in an acyclic manner. Each node can play the role of an encoder or a decoder (or both) in source coding or channel coding settings. This is a one-shot version of the asymptotic acyclic discrete memoryless network studied by Lee and Chung~\cite{lee2018unified}, and includes a wide range of settings as special cases, such as source and channel coding, primitive relay channel~\cite{el2011network, kim2007coding, mondelli2019new, el2021achievable, el2022strengthened}, Gelfand-Pinsker~\cite{gelfand1980coding,Heegard1980}, relay-with-unlimited-look-ahead~\cite{el2005relay, el2007relay}, Wyner-Ziv~\cite{wyner1976rate, wyner1978rate}, coding for computing~\cite{yamamoto1982wyner},  multiple access channels~\cite{ahlswede1971multi, liao1972multiple, ahlswede1974capacity}, broadcast channels~\cite{marton1979coding} and cascade multiterminal source coding~\cite{cuff2009cascade}. In a broader context, our one-shot achievability results are general enough to be applicable to any
combination of source coding, channel coding and coding for computing problems. 
This chapter is based on~\cite{liu2024one}.

In order to alleviate the difficulty of keeping track of a large number of auxiliary random variables in a general $N$-node network, we propose a tool called the \emph{exponential process refinement lemma} based on the Poisson matching lemma~\cite{li2021unified},\footnote{We only present the discrete case in this chapter for the sake of simplicity. Hence, instead of Poisson processes, we may use an i.i.d. exponential processes 
instead~\cite{li2018strong}. While we expect the results to be extended to the continuous case, this is left for future studies. 
For the use of Poisson functional representation in continuous case in other settings, see Chapter~\ref{chp:hiding} and Chapter~\ref{chp:ppr}. }
which simplifies the analyses of the evolution of the posterior distribution of the sources, messages and/or auxiliary random variables at the decoder.
We utilize the lemma to prove a one-shot achievability result for general acyclic discrete networks, which recovers existing one-shot results in a range of settings in~\cite{li2021unified, verdu2012non, yassaee2013technique, watanabe2015nonasymptotic}, and also give novel one-shot results for various multi-hop settings, namely primitive relay channels~\cite{el2011network, kim2007coding, mondelli2019new, el2021achievable, el2022strengthened}, relay-with-unlimited-look-ahead~\cite{el2005relay, el2007relay}, and cascade multiterminal source coding~\cite{cuff2009cascade}.

The chapter is organized as follows. 
We present our proof technique, called the exponential process refinement lemma, in Section~\ref{sec:expon}. 
We describe our general acyclic discrete network in Section~\ref{sec::net_model}, and prove our main theorem in Section~\ref{sec::main}. 
In Section~\ref{sec::relay}, we use a one-shot relay channel and related settings to elaborate our coding scheme in detail. 
We then discuss a novel one-shot cascade multiterminal source coding problem in Section~\ref{sec::cascade}. 
We also show our coding scheme provides one-shot bounds on various network information theory settings in Section~\ref{sec::NIT}.

\section{Exponential Process Refinement Lemma}
\label{sec:expon}

Recall we have introduced the Exponential functional representation~\cite{li2018strong} in Chapter~\ref{chp:existing_techniques}. 
We briefly review it here together with the Poisson matching lemma~\cite{li2021unified} for the sake of completeness. 
We then design a tool for proving one-shot achievability results over noisy multi-hop networks based on the Poisson matching lemma,  
called the \emph{exponential process refinement lemma}.

Consider a finite set $\mathcal{U}$. Let $\mathbf{U}:=(Z_{u})_{u\in\mathcal{U}}$
be i.i.d. $\mathrm{Exp}(1)$ random variables.\footnote{When the space $\mathcal{U}$ is continuous, 
a Poisson process is used in~\cite{li2018strong, li2021unified}.
} 
Given a distribution $P$ over $\mathcal{U}$, 
\begin{equation}
\mathbf{U}_{P} :=\mathrm{argmin}_{u}\frac{Z_{u}}{P(u)} \label{eq:pfr}
\end{equation}
is called the exponential functional representation in Chapter~\ref{chp:existing_techniques}.\footnote{In~\cite{li2018strong}, even for discrete case using exponential random variables, \eqref{eq:pfr} was still called the \emph{Poisson} functional representation. Here and similar to~\cite{li2024channel} we call it exponential functional representation to distinguish it with the Poisson functional representation that works for continuous cases in Chapter~\ref{chp:hiding} and Chapter~\ref{chp:ppr}.}
    
As explained in Chapter~\ref{chp:existing_techniques} and also~\cite{li2018strong}, we have $\mathbf{U}_{P} \sim P$. 

We can generalize this by letting $\mathbf{U}_{P}(1),\ldots,\mathbf{U}_{P}(|\mathcal{U}|)\in\mathcal{U}$
be the elements of $\mathcal{U}$ sorted in ascending order of $Z_{u}/P(u)$:
\[
\frac{Z_{\mathbf{U}_{P}(1)}}{P(\mathbf{U}_{P}(1))}\le\cdots\le\frac{Z_{\mathbf{U}_{P}(|\mathcal{U}|)}}{P(\mathbf{U}_{P}(|\mathcal{U}|))}.
\]

We break ties arbitrarily and treat $1/0=\infty$. 
This is similar to the mapped Poisson process in the generalized Poisson matching lemma \cite{li2021unified}, though unlike \cite{li2021unified},
$\mathbf{U}_{P}(1),\ldots,\mathbf{U}_{P}(|\mathcal{U}|)$ is not an
i.i.d. sequence following $P$. Write $\mathbf{U}_{P}^{-1}: \mathcal{U} \to [|\mathcal{U}|]$ for
the inverse function of $i \mapsto \mathbf{U}_{P}(i)$. 
The following is a
direct corollary of the generalized Poisson matching lemma~\cite{li2021unified}.

\begin{lemma}
\label{lem::GPML}
For distributions $P,Q$ over $\mathcal{U}$, we have the following almost surely:
\[
\mathbf{E}\left[\mathbf{U}_{Q}^{-1}(\mathbf{U}_{P})\,\Big|\,\mathbf{U}_{P}\right]\le\frac{P(\mathbf{U}_{P})}{Q(\mathbf{U}_{P})}+1.
\]
\end{lemma}

We now define a convenient tool. 
\begin{definition}[Refining a distribution by an exponential process]
\label{def:refine}
For a joint distribution $Q_{V,U}$
over $\mathcal{V}\times\mathcal{U}$, the refinement of $Q_{V,U}$ by $\mathbf{U}$, denoted as $Q_{V,U}^\mathbf{U}$, is a joint distribution 
\[
Q_{V,U}^\mathbf{U}(v,u):= 
\frac{Q_{V}(v)}{\mathbf{U}_{Q_{U|V}(\cdot|v)}^{-1}(u)\sum_{i=1}^{|\mathcal{U}|}i^{-1}}
\]
for all $(v, u)$ in the support of $Q_{V,U}$, 
where $Q_{V}$ is the
$V$-marginal of $Q_{V,U}$ and $Q_{U|V}$ is the conditional distribution of $U$ given $V$. 
When $V = \emptyset$, the above definition becomes
\[
Q_{U}^\mathbf{U}(u)=\frac{1}{\mathbf{U}_{Q_{U}}^{-1}(u)\sum_{i=1}^{|\mathcal{U}|}i^{-1}}.
\]
\end{definition}

While the exponential functional representation $\mathbf{U}_{Q_U}$ (which only gives one value of $U$) is used for the unique decoding of $U$, the refinement $Q_{U}^\mathbf{U}(u)$ is for the \emph{soft decoding} of $U$, which gives a distribution over $U$, with $\mathbf{U}_{Q_U}$ having the largest probability. This is useful in non-unique decoding. For example, if we want to decode $U_1$ uniquely, while utilizing $U_2$ via non-unique decoding, we can first obtain the distribution $(\mathbf{U}_2)_{Q_{U_2}}$, and then compute the marginal distribution of $U_1$ in $(\mathbf{U}_2)_{Q_{U_2}}P_{U_1|U_2}$ and use this marginal distribution to recover $U_1$ via the exponential functional representation.

Loosely speaking, if the distribution $Q_{V,U}$ represents our ``prior distribution'' of $(V,U)$, then the refinement $Q_{V,U}^\mathbf{U}$ is our updated ``posterior distribution'' after taking the exponential process $\mathbf{U}$ into account. 
In multiterminal coding settings that a node decodes multiple random variables, the prior distribution of those random variables will be refined by multiple exponential processes. 
To keep track of the evolution of the ``posterior probability'' of the correct values of those random variables through the refinement process, we use the following lemma, called the \emph{exponential process refinement lemma}. 
Although its proof still relies on the Poisson matching lemma~\cite{li2021unified}, it significantly simplifies our analyses. 

\begin{lemma}[Exponential Process Refinement Lemma]
\label{lemma::PML2}
For a distribution $P$ over $\mathcal{U}$ and a joint distribution
$Q_{V,U}$ over a finite $\mathcal{V}\times\mathcal{U}$, for every
$v\in\mathcal{V}$, we have, almost surely,
\begin{align*}
 \mathbf{E}\bigg[\frac{1}{Q_{V,U}^\mathbf{U}(v,\mathbf{U}_{P})}\bigg|\mathbf{U}_{P}\bigg] 
 \le \frac{\ln|\mathcal{U}|+1}{Q_{V}(v)}\left(\frac{P(\mathbf{U}_{P})}{Q_{U|V} (\mathbf{U}_{P}|v)}+1\right).
\end{align*}
\end{lemma}

\begin{proof}
We have 
\begin{align*} & \mathbf{E}\bigg[\frac{1}{Q_{V,U}^\mathbf{U}(v,\mathbf{U}_{P})}\,\bigg|\,\mathbf{U}_{P}\bigg] \\ 
& \stackrel{(a)}{=}   \mathbf{E}\bigg[  \frac{ \mathbf{U}_{Q_{U|V(\cdot | v)}}^{-1}(\mathbf{U}_{P})\sum_{i=1}^{|\mathcal{U}|}i^{-1}}{Q_{V}(v)}\Bigg|\mathbf{U}_P \bigg] \\ 
& \stackrel{(b)}{\leq }  \frac{\sum_{i=1}^{|\mathcal{U}|}i^{-1}}{Q_{V}(v)} \left( \frac{P(\mathbf{U}_{P})}{Q_{U|V}(\mathbf{U}_{P}|v)} + 1 \right) \\  
& \stackrel{(c)}{\leq }  \frac{\ln|\mathcal{U}|+1}{Q_{V}(v)}\left(\frac{P(\mathbf{U}_{P})}{Q_{U|V}(\mathbf{U}_{P}|v)}+1\right),
\end{align*} 
where $(a)$ is by Definition~\ref{def:refine}, $(b)$ is by Lemma~\ref{lem::GPML} and $(c)$ is by $\sum_{i=1}^n i^{-1} \leq \int_1^n x^{-1} dx + 1 = \ln n+1 $.  
\end{proof}

\section{Network Model}
\label{sec::net_model}

We describe a general $N$-node network model, which is the one-shot version of the \emph{acyclic discrete memoryless network (ADMN)}~\cite{lee2018unified}. 
There are $N$ nodes labelled $1,\ldots,N$. 
Node $i$ observes $Y_i \in \mathcal{Y}_i$ and produces $X_i \in \mathcal{X}_i$ (while we assume $\mathcal{X}_i,\mathcal{Y}_i$ are finite). 
Unlike conventional asymptotic settings (e.g. \cite{lee2018unified}), here $X_i$ is only \emph{one symbol}, instead of a sequence $(X_{i,1},\ldots,X_{i,n})$.
The transmission is performed sequentially, and each $Y_i$ is allowed to depend on all previous inputs and outputs (i.e., $X^{i-1},Y^{i-1}$) in a stochastic manner, as shown in Figure~\ref{fig::ADN_net}. 
Therefore, we can formally define an $N$\emph{-node acyclic discrete network (ADN)} as a collection of channels 
$(P_{Y_i | X^{i-1},Y^{i-1}})_{i\in [N]}$, where $P_{Y_i | X^{i-1},Y^{i-1}}$ is a conditional distribution from $(\prod_{j=1}^{i-1}\mathcal{X}_j) \times (\prod_{j=1}^{i-1}\mathcal{Y}_j)$ to $\mathcal{Y}_i$. 
In particular, $Y_1$ follows $P_{Y_1}$ and does not depend on any other random variable.
The asymptotic ADMN \cite{lee2018unified} can be considered as the $n$-fold ADN $(P^n_{Y_i | X^{i-1}, Y^{i-1}})_{i\in [N]}$, where $P^n_{Y_i | X^{i-1},Y^{i-1}}$ denotes the $n$-fold product conditional distribution (i.e., $n$ copies of a memoryless channel), and we take the blocklength $n \to \infty$.

\begin{figure}[htpb]
	\centering
	\includegraphics[scale=0.28]{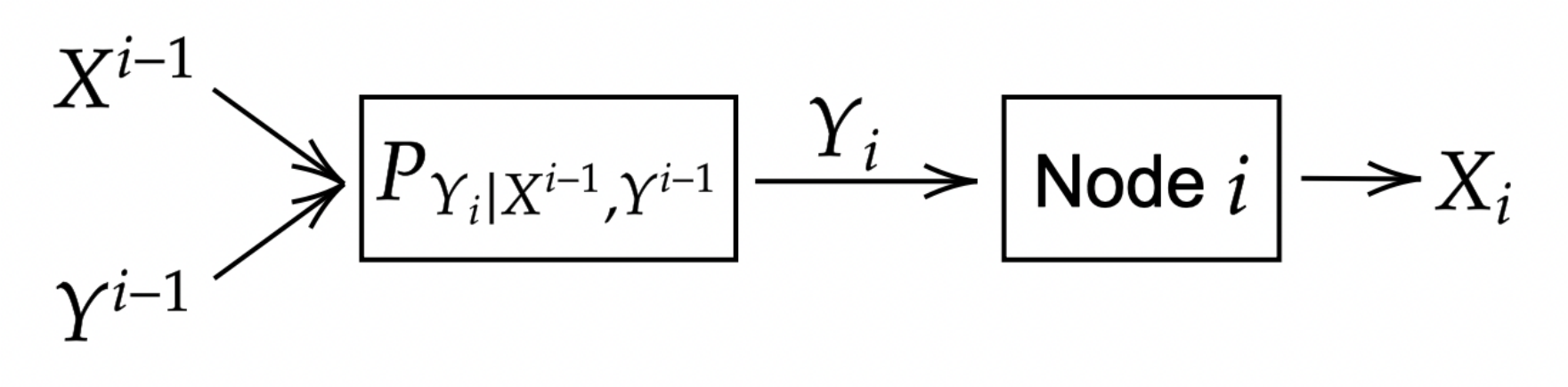}
	\caption{Acyclic discrete memoryless network.} 
	\label{fig::ADN_net}
\end{figure}

We remark that, similar to the asymptotic unified random coding bound~\cite{lee2018unified}, the $X_i$'s and $Y_i$'s can represent sources, states, channel inputs, outputs and messages in source coding and channel coding settings. 
For example, for point-to-point channel coding, we take $Y_1$ to be the message, which the encoder (node $1$) encodes into the channel input $X_1$, which in turn is sent through the channel $P_{Y_2|X_1}$. 
The decoder (node $2$) observes $Y_2$ and outputs $X_2$, which is the decoded message. 
For lossless source coding, $Y_1$ is the source, $X_1=Y_2$ is the description by the encoder, and $X_2$ is the reconstruction.

\begin{figure}[htpb]
    \centering
    \includegraphics[scale=0.35]{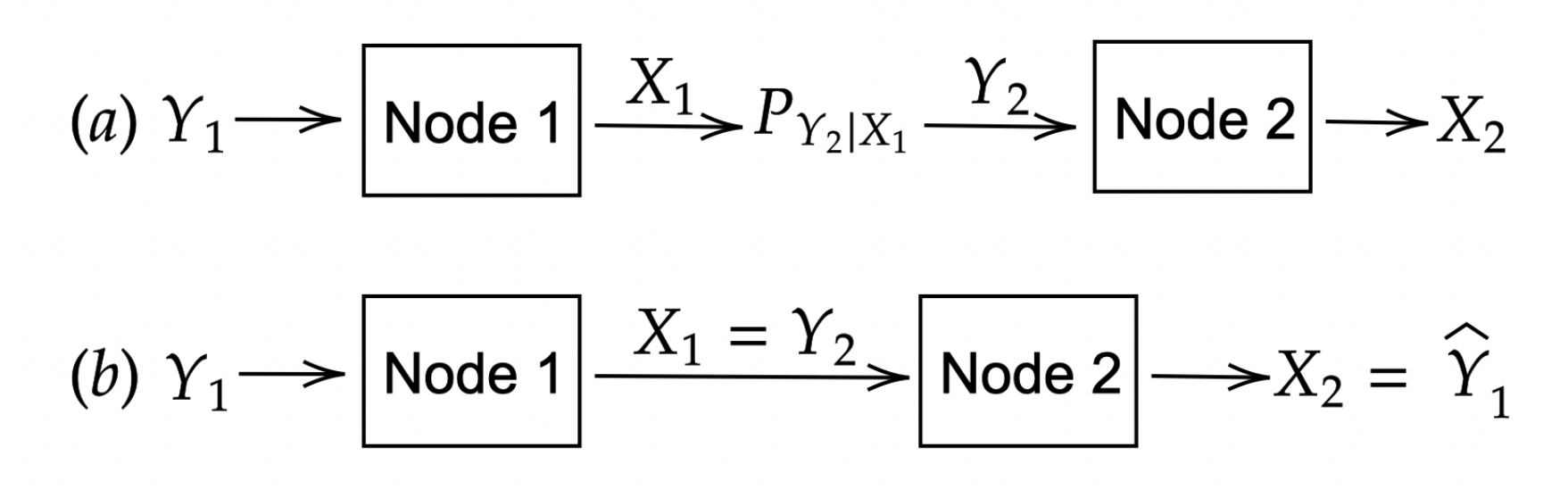}
    \caption{(a) Channel coding. (b) Source coding.} 
    \label{fig::channel_src}
\end{figure}

We give the definition of a coding scheme below.

\begin{definition}
A \emph{deterministic coding scheme} consists of a sequence of encoding functions $(f_i)_{i\in [N]}$, where $f_i : \mathcal{Y}_i \to \mathcal{X}_i$.
For $i=1,\ldots,N$, the following operations are performed:
\begin{itemize}
    \item \textbf{Noisy channel.} The output $\tilde{Y}_i$ is generated conditional on $\tilde{X}^{i - 1},\tilde{Y}^{i-1}$ according to $P_{Y_i | X^{i-1},Y^{i-1}}$. 
    For $i=1$, $\tilde{Y}_1 \sim P_{Y_1}$  can be regarded as a source or a channel state. 

    \item \textbf{Node operation.} Node $i$ observes $\tilde{Y}_i$ and outputs $\tilde{X}_i = f_i(\tilde{Y}_i)$.
\end{itemize}
\end{definition}

We sometimes allow an additional unlimited public randomness available to all nodes. 

\begin{definition}
A \emph{public-randomness coding scheme} for the network consists of a pair $(P_W, (f_i)_{i\in [N]})$, where $P_W$ is the distribution of the public randomness $W \in \mathcal{W}$ available to all nodes and $f_i : \mathcal{Y}_i \times \mathcal{W} \to \mathcal{X}_i$ is the encoding function of node $i$ mapping its observation $Y_i$ and the public randomness $W$ to its output $X_i$.
The operations are as follows. 
First, generate $W \sim P_W$. 
For $i=1,\ldots,N$, generate $\tilde{Y}_i$ conditional on $\tilde{X}^{i-1},\tilde{Y}^{i-1}$ according to $P_{Y_i | X^{i-1},Y^{i-1}}$, 
and take $\tilde{X}_i = f_i(\tilde{Y}_i,W)$.
\end{definition}

We do not impose any constraint on the public randomness $W$. 
In reality, to carry out a public-randomness coding scheme, the nodes share a common random seed to initialize their pseudorandom number generators before the scheme commences.

We use $\tilde{X}_i,\tilde{Y}_i$ to denote the actual random variables from the coding scheme. 
In contrast, $X_i,Y_i$ usually denote the random variables following an ideal distribution.
For example, in channel coding, the ideal distribution is $Y_1=X_2 \sim \mathrm{Unif}[\mathsf{L}]$ (i.e., the message is decoded without error), independent of $(X_1,Y_2) \sim P_{X_1}P_{Y_2|X_1}$. 
If we ensure that the actual $\tilde{X}^2,\tilde{Y}^2$ is ``close to'' the ideal $X^2,Y^2$, this would imply that $\tilde{Y}_1=\tilde{X}_2$ with high probability as well, giving a small error probability.

The goal (the ``achievability'') is to make the actual joint distribution $P_{\tilde{X}^N,\tilde{Y}^N}$ ``approximately as good as'' the ideal joint distribution $P_{X^N,Y^N}$. 
If we have an ``error set'' $\mathcal{E} \subseteq \big(\prod_{i=1}^N \mathcal{X}_i\big) \times \big(\prod_{i=1}^N \mathcal{Y}_i\big)$ that we do not want $(\tilde{X}^N,\tilde{Y}^N)$ to fall into (e.g., for channel coding, $\mathcal{E}$ is the set where $\tilde{Y}_1 \neq \tilde{X}_2$, i.e., an error occurs; for lossy source coding, $\mathcal{E}$ is the set where $d(\tilde{Y}_1, \tilde{X}_2)>\mathsf{D}$, i.e., the distortion exceeds the limit), we want 
\begin{align}
\mathbf{P}\big((\tilde{X}^N,\tilde{Y}^N) \in \mathcal{E} \big) \lesssim \mathbf{P}\big((X^N,Y^N) \in \mathcal{E} \big).\label{eq:errorset}
\end{align}
If $P_{\tilde{X}^N,\tilde{Y}^N}$ is close to $P_{X^N,Y^N}$ in total variation distance, i.e., 
\begin{align}
\delta_{\mathrm{TV}}\big(P_{X^N,Y^N},\, P_{\tilde{X}^N,\tilde{Y}^N} \big) \approx 0, 
\label{eq:error_tv}
\end{align}
then \eqref{eq:errorset} is guaranteed. 
For public-randomness coding, we show that \eqref{eq:error_tv} can be achieved, which can be seen as a channel simulation~\cite{bennett2002entanglement, cuff2013distributed} or a coordination~\cite{cuff2010coordination} result. 
For deterministic coding, since the node operations are deterministic, there might not be sufficient randomness to make $P_{\tilde{X}^N,\tilde{Y}^N}$ close to $P_{X^N,Y^N}$, and hence we use the  error bound in \eqref{eq:errorset}.

\section{Main Theorem for Acyclic Discrete Networks} 
\label{sec::main}

We show a one-shot achievability result for ADN via public-randomness coding scheme. 

\begin{theorem}
\label{thm::network_achievability}
Fix any ADN $(P_{Y_i | X^{i-1},Y^{i-1}})_{i\in [N]}$. 
For any collection of indices $(a_{i,j})_{i\in [N], j \in [d_i]}$ where $(a_{i,j})_{j \in [d_i]}$ is a sequence of distinct indices in $[i-1]$ for each $i$, any sequence $(d'_i)_{i\in [N]}$ with $0\le d'_i \le d_i$ and any collection of conditional distributions $(P_{U_i|Y_i, \overline{U}'_i}, P_{X_i|Y_i,U_i, \overline{U}'_i})_{i \in [N]}$ (where $\overline{U}_{i,\mathcal{S}} := (U_{a_{i,j}})_{j\in \mathcal{S}}$ for $\mathcal{S} \subseteq [d_i]$ and $\overline{U}'_{i}:=\overline{U}_{i,[d'_i]}$), which induces the joint distribution of $X^N,Y^N,U^N$ (the ``ideal distribution''), there exists a public-randomness coding scheme $(P_W, (f_i)_{i\in [N]})$ such that the joint distribution of $\tilde{X}^N, \tilde{Y}^N$ induced by the scheme (the ``actual distribution'') satisfies 
\begin{equation*}
\delta_{\mathrm{TV}}\big(P_{X^N,Y^N},\, P_{\tilde{X}^N,\tilde{Y}^N} \big)  \le
\mathbf{E}\bigg[\min\bigg\{
\sum_{i=1}^N 
\sum_{j=1}^{d'_i} 
B_{i,j},\, 1
\bigg\}
\bigg],
\end{equation*}
where 
\begin{equation}
B_{i,j}  := \gamma_{i,j} \prod_{k =j}^{d_i} \bigg(
  2^{-\iota (\overline{U}_{i,k}; \overline{U}_{i, [d_i] \backslash [j..k]} , Y_i) + \iota(\overline{U}_{i,k}; \overline{U}'_{a_{i,k}}, Y_{a_{i,k}}) } + \mathbf{1}\{k \! >\! j\}\bigg) \label{eq::beta}
\end{equation} 
such that\footnote{Note that the logarithmic terms $\gamma_{i,j}$ do not affect the first and second order results.} \begin{align*}
&\gamma_{i,j} := \prod_{k=j+1}^{d_i}
\Big( \ln|\mathcal{U}_{a_{i,k}}| + 1\Big).
\end{align*} 

\end{theorem}

The sequences $(a_{i,j})_j$ control which auxiliaries $U_j$ node $i$ decodes and in which order. 
Node $i$ uniquely decodes $\overline{U}'_i = $ $(U_{a_{i,j}})_{j \in [d'_i]}$ while utilizing $(U_{a_{i,j}})_{j \in [d'_i+1 .. d_i]}$ by non-unique decoding via the exponential process refinement (Definition~\ref{def:refine}). 
For brevity, we say ``the \emph{decoding order} of node $i$ is $\overline{U}_{i,1}, \ldots,$ $\overline{U}_{i,d'_i}, \overline{U}_{i,d'_i+1}?,\ldots,\overline{U}_{i,d_i}?$'' where ``?'' means the random variable is only used in non-unique decoding. 
Node $i$ decodes $\overline{U}'_i$, creates its own $U_i$ by using the exponential functional representation on $P_{U_i|Y_i, \overline{U}'_i}$, and generates $X_i$ from $P_{X_i|Y_i,U_i, \overline{U}'_i}$.

We also have the following result for deterministic coding schemes. 

\begin{theorem}\label{thm:det}
Fix any ADN $(P_{Y_i | X^{i-1},Y^{i-1}})_{i\in [N]}$.
For any $(a_{i,j})_{i\in [N], j \in [d_i]}$, $(d'_i)_{i\in [N]}$, $(P_{U_i|Y_i, \overline{U}'_i}, P_{X_i|Y_i,U_i, \overline{U}'_i})_{i \in [N]}$ as defined in Theorem~\ref{thm::network_achievability}, 
which induce the joint distribution of $X^N,Y^N,U^N$, and any set $\mathcal{E} \subseteq \big(\prod_{i=1}^N \mathcal{X}_i\big) \times \big(\prod_{i=1}^N \mathcal{Y}_i\big)$, there is a deterministic coding scheme $(f_i)_{i\in [N]}$ such that $\tilde{X}^N,\tilde{Y}^N$ induced by the scheme satisfy
\begin{align}
&\mathbf{P}\big((\tilde{X}^N,\tilde{Y}^N) \in \mathcal{E} \big) \nonumber \\
&\le \mathbf{E}\bigg[\min\bigg\{
\mathbf{1}\big\{(X^N,Y^N) \in \mathcal{E} \big\} + \sum_{i=1}^N 
\sum_{j=1}^{d'_i} 
B_{i,j},\, 1
\bigg\}
\bigg], 
\label{eq:error_bound_det}
\end{align}
where $B_{i,j}$ is defined in Theorem~\ref{thm::network_achievability}. 
\end{theorem}

Theorem~\ref{thm::network_achievability} implies the following result for the asymptotic ADMN \cite{lee2018unified} by directly applying the law of large numbers. 

\begin{corollary}\label{cor:admn}
Fix any ADN $(P_{Y_i | X^{i-1},Y^{i-1}})_{i\in [N]}$.
Fix any $(a_{i,j})_{i\in [N], j \in [d_i]}$, $(d'_i)_{i\in [N]}$, $(P_{U_i|Y_i, \overline{U}'_i}, P_{X_i|Y_i,U_i, \overline{U}'_i})_{i \in [N]}$ as defined in Theorem~\ref{thm::network_achievability}, 
which induces the joint distribution of $X^N,Y^N,U^N$. If for every $i \in [N]$, $j \in [d'_i]$, 
\begin{align*}
& I(\overline{U}_{i,j}; \overline{U}_{i, [d_i] \backslash \{j\}} , Y_i) - I(\overline{U}_{i,j}; \overline{U}'_{a_{i,j}}, Y_{a_{i,j}})  > \\
&\qquad \sum_{k=j+1}^{d_i} \bigg( \max \big\{ I(\overline{U}_{i,k}; \overline{U}'_{a_{i,k}}, Y_{a_{i,k}})  - I(\overline{U}_{i,k}; \overline{U}_{i, [d_i] \backslash [j..k]} , Y_i),\,0 \big\}\bigg),
\end{align*}
then there is a sequence of public-randomness coding (indexed by $n$) for the $n$-fold ADN $(P^n_{Y_i | X^{i-1},Y^{i-1}})_{i\in [N]}$ such that the induced $\tilde{X}^{N,n},\tilde{Y}^{N,n}$ (write $\tilde{X}^{N,n}=(\tilde{X}_{i,j})_{i\in [N], j \in [n]}$) satisfy
\begin{equation}
\lim_{n \to \infty} \delta_{\mathrm{TV}}\big(P^n_{X^{N},Y^{N}},\, P_{\tilde{X}^{N,n},\tilde{Y}^{N,n}} \big) = 0.
\end{equation}
\end{corollary}

While this result is not as strong as the general asymptotic result in \cite{lee2018unified}, a one-shot analogue of \cite{lee2018unified} will likely be significantly more complicated than Theorem~\ref{thm::network_achievability}. 
We choose to present Theorem~\ref{thm::network_achievability} since it is simple but already general and powerful enough to give a wide range of tight one-shot results.

\section{One-shot Relay Channel}
\label{sec::relay}

To explain our scheme, we first discuss a \emph{one-shot relay channel} in Figure~\ref{fig::setup_relay_one_shot}. 
An encoder observes $M \sim \mathrm{Unif}[\mathsf{L}]$ and outputs $X$, which is passed through the channel $P_{Y_{\mathrm{r}}|X}$.
The relay observes $Y_{\mathrm{r}}$ and outputs $X_{\mathrm{r}}$.
Then $(X,X_{\mathrm{r}},Y_{\mathrm{r}})$ is passed through the channel $P_{Y|X,X_{\mathrm{r}},Y_{\mathrm{r}}}$. 
The decoder observes $Y$ and recovers $\hat{M}$.
For generality, we allow $Y$ to depend on all of $X,X_{\mathrm{r}},Y_{\mathrm{r}}$, and $X_{\mathrm{r}}$ may interfere with $(X,Y_{\mathrm{r}})$, 
which can happen if the relay outputs $X_{\mathrm{r}}$ instantaneously or the channel has a long memory, or it is a storage device. 
It is a one-shot version of the \emph{relay-without-delay} and \emph{relay-with-unlimited-look-ahead}~\cite{el2005relay, el2007relay}, and is an  ADN by taking $Y_1=M$, $X_1=X$, $Y_2=Y_{\mathrm{r}}$, $X_2=X_{\mathrm{r}}$, $Y_3=Y$, and $X_3=M$ (in the ideal distributions).

In case if $Y=(Y',Y'')$ consists of two components and the channel $P_{Y|X,X_{\mathrm{r}},Y_{\mathrm{r}}}=P_{Y'|X,Y_{\mathrm{r}}}P_{Y''|X_{\mathrm{r}}}$ can be decomposed into two orthogonal components (so $X_{\mathrm{r}}$ does not interfere with $(X,Y_{\mathrm{r}})$), this becomes the one-shot version of the \emph{primitive relay channel}~\cite{el2011network, kim2007coding, mondelli2019new, el2021achievable, el2022strengthened} since the $n$-fold version of this ADN (with $n\to \infty$) is precisely the asymptotic primitive relay channel.
However, the $n$-fold version of the ADN in Figure~\ref{fig::setup_relay_one_shot} in general is not the conventional relay channel~\cite{el2011network, van1971three, cover1979capacity} (it is the relay-with-unlimited-look-ahead instead). 
The conventional relay channel, due to its causal assumption that the relay can only look at past $Y_{\mathrm{r},t}$'s, has no one-shot counterpart.

\begin{figure*}[htpb]
\centering
\includegraphics[scale=0.36]{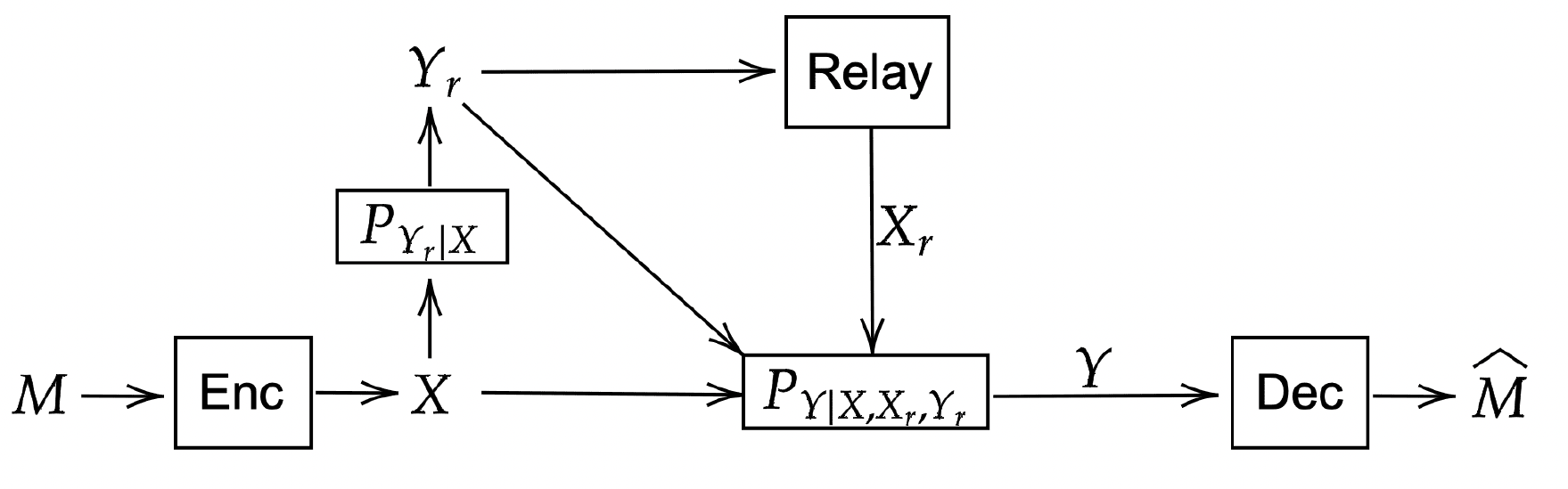}
\caption{One-shot relay channel setting.} \label{fig::setup_relay_one_shot}
\end{figure*}

We use the following corollary of Theorem \ref{thm:det} to demonstrate the use of the exponential process refinement lemma (Lemma~\ref{lemma::PML2}).

\begin{corollary}
\label{cor::one_shot_relay_eg}
For any $P_{X}$, $P_{U|Y_{\mathrm{r}}}$, function $x_{\mathrm{r}}(y_{\mathrm{r}},u)$, 
there is a deterministic coding scheme for the one-shot relay channel such that 
the error probability satisfies
\begin{equation} 
\label{eq::one_shot_relay}
      P_e\leq \mathbf{E} \Big[
    \min\big\{
    \gamma \mathsf{L} 
     2^{-\iota (X;U,Y)}
     \big(2^{-\iota (U;Y) + \iota (U;Y_{\mathrm{r}}) } +1 \big) , 1
    \big\}
    \Big],
\end{equation} 
where $(X,Y_{\mathrm{r}},U, X_{\mathrm{r}},Y) \sim$$P_{X} P_{Y_{\mathrm{r}}|X} P_{U|Y_{\mathrm{r}}} \delta_{x_{\mathrm{r}}(Y_{\mathrm{r}},U)} P_{Y|X,Y_{\mathrm{r}},X_{\mathrm{r}}}$, and $\gamma :=  \ln |\mathcal{U}| + 1$. 
\end{corollary}

\begin{proof}

For the sake of demonstration, we first give a detailed proof via the exponential process refinement lemma without invoking Theorem \ref{thm:det}.
Let $U_1 := (X,M)$, $U_2 := U$. 
Let $\mathbf{U}_1$, $\mathbf{U}_2$ be two independent exponential processes, which serve as the ``random codebooks''. 
The encoder (node 1) uses the exponential functional representation~\eqref{eq:pfr} to compute $U_1 = (\mathbf{U}_1)_{P_{U_1} \times  \delta_M }$ and outputs $X$-component of $U_1$. 
The relay (node 2) computes 
$U_2 = (\mathbf{U}_2)_{P_{U_2|Y_{\mathrm{r}}}(\cdot| Y_{\mathrm{r}})}$ and outputs $X_{\mathrm{r}}=x_{\mathrm{r}}(Y_{\mathrm{r}}, U_2)$. 
Note that $X,Y_{\mathrm{r}},U_2, X_{\mathrm{r}},Y$ follow the ideal distribution in the corollary due to the property of exponential functional representation, and hence we write $X_{\mathrm{r}}$ instead of $\tilde{X}_{\mathrm{r}}$.
The decoder (node 3) observes $Y$, and performs the following steps. 

\begin{enumerate}
\item Refine $P_{U_2|Y}(\cdot|Y)$ (written as $P_{U_2|Y}$ for brevity) to ${Q}_{U_2} := P_{U_2|Y}^{\mathbf{U}_2}$ using Definition~\ref{def:refine}. 
By the exponential process refinement lemma (Lemma~\ref{lemma::PML2}, with $V = \emptyset$), 
\begin{align*}
      \mathbf{E}
     \bigg[\frac{1}{{Q}_{U_2}(U_2)} \bigg| \, U_2,Y,Y_{\mathrm{r}}\bigg]
     \leq   (\ln |\mathcal{U}_2|+1)
    \left(
    \frac{P_{U_2|Y_{\mathrm{r}}}(U_2)}{P_{U_2|Y}(U_2) } + 1
    \right).
\end{align*}

\item Compute the joint distribution ${Q}_{U_2} P_{U_1 | U_2,Y}$ over $\mathcal{U}_1 \times \mathcal{U}_2$, 
	the semidirect product between ${Q}_{U_2}$ and $P_{U_1 | U_2,Y}(\cdot | \cdot , Y)$. 
 Let its $U_1$-marginal be $\tilde{Q}_{U_1}$.

\item Let $\tilde{U}_1 = (\mathbf{U}_1)_{\tilde{Q}_{U_1}\times P_M}$, and output its $M$-component. 
\end{enumerate}
Let $A :=( X,Y_{\mathrm{r}},U_2, X_{\mathrm{r}},Y ,M)$ and $\gamma := \ln|\mathcal{U}_2| + 1$,

\begin{align*}
    &\mathbf{P} (
    \tilde{U}_1 \neq U_1 \, | \, A
    ) \\
    & \stackrel{(a)}{\leq} \mathbf{E}\left[
    \min\left\{
      \frac{P_{U_1}(U_1)  \delta_M(M)}{ 
    P_{U_1 | U_2,Y}(U_1 | U_2,Y)
     {Q} _{U_2}(U_2)  P_M(M) } , 1
    \right\} \,\bigg|\, A   \right]  \\ 
    & \stackrel{(b)}{=} \mathbf{E}\left[
    \min\left\{
      \mathsf{L}\frac{P_{U_1}(U_1)}{ 
    P_{U_1 | U_2,Y}(U_1 | U_2,Y)
     {Q} _{U_2}(U_2)} , 1
    \right\} \,\bigg|\, A   \right]  \\ 
    & \stackrel{(c)}{\leq} 
    \min\bigg\{\mathsf{L}
      \frac{ P_{U_1}(U_1)}{ P_{U_1|U_2,Y} (U_1|U_2,Y)  } \gamma        \bigg(
    \frac{P_{U_2|Y_{\mathrm{r}}}(U_2)}{P_{U_2|Y}(U_2) } + 1
    \bigg)  , 1
    \bigg\}     \\
    & = 
    \min\big\{
    \gamma \mathsf{L}
     2^{-\iota (X;U_2,Y)}
     \big(2^{-\iota (U_2;Y) + \iota (U_2;Y_{\mathrm{r}}) } +1 \big) , 1
    \big\},
\end{align*}
where $(a)$ is by the generalized Poisson matching lemma~\cite{li2021unified} (Lemma \ref{lem::GPML}), $(b)$ is by $\delta_M(M)=1$ and $P_M(M)=1/\mathsf{L}$, and $(c)$ is by step 1) and Jensen's inequality. Taking expectation over $A$ gives the desired error bound. Although the codebooks $\mathbf{U}_1$, $\mathbf{U}_2$ are random (so this is a public-randomness scheme), we can convert it to a deterministic scheme by fixing one particular choice $(\mathbf{u}_1,\mathbf{u}_2)$ that satisfies the error bound.
Alternatively, Theorem \ref{thm:det} allows us to derive bounds for general acyclic discrete networks in a systematic manner, without going through the above arguments for every specific ADN. To prove Corollary \ref{cor::one_shot_relay_eg}, we can invoke Theorem~\ref{thm:det} on the ADN with nodes $1,2,3$, with inputs $Y_i$'s, outputs $X_i$'s, auxiliaries $U_i$'s and the terms $B_{i,j}$'s as follows:
\begin{enumerate}
    \item Node $1$ has input $Y_1 = M$, output $X_1 = X$ and auxiliary $U_1 = (X,M)$. 
    
    \item Node $2$ has input $Y_2 = Y_r$, output $X_2 = X_r$ and auxiliary $U_2 = U$.  
    
    \item Node $3$ has input $Y_3 = Y$, output $X_3 = M$ and decodes with the order ``$U_1, U_2?$''. 
    Applying Theorem~\ref{thm:det} (note that $d_3 =2$ and $d_3' =1$), we have 
    \begin{align*}
        B_{3,1} = (\ln|\mathcal{U}_2| + 1) 
        \mathsf{L} 2^{-\iota (X;U_2,Y)}
         \big(2^{-\iota (U_2;Y) + \iota (U_2;Y_{\mathrm{r}}) } +1 \big),
    \end{align*}
    and hence we obtain the bound \eqref{eq::one_shot_relay} by invoking Theorem \ref{thm:det}.

\end{enumerate} 
\end{proof}

Corollary~\ref{cor::one_shot_relay_eg} yields the following asymptotic achievable rate: 
\begin{align*}
    &R\leq I(X;U, Y) - \max\big\{ I(U; Y_{\mathrm{r}}) - I(U; Y) ,\, 0\big\}
\end{align*}
for some $P_{U|Y_{\mathrm{r}}}$ and function $x_{\mathrm{r}}(y_{\mathrm{r}},u_2)$. 

We also consider a one-shot primitive relay channel (as shown in Figure~\ref{fig::primitive_relay}), where $P_{Y|X,X_{\mathrm{r}}, Y_{\mathrm{r}}}=P_{Y'|X, Y_{\mathrm{r}}} P_{Y''|X_{\mathrm{r}}}$ can be decomposed into two orthogonal components. 
Consider $(X,Y_{\mathrm{r}},Y')$ independent of $(X_{\mathrm{r}},Y'')$ in the ideal distribution and take $U = (U',X_{\mathrm{r}})$ where $U'$ follows $P_{U'|Y_{\mathrm{r}}}$, Corollary~\ref{cor::one_shot_relay_eg} specializes to the following Corollary~\ref{cor::one_shot_primrelay_eg}.

\begin{figure*}[htpb]
	\centering
	\includegraphics[scale=0.4]{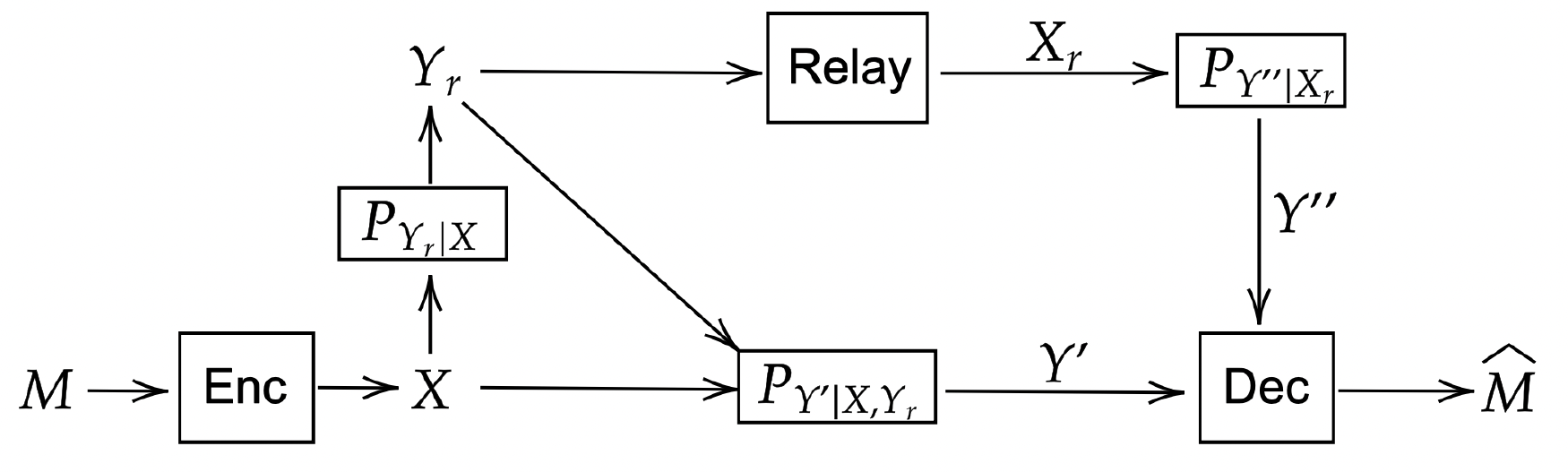}
	\caption{One-shot primitive relay channel setting. } 
    \label{fig::primitive_relay}
\end{figure*}

\begin{corollary}
\label{cor::one_shot_primrelay_eg}
For any $P_{X}$, $P_{X_{\mathrm{r}}}$, $P_{U'|Y_{\mathrm{r}}}$, there is a deterministic coding scheme for the one-shot primitive relay channel with $M\sim\mathrm{Unif}[\mathsf{L}]$ such that the error probability satisfies
\begin{align*} 
    P_e\leq \mathbf{E} \Big[\!
    \min\!\Big\{
    \gamma \mathsf{L} 
    2^{-\iota (X;U',Y')}
    \big(2^{- \iota(X_{\mathrm{r}};Y'') + \iota (U';Y_{\mathrm{r}}|Y') } \! +\! 1 \big) , 1
    \!\Big\}\!
    \Big]\! ,
\end{align*}
where $(X,Y_{\mathrm{r}},U', Y') \sim P_{X} P_{Y_{\mathrm{r}}|X} P_{U'|Y_{\mathrm{r}}}P_{Y'|X,Y_{\mathrm{r}}}$ is independent of $(X_{\mathrm{r}},Y'') \sim P_{X_{\mathrm{r}}} P_{Y''|X_{\mathrm{r}}}$, and $\gamma :=  \ln (|\mathcal{U}'||\mathcal{X}_{\mathrm{r}}|) + 1$.
\end{corollary}

This gives the asymptotic achievable rate $R \le I(X;U',Y') -\max\{ I(U';Y_{\mathrm{r}}|Y') - C_{\mathrm{r}} ,\, 0 \}$ where $C_{\mathrm{r}} = \max_{P_{X_{\mathrm{r}}}}I(X_{\mathrm{r}};Y'')$ is the capacity of the channel $P_{Y''|X_{\mathrm{r}}}$. 
It implies the compress-and-forward bound~\cite{kim2007coding}, which is the maximum of $I(X;U',Y')$ subject to the constraint $C_{\mathrm{r}} \ge I(U';Y_{\mathrm{r}}|Y')$ (where the random variables are distributed as in Corollary~\ref{cor::one_shot_primrelay_eg}). Hence, Corollary~\ref{cor::one_shot_primrelay_eg} can be treated as a one-shot compress-and-forward bound.

\subsection{Partial-Decode-and-Forward Bound}

We extend Corollary~\ref{cor::one_shot_relay_eg} to allow partial decoding of the message~\cite{cover1979capacity,kim2007coding,el2007relay}.
To this end, we split the message and encoder into two. 
The message $M\sim \mathrm{Unif}[\mathsf{L}]$ is split into $M_1\sim \mathrm{Unif}[\mathsf{J}]$ and $M_2\sim \mathrm{Unif}[\mathsf{L}/\mathsf{J}]$ (assume $\mathsf{J}$ is a factor of $\mathsf{L}$).
The encoder controls two nodes (node $1$ and $2$), where node $1$ observes $Y_1 = M_1$, outputs $X_1=V$, and has an auxiliary $U_1=(M_1,V)$; node $2$ observes $Y_2=(M_1,M_2,V)$, outputs $X_2=X$, and has an auxiliary $U_2=(M_1,M_2,X)$.
The relay (node $3$) observes $Y_3 = Y_r$, decodes $U_1$, outputs $X_3 = X_r$, and has an auxiliary $U_3=(M_1,U)$.
The decoder (node $4$) observes $Y_4=Y$ and uses the decoding order ``$U_2,U_3?,U_1?$''. 

\begin{corollary}
\label{cor::pdcf}
Fix any $P_{X, V}$, $P_{U|Y_{\mathrm{r}}, V}$, function $x_{\mathrm{r}}(y_{\mathrm{r}},u,v)$, and $\mathsf{J}$ which is a factor of $\mathsf{L}$.
There exists a deterministic coding scheme for the one-shot relay channel with
\begin{align*}
P_{e} & \leq\mathbf{E}\Big[\min\Big\{\mathsf{J}2^{-\iota(V;Y_{\mathrm{r}})}+\gamma\mathsf{L}\mathsf{J}^{-1}2^{-\iota(X;U,Y|V)} \\
&\;\;\;\;\;\;\;\;\; \cdot\big(2^{-\iota(U;V,Y)+\iota(U;V,Y_{\mathrm{r}})}+1\big)\big(\mathsf{J}2^{-\iota(V;Y)}+1\big),1\Big\}\Big],
\end{align*}
where $(X,V,Y_{\mathrm{r}},U, X_{\mathrm{r}},Y) \sim P_{X, V} P_{Y_{\mathrm{r}}|X,V} P_{U|Y_{\mathrm{r}}, V} \delta_{x_{\mathrm{r}}(Y_{\mathrm{r}},U, V)}  
P_{Y|X,Y_{\mathrm{r}},X_{\mathrm{r}}}$ and $\gamma :=  (\ln (\mathsf{J}|\mathcal{U}|) + 1) ( \ln(\mathsf{J}|\mathcal{V}|) +1)$.
\end{corollary}

Applying the law of large numbers to Corollary~\ref{cor::pdcf}, and Fourier-Motzkin elimination (using the PSITIP software~\cite{li2023automated}), we obtain the following asymptotic achievable rate:
\[
\min\!\left\{\!\!\! \begin{array}{l}
I(V;Y)+I(U,Y;X|V),\\
I(V;Y_{\mathrm{r}})+I(U,Y;X|V),\\
I(V,U;Y)+I(U,Y;X|V)-I(U;Y_{\mathrm{r}}|V),\\
I(V;Y_{\mathrm{r}})+I(U;Y|V)+I(U,Y;X|V)-I(U;Y_{\mathrm{r}}|V)
\end{array}\!\!\!\right\}\! ,
\]
where $(X,V,Y_{\mathrm{r}},U, X_{\mathrm{r}},Y) \sim P_{X, V} P_{Y_{\mathrm{r}}|X,V} P_{U|Y_{\mathrm{r}}, V} \delta_{x_{\mathrm{r}}(Y_{\mathrm{r}},U, V)}  \linebreak[1] P_{Y|X,Y_{\mathrm{r}},X_{\mathrm{r}}}$, 
subject to the constraint $I(U;Y_{\mathrm{r}}|V)\le I(U;Y|V)+I(U,Y;X|V)$. 

Taking $x_\mathrm{r}(y_\mathrm{r},(v',x'_\mathrm{r})) = x'_\mathrm{r}$, $U=\emptyset$, $V $ $=(V',X'_\mathrm{r})$, it gives an achievable rate $\min\{I(X,X_{\mathrm{r}};Y),\,I(V';Y_{\mathrm{r}})+I(X;  Y|X_{\mathrm{r}},V')\}$, recovering the partial noncausal decode-forward bound for relay-with-unlimited-look-ahead~\cite[Prop. 3]{el2007relay}.

Specializing to the primitive relay channel, and again substituting $U=\emptyset$, $V=(V',X'_\mathrm{r})$, $x_\mathrm{r}(y_\mathrm{r},(v',x'_\mathrm{r}))=x'_\mathrm{r}$, we have
\begin{align*}
P_{e} & \leq\mathbf{E}\Big[\min\Big\{\mathsf{J}2^{-\iota(V';Y_{\mathrm{r}})}+2\gamma\mathsf{L}\mathsf{J}^{-1}2^{-\iota(X;Y'|V')} \\
& \;\;\;\;\;\;\;\;\; \cdot\big(\mathsf{J}2^{-\iota(V';Y')-\iota(X_{\mathrm{r}};Y'')}+1\big),1\Big\}\Big],
\end{align*}
where $\gamma :=  (\ln \mathsf{J} + 1) ( \ln(\mathsf{J}|\mathcal{V}'||\mathcal{X}_\mathrm{r}|) +1)$ and $(X,V',Y_{\mathrm{r}},Y') \sim P_{X,V'} P_{Y_{\mathrm{r}}|X} P_{Y'|X,Y_{\mathrm{r}}}$ is independent of $(X_{\mathrm{r}},Y'') \sim P_{X_{\mathrm{r}}} P_{Y''|X_{\mathrm{r}}}$. 
It gives the asymptotic rate $\min\{I(V';Y_{\mathrm{r}})+I(X;Y|V'),\,I(X;Y)+C_{\mathrm{r}}\}$ and recovers the partial decode-forward lower bound for primitive relay channels~\cite{cover1979capacity,kim2007coding}.
One-shot versions of other asymptotic bounds for primitive relay channels (e.g., \cite{mondelli2019new,el2021achievable}) are left for future studies.

\section{Cascade multiterminal source coding with computing} 
\label{sec::cascade}

We consider the cascade multiterminal source coding problem~\cite{cuff2009cascade} (which is also called the \emph{cascade coding for computing} in~\cite[Section 21.4]{el2011network}).
It is similar to the traditional multiterminal source coding problem introduced by Berger and Tung~\cite{berger1978multiterminal, tung1978multiterminal}, where two information sources are encoded in a distributed fashion with loss, though the communication between encoders replaces one of the direct channels to the decoder in the cascade case. 
It can include different variations, e.g., the decoder desires to estimate both $X$ and $Y$, $X$ only, $Y$ only and some functions of both. 
It is tightly related to real problems where it is required to pass messages to neighbors in order to compute functions of data, e.g., distributed data collection, aggregating measurements in sensor networks, interactive coding for computing and distributed lossy averaging (see~\cite{cuff2009cascade} and references therein).

The asymptotic rate-distortion region for the general cascade multiterminal source coding problem is unknown, even for the case where $X$ and $Y$ are independent. 
We study the one-shot setting of this problem, which has not been discussed in literature to the best of our knowledge. 
We provide a novel one-shot bound on the cascade multiterminal source coding problem, and show that our one-shot achievability result recovers the best known asymptotic inner bound, i.e., the \emph{local-computing-and-forwarding inner bound}~\cite{cuff2009cascade} (which in turn recovers various other existing bounds as special cases, also see~\cite[Section 21.4]{el2011network} for a detailed discussion). 

The one-shot cascade multiterminal source coding problem is described as follows (see Figure~\ref{fig::cascade_asym}).   
Consider two sources $X$ and $Y$ that are jointly distributed according to $P_{X,Y}$.  
They will be described by separate encoders and passed to a single decoder in a cascade fashion.
Upon observing $X$, encoder $\mathrm{a}$ sends a message $M\in [\mathsf{L}_1]$ about $X$ to encoder $\mathrm{b}$. 
Encoder $\mathrm{b}$ then creates a final message $M' \in [\mathsf{L}_2]$ summarizing both sources $X$ and $Y$ and sends it to the decoder. 
We investigate the error probability $P_e$, which is the probability of the decoder recovers $\tilde{Z} \in \mathcal{Z}$ with excess distortion $P_e := \mathbf{P}\{d(X,Y,\tilde{Z}) > \mathsf{D} \}$, where $d: \mathcal{X} \times  \mathcal{Y}\times  \mathcal{Z}\rightarrow \mathbb{R}_{\geq 0}$ is a distortion measure. 
Due to the flexibility of the distortion function $d$, in general one can estimate any function of $X$ and $Y$ to tackle various objectives in practice~\cite{cuff2009cascade}. 

\begin{figure}[htpb]
    \centering
    \includegraphics[scale=0.3]{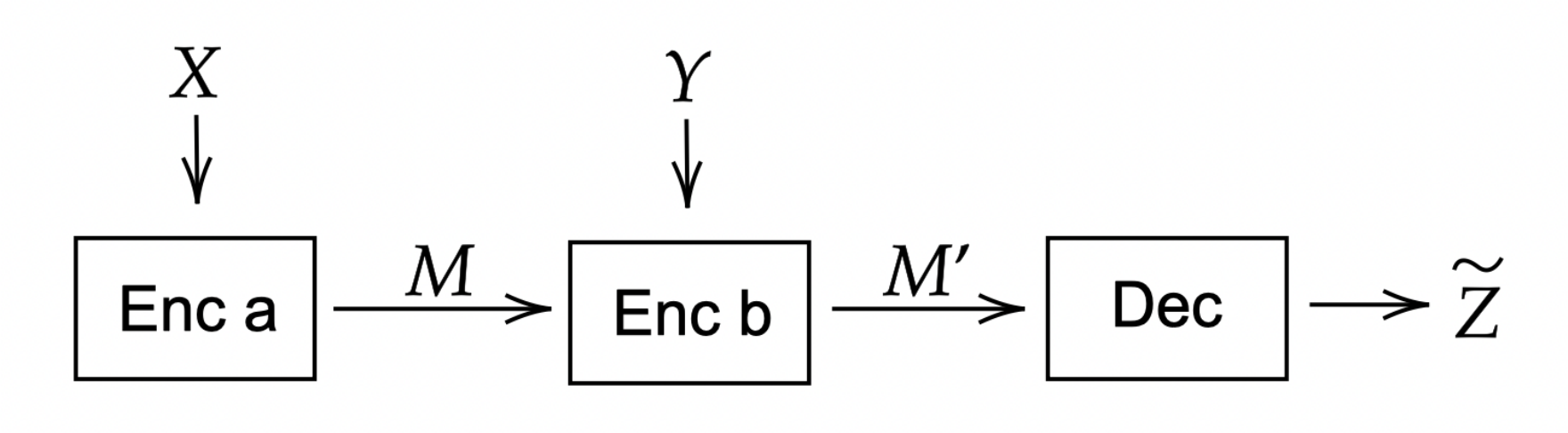}
    \caption{One-shot cascade multiterminal source coding setting. } 
    \label{fig::cascade_asym}
\end{figure}

By Theorem~\ref{thm:det}, we bound $P_e$ by the following corollary.

\begin{corollary}
\label{cor::cascade}
    Fix $P_{X,Y}$, $P_{U,V|X}$, function $z: \mathcal{U} \times \mathcal{V} \times \mathcal{Y} \rightarrow \mathcal{Z}$ and $\tilde{\mathsf{L}}_i, i=1,2,3$ with $\tilde{\mathsf{L}}_1 \tilde{\mathsf{L}}_2 \leq \mathsf{L}_1$ and $\tilde{\mathsf{L}}_2 \tilde{\mathsf{L}}_3 \leq \mathsf{L}_2$, there exists a deterministic coding scheme for the one-shot cascade multiterminal source coding problem such that the probability of excess distortion is bounded by
    \begin{align*}
        P_e 
        & \leq \mathbf{E}\bigg[\min\Big\{ 
        \mathbf{1}\{d(X, Y, Z)>\mathsf{D}\}  + 
        \gamma \tilde{\mathsf{L}}_1^{-1}\tilde{\mathsf{L}}_2^{-1} 2^{\iota(U,V; X | Y)} \\
        & \qquad \qquad \qquad + \gamma \tilde{\mathsf{L}}_1^{-1} 2^{-\iota(V; U,Y) + \iota(V; U,X)} + 
        \tilde{\mathsf{L}}_2^{-1} 2^{- \iota(U; V, Y) + \iota(U; X)}
         \\
        &\qquad\qquad\qquad  + \gamma \tilde{\mathsf{L}}_3^{-1} 2^{\iota(Z;V,Y|U)} 
        \left(\tilde{\mathsf{L}}_2^{-1} 2^{\iota(U; X)} + 1 \right)  
        , \,\, 1
        \Big\}\bigg],
    \end{align*}
where $\gamma = \ln(|\mathcal{U}| \tilde{\mathsf{L}}_2 )+1$ and $X,Y,Z,U,V\sim P_{X}  P_{Y|X}  P_{U,V|X} P_{Z|Y,U,V}$. 
\end{corollary}

\begin{proof}

We adapt the problem into our ADN framework by splitting the encoder $\mathrm{a}$, as shown in Figure~\ref{fig::cascade} (see next page). 
    The encoder $\mathrm{a1}$, encoder $\mathrm{a2}$, encoder $\mathrm{b}$ and decoder are referred to as nodes $1,2,3,4$, respectively.  
    
    \begin{figure*}[htpb]
    \centering \includegraphics[scale=0.42]{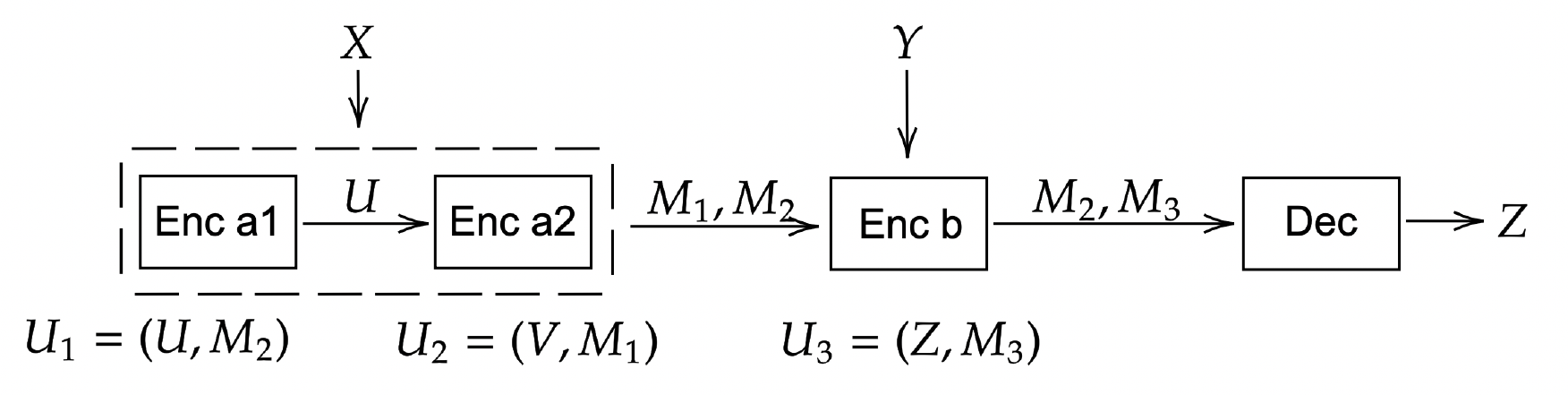}
    \caption{One-shot cascade multiterminal source coding in AND framework by splitting the first encoder. } 
    \label{fig::cascade}
    \end{figure*}
    
    Let $M_i\in [\mathsf{L}_i]$ for $i=1,2,3$. 
    Encoder $\mathrm{a1}$ (node $1$) observes $X$, outputs $U$, and has an auxiliary $U_1 = (U, M_2)$. 
    Encoder $\mathrm{a2}$ (node $2$) observes $(U,X)$, outputs $(M_1, M_2)$, and has an auxiliary $U_2 = (V, M_1)$. 
    Encoder $\mathrm{b}$ (node $3$) observes $M_1$, $M_2$ and $Y$, outputs $(M_2, M_3)$, and has an auxiliary $U_3 = (Z, M_3)$. 
    The decoder observes $M_2, M_3$ and recovers $Z$ by using the function $z$ and our coding scheme. 
    
    For each node $i=1,\ldots,4$ in the ADN, we describe its input $Y_i$, output $X_i$, auxiliary $U_i$ and the terms $B_{i,j}$ in Theorem~\ref{thm:det} as follows:
    \begin{enumerate}
        \item Node $1$ has input $Y_1 = X$, output $X_1 = U$ and auxiliary $U_1 = (U,M_2)$.
        
        \item Node $2$ has input $Y_2 = (U,X)$, output $X_2 = M_1$ and auxiliary $U_2 = (V,M_1)$. 
        
        \item Node $3$ has input $Y_3 = (Y, M_1, M_2)$, output $X_3 = M_3$, auxiliary $U_3 = (Z,M_3)$, and decodes with the order ``$U_2, U_1$'' (i.e., ``$U_{\mathrm{Enc}\, 1b}, U_{\mathrm{Enc}\, \mathrm{1a}}$''). 
        We have $d_3' = d_3 = 2$, and
        \begin{align*}
            B_{3, 1} & = \gamma \tilde{\mathsf{L}}_1^{-1}\tilde{\mathsf{L}}_2^{-1} 2^{\iota(U,V; X | Y)} + \gamma \tilde{\mathsf{L}}_1^{-1} 2^{-\iota(V; U,Y) + \iota(V; U,X)}, \\
            B_{3, 2} & = \tilde{\mathsf{L}}_2^{-1} 2^{- \iota(U; V, Y) + \iota(U; X)},
        \end{align*} 
        where $\gamma = \ln(|\mathcal{U}| \tilde{\mathsf{L}}_2 )+1$. 
        
        \item Node $4$ has input $Y_4 = (M_2, M_3)$, output $X_4 = Z$, and decodes with the order ``$U_3, U_1 ?$'' (i.e., ``$U_{\mathrm{Enc}\, 2}, U_{\mathrm{Enc}\, \mathrm{1a}} ?$''). 
        By applying Theorem~\ref{thm:det} (note that $d_4' = 1$ and $d_4 = 2$), it gives 
        \begin{equation*}
            B_{4, 1} = \gamma \tilde{\mathsf{L}}_3^{-1} 2^{\iota(Z;V,Y|U)} 
            \left(\tilde{\mathsf{L}}_2^{-1} 2^{\iota(U; X)} + 1 \right),
        \end{equation*}
        where $\gamma = \ln(|\mathcal{U}| \tilde{\mathsf{L}}_2 )+1$. 
        
    \end{enumerate}

     Therefore, by applying Theorem~\ref{thm:det}, the probability of distortion exceeds the limit $P_e := \mathbf{P}\{d(X, Y, \tilde{Z}) > \mathsf{D} \}$ can be  bounded as the result stated in this corollary: 
    \begin{align*}
    P_e 
    & \leq \mathbf{E}\bigg[\min\Big\{ \mathbf{1}\{d(X, Y, Z)>\mathsf{D}\} + 
    B_{3,1} + B_{3,2} + B_{4,1}
    , 1\Big\}\bigg] \\
    & = \mathbf{E}\bigg[\min\Big\{ 
    \mathbf{1}\{d(X, Y, Z)>\mathsf{D}\}  + 
    \gamma \tilde{\mathsf{L}}_1^{-1}\tilde{\mathsf{L}}_2^{-1} 2^{\iota(U,V; X | Y)} \\
    & \qquad \qquad \qquad + \gamma \tilde{\mathsf{L}}_1^{-1} 2^{-\iota(V; U,Y) + \iota(V; U,X)}  + 
    \tilde{\mathsf{L}}_2^{-1} 2^{- \iota(U; V, Y) + \iota(U; X)}
     \\
    &\qquad\qquad\qquad  + \gamma \tilde{\mathsf{L}}_3^{-1} 2^{\iota(Z;V,Y|U)} 
    \left(\tilde{\mathsf{L}}_2^{-1} 2^{\iota(U; X)} + 1 \right)  
    , \,\, 1
    \Big\}\bigg],
    \end{align*}
    where $\gamma = \ln(|\mathcal{U}| \tilde{\mathsf{L}}_2 )+1$. 
\end{proof}

Following the one-shot bound as shown above, we let $\tilde{\mathsf{L}}_i = 2^{n \tilde{R}_i}$ for $i=1,2,3$ and apply the law of large numbers. We obtain the asymptotic achievable region
\begin{align*}
    \tilde{R}_1 + \tilde{R}_2 & > I(X;U,V|Y), \\
    \tilde{R}_1               & > I(V; U, X) - I(V; U, Y), \\
    \tilde{R}_2               & > I(X;U) -  I(V, Y; U), \\
    \tilde{R}_2 + \tilde{R}_3 & > I(X;U) + I(Z;V,Y|U), \\
    \tilde{R}_3               & > I(Z;V,Y|U), 
\end{align*}
and $ D > \mathbf{E}[d(X,Y,Z)]$. 
By $\tilde{\mathsf{L}}_1 \tilde{\mathsf{L}}_2 \leq \mathsf{L}_1$ and $\tilde{\mathsf{L}}_2 \tilde{\mathsf{L}}_3 \leq \mathsf{L}_2$ and consider $R_1 = \tilde{R}_1 + \tilde{R}_2$ and $R_2 = \tilde{R}_2 + \tilde{R}_3$, by applying the Fourier-Motzkin elimination (using the PSITIP software~\cite{li2023automated}), we recover an asymptotic achievable region for the cascade multiterminal source coding problem: 
\begin{align*}
    R_1 & > I(X;U,V|Y), \\
    R_2 & > I(X;U)+ I(Z;V,Y|U),
\end{align*}
and $D > \mathbf{E}[d(X,Y,Z)]$, where $X,Y,Z,U,V\sim P_{X}  P_{Y|X}  P_{U,V|X} P_{Z|Y,U,V}$. 
This is the local-computing-and-forwarding inner bound in asymptotic case, as discussed in~\cite{cuff2009cascade} and also in~\cite[Section 21.4]{el2011network}.

\section{Examples of Acyclic Discrete Networks}
\label{sec::NIT}

In this section, to demonstrate the use of our main results, we apply  Theorem~\ref{thm::network_achievability} and Theorem~\ref{thm:det} on several settings in network information theory:  Gelfand-Pinsker~\cite{gelfand1980coding,Heegard1980}, Wyner-Ziv~\cite{wyner1976rate, wyner1978rate}, coding for computing~\cite{yamamoto1982wyner}, multiple access channels~\cite{ahlswede1971multi, liao1972multiple, ahlswede1974capacity} and broadcast channels~\cite{marton1979coding}, recovering similar results as the results in~\cite{li2021unified} and other works.

\subsection{Gelfand-Pinsker Problem}
The one-shot version of the Gelfand-Pinsker problem~\cite{gelfand1980coding} is described as follows. 

Upon observing $M\sim \mathrm{Unif}[\mathsf{L}]$ and $S\sim P_S$, the encoder generates $X$ and sends $X$ through a channel $P_{Y|X,S}$. 
The decoder receives $Y$ and recovers $\hat{M}$. 
This can be considered as an ADN as follows (see Figure~\ref{fig::gp} for an illustration): 
in the ideal situation, let $Y_1 := (M,S)$ represent all the information coming into node $1$, $Y_2 := Y$, $P_{Y_2|Y_1, X_1}$ be $P_{Y|S, X}$, and $X_2 := M$.  
The auxiliary of node $1$ is $U_1 = (U,M)$ for some $U$ following $P_{U|S}$ given $S$. The decoding order of node $2$ is ``$U_1$'' (i.e., it only wants $U_1$). Since node $2$ has decoded $U_1$, $X_2$ is allowed to depend on $U_1=(U,M)$, and hence the choice $X_2 := M$ is valid in the ideal situation. Nevertheless, in the actual situation where we have $\tilde{X},\tilde{Y}$ instead of $X,Y$, the actual output $\tilde{X}_2$ will not be exactly $M$, though the error probability $P_e := \mathbf{P}(\tilde{X}_2 \neq M)$ can still be bounded. 
Applying Theorem~\ref{thm:det}, we obtain the following Corollary~\ref{cor::GP}.

\begin{figure}[htpb]
    \centering
    \includegraphics[scale=0.33]{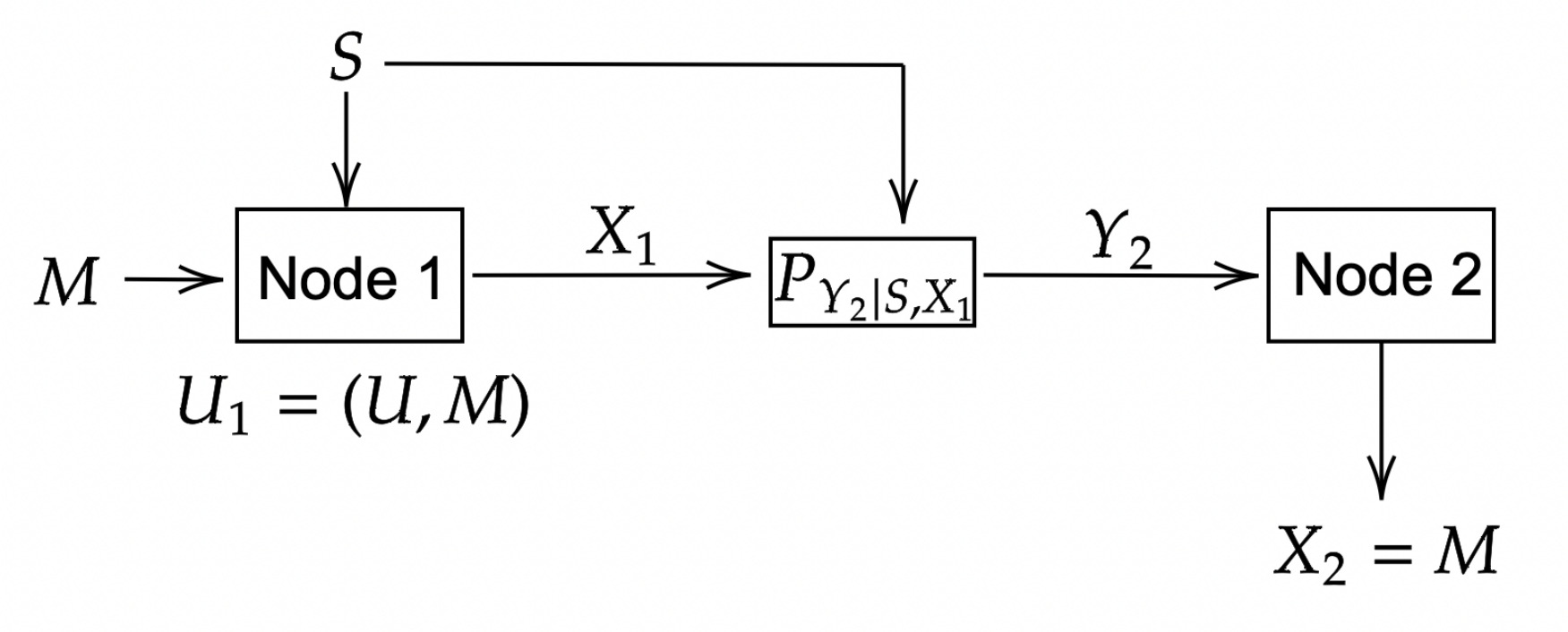}
    \caption{Gelfand-Pinsker problem in ADN framework.} 
    \label{fig::gp}
\end{figure}

\begin{corollary}
\label{cor::GP}
Fix $P_{U|S}$ and function $x:\mathcal{U}\times \mathcal{S}\rightarrow \mathcal{X}$. 
There exists a deterministic coding scheme for the channel $P_{Y|X,S}$ with state $S\sim P_S$ and message $M\sim \mathrm{Unif}[\mathsf{L}]$ such that
\begin{equation*}
	P_e \leq \mathbf{E}\big[  \min\big\{ 
	\mathsf{L}  2^{-\iota (U;Y) + \iota (U;S)}, 1 \big\}
	  \big], 
\end{equation*}
where $S,U,X,Y\sim P_S P_{U|S} \delta_{x(U,S)} P_{Y|X, S}$. 

\end{corollary}

This bound is similar to the one given in~\cite{li2021unified} (which is stronger than the one-shot bounds in~\cite{verdu2012non, yassaee2013technique, watanabe2015nonasymptotic} in the second order). 
Both of them attain the second-order result in~\cite{scarlett2015dispersions}.

\subsection{Wyner-Ziv Problem and Coding for Computing}
The Wyner-Ziv problem~\cite{wyner1976rate, wyner1978rate} in a one-shot setting is described as follows (see Figure~\ref{fig::wz} for an illustration).

Upon observing $X \sim P_X$, the encoder outputs $M \in [\mathsf{L}]$. 
The decoder receives $M$ and the side information $T \sim P_{T|X}$, and recovers $\tilde{Z}\in \mathcal{Z}$ with probability of excess distortion $P_e :=\mathbf{P} \{d(X,\tilde{Z})>\mathsf{D} \}$, where $d: \mathcal{X} \times\mathcal{Z} \rightarrow \mathbb{R}_{\geq 0} $ is a distortion measure. 
This can be considered as an ADN: in the ideal situation, $Y_1:= X$, $X_1 := M$, $Y_2 :=(M,T)$, $X_2 := Z$.
The auxiliary of node $1$ is $U_1 = (U,M)$ for some $U$ following $P_{U|X}$ given $X$. 
By Theorem~\ref{thm:det}, we bound  $P_e$ as follows.

\begin{figure}[htpb]
    \centering
    \includegraphics[scale=0.28]{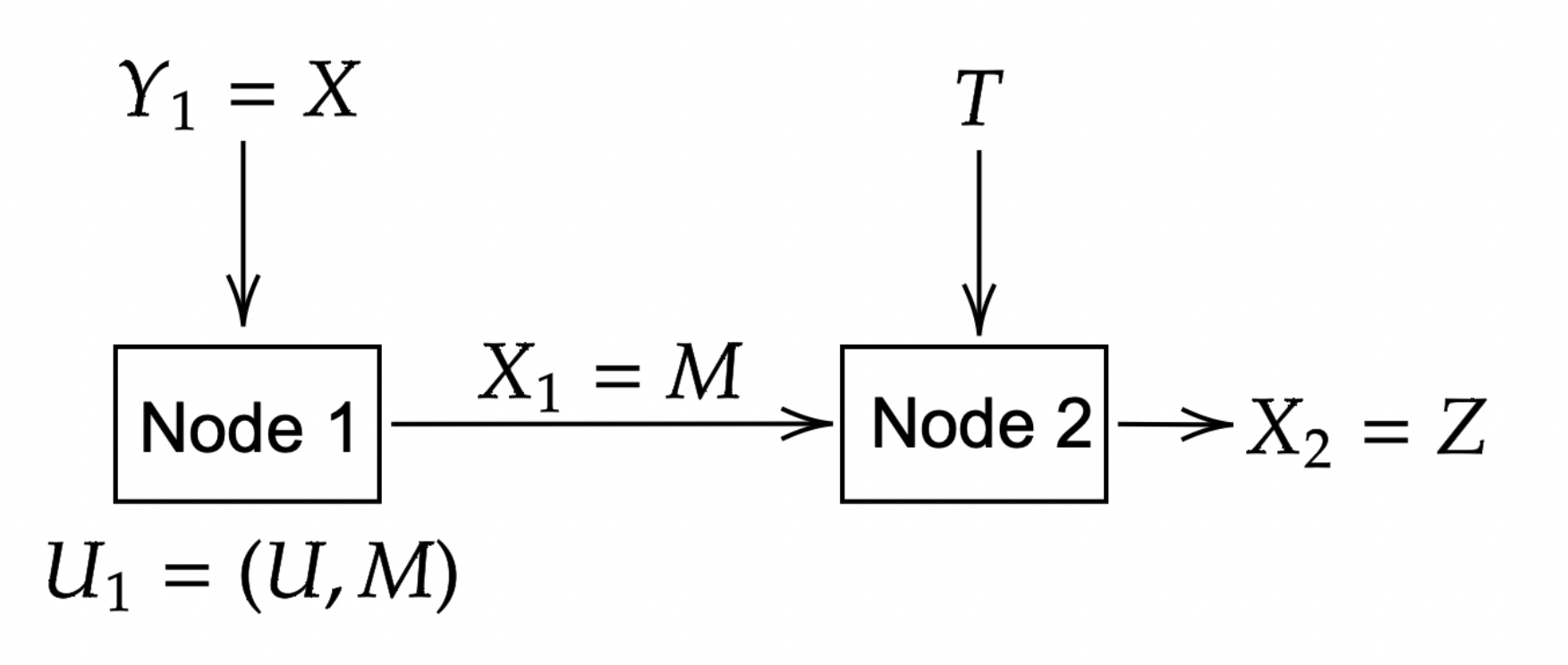}
    \caption{Wyner-Ziv problem in ADN framework.} 
    \label{fig::wz}
\end{figure}

\begin{corollary}
\label{cor::WZ}
Fix $P_{U|X}$ and function $z: \mathcal{U}\times \mathcal{Y}\rightarrow \mathcal{Z}$. 
There exists a deterministic coding scheme  for lossy source coding with source $X\sim P_X$, side information at the decoder $T\sim P_{T|X}$ and description $M\in [\mathsf{L}]$ such that 
\begin{equation}
	P_e \leq \mathbf{E}\Big[  \min\Big\{  
	\mathbf{1}\{d(X, Z)>\mathsf{D}\} + \mathsf{L}^{-1}  2^{-\iota (U;T) + \iota (U;X)}, 1 \Big\}
	  \Big], \label{eq:wyner_ziv}
\end{equation}
where $X,Y,U,Z \sim P_X P_{Y|X}  P_{U|X} \delta_{z(U,Y)}$. 
\end{corollary}

This bound is similar to, though slightly weaker than, the bound given in~\cite{li2021unified} (which improves upon the one-shot bounds in~\cite{verdu2012non,watanabe2015nonasymptotic} in second-order performance). 
Our main contribution lies in the generality of our one-shot coding framework, and we do not always derive bounds identical to those in~\cite{li2021unified}, despite both methods employing Poisson functional representations.

This reduces to lossy source coding with $T=\emptyset$. 
Let $U=Z$, we have $P_e \leq \mathbf{P}(d(X, Z)>\mathsf{D}) + \mathbf{E}\left[ \min\left\{\mathsf{L}^{-1} 2^{\iota (Z;X)}, 1 \right\}\right]$.  

We also consider coding for computing~\cite{yamamoto1982wyner}, where node $2$ recovers a function $f(X, T)$ of $X$ and $T$ with respect to distortion level $\mathsf{D}$ with a distortion measure $d(\cdot, \cdot)$. 
The probability of excess distortion is $P_e := \mathbf{P}\{d(f(X,T),\tilde{Z})>\mathsf{D} \}$.
We obtain a result similar to Corollary \ref{cor::WZ}, where \eqref{eq:wyner_ziv} is changed to \begin{align*}
    P_e 
    & \leq \mathbf{E}\big[  \min \big\{ \mathbf{1}\{d(f(X,T),Z)>\mathsf{D}\} + \mathsf{L}^{-1}  2^{-\iota (U;T) + \iota (U;X)}, \,\, 1 \big\}\big].
\end{align*}

\subsection{Multiple Access Channel}
The multiple access channel~\cite{ahlswede1971multi, liao1972multiple, ahlswede1974capacity} in a one-shot setting is described as follows.

There are two encoders, one decoder, and two independent messages $M_j\sim \mathrm{Unif}[\mathsf{L}_j]$ for $j=1,2$.
Encoder $j$ observes $M_j$ and creates $X_j$ for $j=1,2$. 
The decoder observes the output $Y$ of the channel $P_{Y |X_1, X_2}$ and produces the reconstructions $(\hat{M}_1, \hat{M}_2)$ of the messages. 
The error probability is defined as $P_e := \mathbf{P} \{(M_1, M_2)\neq (\hat{M}_1, \hat{M}_2) \}$. 
To consider this as an ADN, in the ideal situation, we let $Y_1 := M_1$, $Y_2 := M_2$, $Y_3 := Y$ and $X_3 := (M_1, M_2) $. 
We let $U_1 := (X_1, M_1)$ and $U_2 := (X_2, M_2)$. 
The decoding order of node $3$ is 
``$U_2, U_1$'' (i.e., decode $U_1$ (soft), and then $U_2$ (unique), and then $U_1$ (unique)). 
By Theorem~\ref{thm:det}, we have the following result. 

\begin{figure}[htpb]
    \centering
    \includegraphics[scale=0.3]{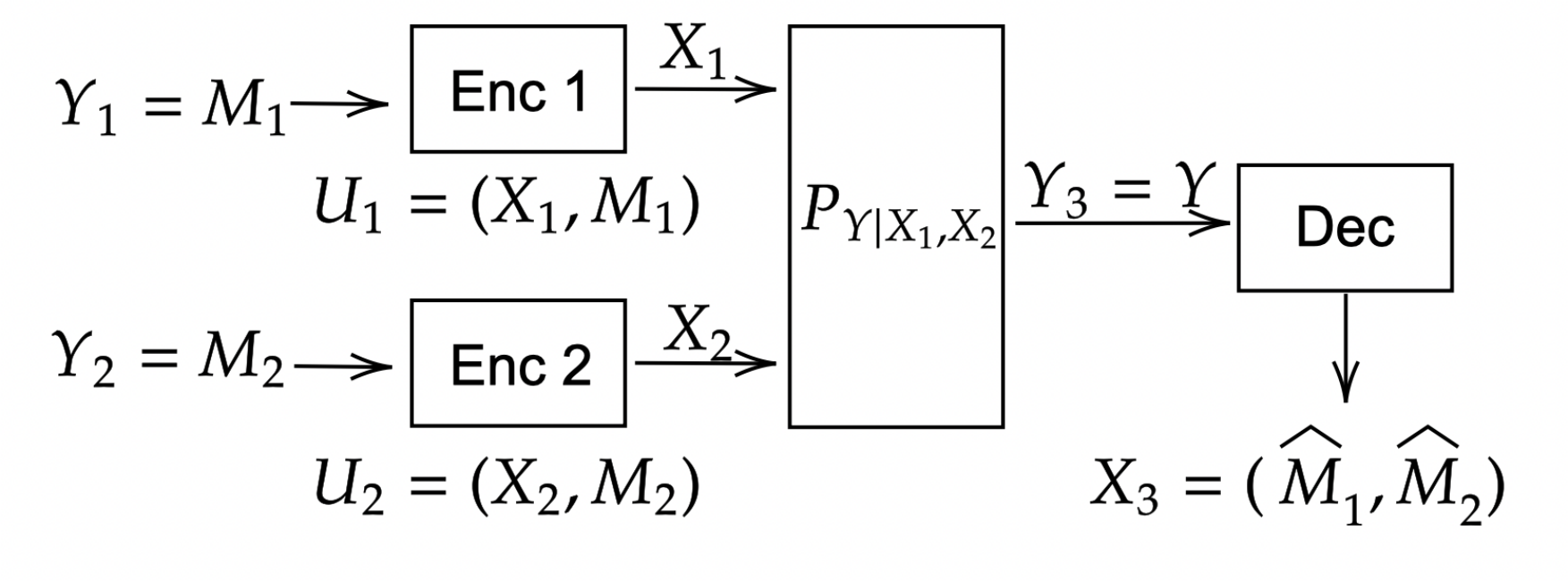}
    \caption{Multiple access channel in ADN framework.} 
    \label{fig::wz}
\end{figure}

\begin{corollary}
\label{cor::MAC_2}
    Fix $P_{X_1}, P_{X_2}$. 
    There exists a deterministic coding scheme for the multiple access channel $P_{Y |X_1, X_2}$ for messages $M_j\sim \mathrm{Unif}[1 : \mathrm{L}_j]$ for $j=1,2$ such that 
\begin{align*}
    P_e  
    \leq  \mathbf{E}\Big[  \min \Big\{ \gamma \mathsf{L}_1 \mathsf{L}_2  2^{-\iota (X_1,X_2;Y)} + \gamma \mathsf{L}_2 2^{-\iota(X_2;Y|X_1)} +  \mathsf{L}_1
    2^{-\iota(X_1;Y|X_2)} ,1\Big\}
    \Big], 
\end{align*}
    where $\gamma := \ln (\mathsf{L}_1 |\mathcal{X}_1|)+1$, $(X_1, X_2, Y) \sim  P_{X_1} P_{X_2} P_{Y|X_1, X_2}$. 
\end{corollary}

This bound is similar to the one-shot bounds in~\cite{li2021unified,verdu2012non}. 
In the asymptotic setting, this will give the region $R_1 <I(X_1;Y|X_2) $, $ R_2 < I(X_2;Y|X_1)$, $ R_1 +R_2 < I(X_1,X_2;Y)$.

\subsection{Broadcast Channel with Private Messages}
The broadcast channel with private messages~\cite{marton1979coding} in a one-shot setting is described as follows. 

Upon observing independent messages $M_j\sim \mathrm{Unif}[\mathsf{L}_j]$ for $j=1,2$, the encoder produces $X$ and sends it through a channel $P_{Y_1,Y_2|X}$. 
Decoder $j$ observes $Y_j$ and reconstructs $\hat{M}_j$ for $j=1,2$. 
By Theorem~\ref{thm:det}, we have the following result. 

\begin{corollary}
\label{cor::Broadcast_2}
Fix any $P_{U_1, U_2}$ and function $x: \mathcal{U}_1 \times \mathcal{U}_2 \rightarrow \mathcal{X}$. 
There exists a deterministic coding scheme for the broadcast channel $P_{Y_1, Y_2 |X}$  for independent messages $M_k \sim \mathrm{Unif}[\mathsf{L}_j]$ for $j = 1,2$, with the error probability bounded by 
\begin{align*}
    P_e  \leq \mathbf{E}\Big[\min\Big\{ \mathsf{L}_1  2^{-\iota (U_1; Y_1) } + \mathsf{L}_2  2^{-\iota (U_2;Y_2) + \iota (U_1; U_2)}, 1 \Big\} \Big],
\end{align*} 
where $(U_1,U_2, X, Y_1, Y_2)\sim P_{U_1, U_2} \delta _{x(U_1,U_2)} P_{Y_1, Y_2|X}$. 
\end{corollary}

In the asymptotic case, this gives a corner point in Marton's region~\cite{marton1979coding}: $R_1 < I(U_1;Y_1), R_2 < I(U_2;Y_2)-I(U_1;U_2)$. 
Another corner point can be obtained by swapping the decoders.


\chapter{One-Shot Coding on Secrecy Problems with Channel Uncertainties} 
\label{chp:hiding}

\section{Overview}
\label{sce::intro}

In this chapter, we study two fundamental information theory problems in the one-shot regime, namely the information hiding problem~\cite{moulin2003information} and the compound wiretap channel~\cite{liang2009compound}. 
The former concerns active attacks during information transmission, while the latter addresses passive eavesdropping and information leakage. 
These two problems have become crucial in the era of data science, where secrecy and privacy are increasingly important due to the growing dependence on and reliance upon large amounts of communicated, analyzed, and utilized data, which inherently contain sensitive and personal information.
This chapter is partially based on~\cite{liu2024hiding}.

For the information hiding problem, \cite{moulin2003information} formulated it as a communication system from a game-theoretic perspective, where an encoder-decoder team seeks to transmit a confidential message \emph{embedded} in a host data source, while the opposing side is an attacker, modeled as a noisy channel, attempts to destroy or degrade the message. 
The information-theoretic limits of different variations of information hiding have been extensively investigated over the past two decades~\cite{somekh2003error, somekh2004capacity, cohen2002gaussian}, due to its wide range of applications, including watermarking, fingerprinting, steganography, and copyright protection. 
Existing analyses of information hiding problems borrow techniques from various fields, including wireless communication, signal processing, cryptography, and game theory.

For the compound wiretap channel, \cite{liang2009compound} modeled the problem as a generalization of Wyner's wiretap channel setting~\cite{wyner1975wire}, where the communication channel can take multiple potential states.
The objective is to ensure reliable transmission and minimize information leakage regardless of which state occurs.
This model is more general and better suited to practical scenarios where the transmitter may not have knowledge of the channel conditions or where channel characteristics change rapidly, yet communication performance must still be guaranteed.

In all existing studies on these two problems, the information-theoretic limits have been analyzed in the \emph{asymptotic} regime, assuming that the signal has a blocklength approaching infinity.  
However, this assumption does not hold in practice, as packets have bounded lengths, which can be quite short in many applications~\cite{ji2018ultra}. 
Similar to Chapter~\ref{chp:nnc}, we study the one-shot achievability results of a generalized information hiding setting and the compound wiretap channel, where the channel or source is arbitrary and used only \emph{once}, the law of large numbers does not apply and conventional typicality-based tools are inapplicable.

Most of the existing asymptotic analyses on the two problems~\cite{moulin2003information, liang2009compound} assume the decoder know the channel condition (in information hiding, the attacker's strategy), since  when the blocklength is large, one can utilize training symbols at the beginning of transmission, whose size becomes negligible compared to the blocklength. 
However, this assumption is questionable in the one-shot case (see Section~\ref{subsec::disc} for discussions).\footnote{Note that this assumption was also removed in~\cite{somekh2004capacity}, although that work assumes the side information is an independent shared key of unlimited size and is chosen as part of the coding scheme, whereas in our information hiding setting, we allow it to be correlated with the host and fixed, as in~\cite{moulin2003information}.} 
We only assume the attack channel belongs to a set (which may have infinite cardinality). 
Our goal is to provide \emph{distributionally robust} coding strategies (also see~\cite{malik2024distributionally}) for the two problems, by utilizing a classical covering argument~\cite{blackwell1959capacity} to handle the uncertainties of channels. 
To the best of the authors' knowledge, our results are novel and have not been studied in the literature.

This chapter is organized as follows.
We begin with a literature review on information hiding, watermarking, compound wiretap channels, and one-shot information theory in Section~\ref{sec::review}.
Next, we present the one-shot generalized information hiding problem in Section~\ref{sec::generalized_info_hiding} and recover existing asymptotic results in Section~\ref{sec::recovery_hiding}.
Within the same framework, we study the one-shot compound wiretap channel in Section~\ref{sec::wiretap}.

\section{Related Work}
\label{sec::review}

We review related literature on the information hiding and the compound wiretap channel in this section.

\subsection{Information Hiding}
\label{subsec::review_infohiding}

The information hiding problem has been studied since~\cite{cohen2002gaussian, moulin2003information, somekh2003error, somekh2004capacity}, due to its wide range of applications, including watermarking, fingerprinting, audio/image/video processing, copyright protection, and steganography.
The goal is to \emph{hide} a message into some host signal (by introducing a certain level of distortion), so that the message can be correctly reconstructed after suffering \emph{attacks} (which introduce another level of distortion).
This problem was modeled as a communication problem, and asymptotic information-theoretic capacity was derived in~\cite{moulin2003information}.
In general, information hiding is closely related to the Gelfand–Pinsker problem~\cite{gelfand1980coding,heegard1983capacity}, and various extensions have been studied, e.g., the case where side information is available to the encoder, decoder, and adversary~\cite{moulin2007capacity}, and the case where the decoder has rate-limited side information~\cite{steinberg2008coding}.
See~\cite{barron2003duality} for its duality with the Wyner–Ziv problem, and~\cite{keshet2008channel} for a comprehensive survey.
We discuss its applications and related settings with different objectives as follows.

\subsubsection{\textbf{Watermarking, Fingerprinting, and Steganography}} 

The setting in~\cite{moulin2003information} can be viewed as \emph{public} watermarking~\cite{somekh2004capacity}, where the host signal is available only at the encoder. In contrast, when it is also available at the decoder, \emph{private} watermarking has been studied in~\cite{somekh2003error, cox1997secure}. In the Gaussian case, public and private watermarking have the same capacity~\cite{cohen2002gaussian}, but this is not true in general. Watermarking problems consider messages containing personal identification information to be protected from attacks, but secrecy is not always required. In comparison, \emph{digital fingerprinting}~\cite{boneh1998collusion, moulin2003information} embeds fingerprints into the host data to uniquely identify users for tracing illegal data usage, which can be more challenging due to potential collusion. A provably good data embedding strategy was introduced by~\cite{chen2001quantization}. Random coding error exponents have been investigated in these problems~\cite{moulin2007capacity, merhav2000random, moulin2008universal}. Although~\cite[Sec. VII.C]{moulin2003information} indicated the applicability of information hiding to steganography, the discussion was later extended by~\cite{moulin2004new, wang2008perfectly} to the capacity of perfectly secure steganographic systems. Other steganographic code designs include using trellis codes~\cite{guan2022double} and polar codes~\cite{li2020designing}.

\subsubsection{\textbf{Host, Stegotext, and Reversibility}} 

In the conventional information hiding setting~\cite{moulin2003information}, the message is embedded into host data by producing an encoded signal (``stegotext''), with the goal of recovering the message only. Other objectives have been considered later, such as conveying the host~\cite{kim2008state} or reconstructing the stegotext~\cite{grover2015information, xu2023information}. \emph{Reversible} information embedding has also been investigated~\cite{kalker2002capacity, steinberg2006reversible, sumszyk2009information}, where the host signal needs to be decoded. However, this can incur a high cost when the host has high entropy~\cite{kalker2002capacity}, making perfect reversibility even impossible for continuous host signals~\cite{sumszyk2009information}. Nevertheless, in practice, the goal is often to enable retransmission of the stegotext, and codes for stegotext recovery have been studied~\cite{grover2015information, xu2023information}. 
For this setting, single-letter capacity-distortion tradeoffs are known only for logarithmic distortion~\cite{kim2008state} and quadratic distortion in the Gaussian case~\cite{sutivong2005channel}.

\subsection{Compound Wiretap Channels}
\label{subsec::cpd_wiretap}

Compound wiretap channels~\cite{liang2009compound} generalize the wiretap channel model by Wyner~\cite{wyner1975wire} by allowing both the legitimate channel and the eavesdropper's channel to have multiple possible states. The objective is to guarantee reliable and secure signal transmission regardless of which state occurs. This is a practical model for channel uncertainty, where the transmitter may have no knowledge of the channel (due to the dynamic nature of the wireless medium or unavoidable implementation/estimation inaccuracies), but zero performance outage is still required (e.g., for ultra-reliable communications~\cite{ji2018ultra}). 
\cite{liang2009compound} proposed achievable and converse results, with the converse bounds shown to be tight in certain cases by~\cite{bjelakovic2013secrecy}. They also studied the achievable secrecy degrees of freedom (s.d.o.f.) region for a multi-input multi-output (MIMO) model, which was later extended to the case of two confidential messages in~\cite{kobayashi2009compound}. The s.d.o.f. of compound wiretap parallel channels was also studied in~\cite{liu2008secrecy}. See~\cite{ekrem2010gaussian, ekrem2012degraded} for discussions on Gaussian MIMO compound wiretap channels.

In~\cite{liang2009compound, bjelakovic2013secrecy}, the results focused on discrete memoryless channels with a countably finite uncertainty set (i.e., the set from which the exact channel realization is drawn). This was later extended to \emph{arbitrary uncertainty sets}, including continuous alphabets, by~\cite{schaefer2015secrecy}, which is also one of our contributions. Moreover,~\cite{boche2015continuity} showed that the secrecy capacity is a continuous function of the uncertainty set.

\section{One-shot Generalized Information Hiding}
\label{sec::generalized_info_hiding}

In this section, we formulate the generalized one-shot information hiding problem, which is more general than the one-shot information hiding setting in~\cite{liu2024hiding} and can be seen as a natural generalization of~\cite{liu2024hiding}, \cite{xu2023information}, and \cite{moulin2007capacity}.  
More specifically, as an extension of our conference version~\cite{liu2024hiding}, we adopt the idea from~\cite{moulin2007capacity} that considers the side information available at both the encoder and the decoder.\footnote{In~\cite{moulin2007capacity}, it was assumed that there were three sources of side information, available at the encoder, the attacker, and the decoder, respectively. 
We model this scenario by considering  sources of side information available at the encoder and the decoder, together with the attack channel being in an unspecified set, since the decoder has no knowledge on the attacker. 
We can recover the Gelfand-Pinsker coding~\cite{gelfand1980coding, heegard1983capacity} by letting $\mathcal{A} = \{A_{Y|X}\}$ be a singleton set. } 
Moreover, similar to~\cite{xu2023information}, we require the decoder to not only reconstruct the message $M$, but also the stegotext $X$.

Our generalized information hiding problem recovers many existing settings as special cases, including the conventional information hiding problem~\cite{moulin2003information, liu2024hiding, somekh2003error}, the information embedding with stegotext reconstruction problem~\cite{xu2023information, grover2015information}, the conventional Gelfand-Pinsker coding~\cite{gelfand1980coding,heegard1983capacity}, the generalized Gelfand–Pinsker family~\cite{moulin2007capacity} and the compound channel~\cite{blackwell1959capacity, dobrushin1959optimum, wolfowitz1980simultaneous, polyanskiy2013dispersion}, and also the special cases recovered therein. 

\begin{figure*}[htpb]
    \centering
    \includegraphics[scale = 0.3]
    {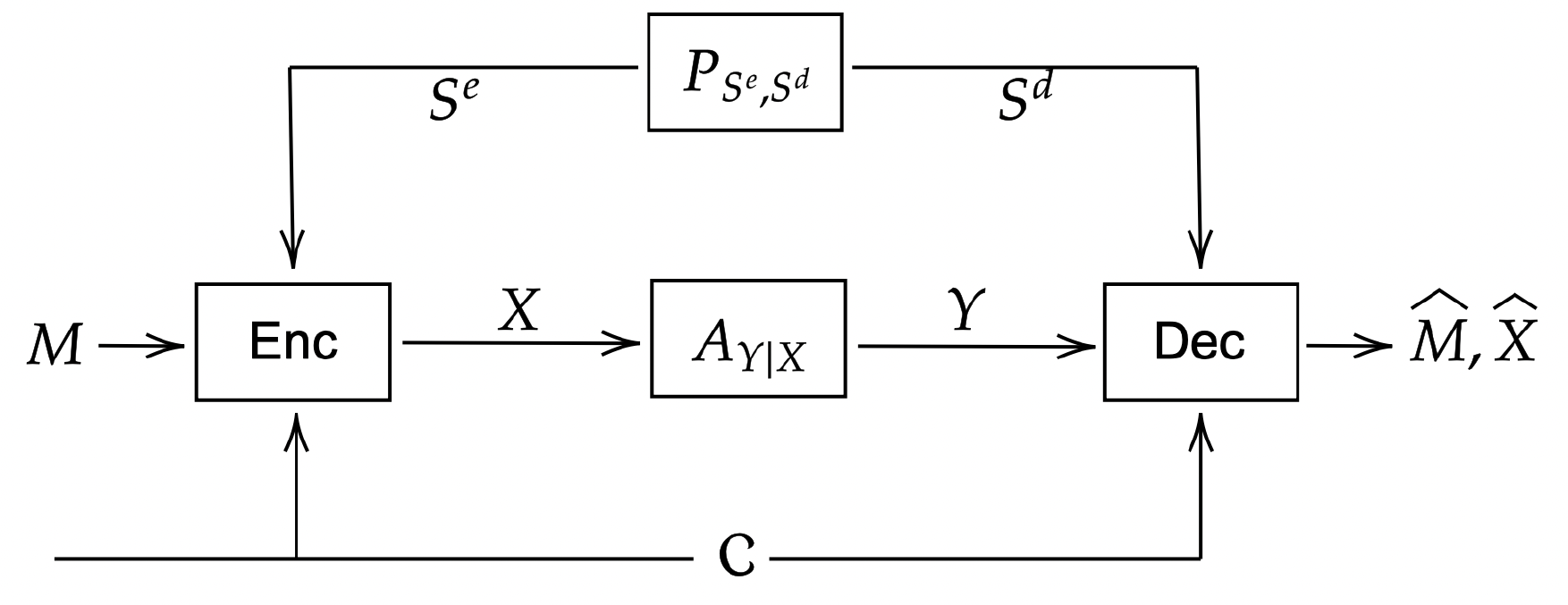}
    \caption{
    One-shot generalized information hiding setting. 
    }
    \label{fig:generalized_info_hiding}
\end{figure*}

\subsection{Problem Formulation}
\label{subsec::gene_infohid_prob_formu}

The one-shot generalized information hiding problem is shown in Figure~\ref{fig:generalized_info_hiding}. 
The goal is to hide a message $M$ into a host signal $S^e$, so that even though there is an attacker during the signal transmission which aims at removing the hidden message, the decoder can still reconstruct the original message and also the stegotext $X$, within a range defined by the distortion functions. 
We further elaborate their roles and assumptions in detail as follows.

\begin{itemize}
    \item \textbf{Encoder}: 
    The encoder observes a message $M$ that is uniformly chosen from the set $[1:\mathsf{L}]$, and the goal of encoding is to \emph{hide} $M$ into a host data source $S^e \in \mathcal{S}^e$ by introducing some tolerable level of distortions. 
    Given $S^e$ and $M$, the encoding function $f:\mathcal{S}^e \times [1:\mathsf{L}] \to \mathcal{X}$ outputs $X=f(S^e,M)$. 
    It is expected that $X$ is close to $S^e$, in the sense that $d_1(S,X)$ is small, where $d_1:\mathcal{S}^e \times \mathcal{X} \to [0,\infty)$ is a distortion measure. 
    We want $d_1(S,X) \le \mathsf{D}_1$ with high probability. This will be elaborated later.
    The encoded signal $X$ is then transmitted through the a channel $A_{Y|X}\in \mathcal{A}$.

    \item \textbf{Attacker}: The attacker is formulated as a noisy channel $P_{Y|X}$. 
    With input $X$, it performs data processing attacks on $X$ by introducing another level of distortion and produces $Y$, a corrupted version of $X$. 
    Its objective is to (partially) remove or degrade the message and/or the stegotext $X$, so that the decoder cannot have a correct reconstruction with a high fidelity. 
    Unlike the conventional asymptotic information hiding~\cite{moulin2003information}, we do not assume the attacker's strategy is known by the encoder and the decoder. 
    Instead, the attacker is free to choose from a class of channels $\mathcal{A}$ (e.g., the class of channels satisfying some distortion constraint between $X$ and $Y$, or the class of memoryless channels in case $X$ and $Y$ are sequences). 
    Both deterministic attacks or randomized attacks could be performed. 
    We assume the attacker has knowledge of the distributions (but not the values) of $M,S^e,S^d$, 
    and also the code $\mathcal{C}$ that is used by the encoder and the decoder.

    \item \textbf{Decoder}: The decoder observes $Y$, the output of the attacker, together with another source of side information $S^d$, and computes $(\hat{M}, \hat{X}) = \phi(S^d, Y)$, the distorted versions of $M$ and $X$, where $\phi: \mathcal{S}^d \times \mathcal{Y}\rightarrow [1:\mathsf{L}] \times \mathcal{X}$. 
    When the stogetext $X$ is expected to be reconstructed, we expect $d_2(X, \hat{X})$ is small, where $d_2:\mathcal{X} \times \mathcal{X} \to [0,\infty)$ is another distortion measure.
    Since we assume that the decoder is \emph{uninformed} of the attacker's strategy (different from~\cite{moulin2003information, xu2023information, grover2015information}), we intend to bound the encoder-decoder team's worst case failure probability
    \begin{equation}
    P_e := \sup_{A_{Y|X} \in \mathcal{A}}\mathbf{P}\big( d_1(S,X) > \mathsf{D}_1 \;\; \mathrm{OR} \;\; d_2(X, \hat{X}) > \mathsf{D}_2 \;\; \mathrm{OR} \;\; M \neq \hat{M}\big). \label{eq:hiding_fail}
    \end{equation}
    to be small, where we assume $(S^e,S^d,M)\sim P_{S^e,S^d} \times \mathrm{Unif}[\mathsf{L}]$, $X=f(S^e,M)$, $Y|X \sim A_{Y|X}$ and $(\hat{M}, \hat{X}) = \phi(S^d, Y)$ in the probability.\footnote{Note that \cite{moulin2003information} imposes a constraint on the expected distortion $\mathbf{E}[d_1(S,X)]$, which is reasonable in the context of \cite{moulin2003information} because the memoryless assumption and the law of large numbers ensure that the actual distortion is close to the expected distortion. 
    Since we are considering a one-shot setting where we only assume the attack channel is chosen from a set $\mathcal{A}$, if constraint need to be specified, it might be more reasonable to consider $d_1(S,X) > \mathsf{D}_1$ as a failure event and bound the probability of failure, i.e., the excess distortion probability instead, compared to expected distortion.
    }

    \item \textbf{Side information}: The side information $S_e, S_a, S_d$ can be viewed as certain \emph{common randomness} (or some resource) available at the encoder, the attacker and the decoder, respectively. 
    The joint distribution $P_{S^e, S^d}$ reveals information about the host data source $S^e$ to the decoder. 
\end{itemize}

\begin{remark}
    Our formulation can be viewed as a one-shot version of the generalized Gelfand-Pinsker problem~\cite{moulin2007capacity} (which only considered discrete case though), or a one-shot compound channel with side information at the encoder and/or the decoder. 
\end{remark}

\begin{remark}
    As noted in~\cite{moulin2003information, somekh2004capacity} (also see~\cite{somekh2003error, cohen2002gaussian}), these settings can be viewed as a \emph{game} between two parties: 
    the first party consists of the encoder (information hider) and the decoder, who are cooperatively transmitting the message $M$; 
    the second party is an attacker, who is trying to destroy or degrade the hidden message $M$ in $S^e$ so that the decoder cannot correctly decode $M$ or reconstruct a good $\hat{X}$. 
    More discussions on such game-theoretic perspective can be found in~\cite{moulin2003information}. 
\end{remark}

\subsection{One-shot Achievability Results}
\label{subsec::one_shot_gene_infohid_bd}

We then provide one-shot achievability results of the generalized information hiding problem. 

Note that in one-shot settings, as we discussed above, the techniques in~\cite{moulin2003information, somekh2004capacity} (e.g., the tools based on the \emph{typical sets}), which have resemblances to the Gelfand-Pinsker coding~\cite{gelfand1980coding,heegard1983capacity} as discussed in~\cite{moulin2003information}) are not suitable. 
Similar to Chapter~\ref{chp:nnc}, we utilize Poisson Matching Lemma has been shown to perform well in various one-shot settings~\cite{li2021unified} and was been introduced in~\ref{sec::PML} as one part of our proof technique. 
Our one-shot results apply to both discrete and continuous cases.

Briefly recall Chapter~\ref{chp:existing_techniques}. 
Fix a distribution $Q$ over $\mathcal{U}$. 
Let $(T_i)_{i=1,2,\ldots}$ be a Poisson process with rate $1$. 
Let $\mathbf{U} := (\bar{U}_i)_i$ be an independent i.i.d. sequence with distribution $Q$.
The ``marked'' Poisson process $(\bar{U}_i,T_i)_i$ supports a ``query operation'' given by the Poisson functional representation, where one can input a distribution $P$ over $\mathcal{U}$, and obtain one sample $\tilde{U}_P$ with distribution $P$. 
The Poisson functional representation is given by
\[
\mathbf{U}_P := \bar{U}_K, \;\; \text{where}\; K:= \underset{i}{\arg \min} \; T_i \cdot \Big(\frac{\mathrm{d}P}{\mathrm{d}Q}(\bar{U}_i)\Big)^{-1}.
\]

Since we let the encoder-decoder team account for all possible attack channels in a set $\mathcal{A}$, the achievability results have to suffer a penalty depending on the ``size'' of $\mathcal{A}$.  
Though the cardinality of $\mathcal{A}$ could be infinite, we can often find a finite subset $\tilde{\mathcal{A}}$ such that every attack channel $A \in \mathcal{A}$ is close enough to some $\tilde{A} \in \tilde{\mathcal{A}}$. 
We capture this notion of size by the $\epsilon$-covering number defined below (see similar covering arguments in~\cite{blackwell1959capacity,moulin2003information}).

\begin{definition}
\label{def::cov_num}
Given a set of channels $\mathcal{A}$ from $\mathcal{X}$ to $\mathcal{Y}$, its $\epsilon$-\emph{covering number} is defined as
\begin{align*}
N_{\epsilon}(\mathcal{A}) := \min \Big\{|\tilde{\mathcal{A}}|:\, \tilde{\mathcal{A}}\subseteq \mathcal{A},\,\,  \sup_{A \in \mathcal{A}} \,\, \min_{\tilde{A} \in \tilde{\mathcal{A}}} \,\,  \sup_{x \in \mathcal{X}} 
\left\Vert A_{Y|X}(\cdot|x) - \tilde{A}_{Y|X}(\cdot|x) \right\Vert_{\mathrm{TV}} \le \epsilon\Big\},
\end{align*}
where $\Vert A_{Y|X}(\cdot|x) - \tilde{A}_{Y|X}(\cdot|x) \Vert_{\mathrm{TV}} \in [0,1]$ denotes the total variation distance between $A_{Y|X}(\cdot|x)$ (the distribution of $Y$ if $X = x$, and $Y$ follows $A_{Y|X}$) and $\tilde{A}_{Y|X}(\cdot|x)$.
\end{definition}

We now present the main result, which is a one-shot achievability result with a bound on the error probability in terms of $N_{\epsilon}(\mathcal{A})$ and information density terms.

\begin{theorem}
\label{thm:info_hiding_cover}
    Fix any $P_{U,X|S^e,S^d}$
    and channel $\hat{A}_{Y|X}$. 
    For any $\epsilon \ge 0$, there exists a scheme for the generalized information hiding problem satisfying
    \begin{align*}
    & P_e  \le N_{\epsilon}(\mathcal{A})\, \sup_{A \in \mathcal{A}} \, \mathbf{E}_{Y|X \sim A}\Bigg[ 1 -  \mathbf{1}\{d_1(S^e, X)\leq \mathsf{D}_1  \} 
     \\
    & \cdot \mathbf{1}\{d_2(X, \hat{X})\leq \mathsf{D}_2  \} \cdot \Big(
    1 + \mathsf{L} \cdot
    2^{ -\hat{\iota}(U; Y, S^d) + \iota(U;S^e)}
    \Big)^{-1}
    \Bigg] + \epsilon,  
\end{align*}
where we assume $(S^e ,S^d,U,X,Y)\sim P_{S^e ,S^d} P_{U,X|S^e} A_{Y|X}$ in the expectation, and $\hat{\iota}(U; Y, S^d)$ is the information density computed by the joint distribution $P_{S,K} P_{U,X|S,K} \hat{A}_{Y|X }$ (instead of $A_{Y|X}$), assuming that $\iota(U;S^e), \hat{\iota}(U; Y, S^d)$ are almost surely finite for every $A_{Y|X} \in \mathcal{A}$.
\end{theorem}

\begin{proof}

The idea is that we design the decoder assuming that the attack channel is fixed to $A_{Y|X}$, and hope that this decoder works for every attack channel $A_{Y|X} \in \mathcal{A}$. 
Let $\mathcal{C} := ((\bar{U}_i, \bar{M}_i),T_i)_i$ where $(T_i)_i$ is a Poisson process, $\bar{U}_i \stackrel{\mathrm{iid}}{\sim} P_U$, and $\bar{M}_i \stackrel{\mathrm{iid}}{\sim} P_M$ (where $P_M = \mathrm{Unif}[\mathsf{L}]$). This will act as a random codebook shared between the encoder and the decoder (and this codebook will be fixed later).

The encoder observes the message $M\sim P_M$ and the encoder-side host signal $S^e$, by the Poisson functional representation~\cite{li2018strong, li2021unified} on the distribution $P_{U|S^e}(\cdot | S^e)\times \delta_M$ over $\mathcal{U} \times [1:\mathsf{L}]$ it produces $U = \mathbf{U}_{P_{U|S^e}(\cdot | S^e)\times \delta_M}$,\footnote{The Poisson functional representation produces a pair $(\mathbf{U},\mathbf{M})$, and $U$ is set to the first component of the pair.} and sends the generated $X | (S^e, U) \sim P_{X|S^e, U}$.
The decoder observes $Y, S^d$, outputs $\hat{M} = \mathbf{M}_{\hat{P}_{U|Y, S^d}(\cdot | Y,K) \times P_M}$ by the Poisson functional representation, and computes the reconstruction sequence $\hat{X}$ by $\hat{X} = \hat{\mathbf{x}}(\hat{M}, Y
)$, where $\hat{P}_{U|Y, S^d}$ is the conditional distribution computed by the joint distribution $P_{S^e , S^d} P_{U,X|S^e} \hat{A}_{Y|X }$. 
When the attack channel is $A_{Y|X } \in \mathcal{A}$, the error probability is bounded as follows: 
\begin{align*}
P_e(A) & :=  1 - \mathbf{P}_{Y|X  \sim A_{Y|X }}\big( d_1(S^e, X)\leq \mathsf{D}_1  \;\;  \mathrm{ AND } \;\; M= \hat{M}  \big)\\
& = \mathbf{E}
\bigg[ 1 - \mathbf{1}\{d_1(S^e, X)\leq \mathsf{D}_1  \} \\
&  \qquad  \qquad  \qquad  \qquad \cdot \mathbf{1}\{d_2(X, \hat{X})\leq \mathsf{D}_2  \}\cdot \mathbf{1}\{M = \hat{M} \} 
\bigg] \\
& 
 \leq \mathbf{E}
\bigg[ 1 - \mathbf{1}\{d_1(S^e, X)\leq \mathsf{D}_1  \} \cdot \mathbf{1}\{d_2(X, \hat{X})\leq \mathsf{D}_2  \}  \\
&  \qquad  \qquad   \qquad  \qquad  \cdot \mathbf{P}\big(M = \hat{M} |  M,S^e,S^d,U,Y \big) \bigg]  \\
& 
 \leq\mathbf{E}
\bigg[ 1 - \mathbf{1}\{d_1(S^e, X)\leq \mathsf{D}_1  \}  \cdot \mathbf{1}\{d_2(X, \hat{X})\leq \mathsf{D}_2  \} \\
&  \cdot \mathbf{P}\big((U,M) = (\mathbf{U}, \mathbf{M})_{\hat{P}_{U|Y, S^d}(\cdot|Y, S^d)\times P_M} | M,S^e,S^d,U,Y \big) \bigg]    \\
& 
 \stackrel{(a)}{\leq} \mathbf{E}
\bigg[ 1 - \mathbf{1}\{d_1(S^e,X)\leq \mathsf{D}_1  \}  \cdot \mathbf{1}\{d_2(X, \hat{X})\leq \mathsf{D}_2  \} \\
&   \qquad  \qquad  \cdot
\Big(
1 +  \frac{\mathrm{d} P_{U|S^e} (\cdot | S^e) \times \delta_M}{\mathrm{d}  \hat{P}_{U|Y, S^d}(\cdot | Y, S^d)\times P_M }
(U, M)
\Big)^{-1} \bigg]  \\
&
=
\mathbf{E}\Bigg[
1 -  \mathbf{1}\{d_1(S^e,X)\leq \mathsf{D}_1  \} \cdot \mathbf{1}\{d_2(X, \hat{X})\leq \mathsf{D}_2  \}  \\
& \qquad  \qquad  \cdot \Big(
1 + \mathsf{L} \cdot
\frac{ P_{U|S^e} (\cdot | S^e)}{  \hat{P}_{U|Y, S^d}(\cdot | Y, S^d) }
(U)
\Big)^{-1}
\Bigg] \\
&
= \mathbf{E}\bigg[
1 -  \mathbf{1}\{d_1(S^e, X)\!\leq\! \mathsf{D}_1  \}  \cdot \mathbf{1}\{d_2(X, \hat{X})\leq \mathsf{D}_2  \} \\
& \qquad  \qquad \cdot  \Big(
1 + \mathsf{L} \cdot 
2^{ -\hat{\iota}(U; Y, S^d) + \iota(U;S^e)}
\Big)^{\! -1}
\bigg]   \\
&
\le  \sup_{A_{Y|X} \in \mathcal{A}}\mathbf{E}_{Y|X  \sim A_{Y|X}}\bigg[
1 - \mathbf{1}\{d_1(S^e,X)\leq \mathsf{D}_1  \}  \\
& \qquad   \cdot \mathbf{1}\{d_2(X, \hat{X})\leq \mathsf{D}_2  \} \cdot \Big(
1 + \mathsf{L} \cdot
2^{ -\hat{\iota}(U; Y, S^d) + \iota(U;S^e)}
\Big)^{-1}
\bigg] 
\end{align*}
where $(a)$ is by the Poisson matching lemma.\footnote{The Poisson matching lemma is applied on the conditional distributions given $M,S^e ,S^d,U,Y$. Also see the conditional Poisson matching lemma~\cite{li2021unified}.} 
If we allow the encoder and the decoder to share unlimited additional common randomness, we can assume the codebook $\mathcal{C} = ((\bar{U}_i, \bar{M}_i),T_i)_i$ is actually shared, and conclude that $P_e = \sup_{A \in \mathcal{A}}P_e(A) \le \overline{P_e}$. Nevertheless, the only actual common randomness between the encoder and the decoder is $K$, which we cannot control. Therefore, we have to fix the codebook.

Let $P_e(A,c)$ be the probability of error when the attack channel is $A$ and the codebook is $\mathcal{C}=c$. We have $P_e(A) = \mathbf{E}_{\mathcal{C}}[P_e(A,\mathcal{C})]$. 
Let $\tilde{\mathcal{A}}\subseteq \mathcal{A}$ attain the minimum in $N_{\epsilon}(\mathcal{A})$.

Consider any $A \in \mathcal{A}$, and let $\tilde{A} \in \tilde{\mathcal{A}}$ satisfy 
\begin{equation*}
    \sup_{x \in \mathcal{X}}
    \left \Vert A_{Y|X} (\cdot|x) - \tilde{A}_{Y|X}(\cdot|x) 
    \right \Vert_{\mathrm{TV}} \le \epsilon.
\end{equation*} 
The total variation distance between the joint distribution of $M,S,K,U,X,Y$ under the attack channel $A$ conditional on $\mathcal{C}=c$ and the joint distribution under the attack channel $\tilde{A}$ conditional on $\mathcal{C}=c$ is also bounded by $\epsilon$. Hence $|P_e(A,c) - P_e(\tilde{A},c)| \le \epsilon$ and
\begin{align*}
P_e(A,c) &\le P_e(\tilde{A},c)+\epsilon \\
&\le \sum_{\tilde{A}\in \tilde{\mathcal{A}}}P_e(\tilde{A},c)+\epsilon.
\end{align*}  
Therefore,
\begin{align*}
\mathbf{E}_{\mathcal{C}}\Big[\sup_{A\in \mathcal{A}} P_e(A,\mathcal{C})\Big] & \le \mathbf{E}_{\mathcal{C}}\Big[\sum_{\tilde{A}\in \tilde{\mathcal{A}}}P_e(\tilde{A},\mathcal{C})+\epsilon\Big]  \\
& = \sum_{\tilde{A}\in \tilde{\mathcal{A}}} P_e(\tilde{A})+\epsilon \\
& \le |\tilde{\mathcal{A}}| \cdot \overline{P_e} +\epsilon.
\end{align*} 
The proof is completed by the existence of a codebook $c$ such that \begin{equation*}
    \sup_{A\in \mathcal{A}} P_e(A,c) \le |\tilde{\mathcal{A}}| \cdot \overline{P_e} +\epsilon.
\end{equation*}
\end{proof}

\begin{remark}
It is straightforward to convert this to a finite blocklength result where $n$ is a fixed number using the Berry-Esseen
theorem~\cite{berry1941accuracy, esseen1942liapunov}. 
\end{remark}

\begin{remark}
In Theorem~\ref{thm:info_hiding_cover}, we use a penalty term $N_{\epsilon}(\mathcal{A})$ to measure the effect of the ``size'' of $\mathcal{A}$, which introduces a degradation in the error probability. The choice of $\epsilon$ can be viewed as a way to balance the two terms in Theorem~\ref{thm:info_hiding_cover}: increasing $\epsilon$ will result in a larger $\epsilon$ but a smaller $N_{\epsilon}(\mathcal{A})$ (see also Proposition~\ref{prop:covering_bound}). 
Although directly investigating $N_{\epsilon}(\mathcal{A})$ following Definition~\ref{def::cov_num} may not be straightforward, it is possible to optimize the one-shot bound in Theorem~\ref{thm:info_hiding_cover} with respect to $\epsilon$ and the bound on $N_{\epsilon}(\mathcal{A})$ in Proposition~\ref{prop:covering_bound}. We leave more detailed analysis of these manipulations and potential second-order results as future work. For now, we only require that our one-shot bound suffices to recover the asymptotic hiding capacity~\cite{moulin2003information} when applied to discrete memoryless channels, as demonstrated in Section~\ref{sec::recovery_hiding} where we take $\epsilon = 1/n$ in the asymptotic analysis.
\end{remark}

\begin{remark}
    Note that when $S^d=S^e=\emptyset$, $d_1(s,x)=0$, and $\mathcal{A}=\{A_{Y|X}\}$ is a singleton set, taking $\hat{A}_{Y|X}=A_{Y|X}$, Theorem~\ref{thm:info_hiding_cover} reduces to the one-shot Gelfand-Pinsker coding result in~\cite{li2021unified}.
\end{remark}

\subsection{Discussions} 
\label{subsec::disc}

In~\cite{moulin2003information}, it is assumed that the attack channel must be memoryless, and hence the decoder can obtain full knowledge about the attack channel, justified by the large blocklength of signals. In this paper, similar to~\cite{somekh2004capacity, moulin2007capacity} (which focus on different targets or are under settings different to us), we drop this assumption, and consider a one-shot setting where the set of possible attack channels $\mathcal{A}$ can be \emph{any} set of channels. 
Also,  we do not assume that the decoder knows the attack channel, which is unrealistic in the one-shot setting where the attacker can be arbitrary.
In~\cite{somekh2004capacity} (which is a specialized information hiding setting that is similar to Section~\ref{sec::recovery_hiding}) the memoryless assumption is also dropped, and an asymptotic hiding capacity expressed as the limit of a sequence of single-letter expressions has been derived using constant composition codes. The key difference between \cite{somekh2004capacity} and our setting in Section~\ref{sec::recovery_hiding} (and also \cite{moulin2003information}) is that the side information in \cite{somekh2004capacity} is a shared key of unlimited size independent of $M,S^e$ that can be chosen as a part of the coding scheme, whereas in our paper and \cite{moulin2003information} the side information is given and may be correlated with the host signal (where the dependence is from the joint distribution), and cannot be changed. 
In some watermarking problems~\cite{cox1997secure, hartung1999multimedia} certain components can be further constrained, e.g., there may exist a mapping from the message $M$ to a codeword $V(M)$ which is independent of the host, and then composite data are obtained by a mapping from $V(M)$ and the side information.

The information hiding can be regarded as a variant of Gelfand-Pinsker coding for channels with side information at the encoder~\cite{gelfand1980coding,heegard1983capacity}, where the channel is fixed and not chosen by the attacker, and there is no shared side information between the encoder and the decoder. 
Since the encoder and the decoder have to account for all possible attack channels, this can be regarded as a combination of Gelfand-Pinsker coding and compound channel~\cite{blackwell1959capacity, dobrushin1959optimum, wolfowitz1980simultaneous}. 
The analyses in~\cite{moulin2003information, somekh2004capacity} utilize techniques such as random binning, joint typicality decoding and constant composition codes, which are also commonly utilized in the asymptotic analyses of Gelfand-Pinsker coding~\cite{gelfand1980coding, scarlett2015dispersions}.
These techniques may not be suitable for our one-shot setting. 
Strong typicality and constant composition codes are inapplicable when the blocklength is $1$. While random binning can be applied to one-shot Gelfand-Pinsker coding~\cite{verdu2012non,yassaee2013technique,watanabe2015nonasymp}, it produces weaker results compared to the Poisson matching lemma~\cite{li2021unified}. 
To obtain tight one-shot bounds for information hiding, we utilize the Poisson matching lemma instead. 
The main tool used to prove the coding theorems of generalized Gelfand-Pinsker problems~\cite{moulin2007capacity} is also the method of types~\cite{csiszar1998method}, which does not work in the one-shot analysis in general.

\section{Recovery of the Asymptotic Information Hiding}
\label{sec::recovery_hiding}

In this section, we discuss a special case of our generalized information hiding setting, which is the information hiding setting that was investigated in~\cite{moulin2003information}. 
We show that one-shot achievability results readily recover their asymptotic results on this setting, when we apply our results on discrete and memoryless channels. 

We first provide a simple bound on the $\epsilon$-covering number in the case that $X$ and $Y$ are discrete and finite.

\begin{proposition}
\label{prop:covering_bound}

If $\mathcal{X}$ and $\mathcal{Y}$ are finite, then 
\[
N_{\epsilon}(\mathcal{A})\le \Big(\frac{1}{2\epsilon} + \frac{|\mathcal{Y}|+1}{2}\Big)^{|\mathcal{X}|\cdot|\mathcal{Y}|}.
\]
\end{proposition}

The proof can be found in Appendix~\ref{appdx::pf_prop_covering_bound}.

We show that we recover the information hiding capacity that was discovered by~\cite{moulin2003information}. 
We can employ similar procedure to recover either the achievable bound of information hiding with stegotext reconstruction~\cite[Theorem 1]{xu2023information} (which in turn is an extension of~\cite{grover2015information} and~\cite{sumszyk2009information}), or the similar bounds in~\cite{moulin2007capacity}. 
For the simplicity, we only show the details of recovering the hiding capacity in~\cite{moulin2003information} here.

The setting is shown in Figure~\ref{fig:info_hiding}, where there exist a host signal $S$ available to the encoder, in which the encoder hides the message, and another source of side information $K$ that is available to the encoder and the decoder. 
By letting $S_e := (S, K)$ and $S_d := K$, we now show that Theorem~\ref{thm:info_hiding_cover} recovers the asymptotic result in \cite{moulin2003information} when $S,K,X,Y$ are finite and discrete, and the attack channel must be memoryless and is subject to a distortion constraint, and hence giving a simple alternative proof to \cite{moulin2003information}. 

\begin{figure*}[htpb]
    \centering
    \includegraphics[scale = 0.28]
    {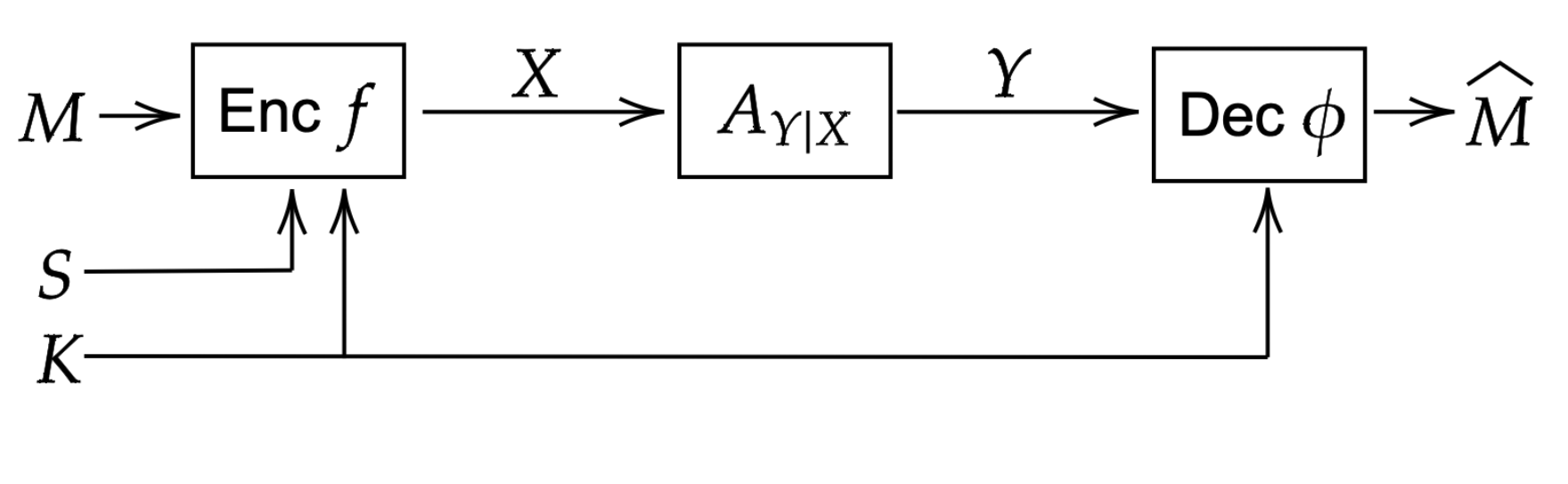}
    \caption{
    Information hiding setting~\cite{moulin2003information, somekh2004capacity}. 
    }
    \label{fig:info_hiding}
\end{figure*}

Consider sequences $S^n=(S_1,\ldots,S_n)$, $K^n$, $X^n$, $Y^n$ where $(S_i,K_i)\stackrel{\mathrm{iid}}{\sim} P_{S,K}$. Consider a channel input distribution $P_X$. The class of attack channels $\mathcal{A}_n=\mathcal{A}_n(P_X)$ (which depends on $P_X$) is taken to be 
\begin{equation*}
    \mathcal{A}_n(P_X) := \big\{A_{Y|X}^n:\, A_{Y|X} \in \mathcal{A}(P_X) \big\},
\end{equation*}
and we let 
\begin{equation*}
    \mathcal{A}(P_X) := \big\{A_{Y|X}:\,  \mathbf{E}_{(X,Y) \sim P_X A_{Y|X}}[d_3(X,Y)] \le \mathsf{D}_3\big\},
\end{equation*}
where $d_3: \mathcal{X} \times \mathcal{Y} \to [0,\infty)$ is a distortion measure, and $\mathsf{D}_3$ is the allowed distortion level. In other words, the attacker can only use memoryless channels $A_{Y|X}^n$ that satisfy the expected distortion constraint $\mathbf{E}[d_3(X,Y)] \le \mathsf{D}_3$.
The asymptotic hiding capacity given in~\cite{moulin2003information} is
\begin{equation*}
	C = \max_{P_{U,X|S,K}} \min_{A_{Y|X}:\, \mathbf{E}[d_2(X,Y)] \le \mathsf{D}_3} \big( I(U;Y|K) - I(U;S|K) \big).
\end{equation*}
where the maximum is over $P_{U,X|S,K}$ with $\mathbf{E}[d_1(S,X)]\le\mathsf{D}_{1}$.

We now show the achievability of the above asymptotic rate as a direct corollary of Theorem~\ref{thm:info_hiding_cover}. Fix $P_{U,X|S,K}$ which achieves the above maximum 
subject to $\mathbf{E}[d_1(S,X)]\le\mathsf{D}'_{1}$ where $\mathsf{D}'_1 < \mathsf{D}_1$.
Take $\hat{A}_{Y|X}$ to be the minimizer of the rate-distortion function
$\min_{A_{Y|X}:\, \mathbf{E}[d_2(X,Y)] \le \mathsf{D}_2}  I(U;Y|K)$,
and assume $(S,K,U,X,Y)\sim P_{S,K} P_{U,X|S,K} \hat{A}_{Y|X}$.
Write the information density and mutual information obtained from this distribution as $\hat{\iota}_{U;Y|K}$ and $\hat{I}(U;Y|K)$, respectively. Fix a coding rate $R < \hat{I}(U;Y|K)- I(U;S|K)$. We want to show that this rate is achievable.

Consider any attack channel $A_{Y|X}$ with $\mathbf{E}[d_3(X,Y)] \le \mathsf{D}_3$. 
Let $A_{Y|X}^{\lambda}:=(1-\lambda)\hat{A}_{Y|X}+\lambda A_{Y|X}$ for $0\le \lambda\le 1$. Write $I_{\lambda}(U;Y|K)$ for the mutual information computed assuming $Y|X \sim A_{Y|X}^{\lambda}$. It is straightforward to check that
\[
\frac{\mathrm{d}}{\mathrm{d}\lambda}I_{\lambda}(U;Y|K)\Big|_{\lambda=0}\!\!=\mathbf{E}_{Y|X\sim A_{Y|X}}[\hat{\iota}(U;Y|K)]-\hat{I}(U;Y|K).
\]
By the optimality of $\hat{A}$, the above derivative is nonnegative, and hence \begin{equation*}
    \mathbf{E}_{Y|X\sim A_{Y|X}}[\hat{\iota}(U;Y|K)] \ge \hat{I}(U;Y|K).
\end{equation*} 
Therefore, when we have i.i.d. sequences $(S^n,K^n,U^n,X^n,Y^n)\sim P^n_{S,K} P^n_{U,X|S,K} A^n_{Y|X}$ and $\mathsf{L}=\lfloor 2^{nR}\rfloor$, by the law of large numbers, 
\begin{align*}
& \mathsf{L} 2^{-\hat{\iota}(U^n;Y^n|K^n)+\iota(U^n;S^n|K^n)} \\
& \le 2^{nR-\sum_{i=1}^n (\hat{\iota}(U_i;Y_i|K_i)-\iota(U_i;S_i|K_i))} \\
& \to \; 0
\end{align*}
exponentially as $n\to \infty$ since \begin{align*}
    & \mathbf{E}[\hat{\iota}(U_i;Y_i|K_i)-\iota(U_i;S_i|K_i))] \\
    & \ge \hat{I}(U;Y|K)-I(U;S|K)\\
    & > R
\end{align*} 
We also have 
\begin{equation*}
    d_1(S^n,X^n) = \sum_{i=1}^n d_1(S_i,X_i) > n\mathsf{D}_1
\end{equation*} with probability approaching $0$ exponentially since $\mathsf{D}'_1 < \mathsf{D}_1$. These convergences are uniform over all such attack channels $A_{Y|X}$ since the random variables are discrete and finite. 

Therefore, to bound $P_e$ using Theorem~\ref{thm:info_hiding_cover}, it is left to bound the $\epsilon$-covering number $N_{\epsilon}(\mathcal{A}_n(P_X))$. Note that 
\begin{align*}
    & \left \Vert A_{Y|X}^n(\cdot | x^n) - \tilde{A}_{Y|X}^n(\cdot | x^n) 
    \right \Vert_{\mathrm{TV}} \\
    & \le \sum_{i=1}^n 
    \left \Vert A_{Y|X}(\cdot | x_i)-\tilde{A}_{Y|X}(\cdot | x_i) 
    \right \Vert_{\mathrm{TV}}, 
\end{align*}
and hence we can construct an $\epsilon$-cover of $\mathcal{A}_n(P_X)$ using an $(\epsilon/n)$-cover of $\mathcal{A}(P_X)$. Therefore, \begin{align*}
    N_{\epsilon}(\mathcal{A}_n(P_X)) 
    & \le N_{\epsilon/n}(\mathcal{A}(P_X)) \\
    & = O((n/\epsilon)^{|\mathcal{X}|\cdot|\mathcal{Y}|})
\end{align*}
by Proposition~\ref{prop:covering_bound}, which grows much slower than the exponential decrease of the expectation in Theorem~\ref{thm:info_hiding_cover}. Therefore, taking $\epsilon =1/n$, we have $P_e \to 0$ as $n\to \infty$. Taking $\mathsf{D}'_1 \to \mathsf{D}_1$ completes the proof.

\section{One-shot Compound Wiretap Channels}
\label{sec::wiretap}

In this section, we consider the compound wiretap channel~\cite{liang2009compound} in the one-shot setting. 
We utilize the Poisson matching lemma~\cite{li2021unified}, under a framework similar to the one-shot codes of the information hiding problem. 
We provide novel one-shot achievablity results for the compound wiretap channel~\cite{liang2009compound}. 
To the best of our knowledge, the one-shot results of this problem has not been discussed in literature, though finite-blocklength bounds on single (without channel uncertainties) wiretap channels can be found in~\cite{hayashi2006general, yassaee2015one, liu2016e_, yang2019wiretap}. 

Unlike the asymptotic analysis of the compound wiretap channel~\cite{liang2009compound}, our results also apply to continuous cases. 
Note that~\cite{schaefer2015secrecy} also studied the continuous case of compound wiretap channels, but the focus in~\cite{schaefer2015secrecy} was mainly on the compound Gaussian MIMO wiretap channels, and the analysis was not one-shot. 
In modern wireless communication, handling continuous cases can be essential in various applications for capturing the inherent variability and nuances of real-world signal propagation, which has a dynamic nature.

\subsection{Problem Formulation}
\label{subsec::comp_wiretap}

The one-shot compound wiretap channel setting is described as follows. 
A message $M$ is uniformly chosen from $\mathrm{Unif}[\mathsf{L}]$. 
Upon observing $M\sim \mathrm{Unif}[\mathsf{L}]$, the encoder produces $X = f(M)$, where $f:[\mathsf{L}]\rightarrow \mathcal{X}$ is a randomized encoding function. 
Then $X$ is sent through a channel $P_{Y,Z|X}$ that is unknown to the encoder and the decoder but known to the eavesdropper. 
A legitimate decoder observes $Y$ and recovers $\hat{M} = g(Y)$, where $g:\mathcal{Y}\rightarrow[\mathsf{L}]$ is a decoding function. 
The eavesdropper observes $Z\in \mathcal{Z}$. 
Justified by~\cite{liang2009compound} and~\cite[Lemma 1]{liang2008multiple}, we can assume the transition probability distribution is $P_{Y | X} P_{Z | X}$ by decomposing $P_{Y,Z|X}$ without loss of optimality.

We assume $P_{Y | X}$ is from a set $\mathcal{D}$ for \textbf{d}ecoding, while $P_{Z | X}$ is from a set $\mathcal{E}$ for \textbf{e}avesdropping. Unlike~\cite{liang2009compound}, we assume $\mathcal{D}, \mathcal{E}$ can be infinite, which captures the infinite variability of real-world signals and their propagation characteristics in practical applications. 
Even though their cardinalities can be infinite, we can often find a finite subset $\tilde{\mathcal{D}}$ (or $\tilde{\mathcal{E}}$) such that every receiver (or eavesdropper) in $\tilde{\mathcal{D}}$ (or $\tilde{\mathcal{E}}$) would be close enough to some $\tilde{D} \in \tilde{\mathcal{D}}$ (or $\tilde{E} \in \tilde{\mathcal{E}}$). 
This idea has appeared in Section~\ref{subsec::one_shot_gene_infohid_bd} and also in~\cite{schaefer2015secrecy}.

The objective is to bound the worst case error probability
\begin{equation}
    P_e := \sup_{P_{Y|X}\in \mathcal{D}} \mathbf{P}\left(
    M\neq \hat{M}
    \right),
\end{equation}
while also ensure the secrecy is guaranteed, which is measured by the total variation distance
\begin{equation}
     \gamma := \sup_{P_{Z|X}\in \mathcal{E}} \left\Vert P_{M,Z} - P_M\times P_Z\right\Vert_{\mathrm{TV}}
\end{equation}
being small.

\subsection{One-Shot Achievability Results}
\label{subsec::comp_wiretap}

We then provide the one-shot achievability results of the compound wiretap channel. 
Note the result can be viewed as a combination of the covering argument appeared in Section~\ref{sec::generalized_info_hiding} and the one-shot soft covering lemma in~\cite[Proposition 3]{li2021unified}. 
Other existing one-shot wiretap channel results~\cite{hayashi2006general, yassaee2015one, liu2016e_} might also be utilized in a similar framework as well.

\begin{theorem}
\label{thm:compound_wiretap}
    Fix any $P_{U,X}$ and any wiretap channel $\hat{P}_{Y|X}\hat{P}_{Z|X}$. 
    For any $\nu \geq 0$, any $\epsilon_1, \epsilon_2 \geq 0$ and $\mathsf{A}, \mathsf{B}\in\mathbb{N}$, there exists a code for the compound wiretap channel setting, with message $M\sim \mathrm{Unif}[\mathsf{L}]$, satisfying
    \begin{align*}
        & P_e +\nu\cdot \gamma \\ 
        & \leq N_{\epsilon_1}(\mathcal{D}) 
        \sup_{P_{Y|X}\in \mathcal{D}} 
        \mathbf{E}_{Y|X\sim P_{Y|X}} 
        \Big[
        \min\left\{
        \mathsf{L}\mathsf{A}2^{- \hat{\iota}(U;Y)}, 1
        \right\}
        \Big] + \epsilon_1 \\
        & + \nu\cdot N_{\epsilon_2}(\mathcal{E}) 
        \Bigg(
        \sup_{P_{Z|X}\in \mathcal{E}} 
        2 \cdot \mathbf{E}_{Z|X\sim P_{Z|X}} 
        \Big[ 
        \big(1 + 2^{-\hat{\iota}(U;Z)} \big)^{-\mathsf{B}} 
        \Big]   + \sqrt{ \mathsf{B}\mathsf{A}^{-1}}
        \Bigg)
        +  \nu\cdot \epsilon_2, 
    \end{align*}
    where we assume $(U, X, Y, Z)\sim P_{U, X} P_{Y|X} P_{Z|X}$ in the expectation, and $\hat{\iota}(U;Y)$, $\hat{\iota}(U;Z)$ are the information densities computed for compound channels by the joint distribution $P_{U,X} \hat{P}_{Y|X} \hat{P}_{Z|X}$, assuming that $\hat{\iota}(U;Y), \hat{\iota}(U;Z)$ are almost surely finite for every $P_{Y|X} \in \mathcal{D}$, $P_{Z|X}\in \mathcal{E}$. 
\end{theorem}

\begin{proof}

We first design our coding strategy assuming that the transmission channel is fixed to $\hat{P}_{X|Y} \in \mathcal{D}$ and the eavesdropping channel is fixed to $\hat{P}_{X|Z} \in \mathcal{E}$.

Let $\mathcal{C} := ((\bar{U}_i, \bar{M}_i),T_i)_i$ where $(T_i)_i$ is a Poisson process, $\bar{U}_i \stackrel{\mathrm{iid}}{\sim} P_U$, and $\bar{M}_i \stackrel{\mathrm{iid}}{\sim} P_M$ (where $P_M = \mathrm{Unif}[\mathsf{L}]$). 
This is a random codebook that is known to both the encoder and the decoder, and it will be fixed later.

Let $A \sim \mathrm{Unif} [ \mathsf{A}]$ be independent of $(M, \mathcal{C})$. 
The encoder observes the message $M\sim P_M$, computes $U = \mathbf{U}_{P_U \times\delta_M}(A)$, and sends the generated $X|U\sim P_{X|U}$. 
The decoder observes $Y$ and recovers $\hat{M} = \mathbf{M}_{\hat{P}_{U|Y}(\cdot |Y) \times P_M}$. 
We have $(M, A, U, X, Y, Z)\sim P_M\times P_A\times P_{U,X}\hat{P}_{Y|X}\hat{P}_{Z|X}$. 
For the fixed $\hat{P}_{Y|X} \in \mathcal{D}$, we have:
\begin{align}
    & \mathbf{P}\left\{
    M\neq \hat{M}
    \right\}\nonumber \\
    & \leq \mathbf{E}
    \bigg[\mathbf{P}\big((U,M) \neq (\mathbf{U}, \mathbf{M})_{\hat{P}_{U|Y}(\cdot|Y) \times P_M} | M, A, U, Y \big) \bigg] \nonumber \\
    & \stackrel{(a)}{\leq} \mathbf{E}
    \bigg[\min \Big\{ \mathsf{A} \frac{\mathrm{d} P_U \times \delta_M}{\mathrm{d} \hat{P}_{U|Y}(\cdot | Y) \times P_M} (U, M), 1  \Big\} \bigg] \nonumber \\
    & = \mathbf{E}
    \bigg[\min \Big\{ \mathsf{L} \mathsf{A} 2^{-\hat{\iota}_{U; Y}(U; Y)} , 1  \Big\} \bigg] \nonumber \\
    & \leq \sup_{P_{Y|X}\in \mathcal{D}} 
    \mathbf{E}_{Y|X \sim P_{Y|X}}\Big[
    \min \Big\{ \mathsf{L} \mathsf{A} 2^{-\hat{\iota}_{U; Y}(U; Y)} , 1  \Big\} 
    \Big] \label{eq::P_e_comp_wiretap} \\
    & =: \overline{P_e}\nonumber 
\end{align}
where $(a)$ is by the conditional
generalized Poisson matching lemma~\cite{li2021unified} applied on $(M, A, (U,M), Y, \hat{P}_{U|Y}\times P_M)$, and we define~\eqref{eq::P_e_comp_wiretap} to be $\overline{P_e}$.

For the secrecy measure, for the fixed $\hat{P}_{Z|X} \in \mathcal{E}$, we have: 
\begin{align}
    & \mathbf{E}\Big[
    \Big\Vert
    \hat{P}_{M,Z|\mathcal{C}}(\cdot, \cdot | \mathcal{C}) - P_{M}\times \hat{P}_{Z|\mathcal{C}}(\cdot | \mathcal{C})
    \Big\Vert_{\mathrm{TV}}
    \Big]  \nonumber \\
    & =  \mathbf{E}\Big[
    \Big\Vert
    \hat{P}_{Z|M, \mathcal{C}}(\cdot, \cdot | \mathcal{C}) -  \hat{P}_{Z|\mathcal{C}}(\cdot | \mathcal{C})
    \Big\Vert_{\mathrm{TV}}
    \Big] \nonumber \\
    & \leq \mathbf{E}\Big[
    \Big\Vert
    \hat{P}_{Z|M, \mathcal{C}}(\cdot, \cdot | \mathcal{C}) - \hat{P}_{Z}(\cdot)
    \Big\Vert_{\mathrm{TV}}
    \Big] + \mathbf{E}\Big[
    \Big\Vert
    \hat{P}_{Z|\mathcal{C}}(\cdot | \mathcal{C}) - \hat{P}_{Z}(\cdot)
    \Big\Vert_{\mathrm{TV}}
    \Big] \nonumber \\
    & \stackrel{(a)}{\leq} 
    2\cdot\mathbf{E}\Big[
    \Big\Vert
    \hat{P}_{Z|M, \mathcal{C}}(\cdot, \cdot | \mathcal{C}) - \hat{P}_{Z}(\cdot)
    \Big\Vert_{\mathrm{TV}}
    \Big]  \nonumber \\
    & = 2 \cdot \mathbf{E}\Big[
    \Big\Vert \mathsf{A}^{-1}
    \sum_{a=1}^\mathsf{A}
    \hat{P}_{Z|U}(\cdot | \mathbf{U}_{P_U\times \delta_M}(a) ) -  \hat{P}_{Z}(\cdot)
    \Big\Vert_{\mathrm{TV}}
    \Big]  \nonumber \\
    & \stackrel{(b)}{\leq}   
    2\cdot\mathbf{E} 
    \Big[ 
    \big(1 + 2^{-\hat{\iota}(U;Z)} \big)^{-\mathsf{B}} 
    \Big] 
    + \sqrt{ \mathsf{B}\mathsf{A}^{-1}}
     \nonumber \\
    & \leq 
    \sup_{P_{Z|X}\in \mathcal{E}} 
    2\cdot\mathbf{E}_{Z|X\sim P_{Z|X}} 
    \Big[ 
    \big(1 + 2^{-\hat{\iota}(U;Z)} \big)^{-\mathsf{B}} 
    \Big] 
    + \sqrt{ \mathsf{B}\mathsf{A}^{-1}}  \label{eq::gamma_comp_wiretap}\\
    & =: \overline{\gamma} \nonumber 
\end{align}
where $(a)$ is by the convexity of total variation distance, $(b)$ is by~\cite[Proposition 3]{li2021unified} since $\left\{ \mathbf{U}_{P_U\times \delta_m(a)} \right\}_{a\in[\mathsf{A}]} \stackrel{\mathrm{iid}}{\sim} P_U $ for any $m$,  and we define~\eqref{eq::gamma_comp_wiretap} to be $\overline{\gamma}$.

Let $P_e\left(P_{Y|X}, c\right)$ denote be the probability of error when the legitimate channel is $P_{Y|X}$ and the codebook is $\mathcal{C} = c$ and also let $\gamma\left(P_{Z|X}, c\right)$ denote the total variation distance $\gamma$ when the wiretap channel is $P_{Z|X}$ and the codebook is $\mathcal{C} = c$. 
Let $P_e(P_{Y|X}, P_{Z|X}, c) = P_e(P_{Y|X}, c) +\nu \cdot \gamma(P_{Z|X}, c)$. 
Let $\tilde{\mathcal{D}} \subseteq \mathcal{D}$ attain the minimum in $N_{\epsilon_1}(\mathcal{D})$ and $\tilde{\mathcal{E}} \subseteq \mathcal{E}$ attain the minimum in $N_{\epsilon_2}(\mathcal{E})$.

Consider any $P_{Y|X} \in \mathcal{D}$ and any $P_{Z|X} \in \mathcal{E}$, and let $\tilde{P}_{Y|X} \in \tilde{\mathcal{D}}, \tilde{P}_{Z|X} \in \tilde{\mathcal{E}}$ satisfy 
\begin{align*}
    & \sup_{x \in \mathcal{X}} 
    \left\Vert 
    P_{Y|X}(\cdot|x) - \tilde{P}_{Y|X}(\cdot|x) 
    \right\Vert_{\mathrm{TV}} 
    \leq \epsilon_1, \\
    & \sup_{x \in \mathcal{X}} 
    \left\Vert 
    P_{Z|X}(\cdot|x) - \tilde{P}_{Z|X}(\cdot|x) 
    \right\Vert_{\mathrm{TV}} 
    \leq \epsilon_2. 
\end{align*} 

The total variation distance between the joint distribution of $M,A,U,X,Y,Z$ under the channel $P_{Y|X}$ (or $P_{Z|X}$) conditional on $\mathcal{C} = c$ and the joint distribution under the channel $\tilde{P}_{Y|X}$ (or $\tilde{P}_{Z|X}$) conditional on $\mathcal{C} = c$ is also bounded by $\epsilon_1$ (or $\epsilon_2$). 
Therefore, we have 
$\left| P_e(P_{Y|X}, c) - P_e(\tilde{P}_{Y|X}, c) \right| \leq \epsilon_1$ and
$\left| \gamma(P_{Z|X}, c) - \gamma(\tilde{P}_{Z|X}, c) \right| \leq \epsilon_2$. 
Hence, 
\begin{align*}
& P_e\left( P_{Y|X}, P_{Z|X}, c \right) \\
& \leq P_e \left(\tilde{P}_{Y|X}, c\right) + \epsilon_1 + \nu \cdot \gamma \left(\tilde{P}_{Z|X}, c\right)  + \nu \cdot \epsilon_2 \\
& \leq \sum_{\tilde{P}_{Y|X} \in \tilde{\mathcal{D}} } P_e \left(\tilde{P}_{Y|X}, c\right)+ \epsilon_1  + \nu \cdot \sum_{\tilde{P}_{Z|X} \in \tilde{\mathcal{E}} } \gamma \left(\tilde{P}_{Z|X}, c\right) + \nu \cdot \epsilon_2. 
\end{align*}
Therefore, 
\begin{align*}
& \mathbf{E}_{\mathcal{C}} \Big[
\sup_{P_{Y|X} \in \mathcal{D}, P_{Z|X} \in \mathcal{E}} 
P_e(P_{Y|X}, P_{Z|X}, \mathcal{C}) 
\Big]\\
& \leq \mathbf{E}_{\mathcal{C}}\Bigg[
\sum_{\tilde{P}_{Y|X} \in \tilde{\mathcal{D}} } P_e \left(\tilde{P}_{Y|X}, \mathcal{C} \right) + \epsilon_1  + \nu\cdot \sum_{\tilde{P}_{Z|X} \in \tilde{\mathcal{E}} } \gamma \left(\tilde{P}_{Z|X}, \mathcal{C} \right) + \nu\cdot \epsilon_2 
\Bigg] \\
& = \sum_{\tilde{P}_{Y|X}\in \tilde{\mathcal{D}} } P_e(\tilde{P}_{Y|X})+ \epsilon_1 + \nu\cdot \sum_{\tilde{P}_{Z|X} \in \tilde{\mathcal{E}} } \gamma \left(\tilde{P}_{Z|X} \right)  + \nu\cdot \epsilon_2 \\
& \leq |\tilde{\mathcal{D}}| \cdot \overline{P_e} + \nu\cdot |\tilde{\mathcal{E}}| \cdot \overline{\gamma} + \epsilon_1 + \nu\cdot \epsilon_2.
\end{align*} 
Hence the proof is completed by the existence of a codebook $c$ such that 
\begin{align*}
    &\sup_{P_{Y|X} \in \mathcal{D}, P_{Z|X} \in \mathcal{E}} 
    P_e(P_{Y|X}, P_{Z|X}, c) \\
    & \leq |\tilde{\mathcal{D}}| \cdot \overline{P_e} +\nu\cdot  |\tilde{\mathcal{E}}| \cdot \overline{\gamma} + \epsilon_1 + \nu\cdot \epsilon_2. 
\end{align*} 
\end{proof}

\begin{remark}
    This scheme can be viewed as a combination of the covering argument that has been discussed in Section~\ref{sec::generalized_info_hiding} and the one-shot soft covering lemma~\cite[Proposition 3]{li2021unified}. 
    One can possibly provide different one-shot achievability results utilizing other existing one-shot results on single wiretap channels~\cite{hayashi2006general, yassaee2015one, liu2016e_}. 
\end{remark}

\subsection{Recovery of the Asymptotic Results}

We recover the existing asymptotic results~\cite{liang2009compound} as follows. 
In~\cite{liang2009compound}, they assume all the random variables are discrete and the channels are memoryless, and $\mathcal{D} := \{P_{Y_1|X}, \ldots, P_{Y_\mathsf{J}|X} \}$ and $\mathcal{E} := \{P_{Z_1|X}, \ldots, P_{Z_\mathsf{K}|X} \}$ for some finite $\mathsf{J}, \mathsf{K}$. 
The setting can be understood as Figure~\ref{fig:wiretap}. 

\begin{figure}[htpb]
	\centering
    \includegraphics[scale = 0.3]
    {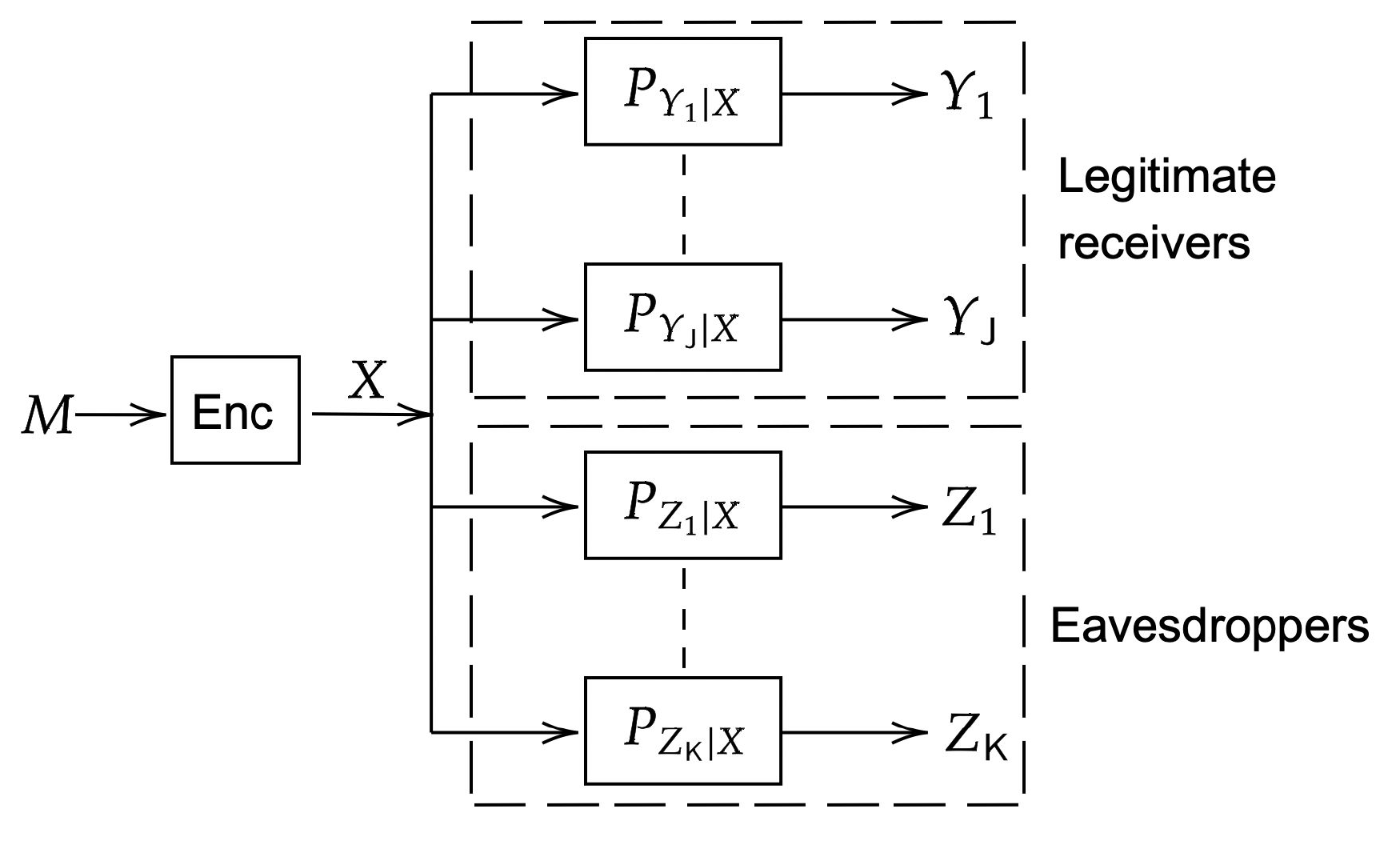}
	\caption{
	Discrete memoryless compound wiretap channel setting in~\cite{liang2009compound}. 
    }
    \label{fig:wiretap}
\end{figure} 

By~\cite{liang2009compound}, for the discrete memoryless channels, the achievable secrecy rate is 
\begin{align}
    R & = \max\left[
    \min_j I(U; Y_j) - \max_k I(U; Z_k)
    \right] \\
    & = \max \min_{j, k} \left[
     I(U; Y_j) - I(U; Z_k)
    \right],
\end{align}
where the maximum is taken over all distributions $P_{U,X}$ such that the auxiliary random variable $U$ satisfies the Markov chain $U \leftrightarrow X \leftrightarrow (Y_j, Z_k)$ for $j = 1, \ldots, \mathsf{J}$ and $k = 1, \ldots, \mathsf{K}$. 
This can be understood as the \emph{worst-case} result, i.e., one considers the worst receiver in $\mathcal{D}$ and the most-powerful eavesdropper in $\mathcal{E}$. 

To recover the above asymptotic results from our Theorem~\ref{thm:compound_wiretap}, fix $P_{U,X}$, take $\hat{P}_{Y|X} \hat{P}_{Z|X}$ which minimizes $I(U; Y) - I(U; Z)$, and assume $(M, A, U, X, Y, Z)\sim P_M\times P_A\times P_{U,X}\hat{P}_{Y|X}\hat{P}_{Z|X}$. 
Write the information density and mutual information obtained terms from this distribution as $\hat{\iota}(U;Y), \hat{\iota}(U;Z), \hat{I}(U;Y), \hat{I}(U;Z)$.  Fix a coding rate $R = \hat{I}(U;Y) - \hat{I}(U;Z) - \epsilon$, we are left to show
that this rate is achievable. 

When we have i.i.d. sequences $(A^n, U^n, X^n, Y^n, Z^n)\sim P_A^n P_{U,X}^n P_{Y|X}^n P_{Z|X}^n $, take $\mathsf{L} = \lfloor 2^{nR} \rfloor$, $\mathsf{A} = 2^{n(I(U;Z) + \epsilon / 2)}$ and $\mathsf{B} = 2^{n(I(U;Z) + \epsilon / 3)}$, by the law of large numbers, the following terms in Theorem~\ref{thm:compound_wiretap}: 
\begin{align*}
    \mathsf{L}\mathsf{A}2^{- \hat{\iota}(U;Y)} & \leq 2^{nR + n( I(U;Z) + \epsilon /2) -\sum^n_{i=1} \hat{\iota}(U_i; Y_i)}   \rightarrow 0, \\ 
    \big(1 + 2^{- \hat{\iota}(U;Z)} \big)^{-\mathsf{B}} & \stackrel{(a)}{\leq} 2^{\sum^n_{i=1} \hat{\iota}(U_i; Z_i) - n(I(U;Z) + \epsilon/3) }\rightarrow 0,  \\ 
    \sqrt{ \mathsf{B}\mathsf{A}^{-1}}  & = \sqrt{2^{n(I(U;Z) + \epsilon / 3) -  n(I(U;Z) + \epsilon/2)}}\rightarrow 0 
\end{align*}
exponentially as $n\rightarrow \infty$, where $(a)$ used  $(1+2^{-x})^{-2^y} \leq 2^{x-y}$.
It is left to bound the $\epsilon$-covering numbers $N_{\epsilon_1}(\mathcal{D}), N_{\epsilon_2}(\mathcal{E})$ in Theorem~\ref{thm:compound_wiretap}. 
Similar to Section~\ref{sec::generalized_info_hiding}, by constructing an $\epsilon$-covers of them and by Proposition~\ref{prop:covering_bound}, we know $N_{\epsilon_1}(\mathcal{D}_n) \leq N_{\epsilon_1/n}(\mathcal{D}) = O((n/\epsilon)^{|\mathcal{X}|\cdot|\mathcal{Y}|})$ and similarly $N_{\epsilon_2}(\mathcal{E}_n) \leq O((n/\epsilon)^{|\mathcal{X}|\cdot|\mathcal{Z}|})$. Hence they grow much slower than the exponential decrease of above terms in Theorem~\ref{thm:compound_wiretap}. 
Therefore we have $P_e + \nu\cdot\gamma \rightarrow 0$ as $n\rightarrow \infty$. 

\newcommand{\lp}{\left(}
\newcommand{\rp}{\right)}
\newcommand{\lb}{\left[}
\newcommand{\rb}{\right]}
\newcommand{\lbp}{\left\{}
\newcommand{\rbp}{\right\}}
\newcommand{\lba}{\left\lvert}
\newcommand{\rba}{\right\rvert}
\newcommand{\lV}{\left\lVert}
\newcommand{\rV}{\right\rVert}
\newcommand{\mv}{\middle\vert}
\newcommand{\ul}{\underline}
\newcommand{\ol}{\overline}
\newcommand{\mcal}{\mathcal}
\newcommand{\mscr}{\mathscr}
\newcommand{\what}{\widehat}
\newcommand{\wtild}{\widetilde}
\newcommand{\mb}{\mathbf}
\newcommand{\bbm}{\mathbbm}
\newcommand{\mbb}{\mathbb}
\newcommand{\msf}{\mathsf}
\newcommand{\la}{\leftarrow}
\newcommand{\ra}{\rightarrow}
\newcommand{\ua}{\uparrow}
\newcommand{\da}{\downarrow}
\newcommand{\lra}{\leftrightarrow}
\newcommand{\lgla}{\longleftarrow}
\newcommand{\lgra}{\longrightarrow}
\newcommand{\lglra}{\longleftrightarrow}
\newcommand{\lan}{\langle}
\newcommand{\ran}{\rangle}
\newcommand{\llan}{\left\langle}
\newcommand{\rran}{\right\rangle}
\newcommand{\lce}{\left\lceil}
\newcommand{\rce}{\right\rceil}
\newcommand{\lfl}{\left\lfloor}
\newcommand{\rfl}{\right\rfloor}

\newcommand{\emphb}{\textcolor{blue}}
\newcommand{\emphg}{\textcolor{Grass}}
\newcommand{\emphr}{\textcolor{red}}
\newcommand{\eqFunc}{\overset{\mathrm{f}}{=}}
\newcommand{\eqDef}{:=}
\newcommand{\diid}{\overset{\text{i.i.d.}}{\sim}}

\newcommand{\etal}{{\it et al.}}
\newcommand{\Cov}{\mathsf{Cov}}
\newcommand{\Bias}{\mathsf{Bias}}
\newcommand{\MSE}{\mathsf{MSE}}
\newcommand{\MLE}{\mathsf{MLE}}
\newcommand{\Risk}{\mathsf{R}}
\newcommand{\Ber}{\mathrm{Ber}}
\newcommand{\Binom}{\mathrm{Binom}}
\newcommand{\Unif}{\mathrm{Unif}}
\newcommand{\SNR}{\mathsf{SNR}}
\newcommand{\INR}{\mathsf{INR}}
\newcommand{\SINR}{\mathsf{SINR}}
\newcommand{\Pe}{\mathsf{P}_{\mathsf{e}}}
\newcommand{\uPe}{\mathsf{\ol{P}}_{\mathsf{e}}}
\newcommand{\lPe}{\mathsf{\ul{P}}_{\mathsf{e}}}
\newcommand{\eps}{\varepsilon}
\newcommand{\Indc}[1]{\mathbbm{1}_{\lbp #1\rbp}}
\newcommand{\Q}[1]{\mathrm{Q}\lp #1 \rp}
\newcommand{\FT}{\breve}
\newcommand{\ZT}{\check}
\newcommand{\sinc}{\mathrm{sinc}}
\newcommand{\rect}{\mathrm{rect}}
\newcommand{\argmin}{\mathop{\mathrm{argmin}}}
\newcommand{\argmax}{\mathop{\mathrm{argmax}}}
\newcommand{\dH}{\mathsf{d_H}}
\newcommand{\wei}{\mathsf{w}}
\newcommand{\herm}{\mathtt{H}}

\newcommand{\MI}[2]{{I}\left( #1\,; #2\, \right)}
\newcommand{\varMI}[1]{{\mathsf{I}}\left( #1\,\right)}
\newcommand{\CMI}[3]{{I}\left( #1\,; #2\, \middle\vert\, #3 \right)}
\newcommand{\ET}[1]{{H}\left( #1\,\right)}
\newcommand{\CET}[2]{{H}\left( #1\, \middle\vert #2\, \right)}
\newcommand{\ETR}[1]{\mcal{H}\left( #1\,\right)}
\newcommand{\tildeETR}[1]{\widetilde{\mcal{H}}\left( #1\,\right)}
\newcommand{\KLD}[2]{{D}_{\msf{KL}}\left( #1\, \middle\Vert #2 \right)}
\newcommand{\TV}[2]{\lVert #1 - #2 \rVert_{\msf{TV}}}
\newcommand{\DET}[1]{{h}\left( #1\,\right)}
\newcommand{\CDET}[2]{{h}\left( #1\, \middle\vert #2\, \right)}

\chapter{One-Shot Channel Simulation with Differential Privacy}
\label{chp:ppr}

\section{Overview}

In this chapter, we introduce a novel ``DP mechanism compressor'', 
called \emph{Poisson private representation}, designed to compress and \emph{exactly} simulate \emph{any} local randomizer while ensuring local DP, through the use of shared randomness. 
The Poisson private representation (PPR) can be viewed as a ``meta-mechanism'', in the sense that it compresses arbitrary differential privacy mechanisms.\footnote{Here ``meta-mechanism'' means a method that takes a privacy mechanism $\mathcal{A}$, and produces a new compressed mechanism $\mathcal{A}'$. While it is intuitively similar to a higher-order function in functional programming, we allow a meta-mechanism to look at the output distribution induced by $\mathcal{A}$, instead of only treating $\mathcal{A}$ as a black box.} 
This chapter is based on~\cite{liu2024universal}.

We elaborate on three main advantages of PPR, namely universality, exactness and communication efficiency.

\begin{enumerate}
    \item \textbf{Universality}: Unlike dithered-quantization-based approaches which can only simulate additive noise mechanisms, PPR can simulate any local or central DP mechanism with discrete or continuous input and output. Moreover, PPR is \emph{universal} in the sense that the user and the server only need to agree on the output space and a proposal distribution, and the user can simulate any DP mechanism with the same output space. The user can choose a suitable DP mechanism and privacy budget according to their communication bandwidth and privacy requirement, without divulging their choice to the server. 
    
    \item \textbf{Exactness}: Unlike previous DP mechanism compressors such as  \citep{feldman2021lossless, shah2022optimal, triastcyn2021dp}, PPR enables \emph{exact} simulation, ensuring that the reproduced distribution perfectly matches the original one. Exact distribution recovery offers several advantages. Firstly, the compressed sample maintains the same statistical properties as the uncompressed one. If the local randomizer is unbiased (a crucial requirement for many machine learning tasks like DP-SGD), the outcome of PPR remains unbiased. In contrast, reconstruction distributions in prior simulation-based compression methods \citep{feldman2021lossless, shah2022optimal} are often biased unless specific debiasing steps are performed (only possible for certain DP mechanisms \citep{shah2022optimal}). Secondly, when the goal is to compute the mean (e.g., for private mean or frequency estimation problems) and the local noise is ``summable'' (e.g., Gaussian noise or other infinitely divisible distributions \citep{kotz2012laplace, goryczka2015comprehensive}), exact distribution recovery of the local noise enables precise privacy accounting for the final \emph{central} DP guarantee, without relying on generic privacy amplification techniques like shuffling \citep{erlingsson2019amplification,feldman2022hiding}. PPR can compress a central DP mechanism (e.g., the Gaussian mechanism \citep{dwork2006our}) and simultaneously achieve weaker local DP (i.e., with a larger $\varepsilon_{\msf{local}}$) and stronger central DP (i.e., with a smaller $\varepsilon_{\msf{central}}$), while maintaining exactly the same privacy-utility trade-offs as the uncompressed Gaussian mechanism. 
    
    \item \textbf{Communication efficiency}:  PPR compresses the output of any DP mechanism to a size close to the theoretical lower bound. For a mechanism on the data $X$ with output $Z$, the compression size of PPR is $I(X;Z)+\log (I(X;Z)+1) + O(1)$, with only a logarithmic gap from the mutual information lower bound $I(X;Z)$.\footnote{This is similar to channel simulation~\cite{harsha2010communication} and the strong functional representation lemma~\cite{li2018strong}, though \cite{harsha2010communication,li2018strong} do not concern privacy.} The ``$O(1)$'' constant can be given explicitly in terms of a tunable parameter $\alpha > 1$ which controls the trade-off between compression size, computational time and privacy. 
An $\alpha$ close to $1$ provides a better local DP guarantee, but requires a larger compression size and longer computational time.
\end{enumerate}

The main technical tool we utilize for PPR is the Poisson functional representation \citep{li2018strong,li2021unified}, which provides precise control over the reconstructed joint distribution in channel simulation problems \citep{bennett2002entanglement, harsha2010communication, li2018strong,flamich2024greedy, goc2024channel,braverman2014public,bennet2014reverse,cuff2013distributed}. 
Channel simulation aims to achieve the minimum communication for simulating a channel (i.e., a specific conditional distribution). 
Typically, these methods rely on shared randomness between the user and server, and privacy is only preserved \emph{when the shared randomness is hidden from the adversary}. 
This setup conflicts with local DP, where the server (which requires access to shared randomness for decoding) is considered adversarial. 
To ensure local DP,
we introduce a randomized encoder based on the Poisson functional representation, which stochastically maps a private local message to its representation. 
Hence, PPR achieves order-wise trade-offs between privacy, communication, and accuracy, while preserving the original distribution of local randomizers.

\section{Related Work\label{sec:related}}

\subsection{Generic compression of local DP mechanisms}

In this work, we consider both central DP \citep{dwork2006calibrating} and local DP \citep{warner1965randomized, kasiviswanathan2011can}. Recent research has explored methods for compressing local DP randomizers when shared randomness is involved. For instance, when $\varepsilon \leq 1$, \citep{bassily2015local} demonstrated that a single bit can simulate any local DP randomizer with a small degradation of utility, as long as the output can be computed using only a subset of the users' data.
\citep{bun18hh} proposed another generic compression technique based on rejection sampling, which compresses a $\varepsilon$-DP mechanism into a $10\varepsilon$-DP mechanism. 
\citep{feldman2021lossless} proposed a distributed simulation approach using rejection sampling with shared randomness, while \citep{shah2022optimal, triastcyn2021dp} utilized importance sampling (or more specifically, minimum random coding \citep{cuff2008communication,song2016likelihood, havasi2019minimal}). However, all these methods only \emph{approximate} the original local DP mechanism, unlike our scheme, which achieves an \emph{exact} distribution recovery.

\subsection{Distributed mean estimation under DP}

Mean estimation is the canonical problems in distributed learning and analytics. They have been widely studied under privacy \citep{duchi2013local, bhowmick2018protection, duchi2019lower, asi2022optimal},  communication \citep{garg2014communication, braverman2016communication, suresh2017distributed}, or both constraints \citep{chen2020breaking, feldman2021lossless, shah2022optimal, guo2023privacy, chaudhuri2022privacy, chen2024privacy}. Among them, \citep{asi2022optimal} has demonstrated that the optimal unbiased mean estimation scheme under local differential privacy is $\msf{privUnit}$ \citep{bhowmick2018protection}. Subsequently, communication-efficient mechanisms introduced by \citep{feldman2021lossless, shah2022optimal, isik2024exact} aimed to construct communication-efficient versions of $\msf{privUnit}$, either through distributed simulation or discretization. However, these approaches only approximate the $\msf{privUnit}$ distribution, while our proposed method ensures exact distribution recovery.

\subsection{Distributed channel simulation}

Our approach relies on the notion of channel simulation \citep{bennett2002entanglement, harsha2010communication, li2018strong,flamich2024greedy, goc2024channel,braverman2014public,bennet2014reverse,cuff2013distributed}. 
One-shot channel simulation is a lossy compression task, which aims to find the minimum amount of communications over a noiseless channel that is in need to ``simulate'' some channel $P_{Z|X}$ (a specific conditional distribution). 
By \citep{harsha2010communication,li2018strong}, the average communication cost is $I(X;Z) + O(\log(I(X;Z)))$. 
In~\cite{harsha2010communication}, algorithms based on rejection sampling are proposed, and it is further generalized in~\cite{flamich2023faster} by introducing the greedy rejection coding. 
Dithered quantization~\cite{ziv1985universal} has also been used to simulate an additive noise channel in~\cite{agustsson2020universally} for neural compression. 
As also shown in~\cite{agustsson2020universally}, the time complexity of channel simulation protocols (e.g., in \cite{li2018strong}) is usually high, and~\cite{theis2022algorithms, flamich2024greedy, goc2024channel} try to improve the runtime under certain assumptions. 
Moreover, channel simulation tools have also been used in neural network compression~\cite{havasi2019minimal}, image compression via variational autoencoders~\cite{flamich2020compressing}, diffusion models with perfect realism~\cite{theis2022lossy} and differentially private federated learning~\cite{shah2022optimal}.

\section{Preliminaries}
\label{sec::preliminary}

We begin by reviewing the formal definitions of differential privacy (DP). We consider two models of DP data analysis. In the central model, introduced in \citep{dwork2006calibrating}, the data of the individuals is stored in a database $X \in \mathcal{X}$ by the server. The server is then trusted to perform data analysis whose output $Z = \mcal{A}(X) \in \mathcal{Z}$ (where $\mcal{A}$ is a randomized algorithm), which is sent to an untrusted data analyst, does not reveal too much information about any particular individual’s data. While this model requires a higher level of trust than the local model, it is possible to design significantly more accurate algorithms. We say that two databases $X,X' \in \mathcal{X}$ are neighboring if they differ in a single data point.
More generally, we can consider a symmetric neighbor relation $\mathcal{N} \subseteq \mathcal{X}^2$, and regard $X,X'$ as neighbors if $(X,X')\in \mathcal{N}$.

On the other hand, in the local model, each individual (or client) randomizes their data before sending it to the server, meaning that individuals are not required to trust the server. A local DP mechanism \citep{kasiviswanathan2011can} is a local randomizer $\mcal{A}$ that maps the local data $X\in \mathcal{X}$ to the output $Z  = \mcal{A}(X) \in \mathcal{Z}$.
Note that here $X$ is the data at one user, unlike central-DP where $X$ is the database with the data of all users.
We now review the notion of $(\varepsilon, \delta)$-central and local DP.

\begin{definition}[Differential privacy \citep{ dwork2006calibrating,kasiviswanathan2011can}]
Given a mechanism $\mcal{A}$ which induces the conditional distribution $P_{Z|X}$ of $Z = \mcal{A}(X)$, we say that it satisfies $(\varepsilon, \delta)$-DP if for any neighboring $(x, x') \in \mathcal{N}$ and $\mathcal{S} \subseteq \mathcal{Z}$, it holds that
\[
\Pr ( Z \in \mathcal{S}\, |\, X=x ) \leq  e^\varepsilon \Pr ( Z \in \mathcal{S}\, |\, X=x' ) + \delta.
\]
In particular, if $\mathcal{N} = \mathcal{X}^2$, we say that the mechanism satisfies $(\varepsilon, \delta)$-local DP \citep{kasiviswanathan2011can}.\footnote{Equivalently, local DP can be viewed as a special case of central DP with dataset size $n=1$.}
\end{definition}

When a mechanism satisfies $(\varepsilon, 0)$-central/local DP, we will refer to it simply as $\varepsilon$-central/local DP. 
$\varepsilon$-DP can be generalized to \emph{metric privacy} by considering a metric $d_{\mathcal{X}}(x,x')$ over $\mathcal{X}$~\cite{chatzikokolakis2013broadening,andres2013geo}.
\begin{definition}[$\varepsilon\cdot d_{\mathcal{X}}$-privacy \citep{chatzikokolakis2013broadening,andres2013geo}]
\label{def:metric_privacy}Given a mechanism $\mcal{A}$ with conditional distribution $P_{Z|X}$, and a metric $d_{\mathcal{X}}$ over $\mathcal{X}$, we say that $\mcal{A}$ satisfies $\varepsilon\cdot d_{\mathcal{X}}$-privacy if for any $x, x' \in \mathcal{X}$, $\mathcal{S} \subseteq \mathcal{Z}$, we have
\[
\Pr ( Z \in \mathcal{S}\, |\, X=x ) \leq  e^{\varepsilon \cdot d_{\mathcal{X}}(x,x')} \Pr ( Z \in \mathcal{S}\, |\, X=x' ).
\]
\end{definition}
This recovers the original $\varepsilon$-central DP by considering $d_{\mathcal{X}}$ to be the Hamming distance among databases, and recovers the original $\varepsilon$-local DP by considering $d_{\mathcal{X}}$ to be the discrete metric~\cite{chatzikokolakis2013broadening}.

The reason we use $X$ to refer to both the database in central DP and the user's data in local DP is that our proposed method can compress both central and local DP mechanisms in exactly the same manner. In the following sections, the mechanism $\mcal{A}$ to be compressed (often written as a conditional distribution $P_{Z|X}$) can be either a central or local DP mechanism, and the neighbor relation $\mathcal{N}$ can be any symmetric relation.
The ``encoder'' refers to the server in central DP, or the user in local DP. The ``decoder'' refers to the data analyst in central DP, or the server in local DP.

\section{Poisson Private Representation\label{sec:ppr}}

As introduced in Chapter~\ref{chp:existing_techniques}, given $(T_{i})_{i}$, a Poisson process with rate $1$ (i.e., $T_{1},T_{2}-T_{1},T_{3}-T_{2},\ldots\stackrel{iid}{\sim}\mathrm{Exp}(1)$) that is independent of $Z_{i}\stackrel{iid}{\sim}Q$ for $i=1,2,\ldots$, the Poisson functional representation~\citep{li2018strong,li2021unified} selects the point $Z := Z_K$ where 
\begin{equation*}
    K = \mathrm{argmin}_{i} T_{i} \cdot \Big(\frac{\mathrm{d}P}{\mathrm{d}Q}(Z_{i})\Big)^{-1}.
\end{equation*}

The calculation of $K$ can be viewed as a search problem over a Poisson process.  
The ``marked'' Poisson process $(Z_i,T_i)_i$ supports a ``query operation'' provided by the Poisson functional representation, where one can input a distribution $P$ over $\mathcal{Z}$, and obtain one sample with distribution $P$, i.e., the Poisson functional representation guarantees that $Z \sim P$~\citep{li2018strong}.

To simulate a DP mechanism with a conditional
distribution $P_{Z|X}$ using the Poisson functional representation, we can use $(Z_{i})_{i}$ as the
shared randomness between the encoder and the decoder.
\footnote{The original Poisson functional representation \cite{li2018strong,li2021unified}
uses the whole $(Z_{i},T_{i})_{i}$ as the shared randomness. It is
clear that $(T_{i})_{i}$ is not needed by the decoder, and hence
we can use only $(Z_{i})_{i}$ as the shared randomness.}
Upon observing $X$, the encoder generates the Poisson process $(T_{i})_{i}$,
computes $\tilde{T}_{i}$ and $K$ using $P=P_{Z|X}$, and transmits
$K$ to the decoder. The decoder simply outputs $Z_{K}$, which follows
the conditional distribution $P_{Z|X}$. The issue is that $K$ is
a function of $X$ and the shared randomness $(Z_{i},T_{i})_{i}$,
and a change of $X$ may affect $K$ in a deterministic manner, and
hence this method cannot be directly used to protect the privacy of $X$.

\textbf{Poisson private representation. }
To ensure privacy, we introduce randomness in the encoder by a generalization of the Poisson functional representation, which we
call\emph{ Poisson private representation (PPR)} with parameter $\alpha\in(1,\infty]$,
proposal distribution $Q$ and the simulated mechanism $P_{Z|X}$. Both $X$ and $Z$ can be discrete or continuous, though as a regularity condition, we require $P_{Z|X}(\cdot |X)$ to be absolutely continuous with respect to $Q$ almost surely. The PPR-compressed mechanism is given as:
\begin{enumerate}
\item We use $(Z_{i})_{i=1,2,\ldots}$, $Z_{i}\stackrel{iid}{\sim}Q$ as the shared
randomness between the encoder and the decoder. Practically, the encoder
and the decoder can share a random seed and generate $Z_{i}\stackrel{iid}{\sim}Q$
from it using a pseudorandom number generator.\footnote{We note that our analyses assume that the adversary knows both the index $K$ and the shared randomness $(Z_i)_i$, and we prove that the mechanism is still private despite the shared randomness between the encoder and the decoder, since the privacy is provided by locally randomizing $K$ in Step \ref{step:generate_K}.}
\item The encoder knows $(Z_{i})_{i},X, P_{Z|X}$ and performs the following steps:
\begin{enumerate}[leftmargin=1em,topsep=0pt]
\item Generates the Poisson process $(T_{i})_{i}$ with rate $1$.
\item Computes $\tilde{T}_{i}:=T_{i} \cdot (\frac{\mathrm{d}P}{\mathrm{d}Q}(Z_{i}))^{-1}$,
where $P:=P_{Z|X}(\cdot|X)$. Take $\tilde{T}_{i}=\infty$ if $\frac{\mathrm{d}P}{\mathrm{d}Q}(Z_{i})=0$.
\item \label{step:generate_K}Generates $K\in\mathbb{Z}_{+}$ using local randomness with 
\[
\Pr(K=k)=\frac{\tilde{T}_{k}^{-\alpha}}{\sum_{i=1}^{\infty}\tilde{T}_{i}^{-\alpha}}.
\]
\item Compress $K$ (e.g., using Elias delta coding \cite{elias1975universal})
and sends $K$.
\end{enumerate}
\item The decoder, which knows $(Z_{i})_{i},K$, outputs $Z=Z_{K}$.
\end{enumerate}

We will provide an algorithm with implementation details later.

Note that when $\alpha=\infty$, we have $K=\mathrm{argmin}_{i}\tilde{T}_{i}$,
and PPR reduces to the original Poisson functional representation
\cite{li2018strong,li2021unified}.
PPR can simulate the
privacy mechanism $P_{Z|X}$ precisely, as shown in the following
proposition. 
The proof is in Appendix~\ref{sec:pf_ppr_output_dist}.

\begin{proposition}\label{prop:ppr_output_dist}
The output $Z$ of PPR follows the conditional distribution $P_{Z|X}$
exactly.
\end{proposition}

Due to the \emph{exactness} of PPR, it guarantees unbiasedness for tasks such as DME. 
If the goal is only to design a stand-alone privacy mechanism, we can focus on the privacy and utility of the mechanism without studying the output distribution. 
However, if the output of the mechanism is used for downstream tasks (e.g., for DME, after receiving information from clients, the server sends information about the aggregated mean to data analysts, where central DP is crucial), having an exact characterization of the conditional distribution of the output given the input allows us to obtain precise (central) privacy and utility guarantees. 

Notably, PPR is \emph{universal} in the sense that only the encoder needs to know the simulated mechanism $P_{Z|X}$. The decoder can decode the index $K$ as long as it has access to the shared randomness $(Z_i)_i$. This allows the encoder to choose an arbitrary mechanism $P_{Z|X}$ with the same $\mathcal{Z}$, and adapt the choice of $P_{Z|X}$ to the communication and privacy constraints without explicitly informing the decoder which mechanism is chosen.

Practically, the algorithm cannot compute the whole infinite sequence
$(\tilde{T}_{i})_i$. We can truncate the method and only
compute $\tilde{T}_{i},\ldots,\tilde{T}_{N}$ for a large $N$ and
select $K\in\{1,\ldots,N\}$, which incurs a small distortion in the
distribution of $Z$.\footnote{To compare to the minimal random coding (MRC) \cite{havasi2019minimal,cuff2008communication,song2016likelihood}
scheme in \cite{shah2022optimal}, which also utilizes a finite number
$N$ of samples $(Z_{i})_{i=1,\ldots,N}$, while truncating the number
of samples to $N$ in both PPR and MRC results
in a distortion in the distribution of $Z$ that tends to $0$ as
$N\to\infty$, the difference is that $\log K$ (which is approximately
the compression size) in MRC
grows like $\log N$, whereas $\log K$ does not grow as $N\to\infty$
in PPR. The size $N$ in truncated PPR merely controls the trade-off
between accuracy of the distribution of $Z$ and the running time
of the algorithm.}
While this method is practically acceptable, it might defeat the purpose of having an exact algorithm that ensures the correct conditional distribution $P_{Z|X}$.
In Appendix \ref{sec:reparametrization}, we will present an exact algorithm for PPR that terminates in a finite amount of time, using a reparametrization that allows the encoder to know when the optimal point $Z_i$ has already been encountered (see Algorithm~\ref{alg:ppr} in Appendix \ref{sec:reparametrization}).

By the lower bound for channel simulation \cite{bennett2002entanglement,li2018strong},
we must have $H(K)\ge I(X;Z)$, i.e., the compression size is at least
the mutual information between the data $X$ and the output $Z$.
The following result shows that the compression provided by PPR is
``almost optimal'', i.e., close to the theoretical lower bound $I(X;Z)$.
The proof is given in Appendix~\ref{sec:pf_logk_bound_simple}.

\begin{theorem}[Compression size of PPR]
\label{thm:logk_bound_simple}For PPR with parameter $\alpha>1$,
when the encoder is given the input $x$, the message $K$ given by
PPR satisfies
\begin{align*}
\mathbb{E}[\log K] & \le D(P\Vert Q)+(\log(3.56)) / \min\{(\alpha-1)/2,1\},
\end{align*}
where $P:=P_{Z|X}(\cdot|x)$. As a result, when the input $X\sim P_{X}$
is random, taking $Q=P_{Z}$, we have
\begin{align*}
\mathbb{E}[\log K] & \le I(X;Z)+(\log(3.56)) / \min\{(\alpha-1)/2,1\}.
\end{align*}
\end{theorem}

Note that the running time complexity (which depends on the number of samples $Z_i$ the algorithm must examine before outputting the index $K$) can be quite high. 
Since $\mathbb{E}[\log K] \approx I(X;Z)$, $K$ (and hence the running time) is at least exponential in $I(X;Z)$. 
See more discussions in Section~\ref{sec:limitations}. 

If a prefix-free encoding of $K$ is required, then the
number of bits needed is slightly larger than $\log_2 K$. For example,
if Elias delta code \cite{elias1975universal} is used, the expected
compression size is $\le\mathbb{E}[\log_{2}K]+2\log_{2}(\mathbb{E}[\log_{2}K]+1)+1$
bits. 
If the Shannon code \cite{shannon1948mathematical} (an almost-optimal prefix-free code) for the Zipf distribution $p(k) \propto k^{-\lambda}$ with $\lambda = 1 + 1/\mathbb{E}[\log_{2}K]$ is used, the expected compression size is $\le\mathbb{E}[\log_{2}K]+\log_{2}(\mathbb{E}[\log_{2}K]+1)+2$ bits (see \cite{li2018strong}). 
Both codes yield an $I(X;Z)+O(\log I(X;Z))$
 size, within a logarithmic gap from the lower bound $I(X;Z)$.
This is similar to some other channel simulation schemes such as \cite{harsha2010communication,braverman2014public,li2018strong},
though these schemes do not provide privacy

\begin{remark}
    Note that PPR requires a variable-length code to encode the index $K\in \mathbb{Z}_+$, which is common in channel simulation~\cite{harsha2010communication,li2018strong} and distributed mean estimation~\cite{suresh2017distributed}. If we impose a fixed limit of $b$ bits on the encoding, since Theorem \ref{thm:logk_bound_simple} and Markov's inequality yields $\Pr(\log_2 K > b) \le P_e := I(X;Z)/b + \log_2(3.56) / (b \min\{(\alpha-1)/2,1\})$
\end{remark}

Note that if $P_{Z|X}$ is $\varepsilon$-DP, then by definition, for any $z \in \mcal{Z}$ and $x, x_0 \in \mcal{X}$, it holds that
\[
D\lp P_{Z|X = x} \middle \Vert P_{Z|X = x_0} \rp = \mathbb{E}_{Z \sim P_{Z|X = x}}\lb \log\lp \frac{\mathrm{d}P_{Z|X = x}}{\mathrm{d}P_{Z|X = x_0}}(Z)  \rp \rb \leq \varepsilon \log e. 
\]
Setting the proposal distribution $Q = P_{Z|X = x_0}$ for an arbitrary $x_0 \in \mcal{X}$ gives the following bound.
\begin{corollary}[Compression size under $\varepsilon$-LDP]\label{cor:generic_compression_bdd}
    Let $P_{Z|X}$ satisfy $\varepsilon$-differential privacy. Let $x_0 \in \mcal{X}$ and $Q = P_{Z|X = x_0}$. Then for PPR with parameter $\alpha > 1$, the expected compression size is at most $\ell + \log_2(\ell+1)+2$ bits,
    where $\ell \eqDef \varepsilon \log_2 e + {(\log_2(3.56))}/{\min\lbp (\alpha - 1)/2, 1 \rbp}$.
\end{corollary}

Next, we analyze the privacy guarantee of PPR. The PPR method induces
a conditional distribution
$P_{(Z_{i})_{i},K|X}$
of the knowledge of the decoder $((Z_{i})_{i},K)$, given the
data $X$. To analyze the privacy guarantee, we study whether
the randomized mapping $P_{(Z_{i})_{i},K|X}$ from $X$ to $((Z_{i})_{i},K)$
satisfies $\varepsilon$-DP or $(\varepsilon,\delta)$-DP.
\footnote{Note that the encoder does not actually send $((Z_{i})_{i},K)$; it
only sends $K$. 
The common randomness $(Z_{i})_{i}$ is independent of the data $X$, and can be pre-generated using a common random seed in practice. 
While this seed must be communicated between the client and the server as a small overhead, the client and the server only ever need to communicate \emph{one} seed to initialize a pseudorandom number generator, that can be used in \emph{all} subsequent privacy mechanisms and communication tasks (to transmit high-dimensional data or use DP mechanisms for many times). The conditional distribution
$P_{(Z_{i})_{i},K|X}$ is only relevant for privacy analysis.}
This is similar to the privacy condition in \cite{shah2022optimal},
and is referred as \emph{decoder privacy} in \cite{shahmiri2024communication},
which is stronger than \emph{database privacy} which concerns the
privacy of the randomized mapping from $X$ to the final output $Z$
\cite{shahmiri2024communication} (which is simply the privacy of
the original mechanism $P_{Z|X}$ to be compressed since PPR simulates
$P_{Z|X}$ precisely). 
Since the decoder knows $((Z_{i})_{i},K)$,
more than just the final output $Z$, we expect that the PPR-compressed
mechanism $P_{(Z_{i})_{i},K|X}$ to have a worse privacy guarantee
than the original mechanism $P_{Z|X}$, which is the price of having
a smaller communication cost.  The following result shows that, if
the original mechanism $P_{Z|X}$ is $\varepsilon$-DP, then the PPR-compressed
mechanism is guaranteed to be $2\alpha\varepsilon$-DP.

\begin{theorem}[$\varepsilon$-DP of PPR]
\label{thm:eps_dp}If the mechanism $P_{Z|X}$ is $\varepsilon$-differentially
private, then PPR $P_{(Z_{i})_{i},K|X}$ with parameter $\alpha>1$
is $2\alpha\varepsilon$-differentially private.
\end{theorem}

Similar results also apply to $(\varepsilon,\delta)$-DP and metric DP.

\begin{theorem}[$(\varepsilon,\delta)$-DP of PPR]
\label{thm:eps_delta_dp_2}If the mechanism $P_{Z|X}$ is $(\varepsilon,\delta)$-differentially
private, then PPR $P_{(Z_{i})_{i},K|X}$ with parameter $\alpha>1$
is $(2\alpha\varepsilon, 2\delta)$-differentially private.
\end{theorem}

\begin{theorem}[Metric privacy of PPR]
\label{thm:metric_privacy}If the mechanism $P_{Z|X}$ satisfies $\varepsilon \cdot d_{\mathcal{X}}$-privacy, then PPR $P_{(Z_{i})_{i},K|X}$ with parameter $\alpha>1$
satisfies $2\alpha\varepsilon \cdot d_{\mathcal{X}}$-privacy.
\end{theorem}

Refer to Appendices~\ref{sec:pf_eps_dp} and \ref{sec:eps_delta_dp_2} for the proofs.
In Theorem \ref{thm:eps_dp} and Theorem \ref{thm:eps_delta_dp_2}, PPR imposes a multiplicative penalty
$2\alpha$ on the privacy parameter $\varepsilon$. 
This penalty can be made arbitrarily
close to $2$ by taking $\alpha$ close to $1$, which increases the
communication cost (see Theorem \ref{thm:logk_bound_simple}). 
Compared to minimal random coding which has a factor $2$ penalty in the DP guarantee \cite{havasi2019minimal,shah2022optimal}, the $2\alpha$ factor in PPR is slightly larger, though PPR ensures exact simulation (unlike  \cite{havasi2019minimal,shah2022optimal} which are approximate). 
The method in \cite{feldman2021lossless} does not have a penalty on $\varepsilon$, but the utility and compression size depends on computational hardness assumptions on the pseudorandom number generator, and there is no guarantee that the compression size is close to the optimum. In comparison, the compression and privacy guarantees of PPR are \emph{unconditional} and does not rely on computational assumptions.
In order to
make the penalty of PPR close to $1$, we have to consider $(\varepsilon,\delta)$-differential
privacy, and allow a small failure probability, i.e., a small increase in $\delta$. 
The following
result shows that PPR can compress any $\varepsilon$-DP mechanism into
a $(\approx \varepsilon,\, \approx 0)$-DP mechanism as long as $\alpha$
is close enough to $1$ (i.e., almost no inflation). 
More generally, PPR can compress an $(\varepsilon,\delta)$-DP
mechanism into an $(\approx\varepsilon,\,\approx2\delta)$-DP mechanism
for $\alpha$ close to $1$. 
The proof is in Appendix \ref{sec:pf_eps_delta_dp}. 

\begin{theorem}[Tighter $(\varepsilon,\delta)$-DP of PPR]
\label{thm:eps_delta_dp}If the mechanism $P_{Z|X}$ is $(\varepsilon,\delta)$-differentially private, then PPR $P_{(Z_{i})_{i},K|X}$ with parameter $\alpha>1$ is $(\alpha\varepsilon+\tilde{\varepsilon},\,2(\delta+\tilde{\delta}))$-differentially
private, for every $\tilde{\varepsilon}\in(0,1]$ and $\tilde{\delta}\in(0,1/3]$
that satisfy
$\alpha\le e^{-4.2}\tilde{\delta}\tilde{\varepsilon}^{2} / (-\ln\tilde{\delta}) +1$.
\end{theorem}

\begin{remark}

For the computation-privacy trade-off, in general, a larger $\alpha$ results in a smaller compression size (i.e., smaller $K$ and hence shorter running time) but larger privacy leakage, while a smaller $\alpha$ leads to worse compression but better privacy guarantee. Regarding the randomness requirement, although in the theory of this paper we assume ``unlimited common randomness'', it is of interest to study the trade-off between the amount of common randomness used and the required communication cost, similar to the study in~\cite{cuff2013distributed}. 
We leave the investigation of the randomness-communication-privacy trade-off for future work.

\end{remark}

\section{PPR on Distributed Mean Estimation}\label{sec:mean_estimation}

We demonstrate the efficacy of PPR by applying it to distributed mean estimation (DME) \citep{suresh2017distributed}. 
Note that this problem is closely related to the federated learning problems~\citep{konevcny2016federated, kairouz2021advances}, or similar stochastic optimization problems, e.g., ~\cite{mcmahan2017communication}. 
In the DME problem, for each iteration, the server sends a message to update the global model by a noisy mean of the local model updates. 
The noisy estimation is usually from some DME framework, hence a distributed differentially-private SGD (or a differentially-private federated learning) can be constructed based on a differentially-private DME framework. 
For such problems, as discussed in~\cite{ghadimi2013stochastic}, if the estimate of the gradient is unbiased in each round, the convergence rates of SGD are dependent on the $\ell_2$ estimation error. 
In short, private DME is the core sub-routine in various private and federated optimization algorithms, such as DP-SGD \citep{abadi2016deep} or DP-FedAvg \citep{mcmahan2017communication}.

In such distributed settings, each local client communicates a length-limited message to the central server, and the privacy (explicit differential privacy guarantee~\citep{dwork2006calibrating}) of the data can be guaranteed by adding noise to the estimated mean at the central server before releasing it to downstream components. 
For example, after estimating the average model update, the central server corrupts it with the addition of Gaussian noise).
This is usually called the trusted server or \emph{central DP} guarantee (see Section~\ref{sec::preliminary} for definitions), since the central server is trusted in privatizing the computed mean, and it is one of the most common methods in practice for federated learning and analytics. 
However, as we have discussed above, our scheme not only achieves the same level of central DP guarantee, but also ensures local DP guarantee.

Consider the following general distributed setting: each of $n$ clients holds a local data point $X_i\in \mcal{X}$, and a central server aims to estimate a function of all local data $\mu\lp X^n\rp$, subject to privacy and local communication constraints. To this end, each client $i$ compresses $X_i$ into a message $Z_i\in\mathcal{Z}_n$ via a local encoder, and we require that each $Z_i$ can be encoded into a bit string with an expected length of at most $b$ bits.
Upon receiving $Z^n:=\left(Z_1,\ldots, Z_n\right)$, the central server decodes it and outputs a DP estimate $\hat{\mu}$. 
Two DP criteria can be considered: the $(\varepsilon,\delta)$-central DP of the randomized mapping from $X^n$ to $\hat{\mu}$, and the $(\varepsilon,\delta)$-local DP of the randomized mapping from $X_i$ to $Z_i$ for each client $i$. 

In the distributed $L_2$ mean estimation problem, 
\begin{equation*}
    \mcal{X} = \mcal{B}_d(C) := \lbp v \in \mbb{R}^d \,\mv\, \lV v \rV_2 \leq C \rbp,
\end{equation*}
and the central server aims to estimate the sample mean $\mu(X^n) := \frac{1}{n} \sum_{i=1}^n X_i$ by minimizing the mean squared error (MSE) $\mathbb{E}[\lV\mu - \hat{\mu}\rV_2^2]$. It is recently proved that under $\varepsilon$-local DP, $\msf{privUnit}$ \citep{bhowmick2018protection, asi2022optimal} is the optimal mechanism. By simulating $\msf{privUnit}$ with PPR and applying Corollary~\ref{cor:generic_compression_bdd} and Theorem~\ref{thm:eps_delta_dp_2}, we immediately obtain the following corollary:

\begin{corollary}[PPR simulating $\msf{privUnit}$] Let $P$ be the density defined by $\varepsilon$-$\msf{privUnit}_2$ \citep[Algorithm~1]{bhowmick2018protection}. Let $Q$ be the uniform density over the sphere $\mbb{S}^{d-1}\lp 1/m \rp$ where the radius $1/m$ is defined in \citep[(15)]{bhowmick2018protection}. Let $r^* := e^\varepsilon$. Then the outcome of PPR (see Algorithm~\ref{alg:ppr}) satisfies (1) $2\alpha\varepsilon$-local DP; and (2) $(\alpha\varepsilon+\tilde{\varepsilon}, 2\tilde{\delta})$-DP for any $\alpha \leq e^{-4.2}\tilde{\delta}\tilde{\varepsilon}^2/\log(1/\tilde{\delta}) +1$. In addition, the average compression size is 
at most $\ell + \log_2(\ell + 1) + 2$ bits where $\ell := \varepsilon + (\log_2\lp 3.56 \rp)/\min\{(\alpha-1)/2, 1\}$.
Moreover, PPR achieves the same MSE as $\varepsilon$-$\msf{privUnit}_2$, which is $O\lp d/\min\lp \varepsilon, \varepsilon^2 \rp \rp$.
\end{corollary}

Note that PPR can simulate arbitrary local DP mechanisms. However, we present only the result of $\msf{privUnit}_2$ because it achieves the optimal privacy-accuracy trade-off. Besides simulating local DP mechanisms, PPR can also compress central DP mechanisms while still preserving some (albeit weaker) local guarantees. 
We give a corollary of Theorems~\ref{thm:logk_bound_simple} and~\ref{thm:eps_delta_dp_2}. The proof is in Appendix~\ref{sec:pf_gaussian_ppr2}.
\begin{corollary}[PPR-compressed Gaussian mechanism]\label{cor:gaussian_ppr2} 
Let $\varepsilon,\delta \in (0,1)$. Consider the Gaussian mechanism $P_{Z|X}(\cdot | x) = \mcal{N}( x, \frac{\sigma^2}{n}\mbb{I}_d)$, and the proposal distribution $Q=\mathcal{N}(0,(\frac{C^{2}}{d}+\frac{\sigma^{2}}{n})\mathbb{I}_{d})$, where $\sigma \geq \frac{C\sqrt{2\ln\lp 1.25/\delta \rp}}{\varepsilon}$.
For each client $i$, let 
$Z_i$ be the output of PPR applied on $P_{Z|X}(\cdot | X_i)$. We have:
\begin{itemize}
    \item $\hat{\mu}(Z^n) := \frac{1}{n}\sum_i Z_i$ yields an unbiased estimator of $\mu( X^n) = \frac{1}{n} \sum_{i=1}^n X_i$ satisfying $(\varepsilon, \delta)$-central DP and has MSE $\mathbb{E}[\lV\mu - \hat{\mu}\rV_2^2] = \sigma^2 d/n^2$.

    \item As long as $\varepsilon < 1/\sqrt{n}$, PPR satisfies $\lp 2\alpha\sqrt{n}\varepsilon, 2\delta\rp$-local DP.\footnote{The restricted range on $\varepsilon < 1/\sqrt{n}$ is due to the simpler privacy accountant \citep{dwork2006differential}. By using the R\'enyi DP accountant instead, one can achieve a tighter result that applies to any $n$. We present the R\'enyi DP version of the corollary in Appendix~\ref{sec::appendix_mean_est}. 
    Moreover, in the context of federated learning, $n$ refers to the number of clients in \emph{each round}, which is typically much smaller than the total number of clients. 
    For example, as observed in~\cite{kairouz2021advances}, the per-round cohort size in Google's FL application typically ranges from $10^3$ to $10^5$, significantly smaller than the number of trainable parameters $d\in [10^6, 10^9]$ or the number of available users $N\in [10^6, 10^8]$. } 
    
    \item The average per-client communication cost is at most $\ell + \log_2(\ell + 1) + 2$ bits where 
\begin{align*}
\ell & :=\frac{d}{2}\log_{2}\Big(\frac{C^{2}n}{d\sigma^{2}}+1\Big)+\eta_{\alpha} \; \le\;\frac{d}{2}\log_{2}\Big(\frac{n\varepsilon^{2}}{2d\ln(1.25/\delta)}+1\Big)+\eta_{\alpha},
\end{align*}
where $\eta_{\alpha}:=(\log_{2}(3.56))/\min\{(\alpha-1)/2,\,1\}$.
\end{itemize}
\end{corollary}

A few remarks are in order. First, notice that when $\alpha$ is fixed,
for an $O(\frac{C^2 d}{n^2\varepsilon^2}\log(1/\delta))$ MSE, the per-client communication cost is 
\[
O\Big(d\log\Big(\frac{n\varepsilon^{2}}{d\log(1/\delta)}+1\Big)+1\Big),
\]
which is at least as good as the $O(n\varepsilon^{2}/\log(1/\delta)+1)$
bound in \citep{suresh2017distributed, chen2024privacy}, and can be better than $O(n\varepsilon^{2}/\log(1/\delta)+1)$
when $n \gg d$. Hence, the PPR-compressed Gaussian mechanism
is order-wise optimal.
Second, compared to other works that also compress the Gaussian mechanism, PPR is the only lossless compressor; schemes based on random sparsification, projection, or minimum random coding (e.g., \citep{triastcyn2021dp, chen2024privacy}) are \emph{lossy}, i.e., they introduce additional distortion on top of the DP noise. Finally, other DP mechanism compressors tailored to local randomizers \citep{feldman2021lossless, shah2022optimal} do not provide the same level of central DP guarantees when applied to local Gaussian noise since the reconstructed noise is no longer Gaussian. Refer to Section~\ref{sec::exp_dme} for experiments.

\section{Empirical Results on Distributed Mean Estimation}
\label{sec::exp_dme}

We empirically evaluate our scheme on the Distributed Mean Estimation (DME) problem (which is formally introduced in Section~\ref{sec:mean_estimation}).

We examine the privacy-accuracy-communication trade-off, and compare it with the Coordinate Subsampled Gaussian Mechanism (CSGM)  \citep[Algorithm 1]{chen2024privacy}, an order-optimal scheme for DME under central DP.
In~\citep{chen2024privacy}, each client only communicates partial information (via sampling a subset of the coordinates of the data vector) about its samples to amplify the privacy, and the compression is mainly from subsampling. 
Moreover, CSGM only guarantees central DP.

\subsection{Experiment}

We use the same setup that has been used in~\cite{chen2024privacy}: 
consider $n=500$ clients, and the dimension of local vectors is $d=1000$, each of which is generated according to $X_i(j)\overset{\mathrm{i.i.d.}}{\sim} 
\left(2\cdot \mathrm{Ber}(0.8) - 1 \right)$, where $\mathrm{Ber}(0.8)$ is a Bernoulli random variable with parameter $p = 0.8$. 
We require $(\varepsilon,\delta)$-central DP with $\delta = 10^{-6}$ and $\varepsilon \in [0.05, 6]$ and apply the PPR with $\alpha =  2$ to simulate the Gaussian mechanism, where the privacy budgets are accounted via R\'enyi DP.

\begin{figure}
	\centering
    \includegraphics[scale = 0.4]{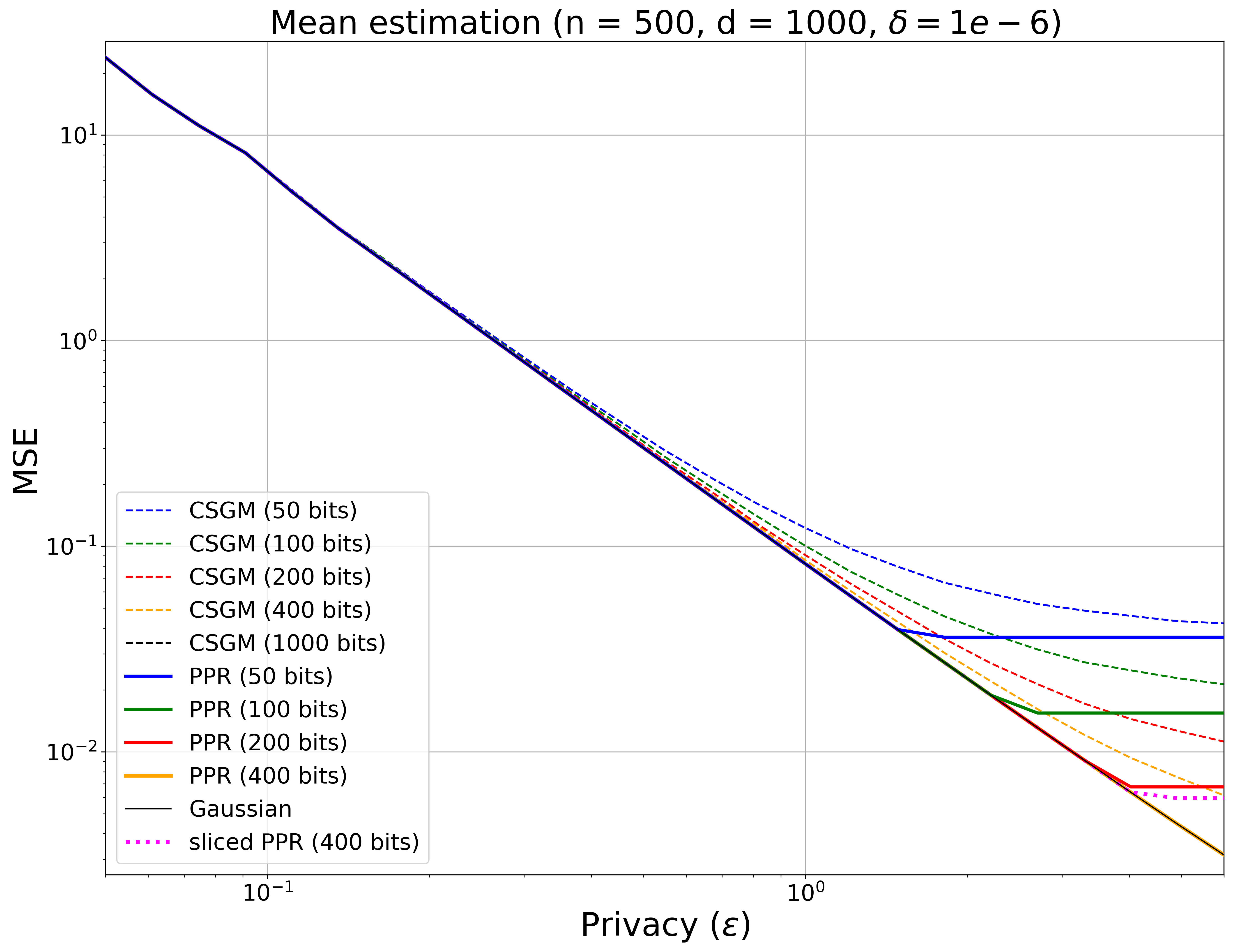}
	\caption{MSE of distributed mean estimation for PPR and CSGM~\citep{chen2024privacy} for different $\varepsilon$'s. }
\label{fig:experiment_log}
\end{figure}

We compare the MSE of PPR ($\alpha =  2$, using Theorem \ref{thm:logk_bound_simple}) and CSGM under various compression sizes in Figure~\ref{fig:experiment_log} (the $y$-axis is in logarithmic scale).\footnote{
Source code: \url{https://github.com/cheuktingli/PoissonPrivateRepr}.
Experiments were executed on M1 Pro Macbook, $8$-core CPU ($\approx 3.2$ GHz) with 16GB memory. For PPR under a privacy budget $\varepsilon$ and communication budget $b$, we find the largest $\varepsilon' \le \varepsilon$ such that the communication cost bound in Theorem \ref{thm:logk_bound_simple} (with Shannon code \cite{shannon1948mathematical}) for simulating the Gaussian mechanism with $(\varepsilon',\delta)$-central DP is at most $b$, and use PPR to simulate this Gaussian mechanism. Thus, MSE of PPR in Figure~\ref{fig:experiment_log} becomes flat for large $\varepsilon$, as PPR falls back to using a smaller $\varepsilon'$ instead of $\varepsilon$ due to the communication budget. 
} 
Note that the MSE of the (uncompressed) Gaussian mechanism coincides with the CSGM with $1000$ bits, and the PPR with only $400$ bits. 
We see that PPR consistently achieves a smaller MSE compared to CSGM for all $\varepsilon$'s and compression sizes considered.
For $\epsilon=1$ and we compress $d=1000$ to $50$ bits, CSGM has an MSE $0.1231$ , while PPR has an MSE $0.08173$, giving a $33.61\%$ reduction. 
For $\epsilon=0.5$ and we compress $d=1000$ to $25$ bits (the case of high compression and conservative privacy), CSGM has an MSE $0.3877$, while PPR has an MSE $0.3011$, giving a $22.33\%$ reduction. 
These reductions are significant, since all considered mechanisms are asymptotically close to optimal and a large improvement compared to an (almost optimal) mechanism is unexpected.
See Section~\ref{sec::mse_size} for more about MSE against the compression sizes.

We also emphasize that PPR provides \emph{both} central and local DP guarantees according to Theorem \ref{thm:eps_dp}, \ref{thm:eps_delta_dp_2} and \ref{thm:eps_delta_dp}, benefiting from the fact that PPR \emph{exactly} compresses the privacy
mechanism and hence control the distributions exactly. 
In contrast, CSGM only provides central DP guarantees.
Another advantage of PPR under conservative privacy (small $\epsilon$) is that the trade-off between $\epsilon$ and MSE of PPR exactly coincides with the trade-off of the Gaussian mechanism for small $\epsilon$ (see Figure~\ref{fig:experiment_log}), and CSGM is only close to (but strictly worse than) the Gaussian mechanism. This means that for small $\epsilon$, PPR provides compression without any drawback in terms of $\epsilon$-MSE trade-off compared to the Gaussian mechanism (which requires an infinite size communication to exactly realize). 

Moreover, although directly applying PPR on the $d$-dimensional vectors is impractical for a large $d$,
one can ensure an efficient $O(d)$ running time (see Section~\ref{sec:limitations} for details) by breaking the vector with $d=1000$ dimensions into small chunks of fixed lengths (we use $d_{\mathrm{chunk}} = 50$ dimensions for each chunk), and apply the PPR to each chunk. 
We call it the \emph{sliced PPR} in Figure~\ref{fig:experiment_log}. 
Though the sliced PPR has a small penalty on the MSE (as shown in Figure~\ref{fig:experiment_log}), it still outperforms the CSGM ($400$ bits) for the range of $\varepsilon$ in the plot. 
For the sliced PPR for one $d=1000$ vector, when $\epsilon = 0.05$, the running time is $1.3348$ seconds on average.\footnote{The running time is calculated by $\frac{1000}{50} \times T_{\mathrm{chunk}}$, where each chunk's running time $T_{\mathrm{chunk}}$ is averaged over $1000$ trials. The estimate of the mean of $T_{\mathrm{chunk}}$ is $0.0667$, whereas the standard deviation is $0.2038$.} 
For larger $\epsilon$'s, we can choose smaller $d_{\mathrm{chunk}}$'s to have reasonable running time: 
For $\epsilon=6$ and $d_{\mathrm{chunk}} = 2$ we have an average running time $0.0127$ seconds and with $d_{\mathrm{chunk}} = 4$ we have an average running time $0.6343$ seconds; for $\epsilon=10$ and $d_{\mathrm{chunk}} = 2$ we have an average running time $0.0128$ seconds and with $d_{\mathrm{chunk}} = 4$ we have an average running time $0.7301$ seconds. 
See Section~\ref{sec::dics_runtime} for more experiments on the running time of the sliced PPR.

\section{Applications to Metric Privacy\label{sec:app_metric}} 

Metric privacy~\cite{chatzikokolakis2013broadening,andres2013geo} (see Definition~\ref{def:metric_privacy}) allows users to send privatized version $Z \in \mathbb{R}^d$ of their data vectors $X \in \mathbb{R}^d$ to an untrusted server, so that the server can know $X$ approximately but not exactly. 
A popular mechanism is the \emph{Laplace mechanism}~\cite{chatzikokolakis2013broadening,andres2013geo,fernandes2019generalised,feyisetan2020privacy}, where a $d$-dimensional Laplace noise is added to $X$. The conditional density function of $Z$ given $X$ is
$f_{Z|X}(z|x) \propto e^{-\varepsilon d_{\mathcal{X}}(x,z)}$, where $\varepsilon$ is the privacy parameter, and the metric $d_{\mathcal{X}}(x,z)=\Vert x-z \Vert_2$ is the Euclidean distance.
The Laplace mechanism achieves $\varepsilon \cdot d_{\mathcal{X}}$-privacy, and has been used, for example, to privatize high-dimensional word embedding vectors~\cite{fernandes2019generalised,feyisetan2020privacy}, or for geo-indistinguishability~\cite{andres2013geo} to privatize the users' locations, where the purpose is to allow users
to send privatized version of their location information to an
untrusted server, so that the server can approximate the locations (to provide some remote services) without knowing the exact locations.

A problem is that the real vector $Z$ cannot be encoded into finitely many bits. To this end, \cite{andres2013geo} studies a \emph{discrete Laplace mechanism} where each coordinate of $Z$ is quantized to a finite number of levels, introducing additional distortion to $Z$.
PPR provides an alternative compression method that preserves the statistical behavior of $Z$ (e.g., unbiasedness) exactly.
We give a corollary of Theorems~\ref{thm:logk_bound_simple} and~\ref{thm:metric_privacy}. The proof is in Appendix~\ref{sec:pf_laplace_ppr}.
Refer to Section~\ref{sec:emprical_metricDP} for an experiment on metric privacy.

\begin{corollary}[PPR-compressed Laplace mechanism]\label{cor:laplace_ppr} 
Consider PPR applied to the Laplace mechanism $P_{Z|X}$ where $X\in\mathcal{B}_{d}(C)=\{x\in\mathbb{R}^{d}\,|\,\Vert x\Vert_{2}\le C\}$,
with a proposal distribution $Q=\mathcal{N}(0,(\frac{C^{2}}{d}+\frac{d+1}{\varepsilon^{2}})\mathbb{I}_{d})$.
It achieves an MSE $\frac{d(d+1)}{\varepsilon^{2}}$, a $2\alpha\epsilon\cdot d_{\mathcal{X}}$-privacy,
and a compression size at most $\ell+\log_{2}(\ell+1)+2$ bits, where
\begin{align*}
\ell & :=\frac{d}{2}\log_{2}\left(\frac{2}{e}\left(\frac{C^{2}\varepsilon^{2}}{d}+d+1\right)\right)-\log_{2}\frac{\Gamma(d+1)}{\Gamma(\frac{d}{2}+1)}+\eta_{\alpha},
\end{align*}
where $\eta_{\alpha}:=(\log_{2}(3.56))/\min\{(\alpha-1)/2,\,1\}$.
\end{corollary}

\section{Empirical Results on Metric Privacy}
\label{sec:emprical_metricDP}

In~\cite{andres2013geo}, to privatize the users' location information for some remote services provided by an untrusted server, $2$-dimensional Laplace noises have been used~\cite{andres2013geo} to obtain metric privacy, where the continuous planar Laplace mechanism~\cite{andres2013geo} is given by the following conditional density function $f_{Z|X}(z|x) = \frac{\varepsilon^2}{2\pi}e^{-\varepsilon d_{\mathcal{X}}(x,z)}$.

We use PPR to simulate the Laplace mechanism~\cite{andres2013geo,fernandes2019generalised,feyisetan2020privacy} $f_{Z|X}(z|x) \propto e^{-\varepsilon d_{\mathcal{X}}(x,z)}$ discussed in Section~\ref{sec:app_metric}.
We consider $X \in \mathcal{B}_d(C)$ where $C=10000$ and $d=500$. A large number of dimensions $d$ is common, for example, in privatizing word embedding vectors~\cite{fernandes2019generalised,feyisetan2020privacy}.
We compare the performance of PPR-compressed Laplace mechanism (Corollary~\ref{cor:laplace_ppr}) with the discrete Laplace mechanism~\cite{andres2013geo}. The discrete Laplace mechanism is described as follows (slightly modified from~\cite{andres2013geo} to work for the $d$-ball $\mathcal{B}_d(C)$): 1) generate a Laplace noise $Y$ with probability density function $f_Y(y) \propto e^{-\varepsilon \Vert y \Vert_2}$; 2) compute $\hat{Z}=X+Y$; 3) truncate $\hat{Z}$ to the closest point $Z$ in $\mathcal{B}_d(C)$; and 4) quantize each coordinate of $Z$ by a quantizer with step size $u>0$. The number of bits required by the discrete Laplace mechanism is $\lceil \log_2(\mathrm{Vol}(\mathcal{B}_d(C)) / u^d) \rceil$, where $\mathrm{Vol}(\mathcal{B}_d(C)) / u^d$ is the number of quantization cells (hypercube of side length $u$) inside $\mathcal{B}_d(C)$. The parameter $u$ is selected to fit the number of bits allowed.

Figure \ref{fig:experiment_laplace} shows the mean squared error of PPR-compressed Laplace mechanism ($\alpha=2$) and the discrete Laplace mechanism for different $\varepsilon$'s, when the number of bits is limited to $500$, $1000$ and $1500$.\footnote{The MSE of PPR is computed using the closed-form formula in Corollary~\ref{cor:laplace_ppr}, which is tractable since $Z$ follows the Laplace conditional distribution $f_{Z|X}$ exactly. The number of bits used by PPR is given by the bound in Corollary~\ref{cor:laplace_ppr}. The MSE of the discrete Laplace mechanism is estimated using $5000$ trials per data point. Although we do not plot the error bars, the largest coefficient of variation of the sample mean (i.e., standard error of the mean divided by the sample mean) is only 0.00117, which would be unnoticeable even if the error bars were plotted.} We can see that PPR performs better for larger $\epsilon$ or smaller MSE, whereas the discrete Laplace mechanism performs better for smaller $\epsilon$ or larger MSE. The performance of discrete Laplace mechanism for smaller $\epsilon$ is due to the truncation step which limits $Z$ to $\mathcal{B}_d(C)$, which reduces the error at the expense of introducing distortion to the distribution of $Z$, and making $Z$ a biased estimate of $X$. In comparison, PPR preserves the Laplace conditional distribution $f_{Z|X}$ exactly, and hence produces an unbiased $Z$.

\begin{figure}
	\centering
    \includegraphics[scale = 0.6]{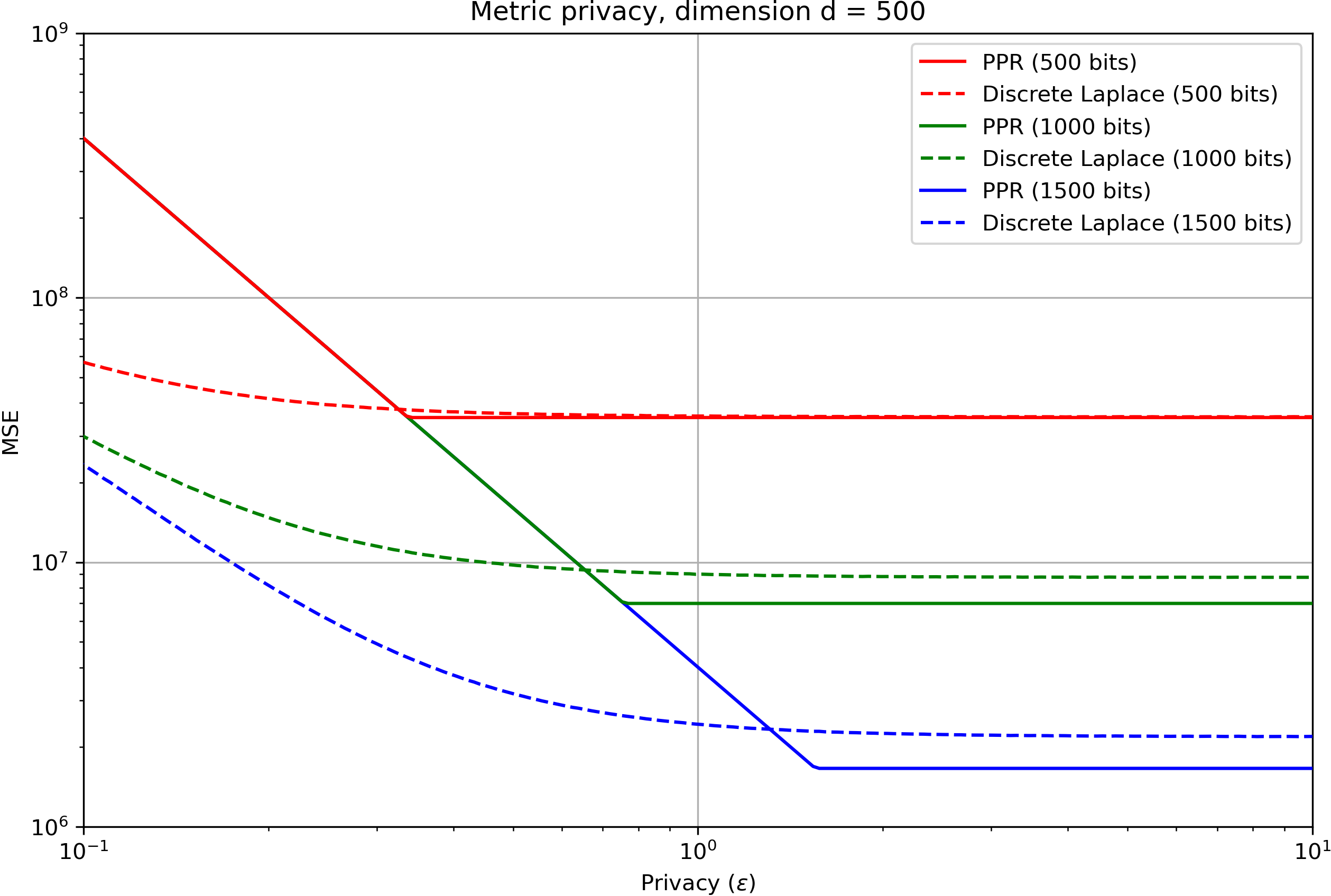}
	\caption{MSE of PPR-compressed Laplace mechanism and discrete Laplace mechanism~\cite{andres2013geo} for different $\varepsilon$'s. }
\label{fig:experiment_laplace}
\end{figure}

\section{Running Time of PPR\label{sec:limitations}}
\label{sec::dics_runtime}

\subsection{Discussions}

While PPR is communication-efficient, having only a logarithmic gap from the theoretical lower bound on the compression size as shown in Theorem~\ref{thm:logk_bound_simple}, the running time complexity can be high. However, we note that an exponential complexity is also needed in sampling methods that do not ensure privacy, such as \citep{maddison2016poisson,havasi2019minimal}. 
It has been proved in~\citep{agustsson2020universally} that no polynomial time general sampling-based method exists (even without privacy constraint), if $RP \neq NP$. All existing polynomial time exact channel simulation methods can only simulate specific noisy channels.\footnote{For example, \cite{flamich2023adaptive} and dithered-quantization-based schemes \citep{hegazy2022randomized,shahmiri2024communication} can only simulate additive noise mechanisms. Among these existing works, only \citep{shahmiri2024communication} ensures local DP.} 
Hence,
a polynomial time algorithm for exactly 
compressing a general DP mechanism is likely nonexistent.

Nevertheless, this is not an obstacle for simulating local DP mechanisms, since the mutual information $I(X;Z)$ for a reasonable local DP mechanism must be small, or else the leakage of the data $X$ in $Z$ would be too large. For an $\varepsilon$-local DP mechanism, we have $I(X;Z) \le \min\{\varepsilon, \varepsilon^2\}$ (in nats) \citep{cuff2016differential}. Hence, the PPR algorithm can terminate quickly even if has a running time exponential in $I(X;Z)$.

Another way to ensure a polynomial running time is to divide the data into small chunks and apply the mechanism to each chunk separately. 
For example, to apply the Gaussian mechanism to a high-dimensional vector, we break it into several shorter vectors and apply the mechanism to each vector. Experiments in Section~\ref{sec::exp_dme} show that this greatly reduces the running time while having only a small penalty on the compression size.

\subsection{Empirical Results}

We show empirical results on the running time of PPR on distributed mean estimation task, as discussed in Section~\ref{sec:mean_estimation} and Section~\ref{sec::exp_dme}.

\subsection{Running Time of Sliced PPR against chunk size}

As discussed in Section~\ref{sec::exp_dme}, we can ensure an $O(d)$ running time for the Gaussian mechanism by using the sliced PPR, where the $d$-dimensional vector $X$ is divided into $\lceil d/d_{\mathrm{chunk}} \rceil$ chunks, each with a fixed dimension $d_{\mathrm{chunk}}$ (possibly except the last chunk if $d_{\mathrm{chunk}}$ is not a factor of $d$). 
The average total running time is $\lceil d/d_{\mathrm{chunk}} \rceil T_{\mathrm{chunk}}$, where $T_{\mathrm{chunk}}$ is the average running time of PPR applied on one chunk.\footnote{Note that the chunks may be processed in parallel for improved efficiency.}
Therefore, to study the running time of the sliced PPR, we study how $T_{\mathrm{chunk}}$ depend on $d_{\mathrm{chunk}}$.
    
In Figure~\ref{fig:time_err} we show the running time $T_{\mathrm{chunk}}$ of PPR applied on one chunk with dimension $d_{\mathrm{chunk}}$, where $d_{\mathrm{chunk}}$ ranges from $40$ to $110$.\footnote{Experiments were executed on M1 Pro Macbook, $8$-core CPU ($\approx 3.2$ GHz) with 16GB memory.} 
With $d=1000$, $n=500$, $\varepsilon = 0.05$ and $\delta = 10^{-6}$, we require a Gaussian mechanism with noise $\mathcal{N}(0, n \tilde{\sigma}^2 \mathbb{I}_{d_{\mathrm{chunk}}})$ where 
$\tilde{\sigma} = 1.0917$ 
at each user in order to ensure $(\varepsilon,\delta)$-central DP.
We record the mean $T_{\mathrm{chunk}}$ and the standard error of the mean\footnote{\label{footnote:stderr}The standard error of the mean is given by $\sigma_{\mathrm{mean}} = \sigma_{\mathrm{time}} / \sqrt{n_{\mathrm{trials}}}$, where $\sigma_{\mathrm{time}}$ is the standard deviation of the running time among the $n_{\mathrm{trials}}=20000$ trials.} of the running time of PPR applied to simulate this Gaussian mechanism (averaged over $20000$ trials). 


\begin{figure}[H]
	\centering
    \includegraphics[scale = 0.27]{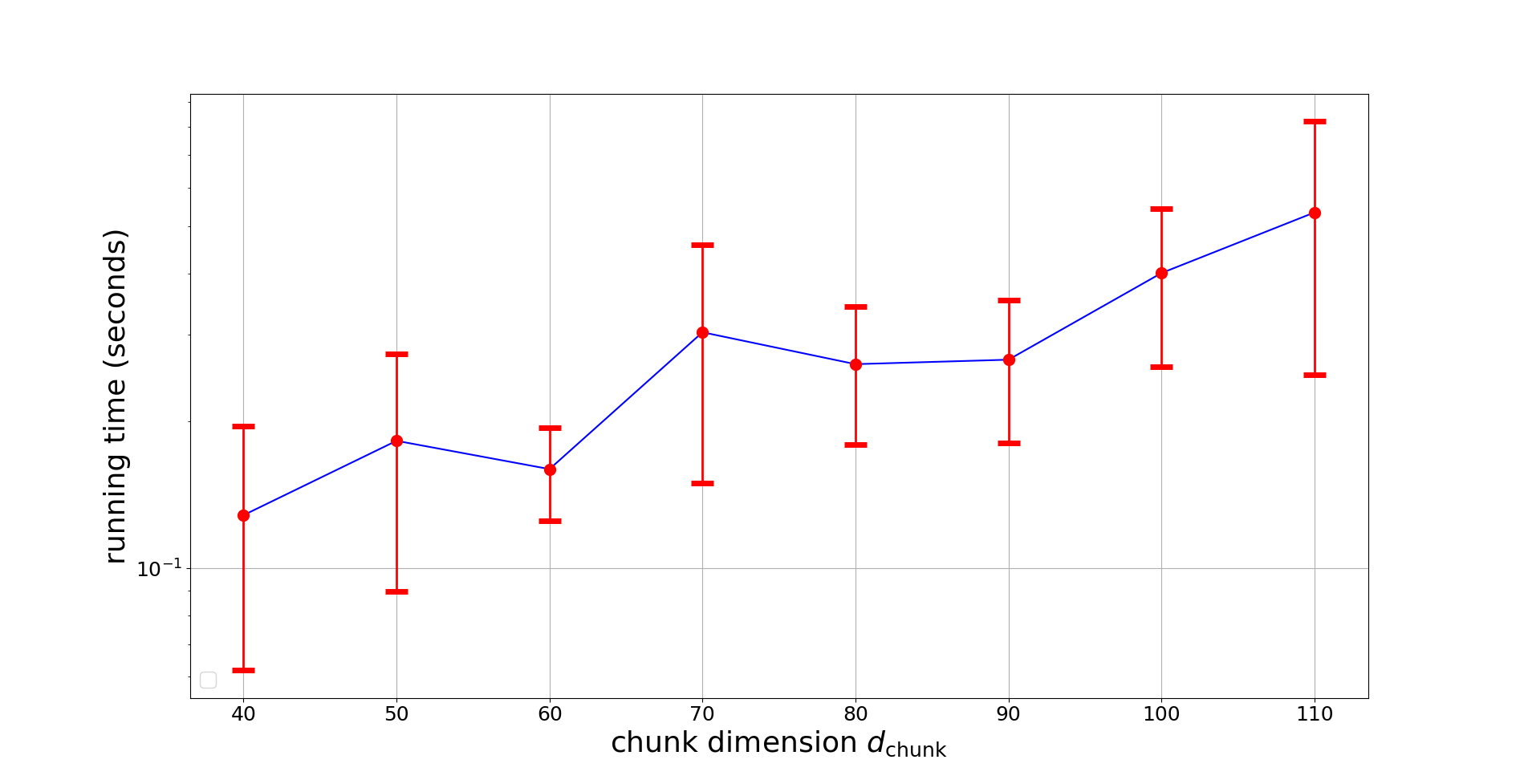}
	\caption{Average running time of PPR applied to a chunk of dimension $d_{\mathrm{chunk}}$, with error bars indicating the interval $T_{\mathrm{chunk}} \pm 2 \sigma_{\mathrm{mean}}$, where $T_{\mathrm{chunk}}$ is the sample mean of the running time, and $\sigma_{\mathrm{mean}}$ is the standard error of the mean (see Footnote \ref{footnote:stderr}). }
\label{fig:time_err}
\end{figure}

\subsection{Running Time of PPR against privacy budget $\epsilon$}

We plot the average running time (over $20000$ trials for each data point) against the values of $\epsilon\in[0.06, 10]$, with $d_{\mathrm{chunk}}$ always chosen to be $4$. 
The average running time is denoted as $T_{\mathrm{chunk}}$, and the standard error of the mean is given by $\sigma_{\mathrm{mean}} = \sigma_{\mathrm{time}} / \sqrt{n_{\mathrm{trials}}}$, where $\sigma_{\mathrm{time}}$  is the standard deviation of the running time among the $\sigma_{\mathrm{time}} = 20000$ trials.

\begin{figure}[H]
	\centering
    \includegraphics[scale = 0.32]{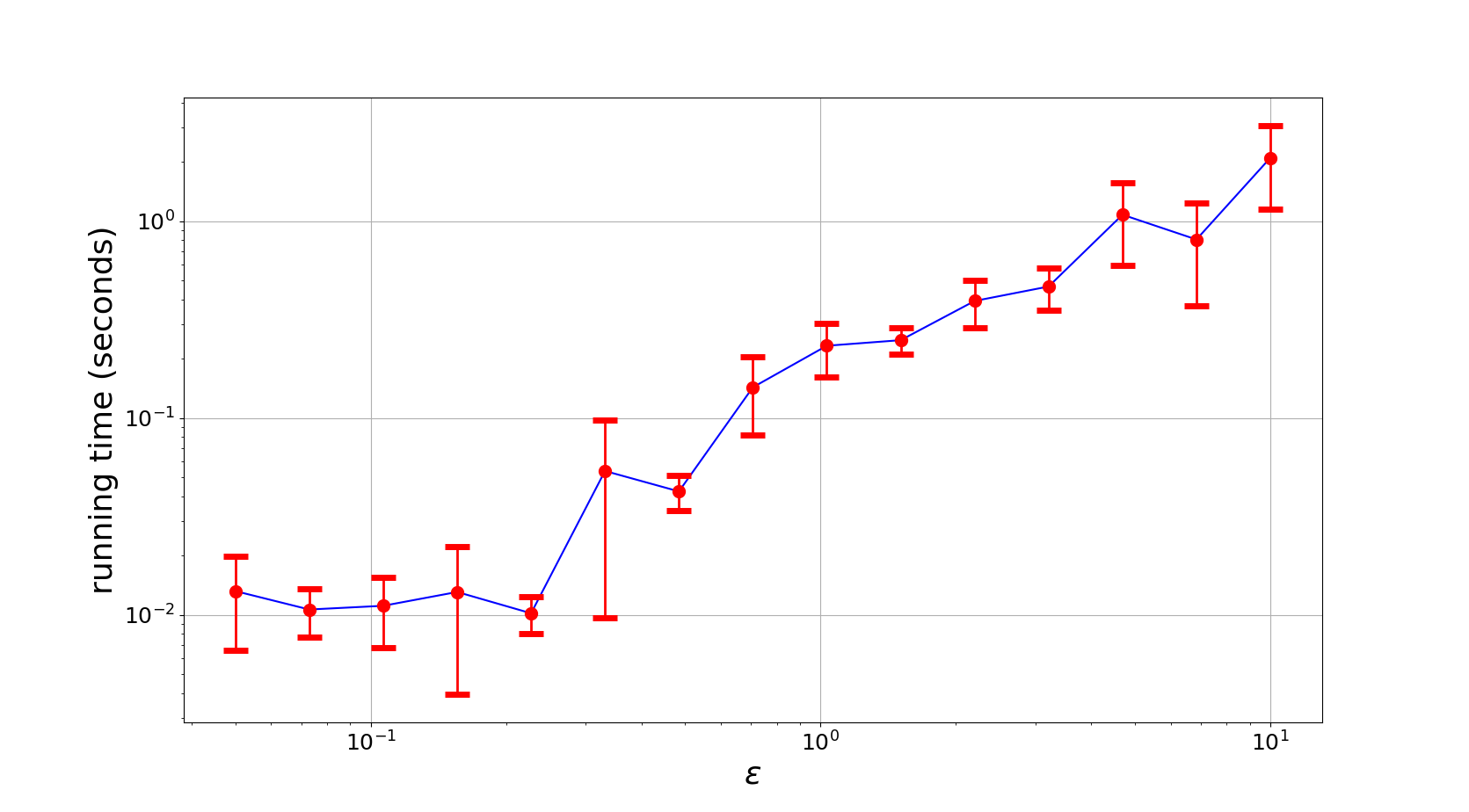}
	\caption{Average running time (over $20000$ trials),  $d_{\mathrm{chunk}}=4$ and $\varepsilon\in [0.06, 10]$, with error bars indicating the interval $T_{\mathrm{chunk}} \pm 2 \sigma_{\mathrm{mean}}$, where $T_{\mathrm{chunk}}$ is the sample mean of the running time, and $\sigma_{\mathrm{mean}}$ is the standard error of the mean.} 
\end{figure}


\chapter{Discussion and Conclusion}
\label{chp::conclu}

In conventional network information theory, the information-theoretic limits are investigated \emph{asymptotically} in the large blocklength regime, based on the law of large numbers. 
However, this assumption is impractical, as packets have bounded lengths. Over the past decade, \emph{finite blocklength}~\cite{kostina2013lossy, tan2013dispersions, polyanskiy2010channel} and \emph{one-shot}~\cite{hayashi2009information, verdu2012non, yassaee2013technique, song2016likelihood, watanabe2015nonasymp, li2021unified} information theory have been widely studied.

In this thesis, we studied one-shot information theory, code constructions, and applications to differential privacy. 
In the one-shot setting, we assumed the channel or source was used only \emph{once} (i.e., it need not be memoryless or ergodic), and the blocklength was $1$. 
Therefore, the blocklength did not approach infinity, and hence existing tools (e.g., \emph{typical sets}~\cite{el2011network}) and the law of large numbers could not be used. 
We provided several novel constructions for one-shot settings and proved achievability results, based on the Poisson functional representation~\cite{li2018strong}. 
These results were expected to recover existing (first-order and second-order) asymptotic bounds when applied to memoryless channels or sources.

In Chapter~\ref{chp:nnc}, we provide a unified one-shot coding framework over a general noisy network, applicable to any combination of source coding, channel coding, and coding for computing problems. 
Compared to the original Poisson matching lemma~\cite{li2021unified}, our scheme works for arbitrary discrete acyclic noisy networks and can be viewed as a one-shot counterpart of the asymptotic acyclic discrete memoryless network studied by Lee and Chung~\cite{lee2018unified}. 
Various one-shot results~\cite{li2021unified} and asymptotic results~\cite{el2011network} have been recovered, and novel one-shot results have been derived, including source coding, channel coding, primitive relay channel~\cite{el2011network, kim2007coding, mondelli2019new, el2021achievable, el2022strengthened}, Gelfand-Pinsker~\cite{gelfand1980coding, Heegard1980}, relay-with-unlimited-look-ahead~\cite{el2005relay, el2007relay}, Wyner-Ziv~\cite{wyner1976rate, wyner1978rate}, coding for computing~\cite{yamamoto1982wyner}, multiple access channels~\cite{ahlswede1971multi, liao1972multiple, ahlswede1974capacity}, broadcast channels~\cite{marton1979coding}, and cascade multiterminal source coding~\cite{cuff2009cascade}.

In Chapter~\ref{chp:hiding}, we derived novel one-shot achievability results for two classical secrecy problems in information theory: the information hiding problem~\cite{moulin2003information} and the compound wiretap channels~\cite{liang2009compound}. 
Our bounds are based on the Poisson matching lemma together with other techniques, and are applicable to  both continuous and discrete cases. 
Our one-shot achievability results apply to any distribution of the source data, and any class of the channels (not necessarily memoryless or ergodic), and can readily recover the existing asymptotic results on both problems when apply to discrete memoryless channels subject to potential distortion constraints, thus providing alternative proofs that are potentially simpler. 
Moreover, we generalized the information hiding setting~\cite{moulin2003information} and extended its reconstruction objective. 
For both the generalized information hiding problem and the compound wiretap channels, unlike most existing studies, we do not assume that the decoder knows the channel state in the one-shot setting.

In Chapter~\ref{chp:ppr}, we proposed a novel scheme for compressing differential privacy mechanisms, called Poisson Private Representation (PPR), to reduce the communication cost of differential privacy mechanisms. 
Unlike previous schemes, which were either constrained to special classes of DP mechanisms or introduced additional distortions to the output, our scheme could compress and exactly simulate arbitrary mechanisms while locally protecting differential privacy, thereby preserving the joint distribution of the data and the output of the original local randomizer. 
PPR achieved a compression size within a logarithmic gap from the theoretical lower bound, and a new order-wise trade-off between communication, accuracy, and central and local differential privacy for distributed mean estimation was derived. 
One possible issue was the running time of PPR, which we discussed and tackled by using sliced PPR, a method that divided long data vectors into small chunks. 
Moreover, we presented experimental results on distributed mean estimation to show that, while providing local differential privacy, PPR and sliced PPR consistently offered a better trade-off between communication, accuracy, and central differential privacy compared to the coordinate-subsampled Gaussian mechanism~\cite{chen2024privacy}.

\section{Future Directions}

While we have discussed the analysis and applications of one-shot codes based on the Poisson functional representation~\cite{li2018strong}, several avenues for future research remain, which we discuss below.

Based on the unified one-shot coding scheme over arbitrary noisy networks described in Chapter~\ref{chp:nnc}, automated theorem provers can potentially be designed. 
For example, an existing automated theorem proving tool~\cite{li2023automated} provides an algorithm for deriving asymptotic inner and outer bounds for general acyclic discrete memoryless networks~\cite{lee2018unified}. 
Considering our coding scheme in Chapter~\ref{chp:nnc} as a one-shot counterpart of~\cite{lee2018unified}, designing an automated theorem prover for one-shot inner bounds appears promising. 
Developing different schemes for one-shot outer bounds could be another potential direction for future work. 
Moreover, we should note that for the sake of universality and simplicity, our one-shot coding scheme sacrifices some performance. 
For example, for broadcast channels, we showed in Section~\ref{sec::NIT} that two corner points can be recovered, but not the entire region of Marton's inner bound, which could be derived using the original Poisson matching lemma with a more complicated analysis~\cite{li2021unified}. 
A possible extension is to generalize our scheme to improve its capability without significantly compromising its universality and simplicity.

In Chapter~\ref{chp:ppr}, we discussed the application of channel simulation schemes (in particular, a variant of the Poisson functional representation~\cite{li2018strong}) for compressing differential privacy mechanisms. 
As detailed in Section~\ref{sec::dics_runtime}, the running time of general channel simulation schemes poses challenges for practical implementation, although certain strategies, such as the sliced PPR method shown in Section~\ref{sec::dics_runtime}, can improve efficiency. 
Although it has been proven in~\cite{agustsson2020universally} that no general polynomial-time exact sampling-based method exists (even without privacy constraints) unless $RP = NP$, fast exact channel simulation schemes can still be designed for specific noisy channels. 
For the channel simulation task itself (without privacy constraints), recent works have proposed using linear error correction codes, such as polar codes~\cite{arikan2009channel}, for fast channel simulation; techniques from~\cite{flamich2022fast, flamich2024greedy} are also useful when $P_{Z|X}$ is unimodal.
On the one hand, developing fast channel simulation schemes for specific channels is an interesting research direction on its own (for example, the additive Gaussian noise channel, which plays an important role in neural compression~\cite{havasi2019minimal} and differential privacy~\cite{yan2023layered, hasircioglu2023communication}; special schemes based on vector quantization have been considered in~\cite{kobus2024gaussian}). 
On the other hand, applying these schemes to compress differential privacy mechanisms is another promising future direction.
Moreover, as discussed in Chapter~\ref{chp:ppr}, our scheme suffers a two-fold increase in the privacy budget, similar to the results based on importance sampling~\cite{shah2022optimal} (which is an approximate but not exact scheme).
It is worth investigating whether this two-fold increase is a fundamental limit for privacy mechanism simulation.

\appendix

\chapter{Proofs for Chapter~\ref{chp:nnc}}

\section{Proof of Theorem~\ref{thm::network_achievability} and Theorem~\ref{thm:det}}
\label{pf:thm}

We first generate $N$ independent exponential processes $\mathbf{U}_i$ for $i\in [N]$ according to Section \ref{sec:expon}, which serve as the random codebooks. Each node $i$ will perform two steps: the \emph{decoding step} and the \emph{encoding step}. 

We describe the decoding step at node $i$. The node observes $Y_i$ and wants to decode $\overline{U}'_i = (U_{a_{i,j}})_{j \in [d'_i]}$, while utilizing $(U_{a_{i,j}})_{j \in [d'_i+1 .. d_i]}$ by non-unique decoding. 
For the sake of notational simplicity, we omit the subscript $i$ and write $d=d_i$, $d'=d'_i$, $a_k=a_{i,k}$, $\overline{U}_{\mathcal{S}}=\overline{U}_{i,\mathcal{S}} =(U_{a_{i,j}})_{j\in \mathcal{S}}$, $\overline{\mathbf{U}}_k := \mathbf{U}_{a_{i,k}}$. 
For each $j=1,\ldots,d'$, the node will perform soft decoding via the exponential process refinement (see Section \ref{sec:expon}) on $\overline{U}_{d}$, and then on $\overline{U}_{d-1}$, and so on up to $\overline{U}_{j+1}$, and then use all the distributions obtained to decode $\overline{U}_{j}$ uniquely using the Poisson functional representation. For example, when $d=3$, $d'=2$, the decoding process will be: $\overline{U}_{3}$ (soft), $\overline{U}_{2}$ (soft), $\overline{U}_{1}$ (unique), $\overline{U}_{3}$ (soft), $\overline{U}_{2}$ (unique). The choice of the sequence $a_{i,k}$ controls the decoding ordering of the random variables. The goal is to obtain the decoded variables $\hat{\overline{U}}_{1}, \ldots, \hat{\overline{U}}_{d'}$ that equal $\overline{U}_{1}, \ldots, \overline{U}_{d'}$ with high probability.

More precisely, for $j=1,\ldots,d'$, the node computes the decoded variable $\hat{\overline{U}}_j \in \overline{\mathcal{U}}_{j}$ by first computing the joint distributions $Q^{(j)}_{ \overline{U}_{[k..d]}}$ over  $\overline{\mathcal{U}}_{k}\times \ldots \times \overline{\mathcal{U}}_{d}$ for $k=d,d-1,\ldots,j+1$ recursively using the exponential process refinement as
\begin{align*}
 Q^{(j)}_{ \overline{U}_{[k..d]}}  :=   \big( Q^{(j)}_{ \overline{U}_{[k+1..d]}} P_{\overline{U}_k| \overline{U}_{[k+1..d]}, \overline{U}_{[j-1]}, Y_i}(\cdot \,|\, \cdot ,\, \hat{\overline{U}}_{[j-1]}, Y_i)\big)^{\overline{\mathbf{U}}_k}, 
\end{align*}
i.e., first compute the semidirect product between $Q^{(j)}_{ \overline{U}_{[k+1..d]}}$ and the conditional distribution $P_{\overline{U}_k| \overline{U}_{[k+1..d]}, \overline{U}_{[j-1]}, Y_i}(\cdot \,|\, \cdot ,\, \hat{\overline{U}}_{[j-1]}, Y_i)$ (computed using the ideal joint distribution of $X^N,Y^N,U^N$) to obtain a distribution over $\overline{\mathcal{U}}_{k}\times \ldots \times \overline{\mathcal{U}}_{d}$, and then refine it by $\overline{\mathbf{U}}_k$ using Definition~\ref{def:refine}. For the base case, we assume $Q^{(j)}_{ \overline{U}_{[d+1..d]}}$ is the degenerate distribution. After we have computed $Q^{(j)}_{ \overline{U}_{[j+1..d]}}$, we can obtain $\hat{\overline{U}}_j$ using the Poisson functional representation~\eqref{eq:pfr} as $\hat{\overline{U}}_j= (\overline{\mathbf{U}}_j)_{\tilde{Q}^{(j)}_{\overline{U}_j}}$, where $\tilde{Q}^{(j)}_{\overline{U}_j}$ is the $\overline{U}_j$-marginal of
\begin{align}
&  Q^{(j)}_{ \overline{U}_{[j+1..d]}} P_{\overline{U}_j| \overline{U}_{[j+1..d]}, \overline{U}_{[j-1]}, Y_i}(\cdot \,|\, \cdot ,\, \hat{\overline{U}}_{[j-1]}, Y_i). \label{eq:Uj_pfr}
\end{align}
The node repeats this process for $j=1,\ldots,d'$ to obtain $\hat{\overline{U}}'_i=(\hat{\overline{U}}_{1}, \ldots, \hat{\overline{U}}_{d'})$. 

We then describe the encoding step at node $i$. It uses the Poisson functional representation (see Section \ref{sec:expon}) to obtain 
\begin{align}
& U_i = (\mathbf{U}_i)_{P_{U_i|Y_i, \overline{U}'_i}(\cdot | Y_i, \hat{\overline{U}}'_{i})}. \label{eq:Ui_pfr}
\end{align}
Finally, it generates $X_i$ from the conditional distribution $P_{X_i|Y_i,U_i, \overline{U}'_i}(\cdot | Y_i,U_i,\hat{\overline{U}}'_{i})$.

For the error analysis, we create a fictitious ``ideal network'' (with $N$ ``ideal nodes'') that is almost identical to the actual network. The only difference is that the ideal node $i$ uses the true $\overline{U}'_i$ (supplied by a genie) instead of the decoded $\hat{\overline{U}}'_i$ for the encoding step.
The random variables induced by the ideal network will have the same distribution as the ideal distribution of $X^N,Y^N,U^N$ in Theorem~\ref{thm::network_achievability}. Hence, we assume $X^N,Y^N,U^N$ are induced by the ideal network.
We couple the channels in the ideal network and the channels in the actual network, such that $Y_i=\tilde{Y}_i$ if $(X^{i-1},Y^{i-1})=(\tilde{X}^{i-1},\tilde{Y}^{i-1})$ (i.e., the ``channel noises'' in the two networks are the same).
If none of the actual nodes makes an error (i.e., $\hat{\overline{U}}'_i=\overline{U}'_i$ for all $i$), the actual network would coincide with the ideal network, and $(\tilde{X}^N,\tilde{Y}^N) = (X^N,Y^N)$. We consider the error probability conditional on $A:=(X^N,Y^N,U^N)$:
\[
F := \mathbf{P}\big(\exists \, i: \; \hat{\overline{U}}'_i=\overline{U}'_i\,\big|\,A \big).
\]
Note that $F$ is a random variable and is a function of $A=(X^N,Y^N,U^N)$. We have
\begin{align*}
F& =\mathbf{P}\big(\exists\, i\in [N], j \in [d'_i]:\, \hat{\overline{U}}_{i,j} \neq \overline{U}_{i,j} \,\big|\,A\big) \\
& = \sum_{i=1}^N \sum_{j=1}^{d'_i} \mathbf{P}\Big( \hat{\overline{U}}'_{[i-1]} = \overline{U}'_{[i-1]},\, \hat{\overline{U}}_{i,[j-1]} = \overline{U}_{i,[j-1]},\,   \hat{\overline{U}}_{i,j} \neq \overline{U}_{i,j} \,\big|\,A\Big).
\end{align*}
For the term inside the summation (which is the probability that the first error we make is at $\overline{U}_{i,j}$), by \eqref{eq:Uj_pfr}, \eqref{eq:Ui_pfr} and the Poisson matching lemma~\cite{li2021unified} (we again omit the subscripts $i$ as in the description of the decoding step, e.g., we write $\overline{U}_{j} =\overline{U}_{i,j} =U_{a_{j}}=U_{a_{i,j}}$; we also simply write $P(\overline{U}_{j} | Y_{a_j}, \overline{U}'_{a_j})=P_{\overline{U}_{j}| Y_{a_j}, \overline{U}'_{a_j}}(\overline{U}_{j} | Y_{a_j}, \overline{U}'_{a_j})$), we have
\begin{align*}
 & \mathbf{P}\big(\hat{\overline{U}}'_{[i-1]}=\overline{U}'_{[i-1]},\,\hat{\overline{U}}_{i,[j-1]}=\overline{U}_{i,[j-1]},\,\hat{\overline{U}}_{i,j}\neq\overline{U}_{i,j} \,\big|\,A \big)\\
 & \le\mathbf{E}\bigg[\frac{P(\overline{U}_{j}|Y_{a_{j}},\overline{U}'_{a_{j}})}{Q^{(j)}(\overline{U}_{[j+1..d]})P(\overline{U}_{j}\,|\,\overline{U}_{[j+1..d]},\,\overline{U}_{[j-1]},Y_{i})} \,\bigg|\,A  \bigg]\\
 & \stackrel{(a)}{\le}\mathbf{E}\bigg[\frac{P(\overline{U}_{j}|Y_{a_{j}},\overline{U}'_{a_{j}})}{P(\overline{U}_{j}\,|\,\overline{U}_{[j+1..d]},\,\overline{U}_{[j-1]},Y_{i})} \\
 &\;\;\;\;\;\;\;\;\;\cdot\mathbf{E}\bigg[\frac{1}{Q^{(j)}(\overline{U}_{[j+1..d]})}\,\bigg|\,\overline{U}_{[d]},Y_{i},Y_{a_{j}},\overline{U}'_{a_{j}},\overline{\mathbf{U}}_{[j+1..d]}\bigg]\,\bigg|\,A\bigg]\\
 & \stackrel{(b)}{\le}\mathbf{E}\bigg[\frac{P(\overline{U}_{j}|Y_{a_{j}},\overline{U}'_{a_{j}})}{P(\overline{U}_{j}\,|\,\overline{U}_{[j+1..d]},\,\overline{U}_{[j-1]},Y_{i})}(\ln|\overline{\mathcal{U}}_{j+1}|+1)\\
 & \;\;\;\cdot\frac{1}{Q^{(j)}(\overline{U}_{[j+2..d]})}\bigg(\frac{P(\overline{U}_{j+1}|Y_{a_{j+1}},\overline{U}'_{a_{j+1}})}{P(\overline{U}_{j+1}|\overline{U}_{[j+2..d]},\overline{U}_{[j-1]},Y_{i})}+1\bigg)\,\bigg|\,A\bigg]\\
 & \stackrel{(c)}{\le}\mathbf{E}\bigg[\frac{P(\overline{U}_{j}|Y_{a_{j}},\overline{U}'_{a_{j}})}{P(\overline{U}_{j}\,|\,\overline{U}_{[j+1..d]},\,\overline{U}_{[j-1]},Y_{i})}\\
 &  \,\cdot\prod_{k=j+1}^{d'}(\ln|\overline{\mathcal{U}}_{k}|+1)\bigg(\frac{P(\overline{U}_{k}|Y_{a_{k}},\overline{U}'_{a_{k}})}{P(\overline{U}_{k}|\overline{U}_{[k+1..d]},\overline{U}_{[j-1]},Y_{i})}+1\bigg)\,\bigg|\,A\bigg]\\
&=B_{i,j}, 
\end{align*}
where (a) is by Jensen's inequality, (b) is due to Lemma \ref{lemma::PML2},
 (c) is by applying the same steps as (a) and (b) $d'-j-1$ times, and $\beta_{i,j}$ is given in \eqref{eq::beta}. 
 The proof of Theorem~\ref{thm::network_achievability} is completed by noting that $\delta_{\mathrm{TV}}(P_{X^N,Y^N},P_{\tilde{X}^N,\tilde{Y}^N}) \le \mathbf{P}((X^N,Y^N) \neq (\tilde{X}^N,\tilde{Y}^N)) \le \mathbf{E}[F]=\mathbf{E}[\min\{F,1\}]$.

We now prove Theorem~\ref{thm:det}. Recall that the scheme we have constructed requires the public randomness $W$, which we have to fix in order to construct a deterministic coding scheme for Theorem~\ref{thm:det}. We have
\begin{align*}
&\mathbf{E}\big[\mathbf{P}\big((\tilde{X}^N,\tilde{Y}^N) \in \mathcal{E} \,\big|\, W \big)\big] \\
& = \mathbf{P}((\tilde{X}^N,\tilde{Y}^N) \in \mathcal{E}) \\
& \le \mathbf{P}\big((X^N,Y^N) \in \mathcal{E} \; \mathrm{or} \; (X^N,Y^N) \neq (\tilde{X}^N,\tilde{Y}^N)\big) \\
& = \mathbf{E}\big[ \mathbf{P}\big((X^N,Y^N) \in \mathcal{E} \; \mathrm{or} \; (X^N,Y^N) \neq (\tilde{X}^N,\tilde{Y}^N) \, \big|\, A\big) \big]\\
& \le \mathbf{E}\big[ \min\big\{ \mathbf{1}\{(X^N,Y^N) \in \mathcal{E}\}   + \mathbf{P}\big((X^N,Y^N) \neq (\tilde{X}^N,\tilde{Y}^N) \, \big|\, A\big),\, 1 \big\} \big]\\
& \le \mathbf{E}\big[ \min\big\{ \mathbf{1}\{(X^N,Y^N) \in \mathcal{E}\} + F,\, 1 \big\} \big].
\end{align*}
Therefore, there exists a value $w$ such that $\mathbf{P}((\tilde{X}^N,\tilde{Y}^N) \in \mathcal{E} \,|\, W=w)$ satisfies the upper bound. Fixing the value of $W$ to $w$ gives a deterministic coding scheme.

\chapter{Proofs for Chapter~\ref{chp:hiding}}

\section{Proof of Proposition~\ref{prop:covering_bound}}
\label{appdx::pf_prop_covering_bound}

\begin{proof}
Write $d(A,\tilde{A}):= \sup_{x \in \mathcal{X}}\Vert Aw_{Y|X}(\cdot|x) - \tilde{A}_{Y|X}(\cdot|x) \Vert_{\mathrm{TV}}$.
We use the standard method to bound the covering number, where we
start with $\tilde{\mathcal{A}}=\emptyset$, and add $A\in\mathcal{A}$
not currently covered by $\tilde{\mathcal{A}}$ (i.e., $\min_{\tilde{A}\in\tilde{\mathcal{A}}}d(A,\tilde{A})>\epsilon$)
to $\tilde{\mathcal{A}}$ one by one until all of $\mathcal{A}$ is
covered. Note that every two different $\tilde{A},\tilde{A}'\in\tilde{\mathcal{A}}$
produced this way must satisfy $d(\tilde{A},\tilde{A}')>\epsilon$,
and hence the $(\epsilon/2)$-balls $\{A:\,d(A,\tilde{A})\le\epsilon/2\}$
must be disjoint for $\tilde{A}\in\tilde{\mathcal{A}}$.

We now treat $A_{Y|X}$ as a transition probability matrix $A \in \mathbb{R}^{|\mathcal{Y}|\times|\mathcal{X}|}$. We have \begin{align*}
    d(A,\tilde{A}) &= \frac{1}{2}\Vert A-\tilde{A}\Vert_1 \\
     &= \frac{1}{2} \max_x \sum_y |A_{y,x} - \tilde{A}_{y,x}|. 
\end{align*}
The volume of the ball $\{A \in \mathbb{R}^{|\mathcal{Y}|\times|\mathcal{X}|}: \,d(A,\tilde{A})\le\epsilon/2\}$ (i.e., its Lebesgue measure in the space $\mathbb{R}^{|\mathcal{Y}|\cdot |\mathcal{X}|}$)
is $((2\epsilon)^{|\mathcal{Y}|}/(|\mathcal{Y}|!))^{|\mathcal{X}|}$, and all these balls are subsets of $\{A \in \mathbb{R}^{|\mathcal{Y}|\times|\mathcal{X}|}: \min_{x,y}A_{y,x}\ge-\epsilon,\,\max_x \sum_{y}A_{y,x}\le1+\epsilon\}$,
which has a volume $((1+(|\mathcal{Y}|+1)\epsilon)^{|\mathcal{Y}|}/(|\mathcal{Y}|!))^{|\mathcal{X}|}$. Hence, the size of $\tilde{\mathcal{A}}$ is upper-bounded by
\begin{align*}
&\frac{\Big(\big(1+(|\mathcal{Y}|+1)\epsilon \big)^{|\mathcal{Y}|}/(|\mathcal{Y}|!)\Big)^{|\mathcal{X}|}}{\big((2\epsilon)^{|\mathcal{Y}|}/(|\mathcal{Y}|!)\big)^{|\mathcal{X}|}} 
=\Big(\frac{1}{2\epsilon} + \frac{|\mathcal{Y}|+1}{2}\Big)^{|\mathcal{X}|\cdot|\mathcal{Y}|}.
\end{align*}
\end{proof}

\chapter{Proofs for Chapter~\ref{chp:ppr}}


\section{Proof of Proposition \ref{prop:ppr_output_dist}\label{sec:pf_ppr_output_dist}}

Write $(X_{i})_{i}\sim\mathrm{PP}(\mu)$ if the points $(X_{i})_{i}$
(as a multiset, ignoring the ordering) form a Poisson point process
with intensity measure $\mu$. Similarly, for $f:[0,\infty)^{n}\to[0,\infty)$,
we write $\mathrm{PP}(f)$ for the Poisson point process with intensity
function $f$ (i.e., the intensity measure has a Radon-Nikodym derivative
$f$ against the Lebesgue measure). 

Let $(T_{i})_{i}\sim\mathrm{PP}(1)$
be a Poisson process with rate $1$, independent of $Z_{1},Z_{2},\ldots\stackrel{iid}{\sim}Q$.
By the marking theorem \cite{last2017lectures}, $(Z_i,T_i)_i \sim \mathrm{PP}(Q \times \lambda_{[0,\infty)})$, where $Q \times \lambda_{[0,\infty)}$ is the product measure between $Q$ and the Lebesgues measure over $[0,\infty)$.
Let $P = P_{Z|X}(\cdot | x)$, and $\tilde{T}_{i} = T_{i} \cdot (\frac{\mathrm{d}P}{\mathrm{d}Q}(Z_{i}))^{-1}$.
By the mapping theorem \cite{last2017lectures} (also see \cite{li2018strong,li2021unified}), 
$(Z_i,\tilde{T}_i)_i \sim \mathrm{PP}(P \times \lambda_{[0,\infty)})$. 
Note that the points $(Z_i,\tilde{T}_i)_i$ may not be sorted in ascending order of $\tilde{T}_i$. Therefore, we will sort them as follows.
Let $j_1$ be the $j$ such that $\tilde{T}_{j}$ is the smallest, $j_2$ be the $j$ other than $j_1$ such that $\tilde{T}_{j}$ is the smallest, and so on. Break ties arbitrarily. Then $(\tilde{T}_{j_i})_i$ is an ascending sequence, and we still have $(Z_{j_i},\tilde{T}_{j_i})_i \sim \mathrm{PP}(P \times \lambda_{[0,\infty)})$ since we are merely rearranging the points. Comparing $(Z_{j_i},\tilde{T}_{j_i})_i \sim \mathrm{PP}(P \times \lambda_{[0,\infty)})$ with the definition of $(Z_i,T_i)_i \sim \mathrm{PP}(Q \times \lambda_{[0,\infty)})$, we can see that $(\tilde{T}_{j_i})_i \sim\mathrm{PP}(1)$
is independent of $Z_{j_1},Z_{j_2},\ldots \stackrel{iid}{\sim}P$.

Recall that in PPR, we generate $K\in\mathbb{Z}_{+}$ with 
\[
\Pr(K=k)=\frac{\tilde{T}_{k}^{-\alpha}}{\sum_{i=1}^{\infty}\tilde{T}_{i}^{-\alpha}},
\]
and the final output is $Z_K$. Rearranging the points according to $(j_i)_i$, the distribution of the final output remains the same if we instead generate $K'\in\mathbb{Z}_{+}$ with 
\[
\Pr(K'=k)=\frac{\tilde{T}_{j_k}^{-\alpha}}{\sum_{i=1}^{\infty}\tilde{T}_{j_i}^{-\alpha}},
\]
and the final output is $Z_{j_{K'}}$. Since $(\tilde{T}_{j_i})_i \sim\mathrm{PP}(1)$
is independent of $Z_{j_i}\stackrel{iid}{\sim}P$, we know that $K'$ is independent of $(Z_{j_i})_i$, and hence $Z_{j_{K'}} \sim P$ follows the desired distribution.

\section{Reparametrization and Detailed Algorithm of PPR\label{sec:reparametrization}}

We now discuss the implementation of the Poisson private representation in Section~\ref{sec:ppr}.
Practically, the algorithm cannot compute the whole infinite sequence
$(\tilde{T}_{i})_i$.
We can truncate the method and only
compute $\tilde{T}_{i},\ldots,\tilde{T}_{N}$ for a large $N$ and
select $K\in\{1,\ldots,N\}$, which incurs a small distortion in the
distribution of $Z$.\footnote{To compare to the minimal random coding (MRC) \cite{havasi2019minimal,cuff2008communication,song2016likelihood}
scheme in \cite{shah2022optimal}, which also utilizes a finite number
$N$ of samples $(Z_{i})_{i=1,\ldots,N}$, while truncating the number
of samples to $N$ in both PPR and MRC results
in a distortion in the distribution of $Z$ that tends to $0$ as
$N\to\infty$, the difference is that $\log K$ (which is approximately
the compression size) in MRC
grows like $\log N$, whereas $\log K$ does not grow as $N\to\infty$
in PPR. The size $N$ in truncated PPR merely controls the tradeoff
between accuracy of the distribution of $Z$ and the running time
of the algorithm.}
While this method is practically acceptable, it might defeat the purpose of having an exact algorithm that ensures the correct conditional distribution $P_{Z|X}$.
We now present an exact algorithm for PPR that terminates in a finite amount of time using a reparametrization.

In the proof of Theorem \ref{thm:logk_bound}, we showed that, letting
$(T_{i})_{i}\sim\mathrm{PP}(1)$, $Z_{1},Z_{2},\ldots\stackrel{iid}{\sim}Q$,
$R_{i}:=(\mathrm{d}P/\mathrm{d}Q)(Z_{i})$, $V_{1},V_{2},\ldots\stackrel{iid}{\sim}\mathrm{Exp}(1)$,
PPR can be equivalently expressed as
\[
K=\underset{k}{\mathrm{argmin}}T_{k}^{\alpha}R_{k}^{-\alpha}V_{k}.
\]
The problem of finding $K$ is that there is no stopping criteria
for the argmin. For example, if we scan the points $(T_{i},R_{i},V_{i})_{i}$
in increasing order of $T_{i}$, it is always possible that there
is a future point with $V_{i}$ so small that it makes $T_{i}^{\alpha}R_{i}^{-\alpha}V_{i}$
smaller than the current minimum. 
If we scan the points in increasing
order of $V_{i}$ instead, it is likewise possible that there is a
future point with a very small $T_{i}$. We can scan the points in
increasing order of $U_{i}:=T_{i}^{\alpha}V_{i}$, but we would not
know the indices of the points in the original process where $T_{1}\le T_{2}\le\cdots$
is in increasing order, which is necessary to find out the $Z_{i}$
corresponding to each point (recall that in PPR, the point with the
smallest $T_{i}$ corresponds to $Z_{1}$, the second smallest $T_{i}$
corresponds to $Z_{2}$, etc.).

Therefore, we will scan the points in increasing order of $B_{i}:=T_{i}^{\alpha}\min\{V_{i},1\}$
instead. By the mapping theorem \cite{last2017lectures}, $(T_{i}^{\alpha})_{i}\sim\mathrm{PP}(\alpha^{-1}t^{1/\alpha-1})$.
By the marking theorem \cite{last2017lectures}, 
\[
(T_{i}^{\alpha},V_{i})_{i}\sim\mathrm{PP}(\alpha^{-1}t^{1/\alpha-1}e^{-v}).
\]
By the mapping theorem,
\[
(T_{i}^{\alpha},T_{i}^{\alpha}V_{i})_{i}\sim\mathrm{PP}(\alpha^{-1}t^{1/\alpha-2}e^{-vt^{-1}}).
\]
Since $B_{i}=\min\{T_{i}^{\alpha},T_{i}^{\alpha}V_{i}\}$, again by the mapping theorem,
\begin{align*}
(B_{i})_{i} & \sim\mathrm{PP}\Bigg(\int_{b}^{\infty}\alpha^{-1}b^{1/\alpha-2}e^{-vb^{-1}}\mathrm{d}v  +\int_{b}^{\infty}\alpha^{-1}t^{1/\alpha-2}e^{-bt^{-1}}\mathrm{d}t\Bigg)\\
 & =\mathrm{PP}\left(\alpha^{-1}b^{1/\alpha-1}e^{-1}+\alpha^{-1}b^{1/\alpha-1}\gamma(1-\alpha^{-1},1)\right)\\
 & =\mathrm{PP}\left(\alpha^{-1}\left(e^{-1}+\gamma_{1}\right)b^{1/\alpha-1}\right),
\end{align*}
where $\gamma_{1}:=\gamma(1-\alpha^{-1},1)$ and $\gamma(\beta,x)=\int_{0}^{x}e^{-\tau}\tau^{\beta-1}\mathrm{d}\tau$
is the lower incomplete gamma function. Comparing the distribution
of $(B_{i})_{i}$ and $(T_{i}^{\alpha})_{i}$, we can generate $(B_{i})_{i}$
by first generating $(U_{i})_{i}\sim\mathrm{PP}(1)$, and then taking
$B_{i}=(U_{i}\alpha/(e^{-1}+\gamma_{1}))^{\alpha}$. The conditional
distribution of $(T_{i},V_{i})$ given $B_{i}=b$ is described as
follows: 
\begin{itemize}
\item With probability $e^{-1}/(e^{-1}+\gamma_{1})$, we have $T_{i}^{\alpha}=b$
and $T_{i}^{\alpha}V_{i}\sim b(\mathrm{Exp}(1)+1)$, and hence $T_{i}=b^{1/\alpha}$
and $V_{i}\sim\mathrm{Exp}(1)+1$.
\item With probability $\gamma_{1}/(e^{-1}+\gamma_{1})$, we have $T_{i}^{\alpha}V_{i}=b$
and 
\[
T_{i}^{\alpha}\sim\frac{\alpha^{-1}t^{1/\alpha-2}e^{-bt^{-1}}}{\alpha^{-1}\gamma(1-\alpha^{-1},1)b^{1/\alpha-1}}.
\]
Hence, for $0<\tau\le1$,
\[
\Pr(V_{i}\le\tau)=\Pr(T_{i}^{\alpha}\ge b/\tau)=\frac{\gamma(1-\alpha^{-1},\tau)}{\gamma(1-\alpha^{-1},1)},
\]
and $V_{i}$ follows the truncated gamma distribution with shape $1-\alpha^{-1}$
and scale $1$, truncated within the interval $[0,1]$. We then have
$T_{i}=(b/V_{i})^{1/\alpha}$.
\end{itemize}
The algorithm is given in Algorithm \ref{alg:ppr}. The encoder and
decoder require a shared random seed $s$. One way to generate $s$
is to have the encoder and decoder maintain two synchronized pseudorandom
number generators (PRNGs) that are always at the same state, and invoke
the PRNGs to generate $s$, guaranteeing that the $s$ at the encoder
is the same as the $s$ at the decoder. The encoder maintains a collection
of points $(T_{i},V_{i},\Theta_{i})$, stored in a heap to allow fast
query and removal of the point with the smallest $T_{i}$. The value
$\Theta_{i}\in\{0,1\}$ indicates whether it is possible that the
point $(T_{i},V_{i})$ attains the minimum of $T_{k}^{\alpha}R_{k}^{-\alpha}V_{k}$.
The encoding algorithm repeats until there is no possible points left
in the heap, and it is impossible for any future point to be better
than the current minimum of $T_{k}^{\alpha}R_{k}^{-\alpha}V_{k}$.
The encoding time complexity is $O(\sup_{z}(\mathrm{d}P/\mathrm{d}Q)(z) \log(\sup_{z}(\mathrm{d}P/\mathrm{d}Q)(z)))$, which is close to other sampling-based channel simulation schemes~\cite{harsha2010communication, flamich2023adaptive}.\footnote{It was shown in \cite{flamich2023adaptive} that greedy rejection sampling \cite{harsha2010communication} runs in $O(\sup_{z}(\mathrm{d}P/\mathrm{d}Q)(z))$ time. The PPR algorithm has an additional log term due to the use of heap.}
The decoding algorithm simply outputs the $k$-th sample generated
using the random seed $s$, which can be performed in $O(1)$ time.\footnote{\label{fn:counterbased}A counter-based PRNG \cite{salmon2011parallel}
allows us to directly jump to the state after $k$ uses of the PRNG,
without the need of generating all $k$ samples, greatly improving
the decoding efficiency. This technique is applicable to greedy rejection
sampling \cite{harsha2010communication} and the original Poisson
functional representation \cite{li2018strong,li2021unified} as well.} 

The PPR is implemented by Algorithm~\ref{alg:ppr}. 
We write $x\leftarrow\mathrm{Exp}_{\mathscr{G}}(1)$ to mean that
we generate an exponential random variate $x$ with rate $1$ using
the pseudorandom number generator $\mathscr{G}$. Write $x\leftarrow\mathrm{Exp}_{\mathrm{local}}(1)$
to mean that $x$ is generated using a local pseudorandom number generator
(not $\mathscr{G}$).

\bigskip 
\bigskip 

\textbf{Algorithm 1}: Poisson private representation
\bigskip 

\textbf{Procedure} $\textsc{PPREncode}(\alpha,Q,r,r^{*},s):$

\textbf{$\;\;\;\;$Input:} parameter $\alpha>1$, distribution $Q$,
density $r(z):=(\mathrm{d}P/\mathrm{d}Q)(z)$,

\textbf{$\;\;\;\;$$\;\;\;\;$$\;\;\;\;$}bound $r^{*}\ge\sup_{z}r(z)$,
random seed $s$

\textbf{$\;\;\;\;$Output:} index $k\in\mathbb{Z}_{>0}$

\smallskip{}

\begin{algorithmic}[1]

\State{Initialize PRNG $\mathscr{G}$ using the seed $s$}

\State{$u\leftarrow0$, $w^{*}\leftarrow\infty$, $k\leftarrow0$,
$k^{*}\leftarrow0$, $n\leftarrow0$}

\State{$\gamma_{1}\leftarrow\gamma(1-\alpha^{-1},1)=\int_{0}^{1}e^{-\tau}\tau^{-\alpha^{-1}}\mathrm{d}\tau$}

\State{$h\leftarrow\emptyset$ (empty heap)}

\While{true}

\State{$u\leftarrow u+\mathrm{Exp}_{\mathrm{local}}(1)$}\Comment{\textit{Generated
using local randomness (not }$\mathscr{G}$\textit{)}}

\State{$b\leftarrow(u\alpha/(e^{-1}+\gamma_{1}))^{\alpha}$}

\If{$n=0$ and $b(r^{*})^{-\alpha}\ge w^{*}$}\Comment{\textit{No
possible points left and future points impossible}}

\State{\Return{$k^{*}$}}

\EndIf

\If{$\mathrm{Unif}_{\mathrm{local}}(0,1)<e^{-1}/(e^{-1}+\gamma_{1})$}\Comment{\textit{Run
with prob. $e^{-1}/(e^{-1}+\gamma_{1})$}}

\State{$t\leftarrow b^{1/\alpha}$, $v\leftarrow\mathrm{Exp}_{\mathrm{local}}(1)+1$}

\Else

\Repeat

\State{$v\leftarrow\mathrm{Gamma}_{\mathrm{local}}(1-\alpha^{-1},1)$}\Comment{\textit{Gamma
distribution}}

\Until{$v\le1$}

\State{$t\leftarrow(b/v)^{1/\alpha}$}

\EndIf

\State{$\theta\leftarrow\mathbf{1}\{(t/r^{*})^{\alpha}v\le w^{*}\}$}\Comment{\textit{Is
it possible for this point to be optimal}}

\State{Push $(t,v,\theta)$ to $h$}

\State{$n\leftarrow n+\theta$}\Comment{\textit{Number of possible
points in heap}}

\While{$h\neq\emptyset$ and $\min_{(t',v',\theta')\in h}t'\le b^{1/\alpha}$}\Comment{\textit{Assign
$Z_{i}$'s to points in heap with small $T_{i}$}}

\State{$(t,v,\theta)\leftarrow\arg\min_{(t',v',\theta')\in h}t'$,
and pop $(t,v,\theta)$ from $h$}

\State{$n\leftarrow n-\theta$}

\State{$k\leftarrow k+1$}

\State{Generate $z\sim Q$ using $\mathscr{G}$}

\State{$w\leftarrow(t/r(z))^{\alpha}v$}

\If{$w<w^{*}$}

\State{$w^{*}\leftarrow w$}

\State{$k^{*}\leftarrow k$}

\EndIf

\EndWhile

\EndWhile

\end{algorithmic}

\bigskip{}

\textbf{Procedure} $\textsc{PPRDecode}(Q,k,s):$

\textbf{$\;\;\;\;$Input:} $Q$, index $k\in\mathbb{Z}_{>0}$, seed
$s$

\textbf{$\;\;\;\;$Output:} sample $z$

\smallskip{}

\begin{algorithmic}[1]

\State{Initialize PRNG $\mathscr{G}$ using the seed $s$}

\For{$i=1,2,\ldots,k$}

\State{Generate $z\sim Q$ using $\mathscr{G}$}\Comment{\textit{See
footnote \ref{fn:counterbased}}}

\EndFor

\State{\Return{$z$}}

\end{algorithmic}


\captionof{algorithm}{Poisson private representation}
\label{alg:ppr}

\section{Proofs of Theorem~\ref{thm:eps_dp} and Theorem~\ref{thm:metric_privacy}\label{sec:pf_eps_dp}}

First prove Theorem~\ref{thm:eps_dp}. Consider a $\varepsilon$-DP mechanism $P_{Z|X}$.
Consider neighbors $x_{1},x_{2}$, and let $P_{j}:=P_{Z|X}(\cdot|x_{j})$,
$\tilde{T}_{j,i}:=T_{i}/(\frac{\mathrm{d}P_{j}}{\mathrm{d}Q}(Z_{i}))$,
and $K_{j}$ be the output of PPR applied on $P_{j}$, for $j=1,2$.
Since $P_{Z|X}$ is $\varepsilon$-DP, 
\begin{equation}
e^{-\varepsilon}\frac{\mathrm{d}P_{2}}{\mathrm{d}Q}(z) \le \frac{\mathrm{d}P_{1}}{\mathrm{d}Q}(z)\le e^{\varepsilon}\frac{\mathrm{d}P_{2}}{\mathrm{d}Q}(z) \label{eq:derivative_ratios}
\end{equation}
for $Q$-almost every $z$,\footnote{$\varepsilon$-DP only implies that \eqref{eq:derivative_ratios} holds for $P_1$-almost every $z$ (or equivalently $P_2$-almost every $z$ since $P_1,P_2$ are absolutely continuous with respect to each other). We now show that \eqref{eq:derivative_ratios} holds for $Q$-almost every $z$. Apply Lebesgue's decomposition theorem to find measures $\tilde{Q}, \hat{Q}$ such that $Q=\tilde{Q}+\hat{Q}$, $\tilde{Q} \ll P_1$ and $\hat{Q} \perp P_1$. There exists $\mathcal{Z}' \subseteq \mathcal{Z}$ such that $ P_1(\mathcal{Z}')=1$ and $\hat{Q}(\mathcal{Z}')=0$. Since $P_1 \ll Q$, we have $P_1 \ll \tilde{Q}$. We have \eqref{eq:derivative_ratios} for $\tilde{Q}$-almost every $z$. Also, we have \eqref{eq:derivative_ratios} for $\hat{Q}$-almost every $z$ since $z \notin \mathcal{Z}'$ gives $\frac{\mathrm{d}P_1}{\mathrm{d}Q}(z) = \frac{\mathrm{d}P_1}{\mathrm{d}\hat{Q}}(z) = 0$ for $\hat{Q}$-almost every $z$, and also $\frac{\mathrm{d}P_2}{\mathrm{d}Q}(z) = 0$ for $\hat{Q}$-almost every $z$ since $P_2 \ll P_1$.} and hence $e^{-\varepsilon} \tilde{T}_{2,i} \le \tilde{T}_{1,i} \le e^{\varepsilon} \tilde{T}_{2,i}$. For $k \in \mathbb{Z}_{+}$, we have, almost surely,
\begin{align*}
\Pr(K_1 = k \, | \, (Z_i,T_i)_i) & = \frac{\tilde{T}_{1,k}^{-\alpha}}{\sum_{i=1}^{\infty} \tilde{T}_{1,i}^{-\alpha}} \\
& \le \frac{e^{\alpha \varepsilon}\tilde{T}_{2,k}^{-\alpha}}{\sum_{i=1}^{\infty} e^{- \alpha \varepsilon}\tilde{T}_{2,i}^{-\alpha}} \\
& = e^{2 \alpha \varepsilon} \Pr(K_2 = k \, | \, (Z_i,T_i)_i).
\end{align*}
For any measurable $\mathcal{S}\subseteq\mathcal{Z}^{\infty}\times\mathbb{Z}_{>0}$,
\begin{align}
& \Pr\left(((Z_{i})_{i},K_{1})\in\mathcal{S}\right) \nonumber \\
& =\mathbb{E}\left[\Pr\left(((Z_{i})_{i},K_{1})\in\mathcal{S}\,\big|\,(Z_{i},T_{i})_{i}\right)\right] \nonumber \\
& =\mathbb{E}\left[\sum_{k:\,((Z_{i})_{i},k)\in\mathcal{S}}\Pr\left(K_{1}=k\,\big|\,(Z_{i},T_{i})_{i}\right)\right] \nonumber \\
& \le e^{2 \alpha \varepsilon} \cdot \mathbb{E}\left[\sum_{k:\,((Z_{i})_{i},k)\in\mathcal{S}}\Pr\left(K_{2}=k\,\big|\,(Z_{i},T_{i})_{i}\right)\right] \nonumber \\
& =  e^{2 \alpha \varepsilon} \Pr\left(((Z_{i})_{i},K_{2})\in\mathcal{S}\right). \label{eq:eps_dp_set_bound}
\end{align}
Hence, $P_{(Z_i)_i,K | X}$ is $2\alpha \varepsilon$-DP.

For Theorem~\ref{thm:metric_privacy}, consider a $\varepsilon \cdot d_{\mathcal{X}}$-private mechanism $P_{Z|X}$, and
consider $x_{1},x_{2} \in \mathcal{X}$. We have
\begin{equation}
e^{-\varepsilon \cdot d_{\mathcal{X}}(x_1,x_2)}\frac{\mathrm{d}P_{2}}{\mathrm{d}Q}(z) \le \frac{\mathrm{d}P_{1}}{\mathrm{d}Q}(z)\le e^{\varepsilon \cdot d_{\mathcal{X}}(x_1,x_2)}\frac{\mathrm{d}P_{2}}{\mathrm{d}Q}(z)
\end{equation}
for $Q$-almost every $z$.
By exactly the same arguments as in the proof of Theorem~\ref{thm:eps_dp}, $\Pr\left(((Z_{i})_{i},K_{1})\in\mathcal{S}\right) \le e^{2 \alpha \varepsilon \cdot d_{\mathcal{X}}(x_1,x_2)} \Pr\left(((Z_{i})_{i},K_{2})\in\mathcal{S}\right)$, and hence $P_{(Z_i)_i,K | X}$ is $2\alpha \varepsilon\cdot d_{\mathcal{X}}$-private.

\section{Proof of Theorem~\ref{thm:eps_delta_dp_2}\label{sec:eps_delta_dp_2}}

Consider a $(\varepsilon,\delta)$-DP mechanism $P_{Z|X}$.
Consider neighbors $x_{1},x_{2}$, and let $P_{j}:=P_{Z|X}(\cdot|x_{j})$,
and $K_{j}$ be the output of PPR applied on $P_{j}$, for $j=1,2$.
By the definition of $(\varepsilon,\delta)$-differential
privacy, we have
\begin{equation}
\int\max\left\{ \rho_{1}(z)-e^{\varepsilon}\rho_{2}(z),\,0\right\} Q(\mathrm{d}z)\le\delta,\label{eq:eps_delta_dp_intbd1}
\end{equation}
\begin{equation}
\int\max\left\{ \rho_{2}(z)-e^{\varepsilon}\rho_{1}(z),\,0\right\} Q(\mathrm{d}z)\le\delta.\label{eq:eps_delta_dp_intbd2}
\end{equation}
Let
\[
\overline{\rho}(z):=\min\left\{ \max\left\{ \rho_{1}(z),\,e^{-\varepsilon}\rho_{2}(z)\right\} ,\,e^{\varepsilon}\rho_{2}(z)\right\} .
\]
Note that $e^{-\varepsilon}\rho_{2}(z)\le\overline{\rho}(z)\le e^{\varepsilon}\rho_{2}(z)$.
We then consider two cases:

Case 1: $\int\overline{\rho}(z)Q(\mathrm{d}z)\le1$. Let $\rho_{3}(z)$
be such that $\int\rho_{3}(z)Q(\mathrm{d}z)=1$ and
\[
\overline{\rho}(z)\le\rho_{3}(z)\le e^{\varepsilon}\rho_{2}(z).
\]
We can always find such $\rho_{3}$ by taking an appropriate convex
combination of the lower bound above (which integrates to $\le1$)
and the upper obund above (which integrates to $\ge1$). We then have
\begin{equation}
e^{-\varepsilon}\rho_{2}(z)\le\rho_{3}(z)\le e^{\varepsilon}\rho_{2}(z).\label{eq:eps_delta_dp_ratio1}
\end{equation}
If $\rho_{1}(z)-e^{\varepsilon}\rho_{2}(z)\le0$, then $\rho_{1}(z)-\rho_{3}(z)\le\rho_{1}(z)-\overline{\rho}(z)\le0$.
If $\rho_{1}(z)-e^{\varepsilon}\rho_{2}(z)>0$, then $\rho_{3}(z)=\overline{\rho}(z)=e^{\varepsilon}\rho_{2}(z)$.
Either way, we have
$\max\left\{ \rho_{1}(z)-\rho_{3}(z),\,0\right\} =\max\left\{ \rho_{1}(z)-e^{\varepsilon}\rho_{2}(z),\,0\right\}$.
By (\ref{eq:eps_delta_dp_intbd1}), we have
\[
\int\max\left\{ \rho_{1}(z)-\rho_{3}(z),\,0\right\} Q(\mathrm{d}z)\le\delta.
\]
Let $P_{3}=\rho_{3}Q$ be the probability measure with $\mathrm{d}P_{3}/\mathrm{d}Q=\rho_{3}$.
Then the total variation distance $d_{\mathrm{TV}}(P_{1},P_{3})$
between $P_{1}$ and $P_{3}$ is at most $\delta$, and by (\ref{eq:eps_delta_dp_ratio1}),
\begin{equation}
e^{-\varepsilon}\frac{\mathrm{d}P_{2}}{\mathrm{d}Q}(z)\le\frac{\mathrm{d}P_{3}}{\mathrm{d}Q}(z)\le e^{\varepsilon}\frac{\mathrm{d}P_{2}}{\mathrm{d}Q}(z).\label{eq:eps_delta_dp_ratio0}
\end{equation}

Case 2: $\int\overline{\rho}(z)Q(\mathrm{d}z)>1$. Let $\rho_{3}(z)$
be such that $\int\rho_{3}(z)Q(\mathrm{d}z)=1$ and
\[
e^{-\varepsilon}\rho_{2}(z)\le\rho_{3}(z)\le\overline{\rho}(z).
\]
We can always find such $\rho_{3}$ by taking an appropriate convex
combination of the lower bound above (which integrates to $\le1$)
and the upper obund above (which integrates to $>1$). We again have
$e^{-\varepsilon}\rho_{2}(z)\le\rho_{3}(z)\le e^{\varepsilon}\rho_{2}(z)$.
If $e^{-\varepsilon}\rho_{2}(z)-\rho_{1}(z)\le0$, then $\rho_{3}(z)-\rho_{1}(z)\le\overline{\rho}(z)-\rho_{1}(z)\le0$.
If $e^{-\varepsilon}\rho_{2}(z)-\rho_{1}(z)>0$, then $\rho_{3}(z)=\overline{\rho}(z)=e^{-\varepsilon}\rho_{2}(z)$.
Either way, we have
$\max\left\{ \rho_{3}(z)-\rho_{1}(z),\,0\right\} =\max\left\{ e^{-\varepsilon}\rho_{2}(z)-\rho_{1}(z),\,0\right\}$.
By (\ref{eq:eps_delta_dp_intbd2}), we have
\[
\int\max\left\{ \rho_{3}(z)-\rho_{1}(z),\,0\right\} Q(\mathrm{d}z)\le e^{-\varepsilon}\delta\le\delta.
\]
Let $P_{3}=\rho_{3}Q$ be the probability measure with $\mathrm{d}P_{3}/\mathrm{d}Q=\rho_{3}$.
Again, we have $d_{\mathrm{TV}}(P_{1},P_{3})\le\delta$ and (\ref{eq:eps_delta_dp_ratio0}).
Therefore, regardless of whether Case 1 or Case 2 holds, we can construct
$P_{3}$ satisfying $d_{\mathrm{TV}}(P_{1},P_{3})\le\delta$ and (\ref{eq:eps_delta_dp_ratio0}). Let $K_{3}$ be
the output of PPR applied on $P_{3}$. 

In the proof of Theorem \ref{thm:logk_bound}, we see that PPR has
the following equivalent formulation. Let $(T_{i})_{i}\sim\mathrm{PP}(1)$
be a Poisson process with rate $1$, independent of $Z_{1},Z_{2},\ldots\stackrel{iid}{\sim}Q$.
Let $R_{i}:=(\mathrm{d}P/\mathrm{d}Q)(Z_{i})$, and let its probability
measure be $P_{R}$. Let $V_{1},V_{2},\ldots\stackrel{iid}{\sim}\mathrm{Exp}(1)$.
PPR can be equivalently expressed as
\begin{align*}
K & =\underset{k}{\mathrm{argmin}}T_{k}^{\alpha}R_{k}^{-\alpha}V_{k} =\underset{k}{\mathrm{argmin}}\frac{T_{k}V_{k}^{1/\alpha}}{R_{k}}.
\end{align*}
Note that $(T_{i}V_{i}^{1/\alpha})_{i}\sim\mathrm{PP}(\int_{0}^{\infty}v^{-1/\alpha}e^{-v}\mathrm{d}v)=\mathrm{PP}(\Gamma(1-\alpha^{-1}))$
is a uniform Poisson process. Therefore PPR is the same as the Poisson
functional representation \cite{li2018strong,li2021unified} applied
on $(T_{i}V_{i}^{1/\alpha})_{i}$. By the grand coupling property of Poisson
functional representation \cite{li2021unified,li2019pairwise} (see \cite[Theorem 3]{li2019pairwise}), if
we apply the Poisson functional representation on $P_{1}$ and $P_{3}$
to get $K_{1}$ and $K_{3}$ respectively, then
\[
\Pr(K_{1}\neq K_{3})\le2d_{\mathrm{TV}}(P_{1},P_{3}) \le 2 \delta.
\]
Therefore, for any measurable $\mathcal{S}\subseteq\mathcal{Z}^{\infty}\times\mathbb{Z}_{>0}$,
\begin{align*}
 \Pr\left(((Z_{i})_{i},K_{1})\in\mathcal{S}\right) & \le\Pr\left(((Z_{i})_{i},K_{3})\in\mathcal{S}\right)+2\delta\\
 & \le e^{2\alpha\varepsilon}\Pr\left(((Z_{i})_{i},K_{2})\in\mathcal{S}\right)+2\delta,
\end{align*}
where the last inequality is by applying \eqref{eq:eps_dp_set_bound} on
$P_{3},P_{2}$ instead of $P_{1},P_{2}$. Hence, $P_{(Z_i)_i,K | X}$ is $(2\alpha\varepsilon, 2\delta)$-DP.

\section{Proof of Theorem~\ref{thm:eps_delta_dp}\label{sec:pf_eps_delta_dp}}

We present the proof of $(\varepsilon,\delta)$-DP of PPR (i.e., Theorem~\ref{thm:eps_delta_dp}).
\begin{proof}
We assume
\begin{equation}
\alpha-1\le\frac{\beta\tilde{\delta}\tilde{\varepsilon}^{2}}{-\ln\tilde{\delta}},\label{eq:eps_d_alpha}
\end{equation}
where $\beta:=e^{-4.2}$. Using the Laplace functional of the Poisson
process $(\tilde{T}_{i})_{i}$ \cite[Theorem 3.9]{last2017lectures},
for $w>0$,
\begin{align}
\mathbb{E}\left[\exp\left(-w\sum_{i}\tilde{T}_{i}^{-\alpha}\right)\right] & =\exp\left(-\int_{0}^{\infty}(1-\exp(-wt^{-\alpha}))\mathrm{d}t\right)\label{eq:laplace}\\
 & =\exp\left(-w^{1/\alpha}\Gamma(1-\alpha^{-1})\right).\nonumber 
\end{align}
We first bound the left tail of $\sum_{i}\tilde{T}_{i}^{-\alpha}$.
By Chernoff bound, for $d\ge0$,

\begin{align*}
 & \Pr\left(\sum_{i}\tilde{T}_{i}^{-\alpha}\le d\right)\\
 & \le\inf_{w>0}e^{wd}\mathbb{E}\left[\exp\left(-w\sum_{i}\tilde{T}_{i}^{-\alpha}\right)\right]\\
 & =\inf_{w>0}\exp\left(wd-w^{1/\alpha}\Gamma(1-\alpha^{-1})\right)\\
 & \le\exp\left(\left(\frac{\Gamma(1-\alpha^{-1})}{\alpha d}\right)^{\frac{\alpha}{\alpha-1}}d-\left(\frac{\Gamma(1-\alpha^{-1})}{\alpha d}\right)^{\frac{1}{\alpha-1}}\Gamma(1-\alpha^{-1})\right)\\
 & =\exp\left(\left(\Gamma(1-\alpha^{-1})\right)^{\frac{\alpha}{\alpha-1}}d^{-\frac{1}{\alpha-1}}\left(\alpha^{-\frac{\alpha}{\alpha-1}}-\alpha^{-\frac{1}{\alpha-1}}\right)\right)\\
 & =\exp\left(-\left(\frac{\alpha d}{\left(\Gamma(1-\alpha^{-1})\right)^{\alpha}}\right)^{-\frac{1}{\alpha-1}}\left(1-\alpha^{-1}\right)\right)\\
 & =\exp\left(-\left(\frac{\alpha d(1-\alpha^{-1})^{\alpha}}{\left(\Gamma(2-\alpha^{-1})\right)^{\alpha}}\right)^{-\frac{1}{\alpha-1}}\left(1-\alpha^{-1}\right)\right)\\
 & =\exp\left(-\left(\frac{(\alpha-1)d}{\left(\Gamma(2-\alpha^{-1})\right)^{\alpha}}\right)^{-\frac{1}{\alpha-1}}\right).
\end{align*}
Therefore, to guarantee $\Pr(\sum_{i}\tilde{T}_{i}^{-\alpha}\le d)\le\tilde{\delta}/3$,
we require
\[
d\le\frac{\Gamma(2-\alpha^{-1})^{\alpha}\left(-\ln(\tilde{\delta}/3)\right)^{-(\alpha-1)}}{\alpha-1},
\]
where
\begin{align*}
 & \Gamma(2-\alpha^{-1})^{\alpha}\left(-\ln(\tilde{\delta}/3)\right)^{-(\alpha-1)}\\
 & \ge\left(\exp\left(-\gamma(\alpha-1)\right)\right)^{\alpha}\left(-\ln(\tilde{\delta}^{2})\right)^{-(\alpha-1)}\\
 & \ge\exp\left(-\gamma\alpha\frac{\beta\tilde{\delta}\tilde{\varepsilon}^{2}}{-\ln\tilde{\delta}}\right)\left(-2\ln\tilde{\delta}\right)^{-\frac{\beta\tilde{\delta}\tilde{\varepsilon}^{2}}{-\ln\tilde{\delta}}}\\
 & \ge\exp\left(-2\gamma\frac{\beta\tilde{\delta}\tilde{\varepsilon}^{2}}{-\ln\tilde{\delta}}-2e^{-1}\beta\tilde{\delta}\tilde{\varepsilon}^{2}\right)\\
 & \ge\exp\left(-\left(\frac{2\gamma}{3\ln2}+\frac{2}{3e}\right)\beta\tilde{\varepsilon}^{2}\right)\\
 & \ge\exp\left(-0.81\cdot\beta\tilde{\varepsilon}^{2}\right)\\
 & \ge e^{-\tilde{\varepsilon}/2},
\end{align*}
since $1<\alpha\le2$, $0<\tilde{\delta}\le1/3$, $\beta=e^{-4.2}$
and $0<\tilde{\varepsilon}\le1$, where $\gamma$ is the Euler-Mascheroni
constant. Hence, we have 
\begin{equation}
\Pr\left(\sum_{i}\tilde{T}_{i}^{-\alpha}\le\frac{e^{-\tilde{\varepsilon}/2}}{\alpha-1}\right)\le\frac{\tilde{\delta}}{3}.\label{eq:sumt_lefttail}
\end{equation}

We then bound the right tail of $\sum_{i}\tilde{T}_{i}^{-\alpha}$.
Unfortunately, the Laplace functional \eqref{eq:laplace} does not
work since the integral diverges for small $t$. Therefore, we have
to bound $t$ away from $0$. Note that $\min_{i}\tilde{T}_{i}\sim\mathrm{Exp}(1)$,
and hence 
\begin{equation}
\Pr(\min_{i}\tilde{T}_{i}\le\tilde{\delta}/3)\le\tilde{\delta}/3.\label{eq:mint_lefttail}
\end{equation}
Write $\tau=\tilde{\delta}/3$. Using the Laplace functional again,
for $w>0$,
\begin{align*}
 & \mathbb{E}\left[\exp\Big(w\sum_{i:\,\tilde{T}_{i}>\tau}\tilde{T}_{i}^{-\alpha}\Big)\right]\\
 & =\exp\left(-\int_{\tau}^{\infty}(1-\exp(wt^{-\alpha}))\mathrm{d}t\right)\\
 & =\exp\left(\int_{\tau}^{\infty}(\exp(wt^{-\alpha})-1)\mathrm{d}t\right)\\
 & \le\exp\left(\int_{\tau}^{\infty}(\exp(w\tau^{-\alpha})-1)\frac{t^{-\alpha}}{\tau^{-\alpha}}\mathrm{d}t\right)\\
 & =\exp\left(\frac{\exp(w\tau^{-\alpha})-1}{\tau^{-\alpha}}\cdot\frac{\tau^{-(\alpha-1)}}{\alpha-1}\right)\\
 & =\exp\left(\frac{\tau(\exp(w\tau^{-\alpha})-1)}{\alpha-1}\right).
\end{align*}
Therefore, by Chernoff bound, for $d\ge0$,

\begin{align}
 & \Pr\Big(\sum_{i:\,\tilde{T}_{i}>\tau}\tilde{T}_{i}^{-\alpha}\ge d\Big)\nonumber \\
 & \le\inf_{w>0}\exp\left(-wd+\frac{\tau(\exp(w\tau^{-\alpha})-1)}{\alpha-1}\right)\nonumber \\
 & \le\exp\left(-d\tau^{\alpha}\ln(d(\alpha-1)\tau^{\alpha-1})+\frac{\tau(\exp(\ln(d(\alpha-1)\tau^{\alpha-1}))-1)}{\alpha-1}\right)\nonumber \\
 & =\exp\left(-d\tau^{\alpha}\ln(d(\alpha-1)\tau^{\alpha-1})+\tau\frac{d(\alpha-1)\tau^{\alpha-1}-1}{\alpha-1}\right)\nonumber \\
 & =\exp\left(-\frac{c\tau}{\alpha-1}\ln c+\tau\frac{c-1}{\alpha-1}\right)\nonumber \\
 & =\exp\left(-\frac{\tau}{\alpha-1}\left(c\ln c-c+1\right)\right)\nonumber \\
 & \le\exp\left(-\frac{\tau(2\ln2-1)(c-1)^{2}}{\alpha-1}\right),\label{eq:sumt_righttail}
\end{align}
where 
\[
c:=d(\alpha-1)\tau^{\alpha-1},
\]
and the last inequality holds whenever $c\in[1,2]$ since in this
range,
\[
c\ln c-c+1\ge(2\ln2-1)(c-1)^{2}.
\]
Substituting
\[
d=\frac{e^{\tilde{\varepsilon}/2}}{\alpha-1},
\]
we have $c=e^{\tilde{\varepsilon}/2}\tau^{\alpha-1}$. By \eqref{eq:sumt_righttail},
to guarantee $\Pr(\sum_{i:\,\tilde{T}_{i}>\tau}\tilde{T}_{i}^{-\alpha}\ge d)\le\tilde{\delta}/3=\tau$,
we require
\[
\frac{\tau(2\ln2-1)(e^{\tilde{\varepsilon}/2}\tau^{\alpha-1}-1)^{2}}{\alpha-1}\ge-\ln\tau,
\]
\begin{equation}
e^{\tilde{\varepsilon}/2}\tau^{\alpha-1}\ge\sqrt{\frac{(\alpha-1)(-\ln\tau)}{\tau(2\ln2-1)}}+1.\label{eq:sumt_righttail_2}
\end{equation}
Substituting \eqref{eq:eps_d_alpha}, we have
\begin{align*}
e^{\tilde{\varepsilon}/2}\tau^{\alpha-1} & \ge e^{\tilde{\varepsilon}/2}\tau^{\frac{\beta\tilde{\delta}\tilde{\varepsilon}^{2}}{-\ln\tilde{\delta}}}\\
 & =\exp\left(\frac{\tilde{\varepsilon}}{2}+\left(\ln\frac{\tilde{\delta}}{3}\right)\frac{\beta\tilde{\delta}\tilde{\varepsilon}^{2}}{-\ln\tilde{\delta}}\right)\\
 & \ge\exp\left(\frac{\tilde{\varepsilon}}{2}+\left(2\ln\tilde{\delta}\right)\frac{\beta\tilde{\varepsilon}}{-3\ln\tilde{\delta}}\right)\\
 & =\exp\left(\tilde{\varepsilon}\left(\frac{1}{2}-\frac{2\beta}{3}\right)\right),
\end{align*}
since $0<\tilde{\delta}\le1/3$. Note that this also guarantees $c=e^{\tilde{\varepsilon}/2}\tau^{\alpha-1}\in[1,2]$
since $\beta=e^{-4.2}$ and $0<\tilde{\varepsilon}\le1$. We also have
\begin{align*}
\frac{(\alpha-1)(-\ln\tau)}{\tau(2\ln2-1)} & \le\frac{\frac{\beta\tilde{\delta}\tilde{\varepsilon}^{2}}{-\ln\tilde{\delta}}(-\ln\tau)}{\tau(2\ln2-1)}\\
 & \le\frac{\frac{\beta\tilde{\delta}\tilde{\varepsilon}^{2}}{-\ln\tilde{\delta}}(-2\ln\tilde{\delta})}{(\tilde{\delta}/3)(2\ln2-1)}\\
 & =\frac{6\beta\tilde{\varepsilon}^{2}}{2\ln2-1}\\
 & \le16\beta\tilde{\varepsilon}^{2}.
\end{align*}
Hence,
\begin{align*}
\sqrt{\frac{(\alpha-1)(-\ln\tau)}{\tau(2\ln2-1)}}+1 & \le4\tilde{\varepsilon}\sqrt{\beta}+1\\
 & \le\exp\left(4\tilde{\varepsilon}\sqrt{\beta}\right)\\
 & \stackrel{(a)}{\le}\exp\left(\tilde{\varepsilon}\left(\frac{1}{2}-\frac{2\beta}{3}\right)\right)\\
 & \le e^{\tilde{\varepsilon}/2}\tau^{\alpha-1},
\end{align*}
where (a) is by $\beta=e^{-4.2}$. Hence \eqref{eq:sumt_righttail_2}
is satisfied, and
\[
\Pr\Big(\sum_{i:\,\tilde{T}_{i}>\tau}\tilde{T}_{i}^{-\alpha}\ge\frac{e^{\tilde{\varepsilon}/2}}{\alpha-1}\Big)\le\frac{\tilde{\delta}}{3}.
\]
Combining this with \eqref{eq:sumt_lefttail} and \eqref{eq:mint_lefttail},
\begin{align*}
 & \Pr\Big(\sum_{i}\tilde{T}_{i}^{-\alpha}\notin\Big[\frac{e^{-\tilde{\varepsilon}/2}}{\alpha-1},\,\frac{e^{\tilde{\varepsilon}/2}}{\alpha-1}\Big]\Big)\\
 & \le\Pr\Big(\sum_{i}\tilde{T}_{i}^{-\alpha}\le\frac{e^{-\tilde{\varepsilon}/2}}{\alpha-1}\Big)+\Pr(\min_{i}\tilde{T}_{i}\le\tilde{\delta}/3)\\
 & \;\;\;+\Pr\Big(\sum_{i:\,\tilde{T}_{i}>\tilde{\delta}/3}\tilde{T}_{i}^{-\alpha}\ge\frac{e^{\tilde{\varepsilon}/2}}{\alpha-1}\Big)\\
 & \le\tilde{\delta}.
\end{align*}

Consider an $(\varepsilon,\delta)$-differentially private mechanism
$P_{Z|X}$. Consider neighbors $x_{1},x_{2}$, and let $P_{j}:=P_{Z|X}(\cdot|x_{j})$,
$\tilde{T}_{j,i}:=T_{i}/(\frac{\mathrm{d}P_{j}}{\mathrm{d}Q}(Z_{i}))$,
and $K_{j}$ be the output of PPR applied on $P_{j}$, for $j=1,2$.
We first consider the case $\delta=0$, which gives $\frac{\mathrm{d}P_{1}}{\mathrm{d}Q}(z)\le e^{\varepsilon}\frac{\mathrm{d}P_{2}}{\mathrm{d}Q}(z)$
for every $z$. For any measurable $\mathcal{S}\subseteq\mathcal{Z}^{\infty}\times\mathbb{Z}_{>0}$,
\begin{align}
 & \Pr\left(((Z_{i})_{i},K_{1})\in\mathcal{S}\right)\nonumber \\
 & =\mathbb{E}\left[\Pr\left(((Z_{i})_{i},K_{1})\in\mathcal{S}\,\big|\,(Z_{i},T_{i})_{i}\right)\right]\nonumber \\
 & =\mathbb{E}\left[\sum_{k:\,((Z_{i})_{i},k)\in\mathcal{S}}\Pr\left(K_{1}=k\,\big|\,(Z_{i},T_{i})_{i}\right)\right]\nonumber \\
 & =\mathbb{E}\left[\sum_{k:\,((Z_{i})_{i},k)\in\mathcal{S}}\frac{\tilde{T}_{1,k}^{-\alpha}}{\sum_{i}\tilde{T}_{1,i}^{-\alpha}}\right]\nonumber \\
 & \le\mathbb{E}\left[\mathbf{1}\left\{ \sum_{i}\tilde{T}_{1,i}^{-\alpha}\in\Big[\frac{e^{-\tilde{\varepsilon}/2}}{\alpha-1},\,\frac{e^{\tilde{\varepsilon}/2}}{\alpha-1}\Big]\right\} \min\left\{ \sum_{k:\,((Z_{i})_{i},k)\in\mathcal{S}}\frac{\tilde{T}_{1,k}^{-\alpha}}{\sum_{i}\tilde{T}_{1,i}^{-\alpha}},\,1\right\} \right]+\tilde{\delta}\nonumber \\
 & \le\mathbb{E}\left[\min\left\{ \sum_{k:\,((Z_{i})_{i},k)\in\mathcal{S}}\frac{\tilde{T}_{1,k}^{-\alpha}}{e^{-\tilde{\varepsilon}/2}/(\alpha-1)},\,1\right\} \right]+\tilde{\delta}\nonumber \\
 & =\mathbb{E}\left[\min\left\{ \sum_{k:\,((Z_{i})_{i},k)\in\mathcal{S}}\frac{(\frac{\mathrm{d}P_{1}}{\mathrm{d}Q}(Z_{k}))^{\alpha}T_{k}^{-\alpha}}{e^{-\tilde{\varepsilon}/2}/(\alpha-1)},\,1\right\} \right]+\tilde{\delta}\nonumber \\
 & \le\mathbb{E}\left[\min\left\{ \sum_{k:\,((Z_{i})_{i},k)\in\mathcal{S}}\frac{(e^{\varepsilon}\frac{\mathrm{d}P_{2}}{\mathrm{d}Q}(Z_{k}))^{\alpha}T_{k}^{-\alpha}}{e^{-\tilde{\varepsilon}/2}/(\alpha-1)},\,1\right\} \right]+\tilde{\delta}\nonumber \\
 & \le\mathbb{E}\left[\mathbf{1}\left\{ \sum_{i}\tilde{T}_{2,i}^{-\alpha}\in\Big[\frac{e^{-\tilde{\varepsilon}/2}}{\alpha-1},\,\frac{e^{\tilde{\varepsilon}/2}}{\alpha-1}\Big]\right\} \min\left\{ \sum_{k:\,((Z_{i})_{i},k)\in\mathcal{S}}\frac{(e^{\varepsilon}\frac{\mathrm{d}P_{2}}{\mathrm{d}Q}(Z_{k}))^{\alpha}T_{k}^{-\alpha}}{e^{-\tilde{\varepsilon}/2}/(\alpha-1)},\,1\right\} \right]+2\tilde{\delta}\nonumber \\
 & \le\mathbb{E}\left[\min\left\{ e^{\alpha\varepsilon}\sum_{k:\,((Z_{i})_{i},k)\in\mathcal{S}}\frac{(\frac{\mathrm{d}P_{2}}{\mathrm{d}Q}(Z_{k}))^{\alpha}T_{k}^{-\alpha}}{e^{-\tilde{\varepsilon}}\sum_{i}\tilde{T}_{2,i}^{-\alpha}},\,1\right\} \right]+2\tilde{\delta}\nonumber \\
 & \le\mathbb{E}\left[e^{\alpha\varepsilon+\tilde{\varepsilon}}\sum_{k:\,((Z_{i})_{i},k)\in\mathcal{S}}\frac{\tilde{T}_{2,k}^{-\alpha}}{\sum_{i}\tilde{T}_{2,i}^{-\alpha}}\right]+2\tilde{\delta}\nonumber \\
 & =e^{\alpha\varepsilon+\tilde{\varepsilon}}\Pr\left(((Z_{i})_{i},K_{2})\in\mathcal{S}\right)+2\tilde{\delta}.\label{eq:pf_eps_dp}
\end{align}
Hence PPR is $(\alpha\varepsilon+\tilde{\varepsilon},\,2\tilde{\delta})$-differentially
private.

For the case $\delta>0$, by the definition of $(\varepsilon,\delta)$-differential
privacy, we have
\[
\int\max\left\{ \frac{\mathrm{d}P_{1}}{\mathrm{d}Q}(z)-e^{\varepsilon}\frac{\mathrm{d}P_{2}}{\mathrm{d}Q}(z),\,0\right\} Q(\mathrm{d}z)\le\delta.
\]
Let $P_{3}$ be a probability measure that satisfies
\[
\min\left\{ \frac{\mathrm{d}P_{1}}{\mathrm{d}Q}(z),\,e^{\varepsilon}\frac{\mathrm{d}P_{2}}{\mathrm{d}Q}(z)\right\} \le\frac{\mathrm{d}P_{3}}{\mathrm{d}Q}(z)\le e^{\varepsilon}\frac{\mathrm{d}P_{2}}{\mathrm{d}Q}(z),
\]
for every $z$. Such $P_{3}$ can be constructed by taking an appropriate
convex combination of the lower bound above (which integrates to $\le1$)
and the upper bound above (which integrates to $\ge1$) such that
$P_{3}$ integrates to $1$. We have
\[
\int\max\left\{ \frac{\mathrm{d}P_{1}}{\mathrm{d}Q}(z)-\frac{\mathrm{d}P_{3}}{\mathrm{d}Q}(z),\,0\right\} Q(\mathrm{d}z)\le\delta,
\]
and hence the total variation distance $d_{\mathrm{TV}}(P_{1},P_{3})$
between $P_{1}$ and $P_{3}$ is at most $\delta$. Let $K_{3}$ be
the output of PPR applied on $P_{3}$. 

In the proof of Theorem \ref{thm:logk_bound}, we see that PPR has
the following equivalent formulation. Let $(T_{i})_{i}\sim\mathrm{PP}(1)$
be a Poisson process with rate $1$, independent of $Z_{1},Z_{2},\ldots\stackrel{iid}{\sim}Q$.
Let $R_{i}:=(\mathrm{d}P/\mathrm{d}Q)(Z_{i})$, and let its probability
measure be $P_{R}$. Let $V_{1},V_{2},\ldots\stackrel{iid}{\sim}\mathrm{Exp}(1)$.
PPR can be equivalently expressed as
\begin{align*}
K & =\underset{k}{\mathrm{argmin}}T_{k}^{\alpha}R_{k}^{-\alpha}V_{k} =\underset{k}{\mathrm{argmin}}\frac{T_{k}V_{k}^{1/\alpha}}{R_{k}}.
\end{align*}
Note that $(T_{i}V_{i}^{1/\alpha})_{i}\sim\mathrm{PP}(\int_{0}^{\infty}v^{-1/\alpha}e^{-v}\mathrm{d}v)=\mathrm{PP}(\Gamma(1-\alpha^{-1}))$
is a uniform Poisson process. Therefore PPR is the same as the Poisson
functional representation \cite{li2018strong,li2021unified} applied
on $(T_{i}V_{i}^{1/\alpha})_{i}$. By the grand coupling property of Poisson
functional representation \cite{li2021unified,li2019pairwise} (see \cite[Theorem 3]{li2019pairwise}), if
we apply the Poisson functional representation on $P_{1}$ and $P_{3}$
to get $K_{1}$ and $K_{3}$ respectively, then
\[
\Pr(K_{1}\neq K_{3})\le2d_{\mathrm{TV}}(P_{1},P_{3}) \le 2 \delta.
\]
Therefore, for any measurable $\mathcal{S}\subseteq\mathcal{Z}^{\infty}\times\mathbb{Z}_{>0}$,
\begin{align*}
 & \Pr\left(((Z_{i})_{i},K_{1})\in\mathcal{S}\right)\\
 & \le\Pr\left(((Z_{i})_{i},K_{3})\in\mathcal{S}\right)+2\delta\\
 & \le e^{\alpha\varepsilon+\tilde{\varepsilon}}\Pr\left(((Z_{i})_{i},K_{2})\in\mathcal{S}\right)+2\tilde{\delta}+2\delta,
\end{align*}
where the last inequality is by applying \eqref{eq:pf_eps_dp} on
$P_{3},P_{2}$ instead of $P_{1},P_{2}$. This completes the proof.

\end{proof}

\section{Proof of Theorem~\ref{thm:logk_bound_simple}\label{sec:pf_logk_bound_simple}}

We now bound the size of the index output by the Poisson private representation.
The following is a refined version of Theorem \ref{thm:logk_bound_simple}.

\medskip{}

\begin{theorem}\label{thm:ldp_compression}
\label{thm:logk_bound}For PPR with parameter $\alpha>1$, when the
encoder is given the input $x$, the message $K$ given by PPR satisfies
\begin{align}
\mathbb{E}[\log K]  &\le D(P\Vert Q) \nonumber \\
&\quad\; +\inf_{\eta\in(0,1]\cap(0,\alpha-1)}\frac{1}{\eta}\log\left(\frac{\Gamma(1-\frac{\eta+1}{\alpha})\Gamma(\eta+1)}{(\Gamma(1-\frac{1}{\alpha}))^{\eta+1}}+1\right)\label{eq:logK_bd1}\\
 & \le D(P\Vert Q)+\frac{\log(3.56)}{\min\{(\alpha-1)/2,1\}},\label{eq:logK_bd2}
\end{align}
where $P:=P_{Z|X}(\cdot|x)$.
\end{theorem}

Note that for $\alpha=\infty$, \eqref{eq:logK_bd1} with $\eta=1$
gives $\mathbb{E}[\log K]\le D(P\Vert Q)+\log2$, recovering the bound
in \cite{li2024lossy_arxiv} (which strengthened \cite{li2018strong}).

\begin{proof}
Write $(X_{i})_{i}\sim\mathrm{PP}(\mu)$ if the points $(X_{i})_{i}$
(as a multiset, ignoring the ordering) form a Poisson point process
with intensity measure $\mu$. Similarly, for $f:[0,\infty)^{n}\to[0,\infty)$,
we write $\mathrm{PP}(f)$ for the Poisson point process with intensity
function $f$ (i.e., the intensity measure has a Radon-Nikodym derivative
$f$ against the Lebesgue measure). Let $(T_{i})_{i}\sim\mathrm{PP}(1)$
be a Poisson process with rate $1$, independent of $Z_{1},Z_{2},\ldots\stackrel{iid}{\sim}Q$.
 Let $R_{i}:=(\mathrm{d}P/\mathrm{d}Q)(Z_{i})$, and let its probability
measure be $P_{R}$. We have $\tilde{T}_{i}=T_{i}/R_{i}$. Let $V_{1},V_{2},\ldots\stackrel{iid}{\sim}\mathrm{Exp}(1)$.
By the property of exponential random variables, for any $p_{1},p_{2},\ldots\ge0$
with $\sum_{i}p_{i}<\infty$, we have $\Pr(\mathrm{argmin}_{k}V_{k}/p_{k}=k)=p_{k}/\sum_{i}p_{i}$.
Therefore, PPRF can be equivalently expressed as
\[
K=\underset{k}{\mathrm{argmin}}T_{k}^{\alpha}R_{k}^{-\alpha}V_{k}.
\]
By the marking theorem \cite{last2017lectures}, $(T_{i},R_{i},V_{i})_{i}$
is a Poisson process over $[0,\infty)^{3}$ with intensity measure
\[
(T_{i},R_{i},V_{i})_{i}\sim\mathrm{PP}\left(e^{-v}P_{R}(r)\right).
\]
By the mapping theorem \cite{last2017lectures}, letting $W_{i}:=T_{i}^{\alpha}R_{i}^{-\alpha}V_{i}$,
we have
\begin{equation}
(T_{i},R_{i},W_{i})_{i}\sim\mathrm{PP}\left(r^{\alpha}t^{-\alpha}e^{-wr^{\alpha}t^{-\alpha}}P_{R}(r)\right).\label{eq:three_poisson}
\end{equation}
Again by the mapping theorem,
\begin{align*}
(W_{i})_{i} & \sim\mathrm{PP}\left(\mathbb{E}_{R\sim P_{R}}\left[\int_{0}^{\infty}R^{\alpha}t^{-\alpha}e^{-wR^{\alpha}t^{-\alpha}}\mathrm{d}t\right]\right)\\
 & =\mathrm{PP}\left(\mathbb{E}\left[\alpha^{-1}(wR^{\alpha})^{1/\alpha-1}\Gamma(1-\alpha^{-1})R^{\alpha}\right]\right)\\
 & =\mathrm{PP}\left(\mathbb{E}\left[\alpha^{-1}w^{1/\alpha-1}\Gamma(1-\alpha^{-1})R\right]\right)\\
 & =\mathrm{PP}\left(\alpha^{-1}w^{1/\alpha-1}\Gamma(1-\alpha^{-1})\right)
\end{align*}
since $\mathbb{E}[R]=\int(\mathrm{d}P/\mathrm{d}Q)(z)Q(\mathrm{d}z)=1$.
Recall that $W_{K}=\min_{i}W_{i}$ by the definition of $K$. We have
\begin{align*}
\Pr(W_{K}>w) & =\exp\left(-\int_{0}^{w}\alpha^{-1}v^{1/\alpha-1}\Gamma(1-\alpha^{-1})\mathrm{d}v\right)\\
 & =\exp\left(-w^{1/\alpha}\Gamma(1-\alpha^{-1})\right).
\end{align*}
Hence the probability density function of $W_{K}$ is
\begin{align}
 & -\frac{\mathrm{d}}{\mathrm{d}w}\exp\left(-w^{1/\alpha}\Gamma(1-\alpha^{-1})\right)\nonumber \\
 & =\alpha^{-1}w^{1/\alpha-1}\Gamma(1-\alpha^{-1})\exp\left(-w^{1/\alpha}\Gamma(1-\alpha^{-1})\right).\label{eq:w_pdf}
\end{align}
By \eqref{eq:three_poisson}, the Radon-Nikodym derivative between
the conditional distribution of $R_{K}$ given $W_{K}=w$ and $P_{R}$
is
\begin{align*}
 & \Pr(R_{K}\in[r,r+\mathrm{d}r)\,|\,W_{K}=w)/P_{R}(\mathrm{d}r)\\
 & =\frac{\int_{0}^{\infty}r^{\alpha}t^{-\alpha}e^{-wr^{\alpha}t^{-\alpha}}\mathrm{d}t}{\mathbb{E}_{R\sim P_{R}}\left[\int_{0}^{\infty}R^{\alpha}t^{-\alpha}e^{-wR^{\alpha}t^{-\alpha}}\mathrm{d}t\right]}\\
 & =\frac{\alpha^{-1}w^{1/\alpha-1}\Gamma(1-\alpha^{-1})r}{\alpha^{-1}w^{1/\alpha-1}\Gamma(1-\alpha^{-1})}\\
 & =r
\end{align*}
does not depend on $w$. Hence $R_{K}$ is independent of $W_{K}$.
By \eqref{eq:three_poisson}, for $0\le\eta<\alpha-1$,

\begin{align}
 & \mathbb{E}[T_{K}^{\eta}\,|\,R_{K}=r,\,W_{K}=w]\nonumber \\
 & =\frac{\int_{0}^{\infty}t^{\eta}r^{\alpha}t^{-\alpha}e^{-wr^{\alpha}t^{-\alpha}}\mathrm{d}t}{\int_{0}^{\infty}r^{\alpha}t^{-\alpha}e^{-wr^{\alpha}t^{-\alpha}}\mathrm{d}t}\nonumber \\
 & =\frac{\alpha^{-1}w^{(\eta+1)/\alpha-1}\Gamma(1-(\eta+1)\alpha^{-1})r^{\eta+1}}{\alpha^{-1}w^{1/\alpha-1}\Gamma(1-\alpha^{-1})r}.\label{eq:ET_gRW}
\end{align}
Since $R_{K}$ is independent of $W_{K}$, using \eqref{eq:ET_gRW}
and \eqref{eq:w_pdf}, for $\eta\in(0,1]\cap(0,\alpha-1)$, 
\begin{align}
 & \mathbb{E}[T_{K}^{\eta}\,|\,R_{K}=r]\nonumber \\
 & =\int_{0}^{\infty}\alpha^{-1}w^{(\eta+1)/\alpha-1}\Gamma(1-(\eta+1)\alpha^{-1})r^{\eta}\exp\left(-w^{1/\alpha}\Gamma(1-\alpha^{-1})\right)\mathrm{d}w\nonumber \\
 & =r^{\eta}\Gamma(1-(\eta+1)\alpha^{-1})\int_{0}^{\infty}\alpha^{-1}w^{(\eta+1)/\alpha-1}\exp\left(-w^{1/\alpha}\Gamma(1-\alpha^{-1})\right)\mathrm{d}w\nonumber \\
 & =r^{\eta}\Gamma(1-(\eta+1)\alpha^{-1})(\Gamma(1-\alpha^{-1}))^{-(\eta+1)}\Gamma(\eta+1)\nonumber \\
 & =:c_{\alpha,\eta}r^{\eta},\label{eq:ET_gR}
\end{align}
where $c_{\alpha,\eta}:=\Gamma(1-(\eta+1)\alpha^{-1})(\Gamma(1-\alpha^{-1}))^{-(\eta+1)}\Gamma(\eta+1)$.
Hence,
\begin{align}
 & \mathbb{E}[\log(T_{K}+1)\,|\,R_{K}=r]\nonumber \\
 & \le\mathbb{E}[\log((T_{K}^{\eta}+1)^{1/\eta})\,|\,R_{K}=r]\nonumber \\
 & =\mathbb{E}[\eta^{-1}\log(T_{K}^{\eta}+1)\,|\,R_{K}=r]\nonumber \\
 & \le\eta^{-1}\log(c_{\alpha,\eta}r^{\eta}+1).\label{eq:logT_gR}
\end{align}

Note that
\[
K-1=\left|\left\{ i:\,T_{i}<T_{K}\right\} \right|,
\]
and hence the expecation of $K-1$ given $T_{K}$ should be around
$T_{K}$. This is not exact since conditioning on $T_{K}$ changes
the distribution of the process $(T_{i},R_{i},V_{i})_{i}$. To resolve
this problem, we define a new process $(T'_{i},R'_{i},V'_{i})_{i}$
which includes all points in $(T_{i},R_{i},V_{i})_{i}$ excluding
the point $(T_{K},R_{K},V_{K})$, together with newly generated points
according to 
\[
\mathrm{PP}\left(e^{-v}P_{R}(r)\mathbf{1}\{t^{\alpha}r^{-\alpha}v<T_{K}^{\alpha}R_{K}^{-\alpha}V_{K}\}\right).
\]
Basically, $\{t^{\alpha}r^{-\alpha}v<T_{K}^{\alpha}R_{K}^{-\alpha}V_{K}\}$
is the ``impossible region'' where the points in $(T_{i},R_{i},V_{i})_{i}$
cannot be located in, since $K$ attains the minimum of $T_{K}^{\alpha}R_{K}^{-\alpha}V_{K}$.
The new process $(T'_{i},R'_{i},V'_{i})_{i}$ removes the point $(T_{K},R_{K},V_{K})$,
and then fills in the impossible region. It is straightforward to
check that $(T'_{i},R'_{i},V'_{i})_{i}\sim\mathrm{PP}(e^{-v}P_{R}(r))$,
independent of $(T_{K},R_{K},V_{K})$. We have
\begin{align*}
 & \mathbb{E}[K\,|\,T_{K}]\\
 & =\mathbb{E}\left[\left|\left\{ i:\,T_{i}<T_{K}\right\} \right|\,\Big|\,T_{K}\right]+1\\
 & \le\mathbb{E}\left[\left|\left\{ i:\,T'_{i}<T_{K}\right\} \right|\,\Big|\,T_{K}\right]+1\\
 & =T_{K}+1.
\end{align*}
Therefore, by \eqref{eq:logT_gR} and Jensen's inequality,
\begin{align*}
 & \mathbb{E}[\log K]\\
 & =\mathbb{E}\left[\mathbb{E}[\log K\,|\,T_{K}]\right]\\
 & \le\mathbb{E}\left[\log(T_{K}+1)\right]\\
 & =\mathbb{E}\left[\mathbb{E}\left[\log(T_{K}+1)\,\big|\,R_{K}\right]\right]\\
 & \le\mathbb{E}\left[\eta^{-1}\log(c_{\alpha,\eta}R_{K}^{\eta}+1)\right]\\
 & =\eta^{-1}\mathbb{E}_{Z\sim P}\left[\log\left(c_{\alpha,\eta}\left(\frac{\mathrm{d}P}{\mathrm{d}Q}(Z)\right)^{\eta}+1\right)\right]\\
 & =\eta^{-1}\mathbb{E}\left[\log\left(\left(\frac{\mathrm{d}P}{\mathrm{d}Q}(Z)\right)^{\eta}\right)\right]+\eta^{-1}\mathbb{E}\left[\log\left(c_{\alpha,\eta}+\Big(\frac{\mathrm{d}P}{\mathrm{d}Q}(Z)\Big)^{-\eta}\right)\right]\\
 & \le D(P\Vert Q)+\eta^{-1}\log\left(c_{\alpha,\eta}+\left(\mathbb{E}\left[\Big(\frac{\mathrm{d}P}{\mathrm{d}Q}(Z)\Big)^{-1}\right]\right)^{\eta}\right)\\
 & \le D(P\Vert Q)+\eta^{-1}\log(c_{\alpha,\eta}+1),
\end{align*}
where the last line is due to $\mathbb{E}[((\mathrm{d}P/\mathrm{d}Q)(Z))^{-1}]=\int((\mathrm{d}P/\mathrm{d}Q)(Z))^{-1}P(\mathrm{d}z)\le1$
(this step appeared in \cite{li2024lossy_arxiv}). The bound \eqref{eq:logK_bd1}
follows from minimizing over $\eta\in(0,1]\cap(0,\alpha-1)$.

To show \eqref{eq:logK_bd2}, substituting $\eta=\min\{(\alpha-1)/2,1\}$,
\begin{align*}
c_{\alpha,\eta} & =\frac{\Gamma(1-(\eta+1)\alpha^{-1})\Gamma(\eta+1)}{(\Gamma(1-\alpha^{-1}))^{\eta+1}}\\
 & \stackrel{(a)}{\le}\frac{(1-\alpha^{-1})^{\eta+1}}{0.885^{\eta+1}\cdot(1-(\eta+1)\alpha^{-1})}\\
 & \le\frac{(1-\alpha^{-1})^{\eta+1}}{0.885^{2}\cdot(1-((\alpha-1)/2+1)\alpha^{-1})}\\
 & =\frac{2}{0.885^{2}}(1-\alpha^{-1})^{\eta}\\
 & \le2.56,
\end{align*}
where (a) is because $0.885\le x\Gamma(x)=\Gamma(x+1)\le1$ for $0<x\le1$.
Hence,
\begin{align*}
\mathbb{E}[\log K] & \le D(P\Vert Q)+\eta^{-1}\log(c_{\alpha,\eta}+1),\\
 & \le D(P\Vert Q)+\frac{\log(3.56)}{\min\{(\alpha-1)/2,1\}}.
\end{align*}
\end{proof}
\medskip{}

\section{Distributed Mean Estimation with R\'enyi DP}\label{sec::appendix_mean_est}
In many machine learning applications, privacy budgets are often accounted in the moment space, and one popular moment accountant is the R\'enyi DP accountant. For completeness, we provide a R\'enyi DP version of Corollary~\ref{cor:gaussian_ppr2} in this section. We begin with the following definition of R\'enyi DP:
\begin{definition}[R\'enyi Differential privacy \citep{abadi2016deep, mironov2017renyi}]
Given a mechanism $\mcal{A}$ which induces the conditional distribution $P_{Z|X}$ of $Z = \mcal{A}(X)$, we say that it satisfies $(\gamma, \varepsilon)$- R\'enyi DP if for any neighboring $(x, x') \in \mathcal{N}$ and $\mathcal{S} \subseteq \mathcal{Z}$, it holds that
$$ D_\gamma\lp P_{Z | X=x} \middle\Vert P_{Z | X=x'} \rp \leq  \varepsilon, $$
where 
$$ D_\gamma\lp P \middle \Vert Q \rp \eqDef \frac{1}{\gamma-1}\log\lp \mathbb{E}_{Q}\lb \lp\frac{P}{Q}\rp^\gamma \rb\rp$$
is the R\'enyi divergence between $P$ and $Q$.
If $\mathcal{N} = \mathcal{X}^2$, we say that the mechanism satisfies $(\gamma, \varepsilon)$-local DP.
\end{definition}

The following conversion lemma from \cite{canonne2020discrete} relates R\'enyi DP to $\lp\varepsilon_{\msf{DP}}(\delta), \delta\rp$-DP.
\begin{lemma}\label{lemma:rdp_to_approx_dp}
    If $\mcal{A}$ satisfies $\lp \gamma, \varepsilon \rp$-R\'enyi DP for some $\gamma \geq 1$, then, for any $\delta > 0$, $\mcal{A}$ satisfies $\lp \varepsilon_{\msf{DP}}(\delta), \delta \rp$-DP, where
    \begin{equation}\label{eq:rdp_to_dp}
        \varepsilon_{\msf{DP}}(\delta) = \varepsilon+\frac{\log\lp 1/\gamma\delta \rp}{\gamma-1}+\log(1-1/\gamma).
    \end{equation}
 
\end{lemma}
The following theorem states that, when simulating the Gaussian mechanism, PPR satisfies the following both central and local DP guarantee:

\begin{corollary}[PPR-compressed Gaussian mechanism]\label{cor:gaussian_ppr_renyi} 
Let $\varepsilon \geq 0$ and $\gamma \geq 1$. Consider the Gaussian mechanism $P_{Z|X}(\cdot | x) = \mcal{N}( x, \frac{\sigma^2}{n}\mbb{I}_d)$, and the proposal distribution $Q=\mathcal{N}(0,(\frac{C^{2}}{d}+\frac{\sigma^{2}}{n})\mathbb{I}_{d})$, where $\sigma \geq \sqrt{\frac{C\gamma}{2\varepsilon}}$.
For each client $i$, let 
$Z_i$ be the output of PPR applied on $P_{Z|X}(\cdot | X_i)$. We have:
\begin{itemize}
    \item $\hat{\mu}(Z^n) := \frac{1}{n}\sum_i Z_i$ yields an unbiased estimator of $\mu( X^n) = \frac{1}{n} \sum_{i=1}^n X_i$ satisfying $(\gamma, \varepsilon)$-(central) R\'enyi DP and $(\varepsilon_{\msf{DP}}(\delta), \delta)$-DP, where $\varepsilon_{\msf{DP}}(\delta)$ is defined in \eqref{eq:rdp_to_dp}.
    
    \item $P_{Z|X_i}$ satisfies $\lp 2\alpha\tilde{\varepsilon}_{\msf{DP}}(\delta), 2\delta\rp$-local DP, where
    $$ \tilde{\varepsilon}_{\msf{DP}}(\delta) \eqDef  \sqrt{n}\varepsilon+\frac{\log\lp 1/\gamma\delta \rp}{\gamma-1}+\log(1-1/\gamma).$$
    
    \item $\hat{\mu}(Z^n)$ has MSE $\mathbb{E}[\lV\mu - \hat{\mu}\rV_2^2] = \sigma^2 d/n^2$.
    
    \item The average per-client communication cost is at most $\ell + \log_2(\ell + 1) + 2$ bits where 
\begin{align*}
\ell & :=\frac{d}{2}\log_{2}\Big(\frac{C^{2}n}{d\sigma^{2}}+1\Big)+\eta_{\alpha} \; \le\;\frac{d}{2}\log_{2}\Big(\frac{n\varepsilon^{2}}{2d\ln(1.25/\delta)}+1\Big)+\eta_{\alpha},
\end{align*}

where $\eta_{\alpha}:=(\log_{2}(3.56))/\min\{(\alpha-1)/2,\,1\}$.
\end{itemize}
\end{corollary}

\textbf{Proof.}
    The central DP guarantee follows from \citep{mironov2017renyi} and Lemma~\ref{lemma:rdp_to_approx_dp}. The local DP guarantee follows from Lemma~\ref{lemma:rdp_to_approx_dp} and Theorem~\ref{thm:eps_delta_dp}. Finally, the communication bound can be obtained from the same analysis as in Corollary~\ref{cor:gaussian_ppr2}.
\qed

\section{Proof of Corollary~\ref{cor:gaussian_ppr2}\label{sec:pf_gaussian_ppr2}}

Consider the PPR applied on the Gaussian mechanism $P_{Z|X}(\cdot|x)=\mathcal{N}(x,\frac{\sigma^{2}}{n}\mathbb{I}_{d})$,
with the proposal distribution $Q=\mathcal{N}(0,(\frac{C^{2}}{d}+\frac{\sigma^{2}}{n})\mathbb{I}_{d})$.
PPR ensures that $Z_{i}$ follows the distribution $\mathcal{N}(X_{i},\frac{\sigma^{2}}{n}\mathbb{I}_{d})$.
Therefore the MSE is
\begin{align*}
\mathbb{E}\left[\Vert\mu-\hat{\mu}\Vert_{2}^{2}\right] & =\mathbb{E}\left[\left\Vert \frac{1}{n}\sum_{i=1}^{n}(X_{i}-Z_{i})\right\Vert _{2}^{2}\right]\\
 & =\frac{1}{n}\cdot d\cdot\frac{\sigma^{2}}{n}\\
 & =\frac{\sigma^{2}d}{n^{2}}.
\end{align*}
For the compression size, for $x\in\mathbb{R}^{d}$ with $\Vert x\Vert_{2}\le C$,
we have
\begin{align*}
 & D(P_{Z|X}(\cdot|x)\Vert Q)\\
 & =\mathbb{E}_{Z\sim P_{Z|X}(\cdot|x)}\left[\log\frac{\mathrm{d}P_{Z|X}(\cdot|x)}{\mathrm{d}Q}(Z)\right]\\
 & =\mathbb{E}_{Z\sim P_{Z|X}(\cdot|x)}\left[\log\frac{(2\pi\sigma^{2}/n)^{-d/2}\exp(-\frac{1}{2}\Vert Z-x\Vert_{2}^{2}/(\sigma^{2}/n))}{(2\pi(\frac{C^{2}}{d}+\frac{\sigma^{2}}{n}))^{-d/2}\exp(-\frac{1}{2}\Vert Z\Vert_{2}^{2}/(\frac{C^{2}}{d}+\frac{\sigma^{2}}{n}))}\right]\\
 & =\mathbb{E}_{Z\sim P_{Z|X}(\cdot|x)}\left[\frac{d}{2}\log\frac{\frac{C^{2}}{d}+\frac{\sigma^{2}}{n}}{\sigma^{2}/n}+\frac{1}{2}\left(\frac{\Vert Z\Vert_{2}^{2}}{\frac{C^{2}}{d}+\frac{\sigma^{2}}{n}}-\frac{\Vert Z-x\Vert_{2}^{2}}{\sigma^{2}/n}\right)\right]\\
 & \le\frac{d}{2}\log\frac{\frac{C^{2}}{d}+\frac{\sigma^{2}}{n}}{\sigma^{2}/n}+\frac{1}{2}\left(\frac{C^{2}+\sigma^{2}d/n}{\frac{C^{2}}{d}+\frac{\sigma^{2}}{n}}-\frac{\sigma^{2}d/n}{\sigma^{2}/n}\right)\\
 & =\frac{d}{2}\log\left(\frac{C^{2}n}{d\sigma^{2}}+1\right).
\end{align*}
Hence, by Theorem \ref{thm:logk_bound_simple}, the compression size is at most $\ell+\log_{2}(\ell+1)+2$
bits, where
\begin{align*}
\ell & :=\frac{d}{2}\log_{2}\Big(\frac{C^{2}n}{d\sigma^{2}}+1\Big)+\eta_{\alpha}\\
 & \le\frac{d}{2}\log_{2}\Big(\frac{n\epsilon^{2}}{2d\ln(1.25/\delta)}+1\Big)+\eta_{\alpha}\\
 & \le\frac{n\epsilon^{2}\log_{2}(e)}{4\ln(1.25/\delta)}+\eta_{\alpha},
\end{align*}
where $\eta_{\alpha}:=(\log_{2}(3.56))/\min\{(\alpha-1)/2,\,1\}$.

The central-DP guarantee follows from $(\varepsilon,\delta)$-DP of Gaussian mechanism~\cite[Appendix A]{dwork2014algorithmic} since the output distribution of PPR is exactly the same as the Gaussian mechanism, whereas the local-DP guarantee follows from Theorem~\ref{thm:eps_delta_dp_2} and \cite[Appendix A]{dwork2014algorithmic}.

\section{Proof of Corollary~\ref{cor:laplace_ppr}\label{sec:pf_laplace_ppr}}

Let $\Vert X-Z\Vert_{2}=RS$ where $R\in[0,\infty)$ is the magnitude of $X-Z$, and $\Vert S\Vert_{2}=1$.
As shown in \cite{fernandes2019generalised}, $R$ follows the Gamma distribution with shape
$d$ and scale $1/\varepsilon$. Hence the MSE is

\begin{align*}
\mathbb{E}\left[\Vert X-Z\Vert_{2}^{2}\right]=\mathbb{E}\left[R^{2}\right] & =\left(\frac{d}{\varepsilon}\right)^{2}+\frac{d}{\varepsilon^{2}}=\frac{d(d+1)}{\varepsilon^{2}}.
\end{align*}
The conditional differential entropy (in nats) of $Z$ given $X$
is
\begin{align*}
h(Z|X) & =h(R)+h(S|R)\\
 & =d+\ln\Gamma(d)-(d-1)\psi(d)-\ln\varepsilon+\mathbb{E}\left[\ln(nR^{d-1}\mathrm{Vol}(\mathcal{B}_{d}(1)))\right]\\
 & =d+\ln\Gamma(d)-(d-1)\psi(d)-\ln\varepsilon+\ln d+\ln(\mathrm{Vol}(\mathcal{B}_{d}(1)))+(d-1)\mathbb{E}\left[\ln R\right]\\
 & =d+\ln\Gamma(d)-(d-1)\psi(d)-\ln\varepsilon+\ln d+\frac{d}{2}\ln\pi-\ln\Gamma\left(\frac{d}{2}+1\right) \\
 &\qquad -(d-1)\ln\epsilon+(d-1)\psi(d)\\
 & =d\ln\frac{e\sqrt{\pi}}{\varepsilon}+\ln\frac{d\Gamma(d)}{\Gamma(\frac{d}{2}+1)},
\end{align*}
where $\psi$ is the digamma function.
Therefore, the KL divergence between $P_{Z|X}(\cdot|x)$ and $Q$
(in nats) is
\begin{align*}
 & D(P_{Z|X}(\cdot|x)\Vert Q)\\
 & =-\mathbb{E}_{Z\sim P_{Z|X}(\cdot|x)}\left[\ln\left(\left(2\pi\left(\frac{C^{2}}{d}+\frac{d+1}{\varepsilon^{2}}\right)\right)^{-d/2}\exp\left(-\frac{\Vert Z\Vert_{2}^{2}}{2(\frac{C^{2}}{d}+\frac{d+1}{\varepsilon^{2}})}\right)\right)\right]-h(Z|X)\\
 & =\frac{d}{2}\ln\left(2\pi\left(\frac{C^{2}}{d}+\frac{d+1}{\varepsilon^{2}}\right)\right)+\frac{\mathbb{E}_{Z\sim P_{Z|X}(\cdot|x)}\left[\Vert Z\Vert_{2}^{2}\right]}{2(\frac{C^{2}}{d}+\frac{d+1}{\varepsilon^{2}})}-d\ln\frac{e\sqrt{\pi}}{\varepsilon}-\ln\frac{d\Gamma(d)}{\Gamma(\frac{d}{2}+1)}\\
 & \le\frac{d}{2}\ln\left(2\pi\left(\frac{C^{2}}{d}+\frac{d+1}{\varepsilon^{2}}\right)\right)+\frac{C^{2}+\frac{d(d+1)}{\varepsilon^{2}}}{2(\frac{C^{2}}{d}+\frac{d+1}{\varepsilon^{2}})}-d\ln\frac{e\sqrt{\pi}}{\varepsilon}-\ln\frac{d\Gamma(d)}{\Gamma(\frac{d}{2}+1)}\\
 & =\frac{d}{2}\ln\left(\frac{2}{e}\left(\frac{C^{2}\varepsilon^{2}}{d}+d+1\right)\right)-\ln\frac{\Gamma(d+1)}{\Gamma(\frac{d}{2}+1)}.
\end{align*}
Hence, by Theorem \ref{thm:logk_bound_simple}, the compression size is at most $\ell+\log_{2}(\ell+1)+2$
bits.
The metric privacy guarantee follows from Theorem~\ref{thm:metric_privacy}.

\section{MSE against Compression Size\label{sec::mse_size}}

We plot the MSE against the compression size (ranging from $25$ to $1000$ bits) for $\epsilon\in\{0.25, 0.5, 1.0, 2.0 \}$ in Figure~\ref{fig::mse_comp_size} as follows. 

\begin{figure}[H]
    \centering
    \includegraphics[scale = 0.45]{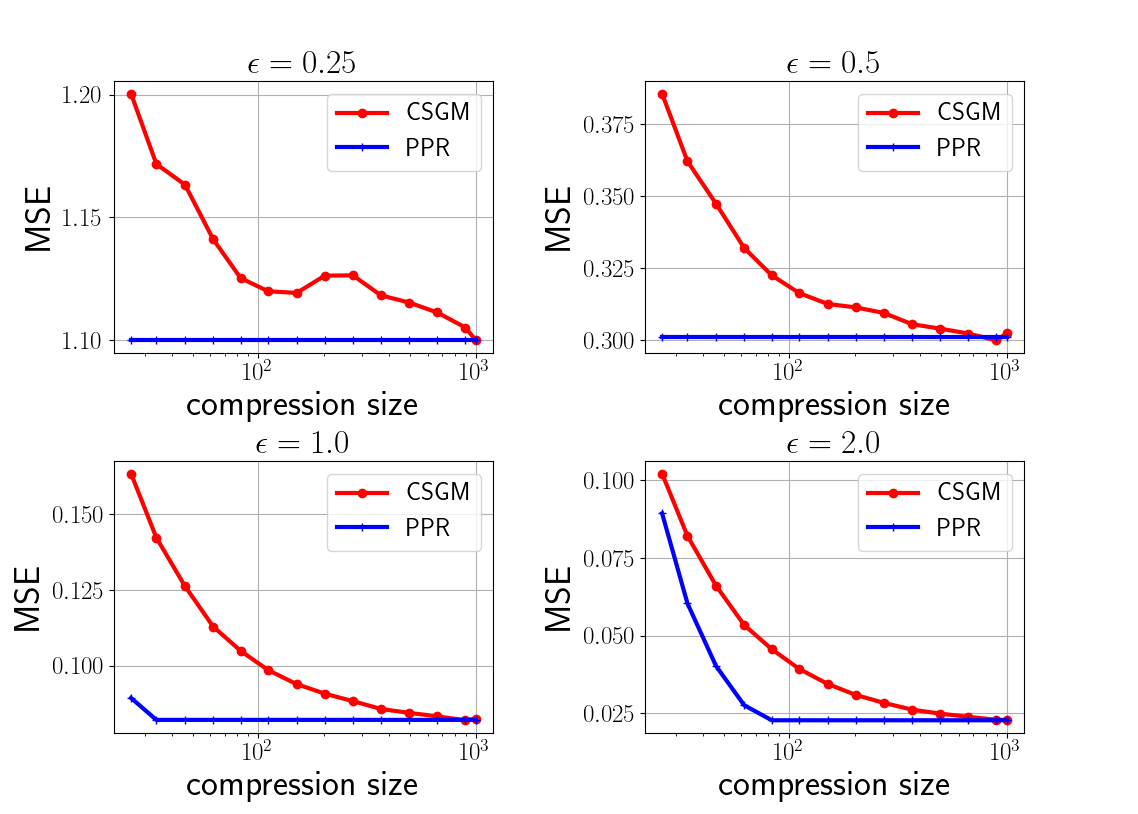}
    \caption{The MSE of PPR and CSGM against the compression size in bits, where $\varepsilon$ is chosen from $\{0.25, 0.5, 1.0, 2.0\}$ and compression sizes vary from $25$ to $1000$ bits. Note that parts of the curves for PPR are flat, because a lower compression size is already sufficient for PPR to exactly simulate the best Gaussian mechanism for that value of $\varepsilon$, so a higher compression size than necessary will not affect the result.}
    \label{fig::mse_comp_size}
\end{figure}

\bibliographystyle{plain}
\addbibtotoc
\bibliography{reference/reference.bib}

\begin{thebibliography}{100}

\bibitem{abadi2016deep}
Martin Abadi, Andy Chu, Ian Goodfellow, H~Brendan McMahan, Ilya Mironov, Kunal Talwar, and Li~Zhang.
\newblock Deep learning with differential privacy.
\newblock In {\em Proceedings of the 2016 ACM SIGSAC conference on computer and communications security}, pages 308--318, 2016.

\bibitem{acharya2019communication}
Jayadev Acharya and Ziteng Sun.
\newblock Communication complexity in locally private distribution estimation and heavy hitters.
\newblock In {\em International Conference on Machine Learning}, pages 51--60. PMLR, 2019.

\bibitem{agustsson2020universally}
Eirikur Agustsson and Lucas Theis.
\newblock Universally quantized neural compression.
\newblock {\em Advances in neural information processing systems}, 33:12367--12376, 2020.

\bibitem{ahlswede1971multi}
Rudolf Ahlswede.
\newblock Multi-way communication channels.
\newblock In {\em 2nd Int. Symp. Inform. Theory, Tsahkadsor, Armenian SSR}, pages 23--52, 1971.

\bibitem{ahlswede1974capacity}
Rudolf Ahlswede.
\newblock The capacity region of a channel with two senders and two receivers.
\newblock {\em The annals of probability}, 2(5):805--814, 1974.

\bibitem{andres2013geo}
Miguel~E Andr{\'e}s, Nicol{\'a}s~E Bordenabe, Konstantinos Chatzikokolakis, and Catuscia Palamidessi.
\newblock Geo-indistinguishability: Differential privacy for location-based systems.
\newblock In {\em Proceedings of the 2013 ACM SIGSAC conference on Computer \& communications security}, pages 901--914, 2013.

\bibitem{arikan2009channel}
Erdal Arikan.
\newblock Channel polarization: A method for constructing capacity-achieving codes for symmetric binary-input memoryless channels.
\newblock {\em IEEE Transactions on information Theory}, 55(7):3051--3073, 2009.

\bibitem{asi2022optimal}
Hilal Asi, Vitaly Feldman, and Kunal Talwar.
\newblock Optimal algorithms for mean estimation under local differential privacy.
\newblock In {\em International Conference on Machine Learning}, pages 1046--1056. PMLR, 2022.

\bibitem{barron2003duality}
Richard~J Barron, Brian Chen, and Gregory~W Wornell.
\newblock The duality between information embedding and source coding with side information and some applications.
\newblock {\em IEEE Transactions on Information Theory}, 49(5):1159--1180, 2003.

\bibitem{bassily2017practical}
Raef Bassily, Kobbi Nissim, Uri Stemmer, and Abhradeep Guha~Thakurta.
\newblock Practical locally private heavy hitters.
\newblock {\em Advances in Neural Information Processing Systems}, 30, 2017.

\bibitem{bassily2015local}
Raef Bassily and Adam Smith.
\newblock Local, private, efficient protocols for succinct histograms.
\newblock In {\em Proceedings of the forty-seventh annual ACM symposium on Theory of computing}, pages 127--135, 2015.

\bibitem{bennet2014reverse}
Charles~H Bennett, Igor Devetak, Aram~W Harrow, Peter~W Shor, and Andreas Winter.
\newblock The quantum reverse shannon theorem and resource tradeoffs for simulating quantum channels.
\newblock {\em IEEE Transactions on Information Theory}, 60(5):2926--2959, 2014.

\bibitem{bennett2002entanglement}
Charles~H Bennett, Peter~W Shor, John Smolin, and Ashish~V Thapliyal.
\newblock Entanglement-assisted capacity of a quantum channel and the reverse {S}hannon theorem.
\newblock {\em IEEE Trans. Inf. Theory}, 48(10):2637--2655, 2002.

\bibitem{berger1978multiterminal}
Toby Berger.
\newblock Multiterminal source coding.
\newblock In G.~Longo, editor, {\em The Information Theory Approach to Communications}, pages 171--231. Springer-Verlag, New York, 1978.

\bibitem{berry1941accuracy}
Andrew~C Berry.
\newblock The accuracy of the {G}aussian approximation to the sum of independent variates.
\newblock {\em Transactions of the {A}merican {M}athematical {S}ociety}, 49(1):122--136, 1941.

\bibitem{bhowmick2018protection}
Abhishek Bhowmick, John Duchi, Julien Freudiger, Gaurav Kapoor, and Ryan Rogers.
\newblock Protection against reconstruction and its applications in private federated learning.
\newblock {\em arXiv preprint arXiv:1812.00984}, 2018.

\bibitem{bjelakovic2013secrecy}
Igor Bjelakovi{\'c}, Holger Boche, and Jochen Sommerfeld.
\newblock Secrecy results for compound wiretap channels.
\newblock {\em Problems of Information Transmission}, 49(1):73--98, 2013.

\bibitem{blackwell1959capacity}
David Blackwell, Leo Breiman, AJ~Thomasian, et~al.
\newblock The capacity of a class of channels.
\newblock {\em The Annals of Mathematical Statistics}, 30(4):1229--1241, 1959.

\bibitem{boche2015continuity}
Holger Boche, Rafael~F Schaefer, and H~Vincent Poor.
\newblock On the continuity of the secrecy capacity of compound and arbitrarily varying wiretap channels.
\newblock {\em IEEE Transactions on Information Forensics and Security}, 10(12):2531--2546, 2015.

\bibitem{bommasani2021opportunities}
Rishi Bommasani, Drew~A Hudson, Ehsan Adeli, Russ Altman, Simran Arora, Sydney von Arx, Michael~S Bernstein, Jeannette Bohg, Antoine Bosselut, Emma Brunskill, et~al.
\newblock On the opportunities and risks of foundation models.
\newblock {\em arXiv preprint arXiv:2108.07258}, 2021.

\bibitem{boneh1998collusion}
Dan Boneh and James Shaw.
\newblock Collusion-secure fingerprinting for digital data.
\newblock {\em IEEE Transactions on Information Theory}, 44(5):1897--1905, 1998.

\bibitem{braverman2014public}
Mark Braverman and Ankit Garg.
\newblock Public vs private coin in bounded-round information.
\newblock In {\em International Colloquium on Automata, Languages, and Programming}, pages 502--513. Springer, 2014.

\bibitem{braverman2016communication}
Mark Braverman, Ankit Garg, Tengyu Ma, Huy~L Nguyen, and David~P Woodruff.
\newblock Communication lower bounds for statistical estimation problems via a distributed data processing inequality.
\newblock In {\em Proceedings of the forty-eighth annual ACM symposium on Theory of Computing}, pages 1011--1020, 2016.

\bibitem{bun18hh}
Mark Bun, Jelani Nelson, and Uri Stemmer.
\newblock Heavy hitters and the structure of local privacy.
\newblock In {\em Proceedings of the 37th ACM SIGMOD-SIGACT-SIGAI Symposium on Principles of Database Systems}, SIGMOD/PODS ’18, page 435–447, New York, NY, USA, 2018. Association for Computing Machinery.

\bibitem{bun2019heavy}
Mark Bun, Jelani Nelson, and Uri Stemmer.
\newblock Heavy hitters and the structure of local privacy.
\newblock {\em ACM Transactions on Algorithms (TALG)}, 15(4):1--40, 2019.

\bibitem{canonne2020discrete}
Cl{\'e}ment~L Canonne, Gautam Kamath, and Thomas Steinke.
\newblock The discrete {G}aussian for differential privacy.
\newblock {\em Advances in Neural Information Processing Systems}, 33:15676--15688, 2020.

\bibitem{carlini2021extracting}
Nicholas Carlini, Florian Tramer, Eric Wallace, Matthew Jagielski, Ariel Herbert-Voss, Katherine Lee, Adam Roberts, Tom Brown, Dawn Song, Ulfar Erlingsson, et~al.
\newblock Extracting training data from large language models.
\newblock In {\em 30th USENIX Security Symposium (USENIX Security 21)}, pages 2633--2650, 2021.

\bibitem{carlini2023extracting}
Nicolas Carlini, Jamie Hayes, Milad Nasr, Matthew Jagielski, Vikash Sehwag, Florian Tramer, Borja Balle, Daphne Ippolito, and Eric Wallace.
\newblock Extracting training data from diffusion models.
\newblock In {\em 32nd USENIX Security Symposium (USENIX Security 23)}, pages 5253--5270, 2023.

\bibitem{chatzikokolakis2013broadening}
Konstantinos Chatzikokolakis, Miguel~E Andr{\'e}s, Nicol{\'a}s~Emilio Bordenabe, and Catuscia Palamidessi.
\newblock Broadening the scope of differential privacy using metrics.
\newblock In {\em Privacy Enhancing Technologies: 13th International Symposium, PETS 2013, Bloomington, IN, USA, July 10-12, 2013. Proceedings 13}, pages 82--102. Springer, 2013.

\bibitem{chaudhuri2022privacy}
Kamalika Chaudhuri, Chuan Guo, and Mike Rabbat.
\newblock Privacy-aware compression for federated data analysis.
\newblock In {\em Uncertainty in Artificial Intelligence}, pages 296--306. PMLR, 2022.

\bibitem{chen2001quantization}
Brian Chen and Gregory~W Wornell.
\newblock Quantization index modulation: A class of provably good methods for digital watermarking and information embedding.
\newblock {\em IEEE Transactions on Information theory}, 47(4):1423--1443, 2001.

\bibitem{chen2020breaking}
Wei-Ning Chen, Peter Kairouz, and Ayfer {\"O}zg{\"u}r.
\newblock Breaking the communication-privacy-accuracy trilemma.
\newblock {\em Advances in Neural Information Processing Systems}, 33:3312--3324, 2020.

\bibitem{chen2022breaking}
Wei-Ning Chen, Peter Kairouz, and Ayfer {\"O}zg{\"u}r.
\newblock Breaking the communication-privacy-accuracy trilemma.
\newblock {\em IEEE Transactions on Information Theory}, 69(2):1261--1281, 2022.

\bibitem{chen2024privacy}
Wei-Ning Chen, Dan Song, Ayfer {\"O}zg{\"u}r, and Peter Kairouz.
\newblock Privacy amplification via compression: Achieving the optimal privacy-accuracy-communication trade-off in distributed mean estimation.
\newblock {\em Advances in Neural Information Processing Systems}, 36, 2024.

\bibitem{cohen2002gaussian}
Aaron~S Cohen and Amos Lapidoth.
\newblock The gaussian watermarking game.
\newblock {\em IEEE transactions on Information Theory}, 48(6):1639--1667, 2002.

\bibitem{cover1979capacity}
Thomas Cover and Abbas~El Gamal.
\newblock Capacity theorems for the relay channel.
\newblock {\em IEEE Transactions on information theory}, 25(5):572--584, 1979.

\bibitem{cox1997secure}
Ingemar~J Cox, Joe Kilian, F~Thomson Leighton, and Talal Shamoon.
\newblock Secure spread spectrum watermarking for multimedia.
\newblock {\em IEEE transactions on image processing}, 6(12):1673--1687, 1997.

\bibitem{csiszar1998method}
Imre Csisz{\'a}r.
\newblock The method of types [information theory].
\newblock {\em IEEE Transactions on Information Theory}, 44(6):2505--2523, 1998.

\bibitem{cuff2008communication}
Paul Cuff.
\newblock Communication requirements for generating correlated random variables.
\newblock In {\em 2008 IEEE International Symposium on Information Theory}, pages 1393--1397. IEEE, 2008.

\bibitem{cuff2013distributed}
Paul Cuff.
\newblock Distributed channel synthesis.
\newblock {\em IEEE Trans. Inf. Theory}, 59(11):7071--7096, Nov 2013.

\bibitem{cuff2010coordination}
Paul Cuff, Haim Permuter, and Thomas~M Cover.
\newblock Coordination capacity.
\newblock {\em IEEE Trans. Inf. Theory}, 56(9):4181--4206, Sept 2010.

\bibitem{cuff2009cascade}
Paul Cuff, Han-I Su, and Abbas El~Gamal.
\newblock Cascade multiterminal source coding.
\newblock In {\em 2009 IEEE International Symposium on Information Theory}, pages 1199--1203. IEEE, 2009.

\bibitem{cuff2016differential}
Paul Cuff and Lanqing Yu.
\newblock Differential privacy as a mutual information constraint.
\newblock In {\em Proceedings of the 2016 ACM SIGSAC Conference on Computer and Communications Security}, pages 43--54, 2016.

\bibitem{cuff2009communication}
Paul~W Cuff.
\newblock {\em Communication in networks for coordinating behavior}.
\newblock Stanford University, 2009.

\bibitem{dobrushin1959optimum}
RL~Dobrushin.
\newblock Optimum information transmission through a channel with unknown parameters.
\newblock {\em Radio Eng. Electron}, 4(12):1--8, 1959.

\bibitem{duchi2019lower}
John Duchi and Ryan Rogers.
\newblock Lower bounds for locally private estimation via communication complexity.
\newblock In {\em Conference on Learning Theory}, pages 1161--1191. PMLR, 2019.

\bibitem{duchi2013local}
John~C Duchi, Michael~I Jordan, and Martin~J Wainwright.
\newblock Local privacy and statistical minimax rates.
\newblock In {\em 2013 IEEE 54th Annual Symposium on Foundations of Computer Science}, pages 429--438. IEEE, 2013.

\bibitem{durisi2016toward}
Giuseppe Durisi, Tobias Koch, and Petar Popovski.
\newblock Toward massive, ultrareliable, and low-latency wireless communication with short packets.
\newblock {\em Proceedings of the IEEE}, 104(9):1711--1726, 2016.

\bibitem{dwork2006differential}
Cynthia Dwork.
\newblock Differential privacy.
\newblock In {\em International colloquium on automata, languages, and programming}, pages 1--12. Springer, 2006.

\bibitem{dwork2006our}
Cynthia Dwork, Krishnaram Kenthapadi, Frank McSherry, Ilya Mironov, and Moni Naor.
\newblock Our data, ourselves: Privacy via distributed noise generation.
\newblock In {\em Advances in Cryptology-EUROCRYPT 2006: 24th Annual International Conference on the Theory and Applications of Cryptographic Techniques, St. Petersburg, Russia, May 28-June 1, 2006. Proceedings 25}, pages 486--503. Springer, 2006.

\bibitem{dwork2006calibrating}
Cynthia Dwork, Frank McSherry, Kobbi Nissim, and Adam Smith.
\newblock Calibrating noise to sensitivity in private data analysis.
\newblock In {\em Theory of cryptography conference}, pages 265--284. Springer, 2006.

\bibitem{dwork2014algorithmic}
Cynthia Dwork and Aaron Roth.
\newblock The algorithmic foundations of differential privacy.
\newblock {\em Foundations and Trends in Theoretical Computer Science}, 9(3--4):211--407, 2014.

\bibitem{ekrem2010gaussian}
Ersen Ekrem and Sennur Ulukus.
\newblock On gaussian mimo compound wiretap channels.
\newblock In {\em 2010 44th Annual Conference on Information Sciences and Systems (CISS)}, pages 1--6. IEEE, 2010.

\bibitem{ekrem2012degraded}
Ersen Ekrem and Sennur Ulukus.
\newblock Degraded compound multi-receiver wiretap channels.
\newblock {\em IEEE Transactions on Information Theory}, 58(9):5681--5698, 2012.

\bibitem{el2021achievable}
Abbas El~Gamal, Amin Gohari, and Chandra Nair.
\newblock Achievable rates for the relay channel with orthogonal receiver components.
\newblock In {\em 2021 IEEE Information Theory Workshop (ITW)}, pages 1--6. IEEE, 2021.

\bibitem{el2022strengthened}
Abbas El~Gamal, Amin Gohari, and Chandra Nair.
\newblock A strengthened cutset upper bound on the capacity of the relay channel and applications.
\newblock {\em IEEE Transactions on Information Theory}, 68(8):5013--5043, 2022.

\bibitem{el2005relay}
Abbas El~Gamal and Navid Hassanpour.
\newblock Relay-without-delay.
\newblock In {\em Proceedings. International Symposium on Information Theory, 2005. ISIT 2005.}, pages 1078--1080. IEEE, 2005.

\bibitem{el2007relay}
Abbas El~Gamal, Navid Hassanpour, and James Mammen.
\newblock Relay networks with delays.
\newblock {\em IEEE Transactions on Information Theory}, 53(10):3413--3431, 2007.

\bibitem{el2011network}
Abbas El~Gamal and Young-Han Kim.
\newblock {\em Network information theory}.
\newblock Cambridge university press, 2011.

\bibitem{elias1975universal}
Peter Elias.
\newblock Universal codeword sets and representations of the integers.
\newblock {\em IEEE transactions on information theory}, 21(2):194--203, 1975.

\bibitem{erlingsson2019amplification}
{\'U}lfar Erlingsson, Vitaly Feldman, Ilya Mironov, Ananth Raghunathan, Kunal Talwar, and Abhradeep Thakurta.
\newblock Amplification by shuffling: From local to central differential privacy via anonymity.
\newblock In {\em Proceedings of the Thirtieth Annual ACM-SIAM Symposium on Discrete Algorithms}, pages 2468--2479. SIAM, 2019.

\bibitem{esseen1942liapunov}
Carl-Gustaf Esseen.
\newblock On the {L}iapunov limit error in the theory of probability.
\newblock {\em Ark. Mat. Astr. Fys.}, 28:1--19, 1942.

\bibitem{feinstein1954new}
Amiel Feinstein.
\newblock A new basic theorem of information theory.
\newblock {\em IRE Trans. Inf. Theory}, (4):2--22, 1954.

\bibitem{feldman2022hiding}
Vitaly Feldman, Audra McMillan, and Kunal Talwar.
\newblock Hiding among the clones: A simple and nearly optimal analysis of privacy amplification by shuffling.
\newblock In {\em 2021 IEEE 62nd Annual Symposium on Foundations of Computer Science (FOCS)}, pages 954--964. IEEE, 2022.

\bibitem{feldman2021lossless}
Vitaly Feldman and Kunal Talwar.
\newblock Lossless compression of efficient private local randomizers.
\newblock In {\em International Conference on Machine Learning}, pages 3208--3219. PMLR, 2021.

\bibitem{fernandes2019generalised}
Natasha Fernandes, Mark Dras, and Annabelle McIver.
\newblock Generalised differential privacy for text document processing.
\newblock In {\em Principles of Security and Trust: 8th International Conference, POST 2019}, pages 123--148. Springer International Publishing, 2019.

\bibitem{feyisetan2020privacy}
Oluwaseyi Feyisetan, Borja Balle, Thomas Drake, and Tom Diethe.
\newblock Privacy-and utility-preserving textual analysis via calibrated multivariate perturbations.
\newblock In {\em Proceedings of the 13th international conference on web search and data mining}, pages 178--186, 2020.

\bibitem{flamich2024greedy}
Gergely Flamich.
\newblock Greedy poisson rejection sampling.
\newblock {\em Advances in Neural Information Processing Systems}, 36, 2024.

\bibitem{flamich2020compressing}
Gergely Flamich, Marton Havasi, and Jos{\'{e}}~Miguel Hern{\'{a}}ndez-Lobato.
\newblock Compressing images by encoding their latent representations with relative entropy coding.
\newblock {\em Advances in Neural Information Processing Systems}, 33:16131--16141, 2020.

\bibitem{flamich2022fast}
Gergely Flamich, Stratis Markou, and Jos{\'e}~Miguel Hern{\'a}ndez-Lobato.
\newblock Fast relative entropy coding with {A*} coding.
\newblock In {\em International Conference on Machine Learning}, pages 6548--6577. PMLR, 2022.

\bibitem{flamich2023faster}
Gergely Flamich, Stratis Markou, and Jos{\'{e}} Miguel~Hern{\'{a}}ndez Lobato.
\newblock Faster relative entropy coding with greedy rejection coding.
\newblock {\em arXiv preprint arXiv:2309.15746}, 2023.

\bibitem{flamich2023adaptive}
Gergely Flamich and Lucas Theis.
\newblock Adaptive greedy rejection sampling.
\newblock In {\em 2023 IEEE International Symposium on Information Theory (ISIT)}, pages 454--459. IEEE, 2023.

\bibitem{garg2014communication}
Ankit Garg, Tengyu Ma, and Huy Nguyen.
\newblock On communication cost of distributed statistical estimation and dimensionality.
\newblock In {\em Advances in Neural Information Processing Systems}, pages 2726--2734, 2014.

\bibitem{gelfand1980coding}
S.~I. Gel'fand and M.~S. Pinsker.
\newblock Coding for channel with random parameters.
\newblock {\em Probl. Contr. and Inf. Theory}, 9(1):19--31, 1980.

\bibitem{ghadimi2013stochastic}
Saeed Ghadimi and Guanghui Lan.
\newblock Stochastic first-and zeroth-order methods for nonconvex stochastic programming.
\newblock {\em SIAM journal on optimization}, 23(4):2341--2368, 2013.

\bibitem{goc2024channel}
Daniel Goc and Gergely Flamich.
\newblock On channel simulation with causal rejection samplers.
\newblock {\em arXiv preprint arXiv:2401.16579}, 2024.

\bibitem{goryczka2015comprehensive}
Slawomir Goryczka and Li~Xiong.
\newblock A comprehensive comparison of multiparty secure additions with differential privacy.
\newblock {\em IEEE transactions on dependable and secure computing}, 14(5):463--477, 2015.

\bibitem{grover2015information}
Pulkit Grover, Aaron~B Wagner, and Anant Sahai.
\newblock Information embedding and the triple role of control.
\newblock {\em IEEE Transactions on Information Theory}, 61(4):1539--1549, 2015.

\bibitem{guan2022double}
Qingxiao Guan, Peng Liu, Weiming Zhang, Wei Lu, and Xinpeng Zhang.
\newblock Double-layered dual-syndrome trellis codes utilizing channel knowledge for robust steganography.
\newblock {\em IEEE Transactions on Information Forensics and Security}, 18:501--516, 2022.

\bibitem{guo2023privacy}
Chuan Guo, Kamalika Chaudhuri, Pierre Stock, and Michael Rabbat.
\newblock Privacy-aware compression for federated learning through numerical mechanism design.
\newblock In {\em International Conference on Machine Learning}, pages 11888--11904. PMLR, 2023.

\bibitem{guo2024hypothesis}
Yuanxin Guo, Sadaf Salehkalaibar, Stark~C Draper, and Wei Yu.
\newblock One-shot achievability region for hypothesis testing with communication constraint.
\newblock In {\em 2024 IEEE Information Theory Workshop (ITW)}, pages 55--60. IEEE, 2024.

\bibitem{harsha2007communication}
Prahladh Harsha, Rahul Jain, David McAllester, and Jaikumar Radhakrishnan.
\newblock The communication complexity of correlation.
\newblock In {\em Twenty-Second Annual IEEE Conference on Computational Complexity (CCC'07)}, pages 10--23. IEEE, 2007.

\bibitem{harsha2010communication}
Prahladh Harsha, Rahul Jain, David McAllester, and Jaikumar Radhakrishnan.
\newblock The communication complexity of correlation.
\newblock {\em IEEE Trans. Inf. Theory}, 56(1):438--449, Jan 2010.

\bibitem{hartung1999multimedia}
Frank Hartung and Martin Kutter.
\newblock Multimedia watermarking techniques.
\newblock {\em Proceedings of the IEEE}, 87(7):1079--1107, 1999.

\bibitem{hasircioglu2023communication}
Burak Has{\i}rc{\i}o{\u{g}}lu and Deniz G{\"u}nd{\"u}z.
\newblock Communication efficient private federated learning using dithering.
\newblock In {\em ICASSP 2024-2024 IEEE International Conference on Acoustics, Speech and Signal Processing (ICASSP)}, pages 7575--7579. IEEE, 2024.

\bibitem{havasi2019minimal}
Marton Havasi, Robert Peharz, and Jos{\'{e}}~Miguel Hern{\'{a}}ndez-Lobato.
\newblock Minimal random code learning: Getting bits back from compressed model parameters.
\newblock In {\em 7th International Conference on Learning Representations, ICLR 2019}, 2019.

\bibitem{hayashi2006general}
Masahito Hayashi.
\newblock General nonasymptotic and asymptotic formulas in channel resolvability and identification capacity and their application to the wiretap channel.
\newblock {\em IEEE Transactions on Information Theory}, 52(4):1562--1575, 2006.

\bibitem{hayashi2009information}
Masahito Hayashi.
\newblock Information spectrum approach to second-order coding rate in channel coding.
\newblock {\em IEEE Transactions on Information Theory}, 55(11):4947--4966, 2009.

\bibitem{Heegard1980}
C.~Heegard and A.~El~Gamal.
\newblock On the capacity of computer memory with defects.
\newblock {\em IEEE Transactions on Information Theory}, 29(5):731--739, 1983.

\bibitem{heegard1983capacity}
Chris Heegard and A~El~Gamal.
\newblock On the capacity of computer memory with defects.
\newblock {\em IEEE transactions on Information theory}, 29(5):731--739, 1983.

\bibitem{hegazy2024compression}
Mahmoud Hegazy, R\'{e}mi Leluc, Cheuk~Ting Li, and Aymeric Dieuleveut.
\newblock Compression with exact error distribution for federated learning.
\newblock In {\em Proceedings of The 27th International Conference on Artificial Intelligence and Statistics}, volume 238 of {\em Proceedings of Machine Learning Research}, pages 613--621. PMLR, 02--04 May 2024.

\bibitem{hegazy2022randomized}
Mahmoud Hegazy and Cheuk~Ting Li.
\newblock Randomized quantization with exact error distribution.
\newblock In {\em 2022 IEEE Information Theory Workshop (ITW)}, pages 350--355. IEEE, 2022.

\bibitem{hentila2024communication}
Henri Hentil{\"a}, Yanina~Y Shkel, and Visa Koivunen.
\newblock Communication-constrained secret key generation: Second-order bounds.
\newblock {\em IEEE Transactions on Information Theory}, 2024.

\bibitem{isik2024exact}
Berivan Isik, Wei-Ning Chen, Ayfer {\"O}zg{\"u}r, Tsachy Weissman, and Albert No.
\newblock Exact optimality of communication-privacy-utility tradeoffs in distributed mean estimation.
\newblock {\em Advances in Neural Information Processing Systems}, 36, 2024.

\bibitem{ji2018ultra}
Hyoungju Ji, Sunho Park, Jeongho Yeo, Younsun Kim, Juho Lee, and Byonghyo Shim.
\newblock Ultra-reliable and low-latency communications in 5g downlink: Physical layer aspects.
\newblock {\em IEEE Wireless Communications}, 25(3):124--130, 2018.

\bibitem{kairouz2021advances}
Peter Kairouz, H~Brendan McMahan, Brendan Avent, Aur{\'e}lien Bellet, Mehdi Bennis, Arjun~Nitin Bhagoji, Kallista Bonawitz, Zachary Charles, Graham Cormode, Rachel Cummings, et~al.
\newblock Advances and open problems in federated learning.
\newblock {\em Foundations and trends{\textregistered} in machine learning}, 14(1--2):1--210, 2021.

\bibitem{kalker2002capacity}
TON Kalker and Frans~MJ Willems.
\newblock Capacity bounds and constructions for reversible data-hiding.
\newblock In {\em 2002 14th International Conference on Digital Signal Processing Proceedings. DSP 2002 (Cat. No. 02TH8628)}, volume~1, pages 71--76. IEEE, 2002.

\bibitem{kasiviswanathan2011can}
Shiva~Prasad Kasiviswanathan, Homin~K Lee, Kobbi Nissim, Sofya Raskhodnikova, and Adam Smith.
\newblock What can we learn privately?
\newblock {\em SIAM Journal on Computing}, 40(3):793--826, 2011.

\bibitem{keshet2008channel}
Guy Keshet, Yossef Steinberg, Neri Merhav, et~al.
\newblock Channel coding in the presence of side information.
\newblock {\em Foundations and Trends{\textregistered} in Communications and Information Theory}, 4(6):445--586, 2008.

\bibitem{khisti2024unequal}
Ashish Khisti, Arash Behboodi, Gabriele Cesa, and Pratik Kumar.
\newblock Unequal message protection: One-shot analysis via poisson matching lemma.
\newblock In {\em 2024 IEEE International Symposium on Information Theory (ISIT)}. IEEE, 2024.

\bibitem{kim2007coding}
Young-Han Kim.
\newblock Coding techniques for primitive relay channels.
\newblock In {\em Proc. Forty-Fifth Annual Allerton Conf. Commun., Contr. Comput}, page 2007, 2007.

\bibitem{kim2008state}
Young-Han Kim, Arak Sutivong, and Thomas~M Cover.
\newblock State amplification.
\newblock {\em IEEE Transactions on Information Theory}, 54(5):1850--1859, 2008.

\bibitem{kobayashi2009compound}
Mari Kobayashi, Yingbin Liang, Shlomo Shamai, and M{\'e}rouane Debbah.
\newblock On the compound mimo broadcast channels with confidential messages.
\newblock In {\em 2009 IEEE International Symposium on Information Theory}, pages 1283--1287. IEEE, 2009.

\bibitem{kobus2024gaussian}
Szymon Kobus, Lucas Theis, and Deniz G{\"u}nd{\"u}z.
\newblock Gaussian channel simulation with rotated dithered quantization.
\newblock In {\em 2024 IEEE International Symposium on Information Theory (ISIT)}, pages 1907--1912. IEEE, 2024.

\bibitem{konevcny2016federated}
Jakub Kone{\v{c}}n{\`y}, H~Brendan McMahan, Felix~X Yu, Peter Richt{\'a}rik, Ananda~Theertha Suresh, and Dave Bacon.
\newblock Federated learning: Strategies for improving communication efficiency.
\newblock {\em arXiv preprint arXiv:1610.05492}, 2016.

\bibitem{kostina2013lossy}
Victoria Kostina and Sergio Verd{\'u}.
\newblock Lossy joint source-channel coding in the finite blocklength regime.
\newblock {\em IEEE Transactions on Information Theory}, 59(5):2545--2575, 2013.

\bibitem{kotz2012laplace}
Samuel Kotz, Tomasz Kozubowski, and Krzystof Podgorski.
\newblock {\em The Laplace distribution and generalizations: a revisit with applications to communications, economics, engineering, and finance}.
\newblock Springer Science \& Business Media, 2012.

\bibitem{lang2023joint}
Natalie Lang, Elad Sofer, Tomer Shaked, and Nir Shlezinger.
\newblock Joint privacy enhancement and quantization in federated learning.
\newblock {\em IEEE Transactions on Signal Processing}, 71:295--310, 2023.

\bibitem{last2017lectures}
G{\"u}nter Last and Mathew Penrose.
\newblock {\em Lectures on the {P}oisson process}, volume~7.
\newblock Cambridge University Press, 2017.

\bibitem{lee2018unified}
Si-Hyeon Lee and Sae-Young Chung.
\newblock A unified random coding bound.
\newblock {\em IEEE Transactions on Information Theory}, 64(10):6779--6802, 2018.

\bibitem{lei2022neural}
Eric Lei, Hamed Hassani, and Shirin~Saeedi Bidokhti.
\newblock Neural estimation of the rate-distortion function with applications to operational source coding.
\newblock {\em IEEE Journal on Selected Areas in Information Theory}, 3(4):674--686, 2022.

\bibitem{li2023automated}
Cheuk~Ting Li.
\newblock An automated theorem proving framework for information-theoretic results.
\newblock {\em IEEE Transactions on Information Theory}, 69(11):6857--6877, 2023.

\bibitem{li2024pointwise}
Cheuk~Ting Li.
\newblock Pointwise redundancy in one-shot lossy compression via {P}oisson functional representation.
\newblock In {\em International Zurich Seminar on Information and Communication (IZS 2024)}, 2024.

\bibitem{li2024lossy_arxiv}
Cheuk~Ting Li.
\newblock Pointwise redundancy in one-shot lossy compression via {P}oisson functional representation.
\newblock {\em arXiv preprint}, 2024.

\bibitem{li2025discrete}
Cheuk~Ting Li.
\newblock Discrete layered entropy, conditional compression and a tighter strong functional representation lemma.
\newblock {\em arXiv preprint arXiv:2501.13736}, 2025.

\bibitem{li2019pairwise}
Cheuk~Ting Li and Venkat Anantharam.
\newblock Pairwise multi-marginal optimal transport and embedding for earth mover's distance.
\newblock {\em arXiv preprint arXiv:1908.01388}, 2019.

\bibitem{li2021unified}
Cheuk~Ting Li and Venkat Anantharam.
\newblock A unified framework for one-shot achievability via the {P}oisson matching lemma.
\newblock {\em IEEE Transactions on Information Theory}, 67(5):2624--2651, 2021.

\bibitem{li2018strong}
Cheuk~Ting Li and Abbas El~Gamal.
\newblock Strong functional representation lemma and applications to coding theorems.
\newblock {\em IEEE Transactions on Information Theory}, 64(11):6967--6978, 2018.

\bibitem{li2024channel}
Cheuk~Ting Li et~al.
\newblock Channel simulation: Theory and applications to lossy compression and differential privacy.
\newblock {\em Foundations and Trends{\textregistered} in Communications and Information Theory}, 21(6):847--1106, 2024.

\bibitem{li2018minimax}
Cheuk~Ting Li, Xiugang Wu, Ayfer {\"O}zg{\"u}r, and Abbas {El Gamal}.
\newblock Minimax learning for remote prediction.
\newblock In {\em 2018 IEEE ISIT}, pages 541--545, June 2018.

\bibitem{li2020minimax}
Cheuk~Ting Li, Xiugang Wu, Ayfer Ozgur, and Abbas El~Gamal.
\newblock Minimax learning for distributed inference.
\newblock {\em IEEE Transactions on Information Theory}, 66(12):7929--7938, 2020.

\bibitem{li2020designing}
Weixiang Li, Weiming Zhang, Li~Li, Hang Zhou, and Nenghai Yu.
\newblock Designing near-optimal steganographic codes in practice based on polar codes.
\newblock {\em IEEE Transactions on Communications}, 68(7):3948--3962, 2020.

\bibitem{liang2009compound}
Yingbin Liang, Gerhard Kramer, and H~Vincent Poor.
\newblock Compound wiretap channels.
\newblock {\em EURASIP Journal on Wireless Communications and Networking}, 2009:1--12, 2009.

\bibitem{liang2008multiple}
Yingbin Liang and H~Vincent Poor.
\newblock Multiple-access channels with confidential messages.
\newblock {\em IEEE Transactions on Information Theory}, 54(3):976--1002, 2008.

\bibitem{liao1972multiple}
H~Liao.
\newblock {\em Multiple Access Channels}.
\newblock PhD thesis, Department of Electrical Engineering, University of Hawaii, Honolulu, HI, 1972.

\bibitem{lim2011noisy}
Sung~Hoon Lim, Young-Han Kim, Abbas El~Gamal, and Sae-Young Chung.
\newblock Noisy network coding.
\newblock {\em IEEE Transactions on Information Theory}, 57(5):3132--3152, 2011.

\bibitem{ling2022weighted}
Chih~Wei Ling, Yanxiao Liu, and Cheuk~Ting Li.
\newblock Weighted parity-check codes for channels with state and asymmetric channels.
\newblock In {\em 2022 IEEE International Symposium on Information Theory (ISIT)}, pages 3103--3108. IEEE, 2022.

\bibitem{ling2023weighted_tit}
Chih~Wei Ling, Yanxiao Liu, and Cheuk~Ting Li.
\newblock Weighted parity-check codes for channels with state and asymmetric channels.
\newblock {\em IEEE Transactions on Information Theory}, 2024.

\bibitem{liu2015one}
Jingbo Liu, Paul Cuff, and Sergio Verd{\'u}.
\newblock One-shot mutual covering lemma and {M}arton's inner bound with a common message.
\newblock In {\em 2015 IEEE International Symposium on Information Theory (ISIT)}, pages 1457--1461. IEEE, 2015.

\bibitem{liu2016e_}
Jingbo Liu, Paul Cuff, and Sergio Verd{\'u}.
\newblock $e_\gamma$-resolvability.
\newblock {\em IEEE Transactions on Information Theory}, 63(5):2629--2658, 2016.

\bibitem{liu2008secrecy}
Tie Liu, Vinod Prabhakaran, and Sriram Vishwanath.
\newblock The secrecy capacity of a class of parallel gaussian compound wiretap channels.
\newblock In {\em 2008 IEEE International Symposium on Information Theory}, pages 116--120. IEEE, 2008.

\bibitem{liu2024universal}
Yanxiao Liu, Wei-Ning Chen, Ayfer {\"O}zg{\"u}r, and Cheuk~Ting Li.
\newblock Universal exact compression of differentially private mechanisms.
\newblock {\em Advances in Neural Information Processing Systems}, 37:91492--91531, 2024.

\bibitem{liu2024one}
Yanxiao Liu and Cheuk~Ting Li.
\newblock One-shot coding over general noisy networks.
\newblock In {\em 2024 IEEE International Symposium on Information Theory (ISIT)}, pages 3124--3129. IEEE, 2024.
\newblock (c) 2024 IEEE. Reprinted, with permission.

\bibitem{liu2024hiding}
Yanxiao Liu and Cheuk~Ting Li.
\newblock One-shot information hiding.
\newblock In {\em 2024 IEEE Information Theory Workshop (ITW)}, pages 169--174. IEEE, 2024.
\newblock (c) 2024 IEEE. Reprinted, with permission.

\bibitem{maddison2016poisson}
Chris~J Maddison.
\newblock A {P}oisson process model for {M}onte {C}arlo.
\newblock {\em Perturbation, Optimization, and Statistics}, pages 193--232, 2016.

\bibitem{maddison2014sampling}
Chris~J Maddison, Daniel Tarlow, and Tom Minka.
\newblock A* sampling.
\newblock {\em Advances in neural information processing systems}, 27, 2014.

\bibitem{malik2024distributionally}
Vikrant Malik, Taylan Kargin, Victoria Kostina, and Babak Hassibi.
\newblock A distributionally robust approach to shannon limits using the wasserstein distance.
\newblock In {\em 2024 IEEE International Symposium on Information Theory (ISIT)}, pages 861--866. IEEE, 2024.

\bibitem{marton1979coding}
Katalin Marton.
\newblock A coding theorem for the discrete memoryless broadcast channel.
\newblock {\em IEEE Transactions on Information Theory}, 25(3):306--311, 1979.

\bibitem{mcmahan2017communication}
Brendan McMahan, Eider Moore, Daniel Ramage, Seth Hampson, and Blaise~Aguera y~Arcas.
\newblock Communication-efficient learning of deep networks from decentralized data.
\newblock In {\em Artificial intelligence and statistics}, pages 1273--1282. PMLR, 2017.

\bibitem{merhav2000random}
Neri Merhav.
\newblock On random coding error exponents of watermarking systems.
\newblock {\em IEEE Transactions on Information Theory}, 46(2):420--430, 2000.

\bibitem{mironov2017renyi}
Ilya Mironov.
\newblock R{\'e}nyi differential privacy.
\newblock In {\em 2017 IEEE 30th computer security foundations symposium (CSF)}, pages 263--275. IEEE, 2017.

\bibitem{mondelli2019new}
Marco Mondelli, S~Hamed Hassani, and R{\"u}diger Urbanke.
\newblock A new coding paradigm for the primitive relay channel.
\newblock {\em Algorithms}, 12(10):218, 2019.

\bibitem{moulin2008universal}
Pierre Moulin.
\newblock Universal fingerprinting: Capacity and random-coding exponents.
\newblock In {\em 2008 IEEE International Symposium on Information Theory}, pages 220--224. IEEE, 2008.

\bibitem{moulin2003information}
Pierre Moulin and Joseph~A O'Sullivan.
\newblock Information-theoretic analysis of information hiding.
\newblock {\em IEEE Transactions on information theory}, 49(3):563--593, 2003.

\bibitem{moulin2004new}
Pierre Moulin and Ying Wang.
\newblock New results on steganographic capacity.
\newblock In {\em Proc. CISS Conference}. Citeseer, 2004.

\bibitem{moulin2007capacity}
Pierre Moulin and Ying Wang.
\newblock Capacity and random-coding exponents for channel coding with side information.
\newblock {\em IEEE Transactions on Information Theory}, 53(4):1326--1347, 2007.

\bibitem{phan2024importance}
Buu Phan, Ashish Khisti, and Christos Louizos.
\newblock Importance matching lemma for lossy compression with side information.
\newblock In {\em International Conference on Artificial Intelligence and Statistics}, pages 1387--1395. PMLR, 2024.

\bibitem{polyanskiy2013dispersion}
Yury Polyanskiy.
\newblock On dispersion of compound dmcs.
\newblock In {\em 2013 51st Annual Allerton Conference on Communication, Control, and Computing (Allerton)}, pages 26--32. IEEE, 2013.

\bibitem{polyanskiy2010channel}
Yury Polyanskiy, H~Vincent Poor, and Sergio Verd{\'u}.
\newblock Channel coding rate in the finite blocklength regime.
\newblock {\em IEEE Trans. Inf. Theory}, 56(5):2307--2359, 2010.

\bibitem{salmon2011parallel}
John~K Salmon, Mark~A Moraes, Ron~O Dror, and David~E Shaw.
\newblock Parallel random numbers: as easy as 1, 2, 3.
\newblock In {\em Proceedings of 2011 international conference for high performance computing, networking, storage and analysis}, pages 1--12, 2011.

\bibitem{scarlett2015dispersions}
Jonathan Scarlett.
\newblock On the dispersions of the {G}el’fand--{P}insker channel and dirty paper coding.
\newblock {\em IEEE Transactions on Information Theory}, 61(9):4569--4586, 2015.

\bibitem{schaefer2015secrecy}
Rafael~F. Schaefer and Sergey Loyka.
\newblock The secrecy capacity of compound gaussian mimo wiretap channels.
\newblock {\em IEEE Transactions on Information Theory}, 61(10):5535--5552, 2015.

\bibitem{shah2022optimal}
Abhin Shah, Wei-Ning Chen, Johannes Balle, Peter Kairouz, and Lucas Theis.
\newblock Optimal compression of locally differentially private mechanisms.
\newblock In {\em International Conference on Artificial Intelligence and Statistics}, pages 7680--7723. PMLR, 2022.

\bibitem{shahmiri2024communication}
Ali~Moradi Shahmiri, Chih~Wei Ling, and Cheuk~Ting Li.
\newblock Communication-efficient laplace mechanism for differential privacy via random quantization.
\newblock In {\em ICASSP 2024-2024 IEEE International Conference on Acoustics, Speech and Signal Processing (ICASSP)}, pages 4550--4554. IEEE, 2024.

\bibitem{shannon1948mathematical}
Claude~E Shannon.
\newblock A mathematical theory of communication.
\newblock {\em Bell system technical journal}, 27(3):379--423, 1948.

\bibitem{shannon1957certain}
Claude~E Shannon.
\newblock Certain results in coding theory for noisy channels.
\newblock {\em Information and control}, 1(1):6--25, 1957.

\bibitem{somekh2003error}
Anelia Somekh-Baruch and Neri Merhav.
\newblock On the error exponent and capacity games of private watermarking systems.
\newblock {\em IEEE Transactions on Information Theory}, 49(3):537--562, 2003.

\bibitem{somekh2004capacity}
Anelia Somekh-Baruch and Neri Merhav.
\newblock On the capacity game of public watermarking systems.
\newblock {\em IEEE Transactions on Information Theory}, 50(3):511--524, 2004.

\bibitem{song2016likelihood}
Eva~C Song, Paul Cuff, and H~Vincent Poor.
\newblock The likelihood encoder for lossy compression.
\newblock {\em IEEE Trans. Inf. Theory}, 62(4):1836--1849, 2016.

\bibitem{steinberg2006reversible}
Yossef Steinberg.
\newblock Reversible information embedding with compressed host at the decoder.
\newblock In {\em 2006 IEEE International Symposium on Information Theory}, pages 188--191. IEEE, 2006.

\bibitem{steinberg2008coding}
Yossef Steinberg.
\newblock Coding for channels with rate-limited side information at the decoder, with applications.
\newblock {\em IEEE transactions on information theory}, 54(9):4283--4295, 2008.

\bibitem{sumszyk2009information}
Orna Sumszyk and Yossef Steinberg.
\newblock Information embedding with reversible stegotext.
\newblock In {\em 2009 IEEE International Symposium on Information Theory}, pages 2728--2732. IEEE, 2009.

\bibitem{suresh2017distributed}
Ananda~Theertha Suresh, X~Yu Felix, Sanjiv Kumar, and H~Brendan McMahan.
\newblock Distributed mean estimation with limited communication.
\newblock In {\em International conference on machine learning}, pages 3329--3337. PMLR, 2017.

\bibitem{sutivong2005channel}
Arak Sutivong, Mung Chiang, Thomas~M Cover, and Young-Han Kim.
\newblock Channel capacity and state estimation for state-dependent gaussian channels.
\newblock {\em IEEE Transactions on Information Theory}, 51(4):1486--1495, 2005.

\bibitem{tan2013dispersions}
Vincent~YF Tan and Oliver Kosut.
\newblock On the dispersions of three network information theory problems.
\newblock {\em IEEE Transactions on Information Theory}, 60(2):881--903, 2013.

\bibitem{theis2022algorithms}
Lucas Theis and Noureldin~Y Ahmed.
\newblock Algorithms for the communication of samples.
\newblock In {\em International Conference on Machine Learning}, pages 21308--21328. PMLR, 2022.

\bibitem{theis2022lossy}
Lucas Theis, Tim Salimans, Matthew~D Hoffman, and Fabian Mentzer.
\newblock Lossy compression with {G}aussian diffusion.
\newblock {\em arXiv preprint arXiv:2206.08889}, 2022.

\bibitem{triastcyn2021dp}
Aleksei Triastcyn, Matthias Reisser, and Christos Louizos.
\newblock {DP}-{REC}: Private \& communication-efficient federated learning.
\newblock {\em arXiv preprint arXiv:2111.05454}, 2021.

\bibitem{tung1978multiterminal}
S.-Y. Tung.
\newblock {\em Multiterminal source coding}.
\newblock PhD thesis, School of Electrical Engineering, Cornell University, Ithaca, NY, 1978.

\bibitem{unsal2023information}
Ay{\c{s}}e {\"U}nsal and Melek {\"O}nen.
\newblock Information-theoretic approaches to differential privacy.
\newblock {\em ACM Computing Surveys}, 56(3):1--18, 2023.

\bibitem{van1971three}
Edward~C Van Der~Meulen.
\newblock Three-terminal communication channels.
\newblock {\em Advances in applied Probability}, 3(1):120--154, 1971.

\bibitem{verdu2012non}
Sergio Verd{\'u}.
\newblock Non-asymptotic achievability bounds in multiuser information theory.
\newblock In {\em 2012 50th Annual Allerton Conference on Communication, Control, and Computing (Allerton)}, pages 1--8. IEEE, 2012.

\bibitem{verdu1994general}
Sergio Verd{\'{u}} and Te~Sun Han.
\newblock A general formula for channel capacity.
\newblock {\em IEEE Trans. Inf. Theory}, 40(4):1147--1157, 1994.

\bibitem{wang2011dispersion}
D.~{Wang}, A.~{Ingber}, and Y.~{Kochman}.
\newblock The dispersion of joint source-channel coding.
\newblock In {\em 2011 49th Annual Allerton Conference on Communication, Control, and Computing (Allerton)}, pages 180--187, Sep. 2011.

\bibitem{wang2008perfectly}
Ying Wang and Pierre Moulin.
\newblock Perfectly secure steganography: Capacity, error exponents, and code constructions.
\newblock {\em IEEE Transactions on Information Theory}, 54(6):2706--2722, 2008.

\bibitem{warner1965randomized}
Stanley~L Warner.
\newblock Randomized response: A survey technique for eliminating evasive answer bias.
\newblock {\em Journal of the American Statistical Association}, 60(309):63--69, 1965.

\bibitem{watanabe2015nonasymp}
S.~Watanabe, S.~Kuzuoka, and V.~Y.~F. Tan.
\newblock Nonasymptotic and second-order achievability bounds for coding with side-information.
\newblock {\em IEEE Trans. Inf. Theory}, 61(4):1574--1605, April 2015.

\bibitem{watanabe2015nonasymptotic}
Shun Watanabe, Shigeaki Kuzuoka, and Vincent~YF Tan.
\newblock Nonasymptotic and second-order achievability bounds for coding with side-information.
\newblock {\em IEEE Transactions on Information Theory}, 61(4):1574--1605, 2015.

\bibitem{wolfowitz1980simultaneous}
J~Wolfowitz.
\newblock {\em Simultaneous Channels}.
\newblock New York: Springer-Verlag, 1980.

\bibitem{wyner1976rate}
Aaron Wyner and Jacob Ziv.
\newblock The rate-distortion function for source coding with side information at the decoder.
\newblock {\em IEEE Transactions on information Theory}, 22(1):1--10, 1976.

\bibitem{wyner1975wire}
Aaron~D Wyner.
\newblock The wire-tap channel.
\newblock {\em Bell system technical journal}, 54(8):1355--1387, 1975.

\bibitem{wyner1978rate}
Aaron~D Wyner.
\newblock The rate-distortion function for source coding with side information at the decoder-ii. general sources.
\newblock {\em Information and control}, 38(1):60--80, 1978.

\bibitem{xu2023information}
Yinfei Xu, Jian Lu, Xuan Guang, and Wei Xu.
\newblock Information embedding with stegotext reconstruction.
\newblock {\em IEEE Transactions on Information Forensics and Security}, 19:1415--1428, 2023.

\bibitem{yamamoto1982wyner}
Hirosuke Yamamoto.
\newblock Wyner-ziv theory for a general function of the correlated sources (corresp.).
\newblock {\em IEEE Transactions on Information Theory}, 28(5):803--807, 1982.

\bibitem{yan2023layered}
Guangfeng Yan, Tan Li, Tian Lan, Kui Wu, and Linqi Song.
\newblock Layered randomized quantization for communication-efficient and privacy-preserving distributed learning.
\newblock {\em arXiv preprint arXiv:2312.07060}, 2023.

\bibitem{yang2019wiretap}
Wei Yang, Rafael~F Schaefer, and H~Vincent Poor.
\newblock Wiretap channels: Nonasymptotic fundamental limits.
\newblock {\em IEEE Transactions on Information Theory}, 65(7):4069--4093, 2019.

\bibitem{yassaee2015one}
Mohammad~Hossein Yassaee.
\newblock One-shot achievability via fidelity.
\newblock In {\em Proc. IEEE Int. Symp. Inf. Theory}, pages 301--305. IEEE, 2015.

\bibitem{yassaee2013non}
Mohammad~Hossein Yassaee, Mohammad~Reza Aref, and Amin Gohari.
\newblock Non-asymptotic output statistics of random binning and its applications.
\newblock In {\em 2013 IEEE International Symposium on Information Theory}, pages 1849--1853. IEEE, 2013.

\bibitem{yassaee2013technique}
Mohammad~Hossein Yassaee, Mohammad~Reza Aref, and Amin Gohari.
\newblock A technique for deriving one-shot achievability results in network information theory.
\newblock In {\em 2013 IEEE International Symposium on Information Theory}, pages 1287--1291. IEEE, 2013.

\bibitem{zamir2002universal}
Ram Zamir and Meir Feder.
\newblock On universal quantization by randomized uniform/lattice quantizers.
\newblock {\em IEEE Transactions on Information Theory}, 38(2):428--436, 2002.

\bibitem{ziv1985universal}
Jacob Ziv.
\newblock On universal quantization.
\newblock {\em IEEE Transactions on Information Theory}, 31(3):344--347, 1985.

\end{thebibliography}

\end{document}